\documentclass{ttumsthesis}

\usepackage[utf8]{inputenc}
\usepackage{amsfonts}
\usepackage{mathtools}
\usepackage{braket}
\usepackage{soul}
\usepackage{xcolor}
\usepackage{dsfont}

\usepackage[colorlinks=true,linkcolor=black,citecolor=black,urlcolor=black]{hyperref}

\usepackage[square,numbers,sort]{natbib}
\usepackage{cleveref}
\newcommand{\Var}{\operatorname{Var}}

\newcommand{\Tr}{\operatorname{Tr}}
\newcommand{\Ad}{\operatorname{Ad}}
\newcommand{\ad}{\operatorname{ad}}
\newcommand{\End}{\operatorname{End}}
\newcommand{\op}{\operatorname{op}}
\newcommand{\Hom}{\operatorname{Hom}}

\newcommand{\Ent}{\operatorname{Ent}}

\crefname{appendix}{appendix}{appendices}
\Crefname{appendix}{Appendix}{Appendices}

\crefname{remark}{remark}{remarks}
\Crefname{remark}{Remark}{Remarks}

\newtheorem{theorem}{Theorem}[section]
\newtheorem{lemma}[theorem]{Lemma}
\newtheorem{proposition}[theorem]{Proposition}
\newtheorem{corollary}[theorem]{Corollary}

\theoremstyle{definition}
\newtheorem{definition}[theorem]{Definition}
\newtheorem{example}[theorem]{Example}

\theoremstyle{remark}
\newtheorem{remark}[theorem]{Remark}

\usepackage{enumitem}
\setlist[enumerate,1]{label=(\arabic*)}

\thesistitle{Krylov-Lie Algebras for Variational Quantum Algorithms: Geometric, Depth-Aware Insights into Expressivity and Trainability}
\thesisauthor{Anžej Margeta-Cacace}
\thesisdepartment{Mathematics and Statistics}
\thesismonth{August}
\thesisyear{2026}
\thesischair{Dmitri Pavlov, Ph.D.}
\thesiscommitteeone{Lucas T. Braydwood, Ph.D.}
\thesiscommitteetwo{Travis Thompson, Ph.D.}

\hypersetup{pdftitle={Krylov-Lie Algebras for Variational Quantum Algorithms: Geometric, Depth-Aware Insights into Expressivity and Trainability}, pdfauthor={Anžej Margeta-Cacace}, pdflang={en-US}}

\begin{document}

\maketitle
\makecopyright

\begin{frontmatter}

\begin{acknowledgments}
This thesis would have been entirely impossible without the assistance of so many in my life, and so it would be incomplete without appropriate acknowledgments. 

Throughout all the trials and tribulations, Mom and Dad, you have been there for me, and when I needed it most, you came to help. Thank you.

I would particularly like to thank my advisor, Dmitri Pavlov, for his assistance in the process of writing this thesis and for his especially shrewd insights into how I should frame my work. I do not believe I would have been even as close to being as satisfied with the outcome of this text without his feedback. Dmitri, you have been a great advisor, and thank you for trusting in my ability to do this. 

To all those who never doubted me for a moment---Cole, Enzel, Galilea, Evelyn, Amaris---your presence and friendship are invaluable to me, and I hope I get to keep you all around as long as I can. Life would not mean as much without you.

James, talking math with you is a joy. I do not know anyone else with whom I can speak about research so fluently. I long for the next opportunity I have to hear what you are doing, to tell you what has piqued my interest. You are a wonderful friend, and a great colleague.

I also must give special thanks to my committee members, Lucas T. Braydwood and Travis Thompson. Besides being wonderful, incredibly supportive people, they very much saved the defense of this thesis from the horrors of bureaucracy. I cannot thank either of you enough for everything you have done for me over the past several years.

Finally, I would like to thank Michael Ragone and Marco Cerezo for their enlightening correspondence on the dynamical Lie algebraic theory of VQA loss landscapes.
\end{acknowledgments}

\tableofcontents

\begin{abstract}

Variational quantum algorithms (VQAs) are a promising method for potentially realizing quantum advantage on the noisy intermediate-scale quantum hardware of today by leveraging methodological wisdom from classical machine learning and optimization. However, their utility is impeded by features such as barren plateaus and other nontrivial pathologies in their loss landscapes. Much work has therefore focused on improving our understanding of VQA landscapes, leading to the development of dynamical Lie algebraic theories of barren plateaus and Jordan-algebraic Wishart systems. Though valuable and foundational, these models typically employ the assumptions of $\epsilon$-approximate $t$-designs and Haar-random circuits, thereby neglecting the depth and geometric content of realistic variational ans\"atze. This makes such approaches ill-suited to the shallow-depth regime, where VQAs are known to be poor approximators of $2$-designs and where trainability is often most practically relevant.

To address these limitations of existing VQA landscape theory, this thesis introduces a novel framework centered around \emph{Krylov algebras}: algebraic structures induced by the Krylov span of a finite generator set acting on one or more seed vectors. We show that the reachable manifolds of VQAs can be approximated by Krylov-Lie algebras and groups, obtaining factorial seed-based error bounds under commutator matching for Hermitian involutions and for prefix constructions. Furthermore, we demonstrate that these structures induce the canonical invariant measures necessary for computing expectation values and variances of loss functions under general sampling measures. In particular, we derive weighted non-Haar variance formulas that recover the usual Lie-algebraic Haar formulas as a special case while isolating all non-Haar effects into explicit correction terms.

We also refine the common heuristic that sufficiently deep depth-$L$ circuit ensembles must converge to Haar by showing that this conclusion does not hold in full generality without additional hypotheses. In the process, we identify concrete obstructions to naive Haar convergence; nevertheless, we find that this convergence may be recovered under natural necessary and sufficient ergodic conditions. Finally, our variance formulas reveal that non-Haar contributions can furnish a possible mechanism for mitigating barren plateaus by reweighting the visible sectors of the landscape, suggesting that VQAs may be more trainable than recent barren plateau literature has posited.
\end{abstract}


\end{frontmatter}

\mainmatter

\chapter{Introduction}\label{chap:intro}

\section{Variational Quantum Algorithms and Barren Plateaus}

Variational quantum algorithms (VQAs), such as the quantum approximate optimization algorithm (QAOA), have emerged as promising approaches to realizing quantum computational advantage \cite{cerezo_variational_2021}.
They avoid implementing a long,
fully error-corrected quantum computation, instead working with a parameterized quantum circuit,
evaluating a loss function on the output, and updating the parameters through a classical
optimization loop. In this way, VQAs combine the expressive power of quantum dynamics
with the flexibility of modern optimization, and they have therefore become central to
proposals for near-term quantum advantage in simulation, optimization, and machine
learning \cite{larocca_barren_2025}.

Simultaneously, the practical promise of VQAs has been shadowed by a substantial and
rapidly growing theory of trainability obstructions. The most famous of these are \emph{barren plateaus}: regimes in which gradients or cost function variances (these are equivalent conditions; \cite{larocca_barren_2025, ragone_lie_2024}) become exponentially
small in the system size, rendering training prohibitively expensive since such training landscapes necessitate an exponential number of measurements to reliably determine a parameter update direction \cite{larocca_barren_2025}. Over the past several years, the study of barren plateaus has developed into a large and sophisticated literature, with important contributions from concentration of measure, unitary design theory, dynamical Lie algebras, and more recently Jordan-algebraic approaches to loss-landscape statistics \cite{anschuetz_2025, cerezo_variational_2021,ragone_lie_2024, larocca_barren_2025}. This body of work has been valuable, and many of its theorems are mathematically
elegant. Nevertheless, a central contention of this thesis is that much of the existing theory has been developed around asymptotic randomization principles that are poorly matched to
the finite-depth, highly structured variational circuits that are most relevant in practice.

A variational quantum algorithm is specified by a parameterized family of unitaries
\[
U(\boldsymbol{\theta})=e^{-i\theta_p H_{i_p}}\cdots e^{-i\theta_1 H_{i_1}},
\]
together with a set of Hermitian circuit generators $\mathcal{S}=\{H_1, \dots, H_n\}$ (the indices in the VQA circuit can refer to any of the generators in $\mathcal{S}$), an initial state $\ket{\psi}$, an observable $O$, and a classical rule (e.g. gradient descent) for updating the circuit parameters \( \boldsymbol \theta\), which is a real vector of length $p$. The number of terms in the product of parameterized unitaries is called the \emph{depth} or \emph{number of layers}, and typically, the function to be optimized is the loss (or cost) function of the form \cite{ragone_lie_2024}:
\begin{equation} \label{eq:loss-fn}
    \ell_{\boldsymbol{\theta}}(\rho,O)=\Tr\bigl[U(\boldsymbol{\theta})\rho U(\boldsymbol{\theta}^\dagger O \bigr].
\end{equation}
The set of all unitaries that can be obtained from the parameterized circuit via some choice of parameters, or the \emph{reachable set} of the ansatz, is hence a highly structured
finite-dimensional manifold determined by the chosen generators, parameterization, depth, and architecture. Any convincing theory of VQA trainability should, therefore, be sensitive
to that geometry.

In hopes of making analysis tractable, much of the modern barren plateau literature instead proceeds by replacing this structured, finite-depth object with a more symmetric random model \cite{anschuetz_2025, ragone_lie_2024, fontana_adjoint_2024}. In many cases, one studies
ensembles that are Haar-random, or close in a suitable sense to Haar-random, on some
large compact Lie group associated with the ansatz. In other cases, one assumes or proves that the circuit behaves as an \(\varepsilon\)-approximate \(t\)-design, which can be thought of as ensembles whose integrals of $t$-degree polynomials of unitaries are within a tolerance $\varepsilon$ of the corresponding integrals over the Haar measure, and then imports the
corresponding Haar moment identities \cite{ragone_lie_2024, anschuetz_2025}. These ideas have been extremely fruitful, especially for deriving clean formulas for expectation values, variances, and concentration phenomena. And yet, they also build a substantial assumption gap between
our mathematical models and the variational circuits we wish to analyze.

\section{Limitations of Existing Theory}

\subsection{The Assumption Gap: Haar Randomness and \(t\)-Designs}

The first major limitation of existing VQA landscape theory is its dependence on
Haar-randomness and unitary design heuristics. The Haar measure is the natural invariant measure on a compact group, and from a purely representation-theoretic perspective, it is
an extremely convenient object. Likewise, \(t\)-designs provide a powerful finite-moment surrogate for Haar randomness. If a circuit ensemble forms an exact or sufficiently constrained approximate \(t\)-design, then one can often compute or bound the low-order moments of the
loss function by appealing to universal formulas that depend only on the surrounding group
representation.

The problem is that realistic variational ans\"atze are generally neither Haar-random nor close to Haar-random in the regimes that matter most for training \cite{anschuetz_2025, Holmes2022, HafterkampRandomCircuits, Park2024hamiltonian, Zhang2024absence}. In the shallow- to intermediate-depth setting, where one typically hopes to avoid concentration pathologies and
retain useful optimization signal, the circuit manifold explored by the ansatz, coupled with its assigned parameter distribution, is usually far too small and too structured to resemble Haar on any ambient group \cite{Zhang2024absence, Park2024hamiltonian, Holmes2022, HafterkampRandomCircuits}. Ansatz properties such as depth, locality, and generator choice all impose strong geometric restrictions on
the accessible region of state space \cite{Zhang2024absence, Park2024hamiltonian, Holmes2022, larocca_barren_2025}. When these features are suppressed in favor of a Haar surrogate, one risks proving theorems about an asymptotic random model rather than about the actual variational family.

This mismatch is especially severe in Lie-algebraic approaches based on the full dynamical Lie algebra. The dynamical Lie algebra captures the infinitesimal closure of the generators,
and it is an important asymptotic object as it can capture the full expressivity of the circuit \cite{ragone_lie_2024}. However, the Lie group it integrates to may
be exponentially larger than the manifold actually explored by a polynomial-depth ansatz, and consequently, formulas derived from the Haar measure on that full group can dramatically overstate the amount of effective mixing present at realistic depths \cite{HafterkampRandomCircuits, Zhang2024absence, Park2024hamiltonian}. This implies that the Lie algebraic viewpoint is incomplete unless it is coupled to a depth-aware description of the part of the Lie-theoretic structure actually visible to
the circuit.

\subsection{The Insufficiency of \(\varepsilon\)-Approximate \(t\)-Designs}

Even when approximate design arguments are available, they are often too coarse to resolve the questions one most wants to ask about trainability. An \(\varepsilon\)-approximate
\(t\)-design controls moments only up to order \(t\), and it does so relative to a chosen ambient Haar model. This can be enough to transfer certain concentration estimates or moment formulas, but it does not identify which representation-theoretic
sectors of the original sampling measure are still visible at finite depth, which sectors are being suppressed, or how the nonuniform sampling distribution induced by the ansatz differs from Haar on the region it actually explores \cite{ragone_lie_2024}. Additionally, examining Equation 21 from Ragone et al., one quickly sees that in cases where the Schatten $1$-norm of our observable is an increasing exponential in the number of qubits, something that is not uncommon in realistic ans\"atze \cite{EHMQAOAPaper}, then at shallow depths these bounds will be vacuous.

For these reasons, constructs based around approximate designs are often better viewed as asymptotic diagnostics than as faithful finite-depth models \cite{anschuetz_2025}. Moreover, once one leaves the exact-Haar setting, the difference between the true ansatz measure and Haar should not be viewed solely as an error term to be discarded. It may itself contain the most important trainability information because the sampling distortion from Haar can amplify or suppress precisely those sectors of the landscape that govern optimization. Any realistic finite-depth theory must therefore replace the design heuristic with an object that approximates the reachable manifold itself rather than its asymptotic random limit.

\subsection{Problematic Empirics: Cragged Terrains Versus Barren Plateaus}

The most serious challenge to the universality of current barren-plateau theory is empirical.
Across a wide range of practically relevant shallow ans\"atze, one does not consistently see
the kind of exponentially vanishing optimization signal predicted by Haar-based arguments \cite{EHMQAOAPaper}.
Instead, one often encounters highly nontrivial landscapes whose variances remain non-concentrated,
and usually even grow polynomially with system size. These are not the nearly featureless surfaces suggested by a naive extrapolation of deep random-circuit theory. They are rough,
structured, and often anisotropic landscapes, which in ongoing work are referred to as cragged terrains \cite{EHMQAOAPaper}.

This empirical picture indicates
that current theory has often been overextended beyond the regime in which it is mathematically justified. If the practically relevant shallow-depth regime overwhelmingly exhibits structured,
non-Haar behavior, then a theory built around universal convergence to Haar cannot be taken as a default explanation of trainability. At best, it describes a limiting regime that may or may not be reached at useful depths. At worst, it obscures the very mechanisms that make
real variational circuits trainable.

This concern becomes more percipient after accounting for the results developed later in this thesis. Specifically, the correction of the common heuristic that increasing the circuit depth should generically force convergence to Haar, and therefore to a tighter approximate \(t\)-design, is not valid in full generality. The issue is structural: one can have visible sectors that fail to mix at all. In the framework of this thesis, those sectors appear explicitly as nontrivial Fourier components of the induced sampling law. In the cycle-graph QAOA, this leads to a concrete obstruction showing that a key contraction lemma used in the literature fails without additional hypotheses \cite{ragone_lie_2024}. Thus, although we do find conditions that do allow us to conclude Haar convergence, the prevailing narrative that depth
alone drives generic variational ensembles toward Haar must be treated with considerable care.

Because of these impediments to the applicability of current theory, the community needs a finite-depth theory that keeps the Lie-theoretic and representation-theoretic insights of past literature while abandoning the assumption that realistic variational circuits are already well approximated by asymptotic randomness. The central idea of this work is
that such a theory should begin from the geometry of the reachable manifold itself.

\section{Thesis Contributions and Approach}

This thesis proposes Krylov--Lie structures as a depth-aware replacement for purely Haar-centric models of variational quantum dynamics. The guiding idea is to combine the graded, seed-dependent flexibility of Krylov constructions with the noncommutative closure properties of Lie theory. Starting from a finite generator set and one or more seed vectors, we build a finite-dimensional Krylov--Lie subspace by repeated commutator action up to a chosen depth, then compress the original generators to that subspace and take the Lie algebra they generate. The resulting Krylov--Lie algebra and its represented compact Lie group furnish a finite-depth Lie-theoretic surrogate for the reachable manifold of the original ansatz. 

The first advantage of this construction is geometric. Instead of comparing a shallow circuit to the full dynamical Lie group it would generate only asymptotically \cite{ragone_lie_2024}, one approximates the manifold actually
explored at the chosen depth by a compact Lie group of controlled dimension. This makes it possible to ask whether the ansatz admits a dimension-matched or dimension-reduced Lie-theoretic model, whether the induced comparison map has full rank almost everywhere, and whether the pushforward of the parameter law becomes absolutely continuous with respect to Haar measure on the approximating group. When this happens, one obtains a canonical Haar background together with a Radon--Nikodym density recording the sampling distortion of the original circuit relative to the approximating geometry. 

The second advantage is structural. Krylov--Lie subspaces and algebras vary with the seed, and this flexibility creates a rich dimensional theory. One can ask which dimensions are attainable, when maximal dimension is generic, and under what conditions dimension matching with the reachable manifold can be achieved. A substantial portion of this thesis is devoted to answering these questions using determinantal stratifications, algebraic genericity arguments, and universal dimension bounds. These results show that the seed dependence of Krylov-Lie structures is a boon of the formalism: it supplies the tunable finite-dimensional geometry needed for realistic manifold approximation.

The third advantage is quantitative. Beyond constructing Lie-theoretic surrogates, the thesis proves a finite-depth approximation theory for them. After establishing a compact, full-rank approximation theorem for reachable manifolds on regular loci, we show that, for Hermitian involution generators, seed-level matching of every commutator word through grade $k$ forces ordinary word preservation through length $\lceil k/2\rceil$. This gives a factorial seed-based error bound on compact regions. For prefix Krylov--Lie algebras, ordinary word preservation gives the stronger bound of order \(B_C^{k+1}/(k+1)!\), without the involution hypothesis (\Cref{thm:seed-KLG-factorial,thm:prefix-KLG-factorial}). When both circuit spectra remain uniformly separated from a common logarithm branch cut, these word-preservation estimates also give geometric approximation of the principal logarithms acting on the seed (\Cref{thm:seed-KLG-logarithm}).

The fourth advantage is measure-theoretic. Once a compact represented Krylov--Lie group is obtained, it carries a canonical normalized Haar measure, and any sampling law induced by the original ansatz that is absolutely continuous with respect to that Haar measure may be expressed through a Radon--Nikodym density relative to it. This makes it possible to derive exact weighted expectation and variance formulas for loss functions under general non-Haar sampling. In the Haar case, these formulas reduce to the familiar Lie-algebraic expressions. Away from Haar, all deviation is isolated into explicit correction operators governed by the visible sectors of the density. In this sense, the framework identifies exactly where non-Haar behavior lives, rather than treating it as an undifferentiated error term.

A further contribution is that the quantitative approximation theory is strong enough to transport observable information from the Krylov--Lie proxy back to the original circuit. For seed-dependent observables, the thesis proves that expectation values, variances, and concentration inequalities for the original ansatz are controlled by the corresponding quantities for the Krylov--Lie surrogate up to the same factorial approximation error under the respective generator and matching hypotheses (\Cref{prop:KLG-concentration-transport}). This yields a clean bridge between geometric concentration properties on Krylov--Lie groups and  those of finite-depth variational quantum circuits. 

Finally, this thesis uses this framework to revisit the interpretation of barren plateaus at finite depth. The resulting theory clarifies how and when Haar-based variance formulas should still be understood as asymptotic limits, but it also shows that non-Haar effects are not peripheral. They may dominate the shallow-depth regime and can in principle mitigate variance concentration by reweighting the visible sectors of the landscape. Relatedly, this thesis corrects the common heuristic that increasing depth alone should generically force convergence to Haar: this is not true in full generality, because one may have sectors visible to our order of moment that fail to mix. We identify concrete obstructions to naive Haar convergence and then recover convergence under natural necessary and sufficient observable ergodic conditions, formulated through a Kawada--It\^o theorem for convergence of observables on unitary $t$-designs. 

The organization of the thesis is as follows. \Cref{chap:background} recalls the minimal background on
dynamical Lie algebras and existing VQA landscape theory needed for this work. \Cref{chap:krylov-lie-algebras}
introduces Krylov algebras and Krylov-Lie algebras, establishes their basic structural
properties, and develops the seed-dependent dimension theory. \Cref{chap:krylov-lie-approximations} proves the main approximation results, including the Krylov--Lie approximation theorem, the factorial seed-based error bounds, approximation of circuit logarithms, and concentration transport for Krylov--Lie proxies. \Cref{chap:applications} applies this machinery to VQA landscape theory by deriving weighted moment and variance formulas, comparing them with existing Haar-based results, and analyzing the failure of naive Haar-convergence heuristics. \Cref{subsec:observable-kawada-ito} finds necessary and sufficient conditions for convergence to the Haar measure of the dynamical Lie group of the $L$-fold convolution of a single-layer measure (in the sense of convergence in all orders of a unitary $t$-design over the dynamical Lie group). \Cref{chap:discussion} concludes
with limitations of the present framework and several directions for future work, including Jordan-Lie extensions, more detailed theory for the dimension of Krylov-Lie algebras, and the role of nonuniform sampling in
trainability.

\chapter{Background}\label{chap:background}

This chapter introduces the algebraic frameworks that will be used as points of reference in our analysis of variational quantum algorithms throughout the thesis. Our focus is on how the operator family generated by a parametrized circuit encodes structural information about the reachable states, the effective observables, and the resulting optimization landscape. We review the aspects of this framework that are most relevant to expressivity, trainability, and barren plateau theory, with an emphasis on the quantities and constructions that will be used later in this thesis; we do not aim for this to be an entirely comprehensive introduction to barren plateau theory and empirics. For such an exposition, see Larocca et al. \cite{larocca_barren_2025}.

\singlespacing
\section{Generators, Reachable Manifolds, and the Dynamical Lie Algebra}

\onehalfspacing
A variational quantum algorithm is specified by a parameterized circuit
\[
U(\boldsymbol{\theta})=e^{-i\theta_p H_{i_p}}\cdots e^{-i\theta_1 H_{i_1}},
\]
with parameter counter $p$ and indices $i_\ell$ corresponding to Hermitian generators $H_1,\dots,H_n$ acting on a finite-dimensional Hilbert space $\mathcal H \cong (\mathbb C^2)^{\otimes N}$ (where $N$ is the number of qubits), together with a density matrix $\rho$ (correlating with an input state) and a Hermitian observable $O$. As mentioned in \Cref{eq:loss-fn}, the loss function is
\[
\ell_{\boldsymbol{\theta}}(\rho,O):=\Tr\bigl[U(\boldsymbol{\theta})\rho U(\boldsymbol{\theta})^\dagger O\bigr].
\]
For example, we have the Quantum Approximate Optimization Algorithm (QAOA), fittingly introduced for approximate combinatorial optimization \cite{brady_iterative_2023, EHMQAOAPaper, allcock2026dynamicalliealgebrasquantum}. Given a cost Hamiltonian \(H_C\) encoding the objective function and a mixer Hamiltonian \(H_M\) driving noncommutativity between layers, the \(L\)-layer QAOA ansatz is
\[
|\boldsymbol{\gamma},\boldsymbol{\beta}\rangle
=
e^{-i\beta_L H_M}e^{-i\gamma_L H_C}\cdots
e^{-i\beta_1 H_M}e^{-i\gamma_1 H_C}|\psi_0\rangle,
\]
where \( |\psi_0\rangle \) is typically an easily prepared reference state. Note that while one can in some sense consider QAOA to be given by a depth-$2L$ quantum circuit, it is most common to group each mixer Hamiltonian gate with a cost Hamiltonian gate and consider those groupings as one layer each. The parameters \( (\boldsymbol{\gamma},\boldsymbol{\beta}) \) are chosen to optimize the expectation value of \(H_C\) so that measurement of the final state yields high-quality approximate solutions. Because of its simple alternating structure and direct connection between generators and circuit geometry, QAOA serves as a useful model for studying expressivity and trainability in variational quantum algorithms \cite{brady_iterative_2023, EHMQAOAPaper, allcock2026dynamicalliealgebrasquantum}.

In VQA landscape theory, borrowing from quantum control \cite{OMeara2014LieTheoretic}, one is interested in the geometry of the reachable family (or reachable manifold)
\begin{equation}\label{eq:reachable-manifold}
M:=\{U(\boldsymbol{\theta}) \, \vert \,\boldsymbol{\theta}\in\boldsymbol{\Theta}\},
\end{equation}
where $\boldsymbol{\Theta}$ is the set of all possible $\boldsymbol{\theta}$ values with periodic structure quotiented out. $M$ is also taken to have a probability distribution $\nu$ over its parameters. The goal of studying this reachable manifold is that it may permit us to classify barren plateaus. We now give a precise definition of this landscape phenomenon.

\begin{definition}[Barren plateau]\label{def:barren-plateau}
A \emph{barren plateau} occurs when there exists a $b > 1$ such that
\begin{equation}\label{eq:barren-plateau}
\Var_{\nu}(\ell_\theta(\rho, O)) \in \mathcal O (1/b^N),
\end{equation}
where $\Var_{\nu}$ indicates a second-moment integral over the appropriate set of parameters and their distribution $\nu$ \cite{ragone_lie_2024, larocca_barren_2025, fontana_adjoint_2024}. 
\end{definition}
In much of the barren plateau literature, the finite-depth reachable set $M$ is studied through a larger Lie theoretic object attached only to the circuit generators.

\begin{definition}[Dynamical Lie algebra]\label{def:DLA-and-DLG}
Let $\mathcal S=\{H_1,\dots,H_n\}$ be the Hermitian generators of the variational circuit.
The \emph{dynamical Lie algebra} (DLA) of the ansatz is the real Lie algebra
\[
\mathfrak g:=\langle i\mathcal S\rangle_{\mathrm{Lie}}
\subset \mathfrak u(2^N),
\]
namely the smallest real linear subspace of $\mathfrak u(2^N)$ containing
$\{iH_1,\dots,iH_n\}$ and closed under the commutator bracket, where $N$ is the number of qubits \cite{ragone_lie_2024, fontana_adjoint_2024}. The connected (and often compact) Lie group integrating the DLA for a given representation is termed the \emph{dynamical Lie group}, or DLG.
\end{definition}

The importance of the DLA is that it captures the asymptotic controllability or expressivity
of the circuit: every circuit unitary $U(\theta)$ lies in the connected Lie group generated by
$\mathfrak g$ and its representation, and sufficiently rich compositions of the circuit generators can explore that group. Since $\mathfrak g$ is a Lie subalgebra of the compact Lie algebra $\mathfrak u(2^N)$,
it is reductive \cite{ragone_lie_2024}, and hence admits a decomposition
\[
\mathfrak g= \mathfrak z \oplus \mathfrak g_1\oplus\cdots\oplus \mathfrak g_{k-1},
\]
where $\mathfrak g_1,\dots,\mathfrak g_{k-1}$ are simple ideals and $\mathfrak z$ is abelian;
equivalently, $\mathfrak z$ is the center and
$[\mathfrak g,\mathfrak g]=\mathfrak g_1\oplus\cdots\oplus \mathfrak g_{k-1}$ is semisimple.

If $G$ is any connected compact Lie-group realization integrating $\mathfrak g$, then the
adjoint representation preserves this decomposition \cite[§3.4 and §6.4]{kirillov_lie_groups}, and one may project operators in the set of all endomorphisms on $\mathcal H$, namely
$\End(\mathcal H)$, onto the corresponding ideals. This is the point of entry for the Lie
algebraic variance formulas of Ragone et al.: the loss landscape is decomposed into
independent algebraic sectors, and its Haar moments are computed sector-by-sector with Schur's lemma and the Weingarten calculus \cite{ragone_lie_2024}.

While the DLA is reductive, this alone does not imply that its represented
connected Lie group is closed. There is, however, a stronger statement
available for its adjoint action. Specifically, once we quotient by the
center, the remaining algebra is semisimple of compact type, and the
corresponding adjoint group is compact. It is useful to make this distinction explicit before
introducing Haar averages over the DLG.

\begin{proposition}[Compactness of the adjoint dynamical Lie group]
\label{prop:compact-adjoint-dlg}
Let $\mathfrak g\subset\mathfrak u(d)$ be the DLA generated by our circuit
Hamiltonians, and let $G$ be
its connected integrating Lie group. Say that $G$ is realized as a possibly
nonclosed subgroup of $U(d)$; equip it with its intrinsic Lie group topology.
Then
\[
\Ad(G)\cong G/Z(G)
\]
is a compact, connected, centerless Lie group with semisimple Lie algebra
\[
\ad(\mathfrak g)
\cong\mathfrak g/\mathfrak z(\mathfrak g)
\cong[\mathfrak g,\mathfrak g].
\]
Thus compactness of the adjoint group does not require compactness of $G$.
\end{proposition}

\begin{proof}
Since $\mathfrak g$ is a real Lie subalgebra of $\mathfrak u(d)$, the
positive-definite invariant form
\[
\beta(X,Y)=-\operatorname{Re}\Tr[XY]
\]
gives an orthogonal decomposition
\[
\mathfrak g=\mathfrak z(\mathfrak g)\oplus\mathfrak s,
\qquad \mathfrak s=[\mathfrak g,\mathfrak g],
\]
where $\mathfrak s$ is semisimple of compact type. In particular, its Killing
form is negative definite. Let $S$ be the simply-connected Lie group with Lie
algebra $\mathfrak s$. Then $S$ is compact, and the simply-connected group
integrating $\mathfrak g$ is
\[
\widetilde G\cong S\times\mathbb R^{\dim\mathfrak z(\mathfrak g)}
\]
\cite[Chs.~3 and 6]{kirillov_lie_groups}.
The adjoint action of the second factor is trivial. Since the covering map
$\widetilde G\to G$ is surjective, the adjoint images of $\widetilde G$ and
$G$ agree, and hence
\[
\Ad(G)\cong\Ad(S)\cong S/Z(S).
\]
This is a compact, connected group with Lie algebra $\mathfrak s$.
To show that it is centerless, let $q:S\to S/Z(S)$ be the quotient map,
and suppose that the coset $q(s)=sZ(S)$ is central in $S/Z(S)$.
Then, for every $x\in S$,
\[
q(sxs^{-1}x^{-1})
=q(s)q(x)q(s)^{-1}q(x)^{-1}
=e_{S/Z(S)}.
\]
Thus $sxs^{-1}x^{-1}\in\ker q=Z(S)$. For this fixed $s$, the map
\[
S\longrightarrow Z(S),\qquad x\longmapsto sxs^{-1}x^{-1}
\]
is continuous. Since $S$ is connected and $Z(S)$ is finite, this map
is constant. Its value at $x=e_S$ is $e_S$, so
$sxs^{-1}x^{-1}=e_S$ for every $x\in S$. Therefore $s\in Z(S)$,
and $q(s)$ is the identity coset. This proves that $S/Z(S)$ is
centerless. Finally, connectedness of $G$ gives $\ker\Ad=Z(G)$,
and consequently
\[
G/Z(G)\cong\Ad(G)\cong S/Z(S),
\]
which proves the claimed quotient identification
\cite[Ch.~3]{kirillov_lie_groups}.
\end{proof}

\begin{remark}
Here the center being removed is the full group center $Z(G)$, and the assertion concerns the adjoint action on $\mathfrak g$ itself. The conjugation action on all Hermitian matrices,
\[
\mathcal C_G(g):Y\longmapsto gYg^{\dagger},
\]
has kernel $G\cap\{e^{it}I:t\in\mathbb R\}$ in the represented setting.
This can be smaller than $Z(G)$, since a nonscalar central element need only
commute with the DLA. Consequently, the compactness assertion above does
not automatically apply to the full observable action.

This distinction is also relevant to the Haar averages used by
Ragone et al.~\cite{ragone_lie_2024}. Their formulation uses $G=e^{\mathfrak g}$
and normalized Haar measures without establishing that the represented
analytic group is closed. Reductivity does not justify this normalization.
The Haar measure exists on a locally compact group and is unique up to a positive
constant; its total mass is finite precisely when the group is compact
\cite[Ch.~2]{Folland1995Harmonic}. In particular, the compact adjoint group
in \Cref{prop:compact-adjoint-dlg} is not the group on which their averages
are stated.

There is nevertheless a way to supply a compact reference for their main
variance calculation. Under its assumption that $O\in i\mathfrak g$ or
$\rho\in i\mathfrak g$, the loss factors through the intrinsic adjoint
action. For instance, if $iO\in\mathfrak g$, then
\[
\ell_g(\rho,O)
=\Tr\!\left[\rho\bigl(-i\Ad_{g^{-1}}(iO)\bigr)\right],
\]
and the analogous expression uses $\Ad_g(i\rho)$ when $i\rho\in\mathfrak g$.
Thus one may formulate the required moment-matching assumption and Haar
averages on the compact group $\Ad(G)$, with the central contribution to
the loss remaining constant. This reformulates the reference group; it
does not supply a normalized Haar measure on a noncompact $G$. The stated
integrals over $G$, and Haar-probability statements concerning the full
represented unitary ensemble, still require a compact target or a
closedness hypothesis.
\end{remark}

One might worry that the DLG construction omits global topological information like discrete lattice structure on the circuit by using a compact, connected Lie group given by the DLA representation. Fortunately, for the Haar averages relevant to the moments of the cost function, this loss is inessential.

\begin{proposition}[Pushforward of the Haar measure]
\label{HaarPushforward}
Let \(G'\) and \(G\) be compact Hausdorff groups, and let $q:G'\to G$ be a continuous surjective homomorphism. Moreover, have \(\mu_{G'}\) and \(\mu_G\) denote the
normalized Haar probability measures on \(G'\) and \(G\), respectively. Then
\[
q_*\mu_{G'}=\mu_G.
\]
Equivalently, for every bounded Borel function \(f:G\to \mathbb C\),
\[
\int_G f(g)\,d\mu_G(g)
=
\int_{G'} f(q(g'))\,d\mu_{G'}(g').
\]
\end{proposition}

\begin{proof}
Define a Borel probability measure \(\nu\) on \(G\) by
\[
\nu(E):=\mu_{G'}\bigl(q^{-1}(E)\bigr)
\]
for every Borel set \(E\subset G\). Since \(q\) is surjective,
\[
\nu(G)=\mu_{G'}\bigl(q^{-1}(G)\bigr)=\mu_{G'}(G')=1.
\]

We claim that \(\nu\) is left-translation invariant. Let \(g\in G\), and choose
\(g_0'\in G'\) such that \(q(g_0')=g\), which is possible because \(q\) is surjective.
Then for every Borel set \(E\subset G\),
\[
q^{-1}(gE)
=
q^{-1}\bigl(q(g_0')E\bigr)
=
g_0'\,q^{-1}(E).
\]
Indeed, if \(x\in q^{-1}(gE)\), then \(q(x)=q(g_0')e\) for some \(e\in E\), so
\(q((g_0')^{-1}x)=e\in E\), hence \(x\in g_0'q^{-1}(E)\). The reverse inclusion is immediate.

Therefore, by left invariance of Haar measure on \(G'\),
\[
\nu(gE)
=
\mu_{G'}\bigl(q^{-1}(gE)\bigr)
=
\mu_{G'}\bigl(g'_0q^{-1}(E)\bigr)
=
\mu_{G'}\bigl(q^{-1}(E)\bigr)
=
\nu(E).
\]
Thus \(\nu\) is a left-invariant probability measure on the compact group \(G\). By uniqueness of the normalized Haar measure on compact groups, \(\nu=\mu_G\). Hence
\(q_*\mu_{G'}=\mu_G\).

The integral identity follows from the defining property of pushforward measure:
for every bounded Borel \(f:G\to\mathbb C\),
\[
\int_G f(g)\,d(q_*\mu_{G'})(g)
=
\int_{G'} f(q(g'))\,d\mu_{G'}(g').
\]
Since \(q_*\mu_{G'}=\mu_G\), this becomes
\[
\int_G f(g)\,d\mu_G(g)
=
\int_{G'} f(q(g'))\,d\mu_{G'}(g').
\]
\end{proof}

\begin{corollary}[Representation-invariant Haar averages of trace-class observables]
\label{TraceHaarTopologyInvariance}
Let \(\mathfrak g\) be a finite-dimensional real Lie algebra, let $\rho:\mathfrak g\to \mathfrak u(\mathcal H)$ be a finite-dimensional unitary Lie-algebra representation, and let $\Pi:\widetilde G\to U(\mathcal H)$ be the unique integrated representation on the simply connected Lie group
\(\widetilde G\) with Lie algebra \(\mathfrak g\). Next, have \(G_1=\widetilde G/\Gamma_1\) and \(G_2=\widetilde G/\Gamma_2\) be connected compact Lie-group realizations such that
\[
\Gamma_1\subset \Gamma_2\subset \ker \Pi.
\]
Further, set
\[
q:G_1\to G_2
\]
as the canonical surjective homomorphism between connected Lie groups with Lie algebra $\mathfrak g$ such that the discrete central kernel of one is a subset of the discrete central kernel of the other (see Kirillov \cite[Cor.~3.49]{kirillov_lie_groups}), and let
\[
\Pi_i:G_i\to U(\mathcal H)
\]
be the descended unitary representations, so that
\[
\Pi_2\circ q=\Pi_1.
\]

Then for every bounded Borel function
\[
\Phi:\mathbb C\to \mathbb C,
\]
every trace-class operator \(T\in \mathcal T_1(\mathcal H)\), and every bounded operator
\(A\in \mathcal B(\mathcal H)\), the functions
\[
F_i(g):=\Phi\bigl(\operatorname{Tr}\bigl[T\,\Pi_i(g)^\dagger A\,\Pi_i(g)\bigr]\bigr),
\]
where $g\in G_i$, satisfy
\[
F_1=F_2\circ q
\]
and we see that
\[
\int_{G_1} F_1(g)\,d\mu_{G_1}(g)
=
\int_{G_2} F_2(g)\,d\mu_{G_2}(g).
\]
\end{corollary}

\begin{proof}
Because \(\Pi_2\circ q=\Pi_1\), for every \(g\in G_1\) we have
\[
\operatorname{Tr}\bigl [T\,\Pi_1(g)^\dagger A\,\Pi_1(g) \bigr ]
=
\operatorname{Tr} \bigl [ T\,\Pi_2(q(g))^\dagger A\,\Pi_2(q(g)) \bigr ].
\]
Applying \(\Phi\) yields \(F_1=F_2\circ q\).

Since \(q:G_1\to G_2\) is a continuous surjective homomorphism of compact groups,
Theorem~\ref{HaarPushforward} gives \(q_*\mu_{G_1}=\mu_{G_2}\). Therefore,
\[
\int_{G_1} F_1(g)\,d\mu_{G_1}(g)
=
\int_{G_1} F_2(q(g))\,d\mu_{G_1}(g)
=
\int_{G_2} F_2(h)\,d\mu_{G_2}(h).
\]
\end{proof}
From the definition of trace-class observables, this corollary shows that, for our purposes, one may work with the unique closed, compact, connected matrix Lie group determined by the given Lie algebra representation without sacrificing the accuracy of the model. Furthermore, as an interesting product of this corollary, we have that $t$-design and barren plateau structure are invariants of the Lie algebra alone and not of any group-specific topology.

\section{The Dynamical Lie Algebraic Theory of Barren Plateaus}

The most systematic DLA-based theory of barren plateaus currently available is due to
Ragone et al., with Fontana et al. describing the canonical adjoint formalism for the DLA theory. Ragone et al. begins with the observation that once the circuit is sufficiently deep, one may replace averaging over parameters by averaging over the compact
Lie group generated by the DLA, and then compute the first two moments of the loss by
representation-theoretic means. In this way they obtain an exact Haar variance formula
which unifies several previously separate mechanisms for barren plateaus.

To tidily state their framework, let
\[
G=\exp(\mathfrak g)
\]
denote the compact, connected Lie-group realization of the DLA given by the relevant representation. Decompose
\[
\mathfrak g= \mathfrak z \oplus \mathfrak g_1\oplus\cdots\oplus \mathfrak g_{k-1}
\]
into simple ideals together with the center. For each ideal $\mathfrak g_j$, let
$P_{\mathfrak g_j}$ denote the Hilbert--Schmidt orthogonal projection from $\End(\mathcal H)$
onto the complexification of $\mathfrak g_j$, and define the generalized purity
of a Hermitian operator $A$ by \cite{ragone_lie_2024}
\[
\mathcal P_{\mathfrak g_j}(A):=\bigl\|P_{\mathfrak g_j}(A)\bigr\|_{HS}^2.
\]
Up to notation, this is the Lie-algebraic quantity that measures how much of the state or
observable is visible in the $j$th simple sector.

The crucial assumption in Ragone et al. is that the circuit is \emph{sufficiently deep}
that its distribution forms a $2$-design on each simple or abelian component. Equivalently, the first two moments of the circuit ensemble are
assumed to agree with those of Haar measure on the corresponding components. Under this
assumption, Haar integration on $G$ becomes available, and one may compute the loss
moments exactly.

\begin{definition}[$t$-design and approximate $t$-design]\label{def:t-design}
Let $\nu$ be a distribution of unitaries associated with ensemble $\mathcal{E}$ in a compact group $G$ acting on
$\mathcal H$. Its $t$-th moment operator is the superoperator
\[
\mathcal M_{\nu}^{(t)}(X)
:=
\int_{\mathcal E} U^{\otimes t}X(U^\dagger)^{\otimes t}\,d\nu(U).
\]
If $\mathcal M_{\mathcal E}^{(t)}=\mathcal M_G^{(t)}$, where $\mathcal M_G^{(t)}$ denotes the Haar moment operator on
$G$, then $\mathcal E$ is an exact $G$ $t$-design. If instead
\[
\|\mathcal  M_{\mathcal E}^{(t)}-\mathcal M_G^{(t)}\|\le \varepsilon
\]
in a chosen operator norm, then $\mathcal E$ (or more accurately, $\nu$) is called an $\varepsilon$-approximate
$G$ $t$-design \cite{ragone_lie_2024, anschuetz_2025}.
\end{definition}

In Ragone et al., the main exact variance theorem applies when either the observable
or the input state lies in the represented Lie algebra: if $O\in i\mathfrak g$
or $\rho\in i\mathfrak g$, then the loss expectation and variance decompose neatly according
to the reductive splitting of $\mathfrak g$.

\begin{theorem}[Ragone et al., schematic form]
Assume that the parameterized circuit is deep enough to form a $2$-design on each
component $G_j=\exp(\mathfrak g_j)$ of the compact Lie group generated by the DLA, and
assume further that either $O\in i\mathfrak g$ or $\rho\in i\mathfrak g$. Then:
\begin{enumerate}
    \item the mean of the loss receives contributions only from the center $\mathfrak z$;
    \item the variance receives contributions only from the simple ideals; and
    \item one has the exact formula
    \[
    \mathrm{Var}_\theta[\ell_\theta(\rho,O)]
    =
    \sum_{j=1}^{k-1}
    \frac{\mathcal P_{\mathfrak g_j}(\rho)\,\mathcal P_{\mathfrak g_j}(O)}{\dim(\mathfrak g_j)}.
    \]
\end{enumerate}
\end{theorem}

This is an incredibly strong result. First, it replaces bounds by an exact
Haar-sector variance formula. Second, it shows that the variance splits over the simple
components of the DLA, so that trainability may be analyzed one ideal at a time. Third,
it packages several familiar sources of barren plateaus into a single algebraic statement:
large expressivity corresponds to large $\dim(\mathfrak g_j)$, highly generalized-entangled
states correspond to small $\mathcal P_{\mathfrak g_j}(\rho)$, and highly generalized-nonlocal
observables correspond to small $\mathcal P_{\mathfrak g_j}(O)$.

This yields the central interpretation promoted in that paper, being that barren plateaus
may arise absent noise from three algebraic sources: excessive circuit expressivity,
insufficient purity of the initial state relative to the DLA, or insufficient purity of the
observable relative to the DLA. In that sense, the DLA furnishes a common language for
expressiveness-, state-, and observable-induced concentration.

Ragone et al. also discuss extensions to SPAM noise (state preparation and measurement noise) and coherent noise (noise that causes deviations from the intended state without the loss of information). In their
formalism, noise can suppress the same generalized purities or enlarge the effective DLA, and
thereby induce additional concentration. Thus the DLA framework is a rather broad algebraic account of several known barren
plateau mechanisms.

Several assumptions of the DLA theory of barren plateaus are essential and should be stated explicitly, since
they are precisely the assumptions whose finite-depth adequacy is later questioned in this
thesis.

\begin{remark}[Assumptions in the DLA theory]
The accomplishments of Ragone et al. depend on several structural facts:
\begin{enumerate}
    \item one passes from the actual finite-depth circuit ensemble to Haar averaging on the
    compact Lie group generated by the DLA, justified by a deep-circuit or $2$-design
    hypothesis;
    \item the main exact variance formula assumes that either the observable $O$ or the
    input state $\rho$ lies in the represented Lie algebra;
    \item the practical approximation statements are expressed through exact or approximate
    $t$-design control, and therefore inherit the usual dependence on convergence-to-Haar
    heuristics; and
    \item the DLA itself is an asymptotic expressivity object, so its associated compact Lie
    group may be much larger than the manifold actually explored by a shallow circuit.
\end{enumerate}
These are not flaws in the theory; on the contrary, they are exactly what makes the theory
clean and powerful. But they do delimit its natural regime of validity.
\end{remark}

For the purposes of the present thesis, Ragone et al. provides the correct Haar-limit benchmark: once one has passed to a sufficiently mixed
ensemble on the compact group generated by the DLA, the variance is governed by the
semisimple ideal decomposition above. Much of what follows in later chapters may be read as an
attempt to retain this Lie-theoretic clarity while replacing the full DLA group by a more
depth-aware geometric object.

It is worth recognizing that Ragone et al. also give a quantitative design-depth theorem
in terms of the $t$-moment expressivity superoperator
\[
\mathcal A_{\mathcal E_L}^{(t)}:=\mathcal M_{\mathcal E_L}^{(t)}-\mathcal M_G^{(t)},
\]
where $\mathcal E_L$ denotes the depth-$L$ ensemble. Their claim is that once $\|\mathcal A_{\mathcal E_L}^{(t)}\|$
is sufficiently small (less than one, so that the application of subsequent layers is a contractive map) and $L$ is modestly large, the circuit meaningfully forms an approximate $t$-design and the Haar variance formula becomes accurate \cite{ragone_lie_2024}. This design-based passage from finite depth to Haar is important background for the present thesis, since \Cref{chap:applications} returns
to precisely this issue and when such convergence is or is not defensible.

\section{Brief Aside on Jordan-Algebraic Wishart Systems}

A more recent and more ambitious generalization of the DLA picture is due to Anschuetz,
who proposes the Jordan-algebraic Wishart systems framework, or \emph{JAWS}, for VQA loss
landscapes. The central observation is that for a variational family $U_\theta$ and observable
$O$, the operators
\[
\{U(\boldsymbol{\theta})^\dagger O U(\boldsymbol{\theta})\}_{\boldsymbol{\theta}}
\]
generate under the anticommutator product a Jordan subalgebra
\[
\mathcal A \subset H_N(\mathbb C).
\]
Using the structure theory of finite-dimensional Jordan algebras, one decomposes
\[
\mathcal A \cong \bigoplus_\alpha \mathcal A_\alpha,
\]
where each simple component $\mathcal A_\alpha$ is isomorphic to
$H_{N_\alpha}(F_\alpha)$ for $F_\alpha\in\{\mathbb R,\mathbb C,\mathbb H\}$.

This framework has real achievements. Most notably, it extends Lie-algebraic barren plateau
formulas beyond the setting in which $O$ or $\rho$ must belong to the DLA---a circuit for which this is the case is called a Lie algebra-supported ansatz (LASA). Anschuetz derives variance formulas whose leading asymptotic term takes the form
\[
\mathrm{Var}_{\boldsymbol{\theta}}[\ell(\rho;\boldsymbol{\theta})]
=
\sum_\alpha
\frac{\Tr \! \big[(O_\alpha)^2\big]\,\Tr \! \big[(\rho_\alpha)^2\big]}
{\dim_{\mathbb R}\mathrm{Aut}(\mathcal A_\alpha)}
+ O(\pi),
\]
and explicitly notes that this generalizes the corresponding LASA formulas of
Ragone et al. Thus the JAWS framework genuinely enlarges the class of observables and states for which one can write asymptotic variance formulas.

Even so, for the purposes of the present thesis, JAWS is best regarded as a refinement and
extension of the same broad program and not a wholly separate foundation for VQA landscape theory. The group that enters its asymptotic formulas is the identity-connected automorphism group of the
Jordan algebra $\mathcal A$, whose connected components are products of the familiar
classical groups $SO(N_\alpha)$, $SU(N_\alpha)$, and $Sp(N_\alpha)$ \cite{anschuetz_2025}. In other words, the role played in Ragone et al. by the compact group generated by the DLA is replaced
by the compact symmetry group naturally attached to the Jordan algebra generated by the
loss observable under the variational orbit.

From this point of view, JAWS preserves much of the architecture of the DLA theory.
One still organizes the variance by a reductive or semisimple decomposition into algebraic
sectors, one still computes asymptotic landscape statistics by passing to a compact symmetry
group and its invariant measure, and one still relies on approximately uniform random
initialization assumptions to access Wishart or Haar-type asymptotic regimes. The main
conceptual gains are that the Jordan algebra sees a larger class of observables than the Lie
algebra alone (therefore being strictly more expressive as a bookkeeping device for the
loss landscape when compared to LASAs), and that it sanctions direct analysis of the density of poor minima \cite{anschuetz_2025}.

For this thesis, however, the additional generality of JAWS is not central. The principal aim
here is to refine the finite-depth geometry of Lie-algebraic barren plateau theory, not to
develop a competing asymptotic framework. Since JAWS largely extends rather than replaces
the DLA and Haar paradigm, and since most of the later arguments in this thesis are framed in
terms of compact Lie-group geometry, Haar measure, and adjoint-sector visibility, we shall
not pursue the Jordan-algebraic formalism further. It will reappear in the discussion
chapter as a natural direction for future extension, especially if one wishes to incorporate
symmetric products or noisy open-system dynamics, but that lies outside the main scope of the
present work.
\chapter{Krylov-Lie Algebras}\label{chap:krylov-lie-algebras} 

\section{Motivation from Measure Theory and Krylov Subspaces}

In light of the deficiencies of mainline VQA landscape theory, it can seem unclear where to go next in terms of a sensible theoretical framework. We can begin carving a path forward by taking lessons from measure theory. Specifically, given some VQA with pushforward measure $\Phi_* \nu_M$ over the image of its circuit parameters (which we conservatively take to be compact, naturally implying $\sigma$-finiteness of $\Phi_* \nu_M$), suppose that we have some ``correct" compact approximation to its manifold structure, and that this structure does have an honest, normalized ($\sigma$-finite) Haar measure over itself, $\mu_\mathcal{G}$ \cite[Thms.~2.10 and 2.20]{Folland1995Harmonic}. By the Lebesgue-Radon-Nikodym theorem \cite[Thm.~3.8]{Folland1999RealAnalysis}, we have that there exist uniquely-determined, $\sigma$-finite $\nu_1$ and $\nu_2$ such that $\nu_1$ is absolutely continuous with respect to $\mu_\mathcal{G}$, $\nu_2$ and $\mu_\mathcal{G}$ are singular, $\Phi_* \nu_\mathcal{M} = \nu_1 + \nu_2$, and 
\begin{equation}
\Phi_* \nu_M(E) = \int_{E} \frac{d \nu_1}{d \mu_\mathcal{G}} d \mu_\mathcal{G} + \nu_2 (E) 
\end{equation}
for all Borel sets $E$ in the $\sigma$-algebra of $\mathcal{G}$, where $\frac{d \nu_1}{d \mu_\mathcal{G}}$ is a unique, nonnegative Lebesgue-integrable Radon-Nikodym derivative for $\nu_1$. We then have the interpretation of $\frac{d \nu_1}{d \mu_\mathcal{G}}$ as a probability density over the (geometric) Haar measure for our approximation to the VQA reachable manifold. Similarly, $\nu_2$ represents the component of our pushforward measure on our VQA that exists externally from distributional effects on the geometry of our approximate structure. It immediately follows that if one were able to find an approximate structure such that $\nu_2$ had total mass $0$, then we would have a semi-canonical choice of measure and approximator for our VQA, since any probabilistic information we would want to investigate from quantities like moments could be purely decoupled into geometric contributions gleaned from the Haar measure of our approximator and sampling distortion over the geometry associated with our approximator. It is in this sense that such an approximator would be a model that is \textit{faithful} to the underlying geometry of the VQA: we are not losing any information for a particular circuit by working with the Haar measure of the approximator instead of the pushforward measure of the VQA, but rather we are merely isolating the VQA's geometric content (as opposed to the supplementary distributional content that lives on top of it). Furthermore, due to the suggestiveness of this representation, we would be able to analyze evolutionary behavior of the VQA by discussing the flow of the Radon-Nikodym derivative under the dynamics of a chosen optimizer on the constant background of our approximator's geometry. We can do even better still by explicitly recognizing the smoothness of the structures between which we are implicitly mapping:


\begin{theorem}\label{thm:DimensionMatchingJacobian}
Let the $m$-dimensional smooth reachable manifold of our VQA be $M$, and equip it with the finite parameter law
\[
\nu_M=\tilde q\, d\operatorname{vol}_M
\]
where $\tilde q\in L^1(M)$ and $\tilde q\ge 0$. Let $\mathcal G^*$ be an $n$-dimensional compact Lie group, $n\le m$, with normalized Haar measure $\mu_{\mathcal G^*}$. Suppose there exists an open set $U\subset M$ such that
\[
\nu_M(M\setminus U)=0
\]
and $\Phi|_U:U\to \mathcal G^*$ is $C^1$. If
\[
\operatorname{rank}(d\Phi_x)=n
\qquad\text{for }\nu_M\text{-a.e. }x\in U,
\]
then the singular component in the Lebesgue--Radon--Nikodym decomposition of $\Phi_*\nu_M$ with respect to $\mu_{\mathcal G^*}$ is zero.
\end{theorem}

\begin{proof}
Let
\[
R:=\{x\in U \, \vert \,\operatorname{rank}(d\Phi_x)=n\},
\qquad
C:=M\setminus R.
\]
Since $\nu_M(M\setminus U)=0$ and $\operatorname{rank}(d\Phi_x)=n$ for $\nu_M$-a.e. $x\in U$, we have
\[
\nu_M(C)=0.
\]
Hence
\[
\Phi_*\nu_M=\Phi_*(\nu_M|_R),
\]
because for every Borel set $E\subset \mathcal G^*$,
\[
\Phi_*(\nu_M|_C)(E)
=
\nu_M(C\cap \Phi^{-1}(E))
\le \nu_M(C)=0.
\]

Now fix $x\in R$. Since $\Phi|_U$ is $C^1$ and $\operatorname{rank}(d\Phi_x)=n=\dim\mathcal G^*$, the differential $d\Phi_x$ is surjective \cite[Thm.~4.12]{Lee2013SmoothManifolds}. Thus $\Phi$ is a submersion at $x$, and by the local submersion theorem there exist charts 
\[
\alpha_x:U_x\to \widetilde U_x\subset \mathbb R^m,
\qquad
\beta_x:V_x\to \widetilde V_x\subset \mathbb R^n,
\]
with $x\in U_x\subset U$ and $\Phi(x)\in V_x$, such that in these coordinates
\[
\beta_x\circ \Phi\circ \alpha_x^{-1}(u,v)=u,
\]
where $u\in \mathbb R^n$ and $v\in \mathbb R^{m-n}$ \cite[Thm.~4.12]{Lee2013SmoothManifolds}.

Because $M$ is second countable, we may choose a countable cover $\{U_i\}_{i=1}^\infty$ of $R$ by such submersion neighborhoods \cite[App.~A]{Lee2013SmoothManifolds}. Define a measurable disjointification by
\[
W_1:=U_1\cap R,
\qquad
W_i:=(U_i\cap R)\setminus \bigcup_{j<i}\,(U_j\cap R)
\quad (i\ge 2).
\]
Then the sets $W_i$ are pairwise disjoint, measurable, and
\[
R=\bigsqcup_{i=1}^\infty W_i.
\]
Therefore
\[
\nu_M|_R=\sum_{i=1}^\infty \nu_M|_{W_i}
\qquad\text{and hence}\qquad
\Phi_*\nu_M=\sum_{i=1}^\infty \Phi_*(\nu_M|_{W_i}).
\]
It therefore suffices to show that each measure
\[
\nu_i:=\Phi_*(\nu_M|_{W_i})
\]
is absolutely continuous with respect to $\mu_{\mathcal G^*}$.

Fix $i$, and let $A\subset \mathcal G^*$ be Borel with $\mu_{\mathcal G^*}(A)=0$. Since the Haar measure on a compact Lie group is locally given by a smooth positive density with respect to the Euclidean Lebesgue measure in any smooth chart \cite[Ch.~16]{Lee2013SmoothManifolds}, it follows that in the chart $\beta_i:V_i\to \widetilde V_i\subset\mathbb R^n$ associated to $U_i$, the set
\[
\beta_i(A\cap V_i)
\]
has Lebesgue measure zero in $\mathbb R^n$.

In the submersion coordinates on $U_i$, the map $\Phi$ is just the projection onto the first $n$ coordinates, so
\[
\alpha_i\bigl(\Phi^{-1}(A\cap V_i)\cap U_i\bigr)
= \{(x,y) \in \alpha_i(U_i)=\beta_i(V_i) \times W_i' \, \big| \, x \in \beta_i(A \cap V_i) \}
=
\beta_i(A\cap V_i)\times W_i'
\]
for some open set $W_i'\subset \mathbb R^{m-n}$.
By Fubini's theorem, this product set has Lebesgue measure zero in $\mathbb R^m$ \cite[Thm.~2.37]{Folland1999RealAnalysis}.

Now $\nu_M=\tilde q\,d\operatorname{vol}_M$ is absolutely continuous with respect to the smooth volume measure $d\operatorname{vol}_M$, and in local coordinates a smooth volume measure is a smooth positive density times the Lebesgue measure. Hence
\[
\nu_M\bigl(\Phi^{-1}(A\cap V_i)\cap U_i\bigr)=0.
\]
Since $W_i\subset U_i$, we obtain
\[
\nu_i(A\cap V_i)
=
\nu_M\bigl(W_i\cap \Phi^{-1}(A\cap V_i)\bigr)
=0.
\]
On the other hand, $\nu_i$ is supported in $\Phi(W_i)\subset V_i$, so
\[
\nu_i(A)=\nu_i(A\cap V_i)=0.
\]
Thus $\nu_i\ll \mu_{\mathcal G^*}$ for every $i$.

Finally, a countable sum of measures absolutely continuous with respect to $\mu_{\mathcal G^*}$ is again absolutely continuous with respect to $\mu_{\mathcal G^*}$ \cite[§3.2]{Folland1999RealAnalysis}. Therefore
\[
\Phi_*\nu_M\ll \mu_{\mathcal G^*}.
\]
By the Lebesgue--Radon--Nikodym decomposition theorem, the singular component of $\Phi_*\nu_M$ with respect to $\mu_{\mathcal G^*}$ is therefore zero \cite[Thm.~3.8]{Folland1999RealAnalysis}.
\end{proof}
Before proceeding, it should be mentioned that by the coarea formula \cite[§3.4]{evans_gariepy_2015}\cite[Ch.~III]{chavel2006riemannian}, when the singular part of the Lebesgue decomposition is zero, we have that
\begin{equation}\label{eq:corea-formula}
 \frac{d \Phi_* \nu_M}{d\mu_\mathcal{G^*}}(g) = \int_{\Phi^{-1}(g)} \frac{\tilde q(x)}{J_\Phi(x)}d\mathcal{H}^{m-n}(x),   
\end{equation}

for $\mu_\mathcal{G^*}$-almost every $g$, where $J_\Phi=\text{sqrt}(\text{det}(d\Phi_x d\Phi^\dagger_x))$ is the metric Jacobian and $\mathcal{H}^{m-n}$ is the Hausdorff measure on the fiber. 

The importance of Theorem \ref{thm:DimensionMatchingJacobian} is that it tells us that in order to even begin to approach the idea of a reasonable Lie group approximator to $M$, we cannot have a positive $\nu_M$-mass set of critical points (which could under the pushforward place mass on a lower-dimensional subset of $\mathcal{G}^*$, thereby creating a genuinely singular component), and more crucially that we should have this full-rank property of the Jacobian of our pushforward. In fact, if we have that $\text{dim} \: \mathcal{G}^* > \text{dim} \: M$ (making full rank impossible), such as in the case of most DLAs, it is easy to see that we can no longer obtain a singular component with zero mass. The question thus becomes how to actually find such a Lie group approximation to $M$, especially one that matches the dimension of $M$, as it is not terribly obvious how to do this, or if a Lie group that satisfies our requirements even exists (although it is relatively unforced in that setting to soften to the problem of finding the Lie group that minimizes the total mass of the singular part). Fortunately, to this end, we can take some inspiration from Krylov subspace methods.

The idea of Krylov subspace methods is elementary: given a linear operator $A$ on a vector space $V$, one selects a seed vector $v\in V$ and studies the finite-dimensional subspace generated by repeated action of $A$ on that seed. The order-$r$ Krylov subspace is
\[
\mathcal{K}_r(A,v)=\operatorname{span}\{v,Av,A^2v,\dots,A^{r-1}v\}.
\]
Equivalently, \(\mathcal{K}_r(A,v)\) is the space of all vectors of the form \(p(A)v\), where \(p\) ranges over polynomials of degree at most \(r-1\) \cite[Prop.~6.2]{saad2003iterative}. In this sense, a Krylov subspace is the natural finite-dimensional receptacle for the action of \(A\) on the seed vector: it contains, by construction, all powers of \(A\) acting on \(v\) up to the prescribed order.

The first important feature of these spaces is that they possess a built-in notion of truncation that is algebraically meaningful. Indeed, while \(\mathcal{K}_r(A,v)\) need not be invariant under arbitrary further action of \(A\) for small \(r\), there is a first index---the grade of \(v\) relative to \(A\), also called the \textit{grade}---at which the sequence stabilizes, and beyond that point no genuinely new directions are produced by applying additional powers of \(A\) \cite{gutknecht2009blockgrade}. Concretely, once one reaches the grade degree, every higher power \(A^kv\) can be written as a linear combination of the previously obtained vectors, so the Krylov subspace has become closed under further multiplication by \(A\) in the sense relevant to the generated dynamics; $\mathcal{K}_r(A,v)$ is a cyclic module. This is precisely the structural phenomenon that is suggestive for our purposes: there is a canonical graded growth process, and at a finite stopping point one attains a canonical algebraic closure with respect to the operation generating the space.

The second important feature is that Krylov constructions are remarkably flexible with respect to dimension. For a fixed operator \(A\), the dimension of \(\mathcal{K}_r(A,v)\) depends strongly on the choice of seed \(v\), and is controlled by the degree of the minimal polynomial of \(A\) relative to that seed \cite{gutknecht2009blockgrade}. Thus different seed vectors can produce Krylov spaces of dramatically different sizes, even for the same ambient operator. In practice this means that, by selecting the seed appropriately, one can realize a wide range of effective dimensions very easily (in fact any dimension up to the degree of $A$'s minimal polynomial, provided that one allows for any choice of initial seed), without ever needing to alter the underlying operator itself. This flexibility is one of the main reasons Krylov methods are so powerful as compression devices: the operator is fixed, while the seed determines which portion of its action is exposed at low dimension \cite[Ch.~11]{antoulas2005approximation}.

It is exactly these two features that make Krylov subspaces so attractive. On one hand, they provide a natural notion of graded approximation, together with a canonical stopping rule given by algebraic closure at the grade degree. On the other, they allow one to tune the dimension of the approximating object simply by changing the seed data. For our problem, however, ordinary Krylov subspaces are still insufficient, because they are fundamentally linear and commutative in their construction: they are built from repeated powers of a single operator, whereas the expressive structure of a variational circuit is governed by noncommutative growth under repeated commutation of its generators.

Thus, it is astute to ask whether one can construct an analogous object to Krylov subspaces in the Lie-theoretic setting. More specifically, one would like a graded family obtained not from powers \(A^k v\), but from repeated adjoint actions or nested commutators of chosen generators acting on some seed element or seed set. Such an object should retain the two desirable properties just discussed. First, it should admit a notion of grade degree at which further commutation produces no new directions, so that one reaches closure under the Lie bracket. Second, by varying the seed data, it should allow substantial control over the dimension of the resulting space. If such a construction can be made precise, then it would furnish exactly the kind of finite-dimensional Lie-algebraic approximator we seek: one that is flexible in dimension, naturally graded, and capable of becoming a bona fide Lie algebra at finite grade.

This aim is now what we shall set out to accomplish.

\section{Construction of Krylov and Krylov-Lie Algebras}

\subsection{Krylov Algebras}

We shall begin by defining a slightly more abstract notion of Krylov algebra before eventually restricting our focus to a Lie-centric structure. To this end, we introduce the notions of \textit{seeds}, \textit{block seeds}, and \(\circ\)\textit{-words}.

\begin{definition}[Seeds and block seeds]
\phantomsection
\label{def:seeds-and-block-seeds}
Let \(V\) be a vector space over a field \(\mathbb{F}\).

A \emph{seed} is an element \(v\in V\setminus\{0\}\).

A \emph{block seed} is an $s$-tuple
\[
\Psi=(\psi_1,\dots,\psi_s)\in (V\setminus\{0\})^s,
\]
\end{definition}

\begin{remark} \label{rem:SurjectiveParamBlockSeeds}
    Clearly, if we have a surjective parameterization for V, then one can always recursively produce a linearly independent block seed for any $s \leq \text{dim}(V)$ by projecting out all previous vectors from the parameterization with $1-\sum^{s-1}_{i=1}\ket{\psi_i} \bra{\psi_i}$ (assuming normalized vectors). Furthermore, it is occasionally valuable to restrict to unit norm seeds (call this space $\Sigma_s$) so that we are working over a compact space. Since this never alters the linear dependence of vectors in our block seed, we can typically perform this restriction without it affecting the validity or complexity of our proofs.
\end{remark}

\begin{definition}[\(\circ\)-words generated by a block seed] \label{def:CircWords}
Let \(V\) be a vector space over \(\mathbb{F}\), $\circ: \mathfrak{a}\times \mathfrak{a}\to \mathfrak{a}$ a bilinear operation on a chosen algebra $\mathfrak{a}$ induced by said operation, $\mathcal{S}=\{H_1, ..., H_n\}$ a finite set of generators for $\mathfrak{a}$, and $\Psi=(\psi_1,\dots,\psi_s)$ a block seed in \( (V \setminus \{ 0\})^s\) (\Cref{def:seeds-and-block-seeds}). For every integer \(r\geq 1\) and every multi-index
$(i_1,\dots,i_r)\in \{1,\dots,n\}^r$
define the associated nested \(\circ\)-word $w_\circ(i_1,...,i_r)$ as
\[
w_\circ(i_1,...,i_r)=w_\circ(H_{i_1},...,H_{i_r})
:=
H_{i_r} \circ (\cdots \circ ( H_{i_3} \circ (H_{i_2}\circ H_{i_1}))\cdots),
\]
and call the set of all such words up to a depth of $k$ in our generators
\[
\mathcal{W}_{\leq k}( \mathcal{S}):= \{ w_\circ(i_1,...,i_r) \:\vert \: 1 \leq r \leq k, \:i_\ell \in \{1,..,n\}\}.
\]
\end{definition}

\begin{remark}
\phantomsection
\label{rmk:words-convention}
On occasion, specifying the specific indices and order of our words will be inconvenient and tedious. In these cases, we write 
\[
w_\circ(H_1, \dots, H_n) \in \mathcal{W}_{\leq k}( \mathcal{S})
\]
to mean any arbitrary word constructed from the generator set $\mathcal S = \{H_1, \dots, H_n \}$ up to a depth of $k$. Also, when unambiguous, we often say $\mathcal W_{\leq k}$ to mean $\mathcal W_{\leq k} (\mathcal S)$ for brevity, we use $\mathcal W_{k}$ to refer to the set of all $\circ$-words of exactly depth $k$, and as a matter of pure convenience let $\mathcal W_0 =I$ ($I$ does not have to be in $\mathcal S$ or its dynamical Lie algebra).
\end{remark}

We now are equipped to introduce Krylov filtrations as well as their algebraically closed cousins, Krylov algebras.
\begin{definition}[Krylov subspace for a binary product]
\phantomsection
\label{def:KrylovSubspace}
Let \(V\), \(\circ\), $\mathcal{S}$, and $\Psi$ be as \Cref{def:CircWords}. The \textit{grade-\(k\) Krylov subspace} (or filtration) defined for $k \geq 1$ and generated by the block seed \(\Psi\) (\Cref{def:seeds-and-block-seeds}) with respect to \(\circ\) and the generators $\mathcal{S}$ is 
\[
\mathcal{K}^{(k)}_{\Psi}(\mathcal{S})
:=
\operatorname{span}_{\mathbb{F}}
\{w \psi_j \, \vert \, w\in I \cup \mathcal{W}_{\leq k} \, , \, 1 \leq j \leq s\}.
\]
If the generators and depth are left implicit, then the Krylov subspace is denoted $\mathcal{K}_{\Psi}$. 

The full Krylov filtration is the increasing family
\[
\mathcal{K}^{(1)}_{\Psi}(\mathcal{S})
\subset
\mathcal{K}^{(2)}_{\Psi}(\mathcal{S})
\subset
\cdots.
\]
Meanwhile, if there exists a least integer \(r_\ast\geq 1\) such that
\[
\mathcal{K}^{\infty}_{\Psi}(\mathcal{S})
=
\mathcal{K}^{(r_*)}_{\Psi}(\mathcal{S})=\mathcal{K}^{*}_{\Psi}(\mathcal{S}),
\]
then \(r_\ast\) is called the \emph{grade degree} of the filtration.
\end{definition}

\begin{remark}
It is not strictly necessary to define the Krylov subspace so that it include the initial seed vectors (which is why the identity operator is present). The vast majority of the results pass by largely unaltered by this, with the slight exception of those in \Cref{sec:dim-matching-seeds}, since the inclusion of the block seed in the Krylov subspace can very well increase its dimension, although the themes remain. For standard cases, especially Lie-theoretic ones, this will be marginally unfavorable as it can sometimes make it harder to achieve dimension-matching at shallow circuit depths (\Cref{prop:witt-upper-bound-krylov-subspace}; the change is quite simple when $\Psi$ is assumed linearly independent).

There are good reasons to include the seeds. Chiefly among these is that it becomes necessary for the derivation of the factorial seed-based error bounds in \Cref{thm:seed-KLG-factorial,thm:prefix-KLG-factorial}, for else we must consider leakage of the seed vectors into the orthogonal complement of the subspace under application of the generators. It also aligns better with standard Krylov theory, which directly employs the seeds in the construction of its subspace.
\end{remark}

\begin{definition}[Krylov algebra] 
\phantomsection
\label{def:KrylovAlgebra}
Assume now that \(V=\mathcal H\) is a finite-dimensional Hilbert space, and let
$\mathcal K_\Psi(\mathcal S)$ be a depth-$k$ Krylov subspace of dimension $r^{(k)}_{\Psi}$ as in \Cref{def:KrylovSubspace}. Say that
\[
B_\Psi
:=
\begin{bmatrix}
b_1 & \cdots & b_{r^{(k)}_{\Psi}}
\end{bmatrix}
\]
is the corresponding basis matrix of vectors than span our Krylov subspace. The orthogonal projector (which can include an indicator for depth if desired) onto
\(\mathcal K_\Psi(\mathcal S)\) is
\[
P_\Psi
:=
B_\Psi B_\Psi^{+}
=
B_\Psi (B_\Psi^\dagger B_\Psi)^{-1} B_\Psi^\dagger,
\]
where \(B_\Psi^{+}\) denotes the Moore-Penrose pseudoinverse \cite[§1.9]{saad2003iterative}.
If \(B_\Psi\) has orthonormal columns, then
\[
P_\Psi=B_\Psi B_\Psi^\dagger.
\]
Define the compressed or reduced generators on the ambient space by
\[
\widetilde H_{j}
:=
P_\Psi H_j P_\Psi,
\qquad j=1,\dots,n.
\]
The \emph{Krylov algebra} associated with the block seed \(\Psi\) and the bilinear operation $\diamond:\mathfrak a \times \mathfrak a \rightarrow \mathfrak a$ (which may or may not be the same as $\circ$) is the
\(\diamond\)-algebra generated by these compressed generators, namely
\[
\mathfrak K_\Psi^{(k)}(\mathcal S)
:=
\operatorname{span}_{\mathbb F}
\Big\{
w_\diamond(\widetilde H_{i_1},\dots,\widetilde H_{i_r})
\,\Big|\,
w_\diamond(\widetilde H_{i_1},\dots,\widetilde H_{i_r})\in\ \widetilde{\mathcal{W}}_{\infty}^{(k)} \Big\},
\]
where $\widetilde{\mathcal{W}}_{\infty}^{(k)}$ is the finite set of all $\diamond$-words of compressed generators induced by the depth of words used to create the Krylov subspace (see \Cref{def:CircWords} for index notation). 
\end{definition}
\begin{remark}
Equivalently, in basis-free notation one may regard \(\mathfrak K^{(k)}_\Psi(\mathcal S)\) as
the algebra generated by the compressed operators $P_\Psi H_j P_\Psi$. Also, it is vital to remember that the depth $k$ does not refer to the depth of the nested $\circ$-products of the compressed generators in $\mathfrak{K}^{(k)}_\Psi$, but rather the depth of the nested $\circ$-products used to generate the Krylov subspace $\mathcal{K}_{\Psi}^{(k)}$.
\end{remark}

\subsection{Krylov-Lie Algebras}

With this setup, Krylov-Lie subspaces and algebras are quite easy to describe.

\begin{definition}[Krylov-Lie subspace] \label{def:krylov-lie-subspace}
Let \(V\) be a vector space over \(\mathbb{F}\), $[\, \cdot \, , \, \cdot \, ]: \mathfrak{g}\times \mathfrak{g}\to \mathfrak{g}$ the Lie bracket commutator  for the Lie algebra $\mathfrak{g}$, $\mathcal{S}=\{H_1, ..., H_n\}$ a finite set of generators for $\mathfrak{g}$, and $\Psi=(\psi_1,\dots,\psi_s)$ a block seed in \( (V \setminus \{ 0\})^s\). The grade-\(k\) Krylov-Lie filtration generated by \(\Psi\), $\mathcal{K}^{(k)}_{\Psi}(\mathcal{S})$, is then the grade-$k$ Krylov subspace generated by $\Psi$ with respect to the Lie bracket and $\mathcal{S}$, with the full Krylov-Lie filtration, and grade degree (which always exists when $\mathcal{H}$ has finite dimension) defined analogously as in \Cref{def:CircWords}.
\end{definition}

\begin{definition}[Krylov-Lie algebra] \label{def:krylov-lie-algebra}
Let \(V\), $\mathcal{S}$, $[\, \cdot \, , \, \cdot \, ]$, and $\Psi$ be as in \Cref{def:krylov-lie-subspace}.  Say our depth-$k$ Krylov subspace (possibly not induced by the Lie bracket) is $\mathcal{K}_\Psi^{(k)}(\mathcal{S})$ with basis matrix of spanning vectors $B_\Psi^{(k)}$ and orthogonal projector $P_\Psi^{(k)}$ (\Cref{def:KrylovAlgebra} and \Cref{def:krylov-lie-subspace}). Then for the compressed (or reduced) generators $\widetilde H_{j}^{(k)}
=
P_\Psi^{(k)} H_j P_\Psi^{(k)}$ (see \Cref{def:KrylovAlgebra}),
\[
\mathfrak{l}_{\Psi}^{(k)}(\mathcal{S})
:=
\operatorname{span}_{\mathbb F}
\Big\{
w_{[ \cdot , \cdot ]}(\widetilde H_{i_1}^{(k)},\dots,\widetilde H_{i_r}^{(k)})
\,\Big|\,
w_{[ \cdot , \cdot ]}(\widetilde H_{i_1}^{(k)},\dots,\widetilde H_{i_r}^{(k)}) \in\ \! \! \widetilde{\mathcal{W}}_{\infty}^{(k)}\Big\}= \operatorname{Lie}_{\mathbb{F}}\langle \widetilde{H}_1^{(k)}, \dots, \widetilde H_n^{(k)}\rangle
\]
is called the \emph{depth-$k$ Krylov-Lie algebra} (KLA) generated by the block seed \(\Psi\). If $k=r_*$, the grade degree of \Cref{def:KrylovAlgebra}, this KLA is the unique and minimal one that has saturated the commutator.
\end{definition}

\begin{remark}[Krylov-Lie algebras are Krylov algebras]
In other words, a depth-$k$ Krylov-Lie algebra $\mathfrak{l}_{\Psi}^{(k)}(\mathcal{S})$ (or $\mathfrak{l}_{\Psi}^{(k)}$ when the generators are left implicit, and $\mathfrak{l}_{\Psi}$ if the depth is implict as well) is a depth-$k$ Krylov algebra whose algebraic product is a Lie bracket, by \Cref{def:KrylovAlgebra}.
\end{remark}

We will find that while these constructions are natural and faithful to the Lie-algebraic dynamics, their error bounds require additional hypotheses. For Hermitian involutive generators, seed-level commutator matching will give a factorial estimate in \Cref{thm:seed-KLG-factorial}. To obtain ordinary word preservation directly, without this generator restriction or a separate matching condition, we introduce a useful variant.

\begin{definition}[Krylov prefix subspace and prefix KLA] 
\phantomsection
\label{def:prefix-subspace-kla}
Take the setup of \Cref{def:krylov-lie-subspace}, but add to it matrix multiplication. The \emph{grade-\(k\) Krylov prefix subspace} $\mathcal P^{(k)}_{\Psi}$ is the grade-$k$ Krylov subspace generated by $\Psi$ with respect to matrix multiplication and $\mathcal{S}$ (cf. \Cref{def:KrylovSubspace}). Moreover, the \emph{grade-$k$ prefix Krylov-Lie algebra} $\mathfrak p^{(k)}_{\Psi}$ is the KLA induced by this (associative) Krylov prefix subspace.
\end{definition}

\begin{remark}
Owing to them both being Krylov-Lie algebras, the prefix KLA and the standard KLA share many of the same properties. In fact, they share so many and the transfer of proofs so trivial that, unless otherwise noted, it can be assumed for the remainder of this text that \emph{any} property held by a KLA is also a property of the prefix KLA (though the converse need not be true). The prefix KLA is not as physical as the usual \Cref{def:krylov-lie-algebra} due to its subspace being crafted from directions not directly reachable by the dynamical Lie commutator structure, but it will serve its purpose well in \Cref{chap:krylov-lie-approximations}.
\end{remark}

\begin{remark}[VQA conventions] \label{rem:VQAConventions}
In the setting of VQAs, the ambient vector space is the Hilbert space
\(
\mathcal{H} \cong \mathbb{C}^{2^N}
\) 
(assumed to have complex dimension $d$), and a block seed
\(
\Psi=(\psi_1,\dots,\psi_s)
\)
consists of state vectors in \(\mathcal{H}\). The generators
\(
\mathcal{S}=\{H_1,\dots,H_n\}
\)
are taken to be Hermitian linear operators acting on \(\mathcal{H}\) so that the associated ansatz unitaries are of the form
\[
U(\boldsymbol \theta)=e^{-i\theta_p H_{i_p}}\cdots e^{-i\theta_1 H_{i_1}},
\]
and $i_1, \dots, i_p$ are indices corresponding to the generators in $\mathcal S$.

From the Lie-theoretic point of view, the more natural algebra is the real Lie algebra of anti-Hermitian operators
\[
\mathfrak{u}(d)
=
\{X\in \operatorname{End}(\mathcal{H}) \, \vert \, X^\dagger=-X\},
\]
or alternatively, the real Lie subalgebra generated by
\(
\{iH_1,\dots,iH_n\}
\).
Indeed, if \(H_i\) and \(H_j\) are Hermitian, then \([H_i,H_j]\) is anti-Hermitian rather than Hermitian.

Accordingly, whenever we work with Hermitian generators, the underlying Lie structure should be understood via the identification
\[
H \longleftrightarrow iH.
\]
Equivalently, one may transport the Lie bracket to the Hermitian side by defining
\[
[ A,B ] := i[A,B],
\qquad A,B \in \operatorname{Herm}(\mathcal{H}),
\]
so that  $[A,B]$ is again Hermitian. Thus, up to the harmless factor of \(i\), the Hermitian and anti-Hermitian conventions encode the same real Lie algebra.

Throughout the remainder of this work, we shall use Hermitian generators for physical clarity. Whenever Krylov-Lie words (see \Cref{def:CircWords} and \Cref{def:krylov-lie-algebra}) are written directly in the \(H_j\) or their compressed variants, the corresponding factors of \(i\) needed to pass to the mathematically natural anti-Hermitian Lie algebra are to be understood implicitly. In particular, all Lie-theoretic spans of generators are understood as real spans, whereas the associated Krylov subspaces of states in \(\mathcal{H}\) may still be taken over \(\mathbb{C}\).
\end{remark}

Next, we shall discuss how to induce a Krylov-Lie group (KLG) from these
Krylov-Lie algebras (KLAs). We first recognize that the required representation is
already supplied by our compressed generators. If $Q$ is an
isometric basis map onto $\mathcal K_\Psi^{(k)}$, then the compressed Hermitian generator has matrix
\[
A_j=Q^{\dagger}H_jQ,
\qquad
A_j^{\dagger}=A_j.
\]
Thus $\mathfrak l_\Psi^{(k)}(\mathcal S)$ is a real Lie subalgebra of
$\mathfrak u(\mathcal K_\Psi^{(k)})$ after accounting for generator Hermiticity conventions. Its action on the Krylov
subspace is faithful by construction. This means that the Lie algebra generated by the compressed matrices already
has a faithful anti-Hermitian realization.

\begin{definition}[Krylov-Lie group]\label{def:krylov-lie-group}
Let $\mathcal K_\Psi^{(k)}(\mathcal S)\subset\mathcal H$ be our Krylov-Lie
subspace, let $P_\Psi^{(k)}$ denote its orthogonal projector, and let
\[
\mathfrak l_\Psi^{(k)}(\mathcal S)
\subset\mathfrak u(\mathcal K_\Psi^{(k)})
\]
be the corresponding Krylov-Lie algebra of \Cref*{def:krylov-lie-algebra}.
By Lie's third theorem, there exists a unique simply-connected
Lie group $\widetilde K_\Psi^{(k)}$, up to the isomorphism identifying its
Lie algebra with $\mathfrak l_\Psi^{(k)}(\mathcal S)$
\cite[Ch.~3]{kirillov_lie_groups}. We call $\widetilde K_\Psi^{(k)}$ the
\textit{initial depth-$k$ Krylov-Lie group} for $(\Psi,\mathcal S)$.

By Lie's second theorem, the inclusion
\[
\mathfrak l_\Psi^{(k)}(\mathcal S)
\hookrightarrow\mathfrak u(\mathcal K_\Psi^{(k)})
\]
integrates uniquely to a Lie-group homomorphism
\[
\Pi_\Psi^{(k)}:\widetilde K_\Psi^{(k)}\longrightarrow
U(\mathcal K_\Psi^{(k)})
\]
\cite[Ch.~3]{kirillov_lie_groups}. We define the \textit{depth-$k$ Krylov-Lie group} associated to \((\Psi,\mathcal S)\) to be the compact, connected matrix Lie group
\[
K_\Psi^{(k)}:=\overline{\Pi_\Psi^{(k)}(\widetilde K_\Psi^{(k)})}=\overline{\left\langle
\exp(d\Pi_\Psi^{(k)}(X)) \: \Big \vert \: X\in \mathfrak l_\Psi^{(k)}(\mathcal S)
\right\rangle}\subset U(\mathcal{K}_\Psi^{(k)}),
\]
where the bar means closure in the topology of $U(\mathcal{K}_\Psi^{(k)})$. Isomorphically, viewing \(\mathcal{K}_\Psi^{(k)}\) as a subspace of \(\mathcal H\), we may realize
\(K_\Psi^{(k)}\) in the ambient unitary group \(U(d)\) through the embedding
\[
\iota_\Psi^{(k)}:U(\mathcal{K}_\Psi^{(k)})\to U(d),\qquad
\iota_\Psi^{(k)}(U)=P_\Psi^{(k)} U P_\Psi^{(k)} + (I-P_\Psi^{(k)}).
\]
When we need to distinguish the Lie algebra of $K_\Psi^{(k)}$ from that of $\widetilde K_\Psi^{(k)}$, we write $\widehat{\mathfrak l}_\Psi^{(k)}(\mathcal S)$ for the Lie algebra of the former. We refer to this object as the \textit{central completion} of the represented KLA. Otherwise, we suppress the hat and understand the KLA to be the Lie algebra of the compact KLG.
\end{definition}

\begin{remark}
The image \(\Pi_\Psi^{(k)}(\widetilde K_\Psi^{(k)})\subset U(\mathcal K_\Psi^{(k)})\) is a connected Lie subgroup
with Lie algebra \(\mathfrak l_\Psi^{(k)}(\mathcal S)\), but it need not be closed \cite[Thm.~3.43]{kirillov_lie_groups}. The topological closure \(K_\Psi^{(k)}\) is therefore the natural compact matrix-group realization of the Krylov-Lie construction. The closure
is compact because it is a closed subset of a finite-dimensional unitary
group, and is a Lie group by the closed subgroup theorem
\cite[Chs.~2--3]{kirillov_lie_groups}. We do not need to assume
another faithful representation in order to obtain compactness. For measure-theoretic purposes below, the normalized Haar
measure on \(K_\Psi^{(k)}\) is the canonical reference measure---by \Cref{TraceHaarTopologyInvariance}, the discrete kernel and global lattice data of the compact connected model (or any Lie-centric model, for that matter) become irrelevant for integrals of observables
that factor through the represented unitary action, which includes the moments we wish to calculate. 

We should,
however, be careful about identifying its tangent algebra with the original
KLA: as we shall prove shortly, taking the closure may add central directions,
so that
\[
\dim K_\Psi^{(k)}
=\dim\widehat{\mathfrak l}_\Psi^{(k)}
\geq\dim\mathfrak l_\Psi^{(k)}.
\]
If the KLA is semisimple (its Killing form will automatically be negative-definite, which we shall also prove below), its simply-connected integrating group is already
compact, so its represented image is closed and the closure can be omitted
\cite[Ch.~6]{kirillov_lie_groups}. More generally, the same is true whenever
the center has the full period lattice described in
\Cref{prop:kla-central-frequencies}.
\end{remark}

\section{Structural results on Krylov-Lie algebras and groups}

With these definitions of Krylov-Lie algebras and Krylov-Lie groups, we
shall prove some germane structural results about our constructions.

\begin{proposition}[Haar measure and unimodularity]
\label{KLGUnimodular}
The compact Lie group \(K_\Psi\) (\Cref{def:krylov-lie-group}) admits a unique, normalized, bi-invariant Haar measure $\mu_{K_{\Psi}}$.
\end{proposition}

\begin{proof}
This follows from every compact Lie group being unimodular \cite[§4.6]{kirillov_lie_groups}\cite[§2.4]{Folland1995Harmonic}.
\end{proof}
\begin{proposition}[Invariant inner products]
\label{Ad-invariant-inner-product}
The KLA $\mathfrak l_\Psi$ admits an $\Ad(K_\Psi)$-invariant positive-definite
inner product. In particular, one may take
\[
\beta(X,Y):=-\operatorname{Re}\Tr[XY].
\]
\end{proposition}

\begin{proof}
The negative real trace form is a positive-definite inner product on
the ambient anti-Hermitian matrix algebra and is invariant under unitary
conjugation \cite[Chs.~5--6]{kirillov_lie_groups}. It therefore suffices
to show that $K_\Psi$ normalizes $\mathfrak l_\Psi$.
The normalizer of $\mathfrak l_\Psi$ in the ambient unitary group is
closed, since it is the stabilizer of a finite-dimensional vector
subspace under the continuous conjugation action. Moreover, it contains
$\Pi_\Psi (\widetilde K_\Psi)$, since a Lie group preserves its Lie algebra under conjugation.
It consequently contains $K_\Psi=\overline{\Pi_\Psi (\widetilde K_\Psi)}$, proving that the
restriction of $\beta$ to $\mathfrak l_\Psi$ has the required invariance.
\end{proof}

\begin{proposition}[Reductivity of the KLA]
\label{ReductiveKLA}
The Krylov-Lie algebra $\mathfrak l_\Psi$ is reductive. More precisely,
\[
\mathfrak l_\Psi
=\mathfrak z(\mathfrak l_\Psi)\oplus[\mathfrak l_\Psi,\mathfrak l_\Psi],
\]
where $\mathfrak z(\mathfrak l_\Psi)$ is abelian and
$[\mathfrak l_\Psi,\mathfrak l_\Psi]$ is semisimple of compact type.
Moreover, this decomposition is orthogonal with respect to any
$\Ad(K_\Psi)$-invariant positive-definite inner product on
$\mathfrak l_\Psi$. The reductivity assertion does not require $\Pi_\Psi (\widetilde K_\Psi)$ to be closed.
\end{proposition}

\begin{proof}
We use the trace inner product from \Cref{Ad-invariant-inner-product}.
Its ad-invariance and nondegeneracy give
\[
[\mathfrak l_\Psi,\mathfrak l_\Psi]^{\perp_\beta}
=\mathfrak z(\mathfrak l_\Psi).
\]
Indeed, $X$ is orthogonal to every commutator $[Y,Z]$ if and only if
$\beta([X,Y],Z)=0$ for every $Y,Z\in\mathfrak l_\Psi$, which is equivalent
to $[X,Y]=0$ for every $Y\in\mathfrak l_\Psi$. Hence positive definiteness
gives the orthogonal direct sum
\[
\mathfrak l_\Psi=\mathfrak z(\mathfrak l_\Psi)\oplus\mathfrak s_\Psi,
\qquad
\mathfrak s_\Psi:=[\mathfrak l_\Psi,\mathfrak l_\Psi].
\]
An element in the center of $\mathfrak s_\Psi$ commutes with both summands
and therefore lies in
$\mathfrak z(\mathfrak l_\Psi)\cap\mathfrak s_\Psi=\{0\}$. Thus
$\mathfrak z(\mathfrak s_\Psi)=0$.

Next, $\ad_X|_{\mathfrak s_\Psi}$ is skew-adjoint with respect to
$\beta|_{\mathfrak s_\Psi}$. The Killing form of $\mathfrak s_\Psi$
therefore satisfies
\[
\kappa_{\mathfrak s_\Psi}(X,X)
=\Tr_{\mathfrak s_\Psi}\bigl[(\ad_X|_{\mathfrak s_\Psi})^2\bigr]
=-\|\ad_X|_{\mathfrak s_\Psi}\|_{\mathrm{HS},\beta}^2<0
\qquad (X\ne0),
\]
where strictness follows from $\mathfrak z(\mathfrak s_\Psi)=0$.
By Cartan's criterion, $\mathfrak s_\Psi$ is semisimple, and the negative
Killing form makes it of compact type
\cite[Chs.~5--6]{kirillov_lie_groups}. This proves reductivity directly
from the anti-Hermitian nature of the representation.

Finally, if $\langle\cdot,\cdot\rangle$ is any
$\Ad(K_\Psi)$-invariant positive-definite inner product, then it is
ad-invariant on $\mathfrak l_\Psi$. For $Z\in\mathfrak z(\mathfrak l_\Psi)$,
\[
\langle[X,Y],Z\rangle=\langle X,[Y,Z]\rangle=0.
\]
The decomposition is therefore orthogonal with respect to this inner
product as well.
\end{proof}

\begin{proposition}[Central enlargement under closure]
\label{prop:kla-central-enlargement}
Let $K_\Psi$ be the KLG of \Cref{def:krylov-lie-group}, and write
$\widehat{\mathfrak l}_\Psi$ for its Lie algebra. Then
\[
\widehat{\mathfrak l}_\Psi
=\mathfrak l_\Psi\oplus\mathfrak c_\Psi,
\qquad
\mathfrak c_\Psi
:=\widehat{\mathfrak l}_\Psi\cap\mathfrak l_\Psi^{\perp_\beta}
\subset\mathfrak z(\widehat{\mathfrak l}_\Psi).
\]
Consequently,
\[
[\widehat{\mathfrak l}_\Psi,\widehat{\mathfrak l}_\Psi]
=[\mathfrak l_\Psi,\mathfrak l_\Psi],
\qquad
\mathfrak z(\widehat{\mathfrak l}_\Psi)
=\mathfrak z(\mathfrak l_\Psi)\oplus\mathfrak c_\Psi.
\]
In particular, taking the closure only adds central directions.
\end{proposition}

\begin{proof}
The proof of \Cref{Ad-invariant-inner-product} shows that $K_\Psi$
normalizes $\mathfrak l_\Psi$. Differentiating this action gives
\[
[\widehat{\mathfrak l}_\Psi,\mathfrak l_\Psi]\subset\mathfrak l_\Psi,
\]
so that $\mathfrak l_\Psi$ is an ideal of
$\widehat{\mathfrak l}_\Psi$. Its orthogonal complement $\mathfrak c_\Psi$
is also an ideal, since for $X\in\widehat{\mathfrak l}_\Psi$,
$C\in\mathfrak c_\Psi$, and $Y\in\mathfrak l_\Psi$,
\[
\beta([X,C],Y)=-\beta(C,[X,Y])=0.
\]
The two complementary ideals commute:
\[
[\mathfrak l_\Psi,\mathfrak c_\Psi]
\subset\mathfrak l_\Psi\cap\mathfrak c_\Psi=\{0\}.
\]
Thus every $C\in\mathfrak c_\Psi$ commutes with all exponentials of
$\mathfrak l_\Psi$, hence with $\Pi_\Psi (\widetilde K_\Psi)$. By continuity it commutes with
$K_\Psi=\overline{\Pi_\Psi (\widetilde K_\Psi)}$ as well. Differentiating once more shows that
$[\widehat{\mathfrak l}_\Psi,C]=0$, so
$\mathfrak c_\Psi\subset\mathfrak z(\widehat{\mathfrak l}_\Psi)$.
The stated commutator and center identities follow from this central
direct sum.
\end{proof}

\begin{corollary}[Central torus and the semisimple case]
\label{cor:kla-central-torus}
Write $\mathfrak z_\Psi=\mathfrak z(\mathfrak l_\Psi)$ and
$\mathfrak s_\Psi=[\mathfrak l_\Psi,\mathfrak l_\Psi]$. Let $S_\Psi$ be
the connected represented subgroup with Lie algebra $\mathfrak s_\Psi$,
and say that 
$T_\Psi:=\overline{\exp(\mathfrak z_\Psi)}$ with
$\mathfrak t_\Psi$ its Lie algebra. Then $S_\Psi$ is compact and closed, $T_\Psi$ is a compact connected abelian
group, and
\[
\Pi_\Psi (\widetilde K_\Psi)=S_\Psi\exp(\mathfrak z_\Psi),\qquad
K_\Psi=S_\Psi T_\Psi,\qquad
\widehat{\mathfrak l}_\Psi=\mathfrak s_\Psi\oplus\mathfrak t_\Psi.
\]
Moreover, $S_\Psi\cap T_\Psi$ is finite and
$\mathfrak t_\Psi=\mathfrak z_\Psi\oplus\mathfrak c_\Psi$. If
$\mathfrak l_\Psi$ is semisimple, then $\Pi_\Psi (\widetilde K_\Psi)=K_\Psi$ and
$\widehat{\mathfrak l}_\Psi=\mathfrak l_\Psi$.
\end{corollary}

\begin{proof}
By \Cref{ReductiveKLA}, $\mathfrak s_\Psi$ is semisimple of compact type.
Its simply-connected integrating group is compact, so the homomorphism
integrating its inclusion into $\mathfrak u(\mathcal K_\Psi)$ has compact,
hence closed, image $S_\Psi$ in a Hausdorff space \cite[Ch.~6]{kirillov_lie_groups}. The central
summand commutes with $\mathfrak s_\Psi$, giving
$\Pi_\Psi (\widetilde K_\Psi)=S_\Psi\exp(\mathfrak z_\Psi)$. Its closure is $S_\Psi T_\Psi$,
since this product is compact and every element of it is a limit of
elements of $\Pi_\Psi (\widetilde K_\Psi)$.

The intersection $S_\Psi\cap T_\Psi$ is contained in the center of
$S_\Psi$. This center is finite because it is compact with zero Lie algebra.
Hence multiplication $S_\Psi\times T_\Psi\to K_\Psi$ has finite kernel,
and its differential gives the stated Lie algebra direct sum. Its center
is $\mathfrak t_\Psi$, which equals
$\mathfrak z_\Psi\oplus\mathfrak c_\Psi$ by
\Cref{prop:kla-central-enlargement}. When $\mathfrak l_\Psi$ is semisimple,
$\mathfrak z_\Psi=0$, so $T_\Psi=\{I\}$ and the closure adds nothing.
\end{proof}

\begin{proposition}[Central frequencies and closedness]
\label{prop:kla-central-frequencies}
Let $Z_1,\ldots,Z_a$ be a real basis of $\mathfrak z_\Psi$ and write
\[
\mathcal K_\Psi=\bigoplus_{\alpha=1}^{q}V_\alpha,
\qquad
Z_b|_{V_\alpha}=i\omega_{\alpha b}I_{V_\alpha}
\]
for their distinct joint eigenspaces, where $\omega_{\alpha b}\in\mathbb R$.
Define
\[
\Omega=(\omega_{\alpha b})\in\mathbb R^{q\times a},\qquad
W=\operatorname{im}_{\mathbb R}\Omega,\qquad
\mathcal R=\ker(\Omega^T)\cap\mathbb Z^q,
\]
and let
\[
W^{\mathrm{rat}}
:=\{v\in\mathbb R^q:n\cdot v=0\text{ for every }n\in\mathcal R\}.
\]
Then
\[
\mathfrak t_\Psi
=\left\{i\bigoplus_{\alpha=1}^{q}v_\alpha I_{V_\alpha}
\;\middle|\;v\in W^{\mathrm{rat}}\right\}.
\]
Furthermore, $\Pi_\Psi (\widetilde K_\Psi)$ is closed if and only if $W=W^{\mathrm{rat}}$, or
equivalently if and only if
\[
\operatorname{span}_{\mathbb R} \{W\cap2\pi\mathbb Z^q\}=W.
\]
\end{proposition}

\begin{proof}
The central generators are commuting normal matrices and can therefore be
simultaneously diagonalized. In the resulting coordinates, they lie in the
Lie algebra of the scalar torus
\[
\left\{\bigoplus_{\alpha=1}^{q}e^{i\theta_\alpha}I_{V_\alpha}
\;\middle|\;\theta\in\mathbb R^q\right\}
\cong\mathbb R^q/(2\pi\mathbb Z^q).
\]
Identify the Lie algebra of
$T_\Psi=\overline{\exp(\mathfrak z_\Psi)}$
(see \Cref{cor:kla-central-torus}) with a subspace
$V\subset\mathbb R^q$ in these coordinates.
We first show that $W\subset V$. For any $w=\Omega c\in W$, let
\[
Z=\sum_{b=1}^a c_bZ_b
=i\bigoplus_{\alpha=1}^q w_\alpha I_{V_\alpha}
\in\mathfrak z_\Psi.
\]
The curve $t\mapsto\exp(tZ)$ lies in $T_\Psi$ and has derivative
$Z$ at the identity. Hence $Z\in\operatorname{Lie}(T_\Psi)$,
which gives $w\in V$.

Since $T_\Psi$ is compact, connected, and abelian, its exponential
map gives
\[
T_\Psi\cong V/(V\cap2\pi\mathbb Z^q),
\]
and the period lattice has full rank in $V$. Thus $V$ is rational.
Conversely, let $L\subset\mathbb R^q$ be any rational subspace
containing $W$. Its exponential image is a compact torus containing
$\exp(\mathfrak z_\Psi)$, and therefore also its closure $T_\Psi$.
Taking Lie algebras gives $V\subset L$. Consequently, $V$ is the
smallest rational subspace containing $W$.

A rational subspace is the common kernel of finitely many integer linear
forms. The integer forms vanishing on $W$ are exactly those with coefficient
vectors in $\mathcal R$, so their common kernel is $W^{\mathrm{rat}}$.
It is rational, contains $W$, and is contained in every rational subspace
containing $W$. This proves the formula for $\mathfrak t_\Psi$.
If $W$ is rational, then $\exp(\mathfrak z_\Psi)$ is already compact,
and \Cref{cor:kla-central-torus} shows that $\Pi_\Psi (\widetilde K_\Psi)$ is closed. Conversely,
if $\Pi_\Psi (\widetilde K_\Psi)$ is closed, then
$\widehat{\mathfrak l}_\Psi=\mathfrak l_\Psi$, so
$\mathfrak t_\Psi=\mathfrak z_\Psi$ and $W^{\mathrm{rat}}=W$.
The full period-lattice condition is the equivalent compactness criterion
for the quotient of $W$ by its periods.
\end{proof}

\begin{remark}
Proposition~\ref{prop:kla-central-frequencies} also gives a direct procedure for finding the
additional directions induced by taking the topological closure. First compute a real basis $L_1,\ldots,L_\ell$ of
$\mathfrak l_\Psi$, then solve $[Z,L_j]=0$ for $Z\in\mathfrak l_\Psi$ to
obtain a basis of its center. If the center is zero, the procedure is
complete. Otherwise, simultaneously diagonalize the central basis and form
$\Omega$ as above. One may then find integer relations spanning the real span of
$\mathcal R=\ker(\Omega^T)\cap\mathbb Z^q$. If these relations are the rows
of a matrix $R$, then $W^{\mathrm{rat}}=\ker_{\mathbb R}R$.

For each vector $v$ in a basis of $W^{\mathrm{rat}}$, form the matrix
$i\bigoplus_\alpha v_\alpha I_{V_\alpha}$ and subtract its
$\beta$-orthogonal projection onto $\mathfrak z_\Psi$. A maximal independent
collection of the resulting matrices gives a basis $iC_1,\ldots,iC_r$ of
$\mathfrak c_\Psi$. Thus we may still express the completed algebra by
ordinary Lie closure:
\[
\widehat{\mathfrak l}_\Psi
=\operatorname{Lie}\langle \widetilde H_1,\ldots,\widetilde H_n,C_1,\ldots,C_r\rangle,
\qquad [C_a,\widetilde H_j]=[C_a,C_b]=0,
\]
where $\widetilde H_j$ are the original compressed Hermitian generators.
The number of added directions is
$r=\dim W^{\mathrm{rat}}-\dim W$.

The only arithmetic step is the determination of the integer relations.
If the entries of $\Omega$ lie in a specified algebraic number field, this
reduces to rational linear algebra by expanding them in a rational field
basis and clearing denominators in a nullspace basis. Arbitrary numerical
approximations do not, however, certify all such relations. Thus the formula
is always an exact characterization, but an exact computation requires
sufficient information about the frequencies. 
\end{remark}

\begin{corollary}[Closedness for an irreducible KLA]
\label{cor:irreducible-kla-closed}
If $\mathfrak z_\Psi\subset i\mathbb R I_{\mathcal K_\Psi}$, then
\[
\Pi_\Psi(\widetilde K_\Psi)=K_\Psi,
\qquad
\widehat{\mathfrak l}_\Psi=\mathfrak l_\Psi.
\]
In particular, this holds whenever $\mathfrak l_\Psi$ acts irreducibly
on $\mathcal K_\Psi$ as a complex representation.
\end{corollary}

\begin{proof}
If the center is scalar, its exponential image is either trivial or
the scalar circle. It is therefore closed, and \Cref{cor:kla-central-torus}
gives the result. If the representation is irreducible, Schur's lemma
implies that every central element is scalar
\cite[Ch.~4]{kirillov_lie_groups}. Since these elements are
anti-Hermitian, $\mathfrak z_\Psi\subset i\mathbb R I_{\mathcal K_\Psi}$.
\end{proof}

\begin{corollary}[Preservation of the represented associative algebra]
\label{cor:kla-associative-preservation}
Let $\pi_\Psi$ and $\widehat{\pi}_\Psi$ denote the representations of
the real universal enveloping algebras
$U_{\mathbb R}(\mathfrak l_\Psi)$ and
$U_{\mathbb R}(\widehat{\mathfrak l}_\Psi)$ induced by their respective
matrix actions on $\mathcal K_\Psi$. Then
\[
\pi_\Psi\bigl(U_{\mathbb R}(\mathfrak l_\Psi)\bigr)
=
\widehat{\pi}_\Psi\bigl(U_{\mathbb R}
(\widehat{\mathfrak l}_\Psi)\bigr).
\]
Consequently, $\mathfrak l_\Psi$ and $\widehat{\mathfrak l}_\Psi$
have the same invariant subspaces and the same commutant on
$\mathcal K_\Psi$.
\end{corollary}

\begin{proof}
By the universal property of the universal enveloping algebra,
the matrix action of $\mathfrak l_\Psi$ extends uniquely to a unital
algebra homomorphism $\pi_\Psi$, whose image
\[
\mathcal A_\Psi
:=\pi_\Psi\bigl(U_{\mathbb R}(\mathfrak l_\Psi)\bigr)
\]
is the unital real associative algebra generated by
$\mathfrak l_\Psi$ \cite[Sec.~5.1]{kirillov_lie_groups}.
As a vector subspace of a finite-dimensional matrix algebra,
$\mathcal A_\Psi$ is closed. The exponential series therefore gives
$\exp(X)\in\mathcal A_\Psi$ for every $X\in\mathfrak l_\Psi$.
Since $\mathcal A_\Psi$ is closed under multiplication, it contains
$\Pi_\Psi(\widetilde K_\Psi)$ and hence $K_\Psi=\overline{\Pi_\Psi(\widetilde K_\Psi)}$.

Now let $X\in\widehat{\mathfrak l}_\Psi=\operatorname{Lie}(K_\Psi)$.
Since $\exp(tX)\in K_\Psi\subset\mathcal A_\Psi$ for every
$t\in\mathbb R$, closedness gives
\[
X=\lim_{t\to0}\frac{\exp(tX)-I}{t}\in\mathcal A_\Psi.
\]
Thus $\widehat{\mathfrak l}_\Psi$ generates an associative algebra
contained in $\mathcal A_\Psi$. The reverse inclusion follows from
$\mathfrak l_\Psi\subset\widehat{\mathfrak l}_\Psi$, proving the
stated equality. Finally, preserving a subspace or commuting with
every element of a Lie algebra is equivalent to doing so for its
generated unital associative algebra, which proves the remaining
claims.
\end{proof}

\begin{remark}
The completion does not alter the original compressed matrices or their
word identities. In particular, the additional central matrices should not
be assigned grade one in the original circuit filtration merely because
they appear in a generating set for the completed algebra. With this
convention, the prefix matching of \Cref*{thm:word-preservation-prefix} is unchanged. We may
therefore use the compact group as a reference while retaining the depth
information carried by the original Krylov construction.
\end{remark}

\begin{theorem}[Restricted DLA action on the Krylov--Lie algebra]
\label{thm:restricted-dla-action}
Let \(\mathfrak{l}_\Psi(\mathcal S)\) be a Krylov-Lie algebra in the sense of Definition \ref{def:krylov-lie-algebra}, with associated Krylov-Lie subspace \(\mathcal K_\Psi(\mathcal S)\subset \mathcal H\) (\Cref{def:krylov-lie-subspace}),
orthogonal projector \(P_\Psi\) (\Cref{def:KrylovAlgebra}), and compressed generators
\[
\widetilde H_j = P_\Psi H_j P_\Psi, \qquad j=1, \dots , n.
\]
Next, let
\[
\mathfrak g:=\operatorname{Lie} \langle H_1,\dots,H_n\rangle = \mathfrak z(\mathfrak g) \oplus [\mathfrak g, \mathfrak g] \subset \mathfrak u(d)
\]
be the reductive dynamical Lie algebra generated by the Hermitian circuit generators, with the
convention of \Cref{rem:VQAConventions}.

Then the following are equivalent:
\begin{enumerate}
\item
\[
H_j \mathcal K_\Psi \subset \mathcal K_\Psi
\qquad \forall \, j=1,\dots,n.
\]
\item Every \(X\in\mathfrak g\) preserves \(\mathcal K_\Psi\), so restriction defines a Lie algebra
homomorphism 
\[
\rho_\Psi:\mathfrak g\to \mathfrak u(\mathcal K_\Psi),\qquad \rho_\Psi(X)=X|_{\mathcal K_\Psi}.
\]
\item The restricted action of \(\mathfrak g\) on \(\mathcal K_\Psi\) is represented
by the Lie algebra homomorphism
\[
\rho_\Psi:\mathfrak g\to \mathfrak u(\mathcal K_\Psi),\qquad
\rho_\Psi(X)=P_\Psi X P_\Psi,
\]
and
\[
\rho_\Psi(\mathfrak g)=\mathfrak l_\Psi(\mathcal S).
\]
\end{enumerate}
Thus by the fact that $\mathfrak l_\Psi$ and $\mathfrak{g}$ are both reductive Lie algebras, $\mathfrak l_\Psi$ must contain an abelian subalgebra or a Lie subalgebra that is an image of the semisimple part of \(\mathfrak{g}\).
\end{theorem}

\begin{proof}
We prove \((1)\Rightarrow(2)\Rightarrow(3)\Rightarrow(1)\).

\((1)\Rightarrow(2)\):
Define
\[
\mathfrak I(\mathcal K_\Psi):=\{X\in\mathfrak u(d) \, \vert \, X(\mathcal K_\Psi)\subset \mathcal K_\Psi\}.
\]
This is a real Lie subalgebra of \(\mathfrak u(d)\): it is closed under real linear
combinations, and if \(X,Y\in \mathfrak I(\mathcal K_\Psi)\), then
\[
[X,Y](\mathcal K_\Psi)=XY(\mathcal K_\Psi)-YX(\mathcal K_\Psi)\subset \mathcal K_\Psi.
\]
Hence \([X,Y]\in \mathfrak I(\mathcal K_\Psi)\).

By hypothesis, \(H_j\in \mathfrak I(\mathcal K_\Psi)\) for all \(j\). Since \(\mathfrak g\) is the smallest real Lie
subalgebra of \(\mathfrak u(d)\) containing \(\{H_1,\dots,H_n\}\), it follows that
\[
\mathfrak g\subset \mathfrak I(\mathcal K_\Psi).
\]
Therefore every \(X\in\mathfrak g\) preserves \(\mathcal K_\Psi\), so restriction defines
\[
\rho_\Psi(X):=X|_{\mathcal K_\Psi}\in \mathfrak u(\mathcal K_\Psi).
\]
For \(X,Y\in\mathfrak g\), because both $X, Y$ preserve $\mathcal K_\Psi$ we note that $\forall \:v \in \mathcal{K}_\Psi$,
\[
[\rho_\Psi(X),\rho_\Psi(Y)](v) = \rho_\Psi(X)\rho_\Psi(Y)(v) - \rho_\Psi(Y)\rho_\Psi(X)(v) =X(Yv)-Y(Xv)=[X,Y]v.
\]
Then by definition,
\[
\rho_\Psi([X,Y])=[X,Y]|_{\mathcal K_\Psi}
=
[\rho_\Psi(X),\rho_\Psi(Y)],
\]
so \(\rho_\Psi\) is a Lie algebra homomorphism.

\((2)\Rightarrow(3)\):
Identify the projector $P_\Psi$ onto the Krylov subspace. The matrix of \(X|_{\mathcal K_\Psi}\) is
\(P_\Psi X P_\Psi\), so the restricted representation is given by
\[
\rho_\Psi(X)=P_\Psi X P_\Psi.
\]
In particular,
\[
\rho_\Psi(H_j)=P_\Psi H_j P_\Psi=\widetilde H_j.
\]
Hence \(\rho_\Psi(\mathfrak g)\) is a real Lie subalgebra containing
\(\widetilde H_1,\dots,\widetilde H_n\), so
\[
\mathfrak l _\Psi(\mathcal S)\subset \rho_\Psi(\mathfrak g).
\]
Conversely, since \(\rho_\Psi\) is a Lie algebra homomorphism and \(\mathfrak g\) is generated
by \(\{H_1,\dots,H_n\}\), the image \(\rho_\Psi(\mathfrak g)\) is generated by
\(\rho_\Psi(H_1),\dots,\rho_\Psi(H_n)\), i.e. by
\(\{\widetilde H_1,\dots,\widetilde H_n\}\). Therefore
\[
\rho_\Psi(\mathfrak g)\subset \mathfrak l _\Psi(\mathcal S),
\]
and it is seen that
\[
\rho_\Psi(\mathfrak g)=\mathfrak l_\Psi(\mathcal S).
\]
Moreover, since \(\rho_\Psi\) is a Lie algebra homomorphism of reductive Lie algebras (by \Cref{ReductiveKLA}), by the First Isomorphism Theorem and the fact that $\ker \rho_\Psi$ is an ideal, we have that 
\[
\rho_\Psi( \mathfrak g ) \cong \mathfrak z( \mathfrak g) / (\mathfrak z( \mathfrak g) \cap \ker  \rho_\Psi) \oplus [\mathfrak g, \mathfrak g] / ([\mathfrak g, \mathfrak g] \cap \ker \rho_\Psi ).  
\]
Thus, if $\ker \rho_\Psi $ does not contain the entire center of $\mathfrak g$, then $\mathfrak z( \mathfrak g) / (\mathfrak z( \mathfrak g) \cap \ker  \rho_\Psi)$ is isomorphic to an abelian subalgebra of $\mathfrak g$, and if $\ker \rho_\Psi $ does contain the entire center of $\mathfrak g$, then by standard theory of semisimple Lie algebras, the quotient $[\mathfrak g, \mathfrak g] / ([\mathfrak g, \mathfrak g] \cap \ker \rho_\Psi )$ is isomorphic to a semisimple ideal in $\mathfrak g$ \cite[§6.7]{kirillov_lie_groups}\cite[Ch.~7]{HallLieGroups2015}. That is to say, we have that $\mathfrak l _\Psi$ is isomorphic to an algebra that contains a nontrivial abelian subalgebra of $\mathfrak{g}$ or is itself isomorphic to a semisimple ideal of \(\mathfrak{g}\).

\((3)\Rightarrow(1)\):
Under (3), \(\rho_\Psi\) is the matrix representation of the restricted action of
\(\mathfrak g\) on \(\mathcal K_\Psi\). In particular, for each generator \(H_j\), the restriction
\(H_j|_{\mathcal K_\Psi}\) is well defined. Consequently,
\[
H_j \mathcal K_\Psi\subset \mathcal K_\Psi
\qquad \forall \, j=1,\dots,n.
\]
\end{proof}

\begin{remark}[Restriction versus compression]
\label{rem:restriction-vs-compression}
The hypothesis of Proposition~\ref{thm:restricted-dla-action} is substantially stronger
than merely having a Krylov-Lie state space \(\mathcal{K}_{\Psi}\subset\mathcal H\). If \(\mathcal{K}_{\Psi}\) fails to be
invariant under the original generators, one may still form compressed operators
\[
\widetilde H_j := P_{\Psi} H_j P_\Psi \in \operatorname{End}(\mathcal{K}_{\Psi}),
\]
where \(P_\Psi\) denotes the orthogonal projector onto \(\mathcal{K}_{\Psi}\), and then study the Lie algebra
generated by the \(\widetilde H_j\). However, this yields a \emph{compressed} or
\emph{reduced} Lie model on \(\mathcal{K}_{\Psi}\), not the restriction of the original DLA to an invariant
subspace.

The relevance of this hypothesis is that it educates us that compression does not in general preserve Lie brackets unless the subspace is invariant. Given the importance of dimension matching in constructing a suitable Lie group approximator to $M$, since invariance forces a homomorphism and necessarily that $\mathfrak{l}_\Psi$ is isomorphic to a subalgebra of $\mathfrak g$ by the proof of \Cref{thm:restricted-dla-action}, greatly restricting the possible dimension(s) of $\mathfrak{l}_\Psi$, we should not expect most usable Krylov-Lie algebras to possess this sort of invariance. Even so, there is a much weaker condition that will grant enough preservation of the Lie bracket under projection to derive very strong error bounds: being a prefix KLA. For involutive generators, seed-level commutator matching also suffices, as we will prove in \Cref{prop:involution-seed-word-matching} and \Cref{thm:seed-KLG-factorial}.
\end{remark}

\begin{theorem}[Word preservation for prefix Krylov-Lie algebras]
\phantomsection
\label{thm:word-preservation-prefix}
Let $\mathfrak p^{(k)}_\Psi$ be a prefix Krylov-Lie algebra with prefix Krylov subspace $\mathcal{P}^{(k)}_{\Psi}$ as in \Cref{def:prefix-subspace-kla}. Then for any prefix words $w(i_1, ..., i_r) \in \mathcal{W}_{\leq k}(\{H_1, \dots, H_n\})=\mathcal{W}_{\leq k}$ and $\widetilde w(i_1, ..., i_r) \in \widetilde{\mathcal{W}}_{\leq k}(\{\widetilde H_1, \dots, \widetilde H_n\})=\widetilde{\mathcal{W}}_{\leq k}$ where $r \leq k$ (see \Cref{def:CircWords}, \Cref{def:KrylovSubspace}, and \Cref{def:KrylovAlgebra}), we have for any seed $\psi \in \Psi$ that
\[
w(i_1, ..., i_r) \psi = \widetilde w(i_1, ..., i_r) \psi.
\]
Additionally, since any commutator word to grade $k$ can be written as a linear combination of prefixes of at most depth $k$, the same statement holds for commutator words.
\end{theorem}

\begin{proof}
Let $
P^{(k)}_\Psi \colon \mathcal{H} \to \mathcal{P}^{(k)}_\Psi$
denote the orthogonal projector onto the prefix Krylov subspace
\[
\mathcal{P}^{(k)}_\Psi
=
\operatorname{span}
\left\{
H_{i_r}\cdots H_{i_1}\psi
:
0 \leq r \leq k,\;
1 \leq i_1,\ldots,i_r \leq n,\;
\psi \in \Psi
\right\},
\]
where $r=0$ corresponds to the identity operator. Recall by \Cref{def:KrylovAlgebra} that
\[
\widetilde H_j:=P^{(k)}_\Psi H_jP^{(k)}_\Psi
\]
for every \(j=1,\ldots,n\).  Set
\[
w(i_1, ..., i_r)=H_{i_r}\cdots H_{i_1}
\in
\mathcal{W}_{\leq k},
\]
where \(0 \leq r \leq k\), and say
\[
\widetilde w(i_1, ..., i_r)
=
\widetilde H_{i_r}\cdots \widetilde H_{i_1}
\in
\widetilde{\mathcal{W}}_{\leq k}
\]
is the corresponding compressed prefix word with the same ordered index sequence.  For \(0 \leq a \leq r\), write
\[
v_a:=H_{i_a}\cdots H_{i_1}\psi,
\]
where \(v_0:=\psi\), and define $\widetilde v_a$ analogously.

By the definition of \(\mathcal{P}^{(k)}_\Psi\), each prefix vector
\(v_a\) belongs to \(\mathcal{P}^{(k)}_\Psi\), since \(a\leq r\leq k\).
Therefore,
\[
P^{(k)}_\Psi v_a=v_a
\]
for every \(0\leq a\leq r\). We now proceed by induction on \(a\).  For \(a=0\), we have $\widetilde v_0=\psi=v_0$. Next, suppose that
\[
\widetilde v_{a-1}=v_{a-1}
\]
for some \(1\leq a\leq r\).  Then
\[
\begin{aligned}
\widetilde v_a
&=
\widetilde H_{i_a}\widetilde v_{a-1} \\
&=
P^{(k)}_\Psi H_{i_a} P^{(k)}_\Psi \widetilde v_{a-1} \\
&=
P^{(k)}_\Psi H_{i_a} P^{(k)}_\Psi v_{a-1} \\
&=
P^{(k)}_\Psi H_{i_a}v_{a-1} \\
&=
P^{(k)}_\Psi v_a \\
&=
v_a.
\end{aligned}
\]
By induction, $\widetilde v_r=v_r$, proving the first claim by \Cref{def:CircWords}. From this identity the final claim follows after expanding grade-$k$ commutator words into finite linear combinations of prefixes with the formula
\[
[A,B]=AB-BA.
\]
\end{proof}

\section{Dimension-Matching and Seed Dependence}\label{sec:dim-matching-seeds}

To elucidate some of the ever-essential dimensional structure of Krylov-Lie algebras and groups, we shall spend a reasonable portion of this text on the algebraic geometry of Krylov-Lie subspaces and Krylov-Lie algebras. Although this sort of language is uncommon in the VQA literature, our arguments rely only on basic facts about polynomial or real-analytic maps, algebraic sets, and genericity, and an introductory exposure to the subject like in the book of Beltrametti et al. or that of Harris \cite{BeltramettiAlgebraicGeometry2009, Harris1992} should prove sufficient to understand the contents of this section.

Before studying seed dependence, we must account for the change in dimension induced by compactification. Although the additional directions are central, they still contribute to the dimension of the group used as an approximator. For this comparison, $\mathfrak l_\Psi$ denotes the original generated algebra and $\widehat{\mathfrak l}_\Psi=\operatorname{Lie}(K_\Psi)$ its completion, as in \Cref{prop:kla-central-enlargement}. The central-frequency description gives both an exact dimension formula and bounds that do not require finding every integer relation.

\begin{corollary}[Dimension of the Krylov-Lie group]
\label{cor:klg-completed-dimension}
In the notation of \Cref{prop:kla-central-frequencies},
\[
\dim K_\Psi
=\dim\widehat{\mathfrak l}_\Psi
=\dim\mathfrak l_\Psi-\dim\mathfrak z_\Psi
+q-\operatorname{rank}_{\mathbb Z}\mathcal R.
\]
In particular,
\[
0\leq\dim K_\Psi-\dim\mathfrak l_\Psi
\leq q-\dim\mathfrak z_\Psi
\leq\dim\mathcal K_\Psi-\dim\mathfrak z_\Psi.
\]
Moreover,
\[
\dim K_\Psi-\dim\mathfrak l_\Psi\leq\dim\mathcal K_\Psi-1,
\qquad
\dim K_\Psi\leq\sum_{\alpha=1}^{q}(\dim V_\alpha)^2
\leq(\dim\mathcal K_\Psi)^2.
\]
\end{corollary}

\begin{proof}
The integer rank of $\mathcal R$ equals the real dimension of its span.
Thus \Cref{prop:kla-central-frequencies} gives
$\dim\mathfrak t_\Psi=q-\operatorname{rank}_{\mathbb Z}\mathcal R$.
Combining this with
$\widehat{\mathfrak l}_\Psi=\mathfrak s_\Psi\oplus\mathfrak t_\Psi$
and $\mathfrak l_\Psi=\mathfrak s_\Psi\oplus\mathfrak z_\Psi$
from \Cref{cor:kla-central-torus} proves the dimension formula.
Since $\mathfrak z_\Psi\subset\mathfrak t_\Psi$ and the $q$ joint
eigenspaces are nonzero, the first inequalities follow. If the center
vanishes, the completion adds nothing. Otherwise $\dim\mathfrak z_\Psi\geq1$,
which proves the bound on the added dimension.

Every element of $\mathfrak l_\Psi$ commutes with its center and therefore
preserves each $V_\alpha$. The same holds for its exponentials and their
closure, so $K_\Psi\subset\prod_{\alpha=1}^{q}U(V_\alpha)$. Taking
dimensions proves the remaining bounds.
\end{proof}

Thus the enlargement is at most linear in the Krylov-subspace dimension, even though the full KLG dimension can be quadratic. We next study the dimensions of the Krylov subspace and the original KLA, then identify conditions under which these results determine the KLG dimension. In dimension-matching arguments, the dimension of the chosen group is always $\dim K_\Psi^{(k)}=\dim\widehat{\mathfrak l}_\Psi^{(k)}$.

Following this excursion, we shall also derive some generic bounds for the grade-$k$ Krylov-Lie subspace dimension as well as the KLA dimension with Witt's formula.

\begin{theorem}[Determinantal stratification of seed-dependent Krylov--Lie dimensions]
\label{thm:determinantal-krylov-dim}
Fix a depth \(k\), a generator set \(\mathcal S=\{H_1,\dots,H_n\}\subset \mathrm{End}(\mathbb C^d)\),
and a block size \(s\ge 1\) (\Cref{def:seeds-and-block-seeds} and \Cref{def:krylov-lie-algebra}). Let \(\mathcal W_{\le k}\) be the finite set of Lie words in \(\mathcal S\) of
depth at most \(k\) (\Cref{def:CircWords}) including $\mathcal W_0=I$ (\Cref{rmk:words-convention}), and define
\[
\mathcal E_k:(\mathbb C^d\setminus\{0\})^s \to \mathrm{Mat}_{d\times sN_k}(\mathbb C)
\]
by
\[
\mathcal E_k(\Psi)
=
\bigl[w(\psi_a)\bigr]_{w\in\mathcal W_{\le k},\,1\le a\le s},
\qquad
\Psi=(\psi_1,\dots,\psi_s),
\]
where \(N_k=\lvert \mathcal W_{\le k}\rvert\) and $N_0=\lvert \mathcal W_{0}\rvert=1$. Then for every \(r\ge 0\), the locus
\[
X_{\le r}
:=
\{\Psi\in(\mathbb C^d\setminus\{0\})^s \, \vert \,\dim \mathcal K^{(k)}_\Psi(\mathcal S)\le r\}
\]
is Zariski closed relative to the block-seed space, cut out by the vanishing of all
\((r+1)\times(r+1)\) minors of \(\mathcal E_k(\Psi)\). Equivalently, the locus
\[
X_{\ge r}
:=
\{\Psi\in(\mathbb C^d\setminus\{0\})^s \, \vert \,\dim \mathcal K^{(k)}_\Psi(\mathcal S)\ge r\}
\]
is Zariski open.
\end{theorem}

\begin{proof}
By construction,
\[
\dim \mathcal K^{(k)}_\Psi(\mathcal S)=\operatorname{rank}\mathcal E_k(\Psi).
\]
A matrix has rank at most \(r\) if and only if all of its \((r+1)\times(r+1)\) minors vanish.
Since each entry of \(\mathcal E_k(\Psi)\) is polynomial in the coordinates of \(\Psi\), each such
minor is a polynomial on \((\mathbb C^d)^s\). Hence
\[
X_{\le r}
\]
is the common zero locus of these minors, and therefore is Zariski closed, as is $X_{\le r-1}$ \cite[Ch.~1]{BeltramettiAlgebraicGeometry2009}\cite[Lecture~1]{Harris1992}. The complement of the latter,
\(X_{> r-1}=X_{\geq r}\), is Zariski open.
\end{proof}

\begin{corollary}[Positive-measure genericity]
\label{cor:positive-measure-genericity}
In the setting of Theorem~\ref{thm:determinantal-krylov-dim}, suppose there exists a block
seed \(\Psi_0\in (\mathbb C^d\setminus\{0\})^s\) (see \Cref{def:seeds-and-block-seeds}) such that
\[
\dim \mathcal K^{(k)}_{\Psi_0}(\mathcal S)\ge r,
\]
where $\mathcal K^{(k)}_\Psi(\mathcal S)$ is the Krylov-Lie subspace of \Cref{def:krylov-lie-subspace}. Then the set
\[
X_{\ge r}
=
\{\Psi\in(\mathbb C^d\setminus\{0\})^s \, \vert \,\dim \mathcal K^{(k)}_\Psi(\mathcal S)\ge r\}
\]
contains a nonempty Zariski open subset of \((\mathbb C^d)^s\). In particular,
\(X_{\ge r}\) has positive Lebesgue measure in \((\mathbb C^d)^s\cong \mathbb R^{2ds}\) \cite[§2.8]{bochnak1998realalgebraicgeometry}\cite{sottile2016realalgebraicgeometry}.
\end{corollary}

\begin{proof}
If there exists \(\Psi_0\) with
\(\operatorname{rank}\mathcal E_k(\Psi_0)\ge r\),
then some \(r\times r\) minor of \(\mathcal E_k\) is nonzero at \(\Psi_0\). Hence that minor is
not the zero polynomial, and its nonvanishing locus is a nonempty Zariski open subset
contained in \(X_{\ge r}\) \cite[Ch.~1]{BeltramettiAlgebraicGeometry2009}.

A proper complex algebraic subset of \(\mathbb C^N\) has empty Euclidean interior and, being
contained locally in the zero set of a nonzero holomorphic polynomial, has Lebesgue measure
zero \cite[§2.8]{bochnak1998realalgebraicgeometry}\cite{sottile2016realalgebraicgeometry}. Therefore a nonempty Zariski open subset of \(\mathbb C^N\) has full measure in some
Euclidean ball, in particular positive Lebesgue measure \cite[§2.8]{bochnak1998realalgebraicgeometry}\cite{sottile2016realalgebraicgeometry}.
\end{proof}

\begin{theorem}[Generic maximal dimension]
\label{thm:generic-maximal-dimension}
Fix $k$, $\mathcal S$, and $s$ as in \Cref{cor:positive-measure-genericity}, and set
\[
r_{\max}:=\max_{\Psi\in(\mathbb C^d\setminus\{0\})^s}\dim \mathcal K^{(k)}_\Psi(\mathcal S),
\]
where $\mathcal K^{(k)}_\Psi(\mathcal S)$ is the Krylov-Lie subspace of \Cref{def:krylov-lie-subspace}. Then the set
\[
X_{\max}
:=
\{\Psi\in(\mathbb C^d\setminus\{0\})^s \, \vert \,\dim \mathcal K^{(k)}_\Psi(\mathcal S)=r_{\max}\}
\]
is a nonempty Zariski open subset of \((\mathbb C^d)^s\). In particular, \(X_{\max}\) is
Euclidean dense and has positive Lebesgue measure.
\end{theorem}

\begin{proof}
    Clearly $r_{\operatorname{max}}$ exists since dimension only assumes integer values and $\dim \mathcal {K}^{(k)}_\Psi(\mathcal S) \leq d$, where $d$ is the dimension of our starting Hilbert space. Then by Corollary \ref{cor:positive-measure-genericity}, as at least one $\Psi$ must give a dimension of $r_{\operatorname{max}}$ for it to be the maximal dimension, we see that 
    \[
    X_\text{max} =\{\Psi\in(\mathbb C^d\setminus\{0\})^s \, \vert \,\dim \mathcal K^{(k)}_\Psi(\mathcal S) \geq r_{\max}\} = \{\Psi\in(\mathbb C^d\setminus\{0\})^s \, \vert \,\dim \mathcal K^{(k)}_\Psi(\mathcal S)=r_{\max}\}
    \]
    has positive Lebesgue measure.
\end{proof}

\begin{proposition}[Determinantal seed stratification]
\label{prop:seed-stratification}
Fix a depth \(k\), a generator set \(\mathcal S\subset \mathrm{End}(\mathbb C^d)\), and a block size
\(s\ge 1\) (see \Cref{def:krylov-lie-algebra} and \Cref{def:seeds-and-block-seeds}). Let \(\mathcal E_k(\Psi)\) be the polynomial evaluation matrix associated to the
depth-\(k\) Krylov-Lie subspace (\Cref{def:krylov-lie-subspace}). Then for every \(r\ge 0\),
\[
X_{=r}:=\{\Psi \in(\mathbb C^d\setminus\{0\})^s \, \vert \, \dim \mathcal K_\Psi^{(k)}(\mathcal S)=r\}
\]
is locally closed.
\end{proposition}

\begin{proof}
As before,
\[
\dim \mathcal K_\Psi^{(k)}(\mathcal S)=\operatorname{rank}\mathcal E_k(\Psi).
\]

Fix \(r\ge 0\). Then define
\[
X_{=r}
=
\{\Psi \in(\mathbb C^d\setminus\{0\})^s \, \vert \,\dim \mathcal K_\Psi^{(k)}(\mathcal S)=r\}
=
\{\Psi \in(\mathbb C^d\setminus\{0\})^s \, \vert \,\operatorname{rank}\mathcal E_k(\Psi)=r\}.
\]
By Theorem \ref{thm:determinantal-krylov-dim} \(X_{\le r-1}\) is Zariski closed, where for \(r=0\) we interpret \(X_{\le -1}=\varnothing\). Therefore
\[
X_{=r}
=
X_{\le r}\setminus X_{\le r-1}.
\]
Thus \(X_{=r}\) is the difference of two Zariski-closed sets and is locally closed \cite[Lecture~3]{Harris1992}.
\end{proof}

\begin{corollary}[Genericity inside a rank-drop closure] \label{cor:generic-on-support}
If \(X_{=r}\) is nonempty, then \(X_{=r}\) contains a nonempty Zariski open subset of its
closure in the Zariski topology \(\overline{X_{=r}}^{\,\mathrm{Zar}}\). Equivalently, rank \(r\) is generic on its support \cite[Lecture~3]{Harris1992}.
\end{corollary}

\begin{proof}
By \Cref{prop:seed-stratification}, the rank-\(r\) stratum
\[
X_{=r}
=
\{\Psi\in(\mathbb C^d\setminus\{0\})^s \, \vert \, \dim \mathcal K_\Psi^{(k)}(\mathcal S)=r\}
\]
is locally closed.

A well-known characterization of locally closed subsets states that a subset is locally closed if and only if it is open in its Zariski closure \cite[Lecture~3]{Harris1992}. Therefore \(X_{=r}\) is Zariski open in \(\overline{X_{=r}}^{\,\mathrm{Zar}}\). 

If \(X_{=r}\neq\varnothing\), then this open subset is nonempty. Hence \(X_{=r}\) contains a nonempty Zariski open subset of its closure, namely itself viewed as an open subset of \(\overline{X_{=r}}^{\,\mathrm{Zar}}\). That is to say, rank \(r\) is generic on its support, in the sense that it holds on a nonempty Zariski open subset of \(\overline{X_{=r}}^{\,\mathrm{Zar}}\). 
\end{proof}

\begin{remark}
Corollary~\ref{cor:generic-on-support} does \emph{not} imply that
\(X_{=r}\) has positive Lebesgue measure in the full ambient space \(\mathcal H^s \setminus \{ 0\}\).
If \(\overline{X_{=r}}^{\,\mathrm{Zar}}\) is a proper algebraic subset of \(\mathcal H^s \setminus \{ 0\}\), then it has
Lebesgue measure zero, and so does \(X_{=r}\) \cite[§2.8]{bochnak1998realalgebraicgeometry}\cite{sottile2016realalgebraicgeometry}. What the corollary asserts is instead that
rank-\(r\) is Zariski-generic after restricting to the natural algebraic locus on which that
rank can occur \cite[Lecture~3]{Harris1992}\cite[§2.8]{bochnak1998realalgebraicgeometry}.
\end{remark}

The corresponding statements for compressed Lie algebras require real algebraic geometry. The projector $B_\Psi(B_\Psi^\dagger B_\Psi)^{-1}B_\Psi^\dagger$ depends on the conjugate seed coordinates, so its entries are rational functions of their real and imaginary parts. For the remainder of this section, Zariski topology and constructibility therefore refer to these real coordinates, unless stated otherwise. We also take the generators to be Hermitian, with real Lie spans understood in their anti-Hermitian realization as in \Cref{rem:VQAConventions}.

\begin{lemma}[Local algebraicity of the KLA dimension]
\label{lem:kla-local-algebraicity}
Fix $k$, $\mathcal S$, $s$, and $r\geq1$. Let $U\subset X_{=r}$ be a nonempty relatively Zariski-open set on which a fixed nonvanishing $r\times r$ minor of $\mathcal E_k(\Psi)$ determines a basis matrix $B_\Psi$ for $\mathcal K_\Psi^{(k)}$. Then the real and imaginary entries of the compressed generators
\[
\widetilde H_j(\Psi)=P_\Psi^{(k)}H_jP_\Psi^{(k)}\in\End(\mathcal H)
\]
are rational functions of the real seed coordinates on $U$. Moreover, for every $\ell$, the locus
\[
U_{\leq\ell}:=\{\Psi\in U\:\vert\:\dim\mathfrak l_\Psi^{(k)}(\mathcal S)\leq\ell\}
\]
is Zariski closed in $U$. The locus of maximal KLA dimension on $U$ is Zariski open.
\end{lemma}

\begin{proof}
The selected columns of $\mathcal E_k(\Psi)$ are polynomial in the seed coordinates and have full column rank on $U$. Hence
\[
B_\Psi^+=(B_\Psi^\dagger B_\Psi)^{-1}B_\Psi^\dagger,
\qquad
P_\Psi^{(k)}=B_\Psi B_\Psi^+
\]
have real and imaginary entries rational in the real seed coordinates, with denominator nonzero on $U$. The same follows for each $\widetilde H_j(\Psi)$.

Starting with the real span of $\widetilde H_1(\Psi),\dots,\widetilde H_n(\Psi)$, successively adjoin commutators with these generators. If a step does not increase the span, the Jacobi identity implies that all further Lie words are already contained in it. Since the generated algebra has dimension at most $r^2$, words of length at most $r^2$ suffice. Let $\mathcal V_r$ be the finite set of these words in the fixed $n$ generators, and form
\[
\mathcal L_U(\Psi)
:=\left[\operatorname{vec}_{\mathbb R}
\left(w(\widetilde H_1(\Psi),\dots,\widetilde H_n(\Psi))\right)\right]_{w\in\mathcal V_r},
\]
where $\operatorname{vec}_{\mathbb R}$ stacks the real and imaginary matrix entries. The ambient matrices vanish on $(\mathcal K_\Psi^{(k)})^\perp$, so restriction to $\mathcal K_\Psi^{(k)}$ does not change their real linear dependence. Thus
\[
\dim\mathfrak l_\Psi^{(k)}(\mathcal S)
=\operatorname{rank}_{\mathbb R}\mathcal L_U(\Psi).
\]
Multiplying by a common denominator $D(\Psi)$, nonzero on $U$, gives a matrix with real polynomial entries and the same rank. Its $(\ell+1)\times(\ell+1)$ minors cut out $U_{\leq\ell}$, proving closedness. Taking the complement at one less than the maximal dimension proves the final assertion.
\end{proof}

\begin{corollary}[Local KLA dimension strata] \label{cor:KLA-generic-rank-drop}
Under the hypotheses of \Cref{lem:kla-local-algebraicity}, the locus
\[
U_{=\ell}:=\{\Psi\in U \, \vert \, \dim \mathfrak l_\Psi^{(k)}(\mathcal S)= \ell\}
\]
is locally closed in the Zariski topology and open in its Zariski closure.
\end{corollary}

\begin{proof}
By \Cref{lem:kla-local-algebraicity}, $U_{=\ell}=U_{\leq\ell}\setminus U_{\leq\ell-1}$ is locally closed. As in \Cref{cor:generic-on-support}, it is therefore open in its Zariski closure.
\end{proof}

\begin{theorem}[Global constructibility on a fixed Krylov stratum]\label{thm:global-krylov-constructability}
Fix \(r\ge 0\), and define
\[
Y_{r,\ell}:=\{\Psi\in (\mathbb C^d\setminus\{0\})^s \, \vert \,
\dim \mathcal K_\Psi^{(k)}(\mathcal S)=r,\ 
\dim \mathfrak l_\Psi^{(k)}(\mathcal S)=\ell\},
\]
where $\mathcal K_\Psi^{(k)}$ and $\mathfrak l_\Psi^{(k)}$ are the Krylov-Lie subspace and algebra of \Cref{def:krylov-lie-subspace,def:krylov-lie-algebra}. Then each $Y_{r,\ell}$ is constructible in the real seed space $\mathbb R^{2ds}$. If $X_{=r}$ is nonempty, its locus of maximal KLA dimension is relatively Zariski open. On each irreducible component of $X_{=r}$, that component's maximal KLA dimension holds on a nonempty relatively Zariski-open dense subset.
\end{theorem}

\begin{proof}
The empty-stratum case is immediate. Otherwise, cover $X_{=r}$ by the finitely many relative open sets $U_\alpha$ on which a fixed $r\times r$ minor of $\mathcal E_k(\Psi)$ is nonzero. The real and imaginary parts of these minors are polynomial in the real seed coordinates. On each $U_\alpha$, \Cref{lem:kla-local-algebraicity,cor:KLA-generic-rank-drop} show that the KLA dimension sublevel sets are closed and its level sets are locally closed.

Since $X_{=r}$ is locally closed in the real seed space, each $(U_\alpha)_{=\ell}$ is locally closed there as well. Thus
\[
Y_{r,\ell}=\bigcup_\alpha (U_\alpha)_{=\ell}
\]
is constructible. The sublevel sets are closed relative to $X_{=r}$ because this holds on its open cover, so its maximal-dimension locus is open. Restricting to an irreducible component gives the same conclusion for that component's maximum. This open set is nonempty because the maximum is attained, and is therefore dense in the component.
\end{proof}

\begin{remark}
The maximal Krylov-rank stratum $X_{\max}$ is a nonempty Zariski-open subset of the real seed space and is irreducible. Its maximal KLA dimension therefore holds outside a proper real algebraic subset, hence for almost every block seed with respect to Lebesgue measure \cite[§2.8]{bochnak1998realalgebraicgeometry}. This is the maximum on $X_{\max}$; on lower-dimensional strata, the preceding genericity statements are relative to their irreducible components.
\end{remark}

These results describe the original KLA dimension. To apply them to KLGs, it suffices to find a family on which the represented groups are already closed. By \Cref{cor:irreducible-kla-closed}, irreducible representations give one such family. Their seed dependence is also governed by a finite rank condition.

\begin{proposition}[Generic closedness from an irreducible seed]
\label{prop:generic-irreducible-kla-closed}
Let $X$ be an irreducible component of the fixed Krylov-rank stratum $X_{=r}$. Suppose some $\Psi_0\in X$ gives a KLA $\mathfrak l_{\Psi_0}^{(k)}(\mathcal S)$ acting irreducibly on $\mathcal K_{\Psi_0}^{(k)}$. Then the same holds on a nonempty relatively Zariski-open dense subset of $X$, and on this subset
\[
\widehat{\mathfrak l}_\Psi^{(k)}=\mathfrak l_\Psi^{(k)},
\qquad
\dim K_\Psi^{(k)}=\dim\mathfrak l_\Psi^{(k)}.
\]
\end{proposition}

\begin{proof}
For a unitary representation, irreducibility is equivalent to the complex commutant consisting only of scalars. One direction is Schur's lemma. Conversely, a proper invariant subspace has an invariant orthogonal complement, so its orthogonal projector is a nonscalar element of the commutant.

On a basis chart from \Cref{lem:kla-local-algebraicity}, the compressed generators have matrices $B_\Psi^+H_jB_\Psi$ in the basis $B_\Psi$. Thus their commutant is the solution space of
\[
[B_\Psi^+H_jB_\Psi,A]=0,
\qquad j=1,\dots,n,
\qquad A\in\mathrm{Mat}_{r\times r}(\mathbb C).
\]
The coefficient matrix of this linear system has entries rational in the real seed coordinates. Its kernel contains the identity, so the commutant is scalar precisely when the matrix has complex rank $r^2-1$. After clearing denominators, nonvanishing of a minor of this size is a real Zariski-open condition. Here a complex minor is nonzero if and only if the sum of the squares of its real and imaginary parts is nonzero.

The finitely many basis charts therefore give a relatively open irreducible-representation locus in $X$. It is nonempty by the hypothesis and hence dense because $X$ is irreducible. The closedness assertion follows from \Cref{cor:irreducible-kla-closed}.
\end{proof}

Intersecting this locus with the maximal KLA-dimension locus of \Cref{thm:global-krylov-constructability} gives a nonempty relatively open dense set on which the KLG dimension equals the maximal original KLA dimension on $X$.

For a general family of seeds, \Cref{cor:klg-completed-dimension} instead requires the integer relations among its central frequencies. Each fixed relation can be tested by an analytic equation when the frequencies vary analytically. Since there are countably many integer relations, this yields the following measure-theoretic genericity statement.

\begin{theorem}[Generic maximal KLG dimension]
\label{thm:generic-completed-klg-dimension}
Fix $k$, $\mathcal S$, and $s$. Let $U\subset(\mathbb C^d\setminus\{0\})^s$ be a nonempty connected real-analytic submanifold on which
\[
\dim\mathcal K_\Psi^{(k)},\qquad
\dim\mathfrak l_\Psi^{(k)},\qquad
\dim\mathfrak z(\mathfrak l_\Psi^{(k)})
\]
are constant. Suppose the central-frequency matrix $\Omega(\Psi)$ of \Cref{prop:kla-central-frequencies} can be chosen with fixed size and real-analytic entries on $U$. Then $\dim K_\Psi^{(k)}$ is maximal on $U$ outside a countable union of proper real-analytic zero sets. In particular, it is maximal for almost every seed with respect to smooth volume on $U$.
\end{theorem}

\begin{proof}
Let $q$ be the number of rows of $\Omega(\Psi)$, as in \Cref{prop:kla-central-frequencies}, and consider the integer relations that hold throughout $U$:
\[
\mathcal R_*:=\bigcap_{\Psi\in U}
\bigl(\ker(\Omega(\Psi)^T)\cap\mathbb Z^q\bigr).
\]
For each $n\in\mathbb Z^q\setminus\mathcal R_*$, the real-analytic function $\|\Omega(\Psi)^Tn\|^2$ is not identically zero. Since $U$ is connected, its zero set
\[
E_n:=\{\Psi\in U\:\vert\:\Omega(\Psi)^Tn=0\}
\]
has empty interior and smooth volume zero \cite{Mityagin2015}. There are countably many such $n$, so their zero sets have a null union. Outside this union,
\[
\ker(\Omega(\Psi)^T)\cap\mathbb Z^q=\mathcal R_*.
\]
At every other seed the relation lattice contains $\mathcal R_*$ and therefore cannot have smaller rank. All other terms in the dimension formula of \Cref{cor:klg-completed-dimension} are constant on $U$, so the completed dimension outside the union is maximal.
\end{proof}

The exceptional locus is controlled by countably many integer relations. A general algebraic classification of these relations and their seed loci is left for future work.

These algebro-geometric results establish generic maximal dimension on the stated families, while still allowing smaller dimensions at exceptional seeds. Varying the seeds and the commutator depth therefore provides possible means of finding dimension-matched approximators, though existence at a prescribed dimension requires further information about the generators. We next derive universal upper bounds for each depth. By \Cref{cor:klg-completed-dimension}, the resulting quadratic bounds also apply to the compact KLGs.

\begin{proposition}[Universal upper bound from Witt's formula]
\label{prop:witt-upper-bound-krylov-subspace}
Let \(\mathcal H\cong \mathbb C^d\), let
\[
\mathcal S=\{H_1,\dots,H_n\}\subset \operatorname{End}(\mathcal H),
\]
let \(\Psi=(\psi_1,\dots,\psi_s)\in \mathcal H^s \setminus\{0\}\) be a block seed (\Cref{def:seeds-and-block-seeds}), and let \(k\ge 1\).
For each integer \(j\ge 1\), define
\[
\ell_n(j):=\frac{1}{j}\sum_{e\mid j}\mu(e)\,n^{j/e},
\]
where \(\mu:\mathbb N\to\{-1,0,1\}\) is the Möbius function, i.e.
\[
\mu(e)=
\begin{cases}
1 & e=1,\\
(-1)^r& e \text{ is a product of } r \text{ distinct primes},\\
0& e \text{ is divisible by the square of a prime}.
\end{cases}
\]
Then \(\ell_n(j)\) is the dimension of the homogeneous degree-\(j\) component of the free Lie
algebra on \(n\) generators, by Witt's formula \cite{Petrogradsky2000}. Consequently,
\[
\dim\mathcal K_\Psi^{(k)}(\mathcal S)
\le
\min\!\left\{
d,\;
s\left(1+\sum_{j=1}^k \ell_n(j)\right)
\right\},
\]
where $\mathcal K_\Psi^{(k)}$ is our Krylov-Lie subspace of \Cref{def:krylov-lie-subspace}.
\end{proposition}

\begin{proof}
The grade-\(k\) Krylov--Lie subspace is spanned by the seeds in $\Psi$ as well as the vectors \(w(\psi_a)\), where \(w\) ranges over all Lie words in the \(n\) generators of bracket depth at most \(k\), and
\(a=1,\dots,s\). Modulo the universal Lie identities, the space of formal Lie words of homogeneous degree \(j\) has dimension \(\ell_n(j)\) by Witt's formula \cite{Petrogradsky2000}. Therefore, for each
fixed seed vector \(\psi_a\), there are at most \(1+\sum_{j=1}^k \ell_n(j)\) linearly independent
word evaluations \(w(\psi_a)\). Summing over the \(s\) seed vectors and using that
\(\mathcal K_\Psi^{(k)}(\mathcal S)\subset \mathcal H\cong \mathbb C^d\), we obtain
\[
\dim \mathcal K_\Psi^{(k)}(\mathcal S)
\le
\min\!\left\{
d,\;
s \left(1+\sum_{j=1}^k \ell_n(j) \right)
\right\}.
\]
\end{proof}

\begin{corollary}
[Universal KLA upper bound from Witt's formula]
\label{cor:witt-upper-bound-kla}
Let \(\Psi\in \mathcal H^s \setminus\{0\}\), let \(k\ge 1\), and set
\[
r_\Psi^{(k)}:=\dim\mathcal K_\Psi^{(k)}(\mathcal S),
\]
where $\mathcal K_\Psi^{(k)}$ is the Krylov-Lie subspace of \Cref{def:krylov-lie-subspace}. Then
\[
\dim \mathfrak l_\Psi^{(k)}(\mathcal S)
\le (r_\Psi^{(k)})^2,
\]
where $\mathfrak l_\Psi^{(k)}$ is the Krylov-Lie algebra of \Cref{def:krylov-lie-algebra}, and more loosely,
\[
\dim \mathfrak l_\Psi^{(k)}(\mathcal S)
\le
\min\!\left\{
d^2,\;
\left(s \left( 1 + \sum_{j=1}^k \ell_n(j) \right)\right)^2
\right\},
\]
where
\[
\ell_n(j)=\frac{1}{j}\sum_{e\mid j}\mu(e)\,n^{j/e}
\]
is given by Witt's formula for the free Lie algebra on \(n\) generators \cite{Petrogradsky2000}.
\end{corollary}

\begin{proof}
Immediate from \Cref{prop:witt-upper-bound-krylov-subspace}.
\end{proof}

Because matrix multiplication is not as constrained as the Lie bracket, which are compelled to obey antisymmetry and the Jacobi identity, the number of independent directions available at a fixed commutator grade can be substantially smaller than the number of arbitrary prefixes of the same depth. Therefore we must derive separate dimension bounds for the Krylov prefix subspace and the prefix KLA.

\begin{proposition}[Universal upper bound for the Krylov prefix subspace]
\phantomsection
\label{prop:prefix-upper-bound-krylov-subspace}
Let \(\mathcal H\cong\mathbb C^d\), let
\[
\mathcal S=\{H_1,\dots,H_n\}
\subset
\operatorname{End}(\mathcal H),
\]
let $\Psi=(\psi_1,\dots,\psi_s) \in \mathcal H^s\setminus\{0\}$ be a block seed, and let \(k\geq 0\). Further, say that $\mathcal P_\Psi^{(k)}$ is the Krylov prefix subspace of \Cref{def:prefix-subspace-kla} with $\mathfrak p_\Psi^{(k)}$ as its prefix KLA. Then if $n>1$
\[
\dim\mathcal P_\Psi^{(k)}(\mathcal S)
\leq
\min\left\{
d,\;
s\sum_{j=0}^k n^j
\right\} = \min\left\{
d,\;
s\frac{n^{k+1}-1}{n-1}
\right\},
\]
and if \(n=1\), then
\[
\dim\mathcal P_\Psi^{(k)}(\mathcal S)
\leq
\min\left\{
d,\;
s(k+1)
\right\}.
\]
\end{proposition}

\begin{proof}
By definition, taking a $j$ of $0$ to mean the identity operator,
\[
\mathcal P_\Psi^{(k)}(\mathcal S)
=
\operatorname{span}
\left\{
H_{i_j}\cdots H_{i_1}\psi_a
:
0\leq j\leq k,\;
1\leq i_1,\dots,i_j\leq n,\;
1\leq a\leq s
\right\}.
\]
By elementary combinatorics, for each fixed depth \(j\) there are exactly \(n^j\) ordered prefixes of length \(j\) in the \(n\) generators. Hence, for each fixed seed vector \(\psi_a\), the Krylov prefix subspace is spanned by at most $\sum_{j=0}^k n^j$ word evaluations. Summing over the \(s\) seed vectors gives
\[
\dim\mathcal P_\Psi^{(k)}(\mathcal S)
\leq
s\sum_{j=0}^k n^j.
\]
Since $\mathcal P_\Psi^{(k)}(\mathcal S) \subset \mathcal H$, we have the trivial upper bound of $d$ on the dimension, so combining inequalities yields
\[
\dim\mathcal P_\Psi^{(k)}(\mathcal S)
\leq
\min\left\{
d,\;
s\sum_{j=0}^k n^j
\right\}.
\]
For \(n>1\), the finite geometric series identity provides the formula for the sum, while the \(n=1\) case is handled by direct substitution.
\end{proof}

\begin{samepage}
\begin{corollary}[Universal upper bound for the prefix KLA]
\phantomsection
\label{cor:prefix-upper-bound-kla}
Let $\mathcal P_\Psi^{(k)}(\mathcal S)$
be the prefix Krylov subspace of
\Cref{def:prefix-subspace-kla}. Then
\[
\dim\mathfrak p_\Psi^{(k)}(\mathcal S)
\leq
\left(
\dim \mathcal P_\Psi^{(k)}(\mathcal S)
\right)^2,
\]
where $\mathfrak p_\Psi^{(k)}(\mathcal S)$ is the prefix Krylov-Lie algebra of \Cref{def:prefix-subspace-kla}. Consequently, for $n>1$,
\[
\dim\mathfrak p_\Psi^{(k)}(\mathcal S)
\leq
\min\left\{
d^2,\;
\left(
s\sum_{j=0}^k n^j
\right)^2
\right\}
=
\min\left\{
d^2,\;
\left(
s\frac{n^{k+1}-1}{n-1}
\right)^2
\right\},
\]
and if $n=1$,
\[
\dim\mathfrak p_\Psi^{(k)}(\mathcal S) \leq \min\left\{
d^2,\;
(s(k+1))^2
\right\}
\]

\end{corollary}

\begin{proof}
Immediate from \Cref{prop:prefix-upper-bound-krylov-subspace}.
\end{proof}
\end{samepage}

\chapter{Approximation of the Reachable Manifold}
\label{chap:krylov-lie-approximations}

At last, we finally have erected all the requisite architecture to characterize the quality of KLG approximations to the reachable manifolds, especially for the dimension-matched case. However, we must take care to not prematurely disqualify of lower-dimensional KLG approximations, as they may give valuable insights of their own: when a lower-dimensional KLG is measure-faithful, the corresponding comparison map will generically be a submersion rather than a local diffeomorphism, and therefore the reachable manifold contains local fiber directions invisible to the effective Lie-group geometry. This is suggestive of a notion of geometric redundancy that may be relevant for ansatz design and trainability. With that acknowledgement complete, we will begin by proving some helpful lemmas to investigate the properties of dimension-reduced and dimension-matched Krylov-Lie approximations to a reachable manifold.

\section{Preparatory Lemmas}
\begin{lemma}[Analytic Jacobian rank criterion]
\label{lem:jacobian-rank-criterion}
Let \( M\) be the \(m\)-dimensional smooth reachable manifold of a VQA (Equation \ref{eq:reachable-manifold}), equipped with
\[
\mu_{ M}=q\,d\mathrm{vol}_{ M},
\qquad q\in L^1( M), \quad q\ge 0.
\]
Let \(K^{(k)}_{\Psi}\) be a depth-\(k\) represented Krylov-Lie group (\Cref{def:krylov-lie-group}) of dimension
\[
m_{\mathfrak l}:=\dim K^{(k)}_{\Psi}=\dim \mathfrak l^{(k)}_{\Psi}\le m,
\]
associated to the depth-\(k\) Krylov-Lie algebra \(\mathfrak l^{(k)}_{\Psi}\) (\Cref{def:krylov-lie-algebra}), and realized in \(U(d)\) by projectors onto our subspace (\Cref{def:KrylovAlgebra}) as follows:
\[
\iota^{(k)}_{\Psi}(U)=P^{(k)}_{\Psi}UP^{(k)}_{\Psi}+I-P^{(k)}_{\Psi}.
\]
Let
\[
\Phi^{(k)}_{\Psi}: M\to K^{(k)}_{\Psi}
\]
be a comparison map. Fix a connected smooth chart \((U,x^1,\dots,x^m)\) on \( M\), and further choose a basis
\(X_1,\dots,X_{m_{\mathfrak l}}\) of \(\mathfrak l^{(k)}_{\Psi}\). Define the projected Jacobian matrix
\[
A(x)=(A_{\alpha\beta}(x))\in \mathrm{Mat}_{m_{\mathfrak l}\times m}(\mathbb R)
\]
by
\[
\left(L_{{\Phi^{(k)}_{\Psi}}(x)^{-1}}\right)_*
\!\left(\frac{\partial \Phi^{(k)}_{\Psi}}{\partial x^\beta}(x)\right)
=
\sum_{\alpha=1}^{m_{\mathfrak l}} A_{\alpha \beta}(x)\,d\iota^{(k)}_{\Psi}(X_\alpha),
\qquad \beta=1,\dots,m,
\]
where $\left(L_{{\Phi^{(k)}_{\Psi}}(x)^{-1}}\right)_* (\cdot)$ denotes the pushforward of the left-translation map by the group element ${\Phi^{(k)}_{\Psi}}(x)^{-1}$ . Assume that the entries of \(A(x)\) are real-analytic on \(U\). If there exists \(x_0\in U\)
such that
\[
\operatorname{rank}A(x_0)=m_{\mathfrak l},
\]
then the full-rank locus
\[
U_{\mathrm{full}}
:=
\{x\in U \, \vert \, \operatorname{rank}d\Phi_{\Psi,x}^{(k)}=m_{\mathfrak l}\}
\]
is a nonempty open subset of \(U\), and \(U\setminus U_{\mathrm{full}}\) has Lebesgue measure
zero in local coordinates. In particular, \(\operatorname{rank}d\Phi_{\Psi,x}^{(k)}=m_{\mathfrak l}\) for
\(\mu_{ M}\)-almost every \(x\in U\).

Hence:
\begin{enumerate}
    \item if \(m_{\mathfrak l}<m\), then \(\Phi^{(k)}_{\Psi}\) is a submersion on \(U_{\mathrm{full}}\);
    \item if \(m_{\mathfrak l}=m\), then \(\Phi^{(k)}_{\Psi}\) is a local diffeomorphism on \(U_{\mathrm{full}}\).
\end{enumerate}

Consequently, if \(M\) admits a finite atlas by such charts, for which the corresponding
rank-drop loci are \(\mu_{M}\)-null, then
\[
\operatorname{rank}d{\Phi^{(k)}_{\Psi, x}}=\dim K^{(k)}_{\Psi}
\qquad\text{for }\mu_{M}\text{-a.e. }x\in  M.
\]
Therefore, by Theorem \ref{thm:DimensionMatchingJacobian}, the singular component in the Lebesgue--Radon--Nikodym
decomposition of \((\Phi^{(k)}_{\Psi})_*\mu_{ M}\) with respect to Haar measure on \(K^{(k)}_{\Psi}\) vanishes.
\end{lemma}

\begin{proof}
Because \(K^{(k)}_{\Psi}\) is realized in the ambient unitary group through
\[
\iota^{(k)}_{\Psi}(U)=P^{(k)}_{\Psi}UP^{(k)}_{\Psi}+I-P^{(k)}_{\Psi},
\]
its tangent space at the identity is represented by
\[
d\iota^{(k)}_{\Psi}(\mathfrak l^{(k)}_{\Psi})\subset \mathfrak u(d).
\]
Since \(d\iota^{(k)}_{\Psi}:\mathfrak l^{(k)}_{\Psi}\to T_eK^{(k)}_{\Psi}\) is the differential of the representing
homomorphism, the vectors
\[
d\iota^{(k)}_{\Psi}(X_1),\dots,d\iota^{(k)}_{\Psi}(X_{m_{\mathfrak l}})
\]
form a basis of \(T_eK^{(k)}_{\Psi}\).

For each \(x\in U\), left translation by \(\Phi^{(k)}_{\Psi}(x)^{-1}\) identifies \(T_{\Phi^{(k)}_{\Psi}(x)}K^{(k)}_{\Psi}\)
with \(T_eK^{(k)}_{\Psi}\). Thus \(d{\Phi^{(k)}_{\Psi, x}}:T_x M\to T_{\Phi^{(k)}_{\Psi}(x)}K^{(k)}_{\Psi}\) has rank \(m_{\mathfrak l}\) if and only if
the vectors
\[
\left(L_{\Phi^{(k)}_{\Psi}(x)^{-1}}\right)_*
\!\left(\frac{\partial \Phi^{(k)}_{\Psi}}{\partial x^\beta}(x)\right),
\qquad \beta=1,\dots,m,
\]
span \(T_eK^{(k)}_{\Psi}\). By definition of \(A(x)\), this is equivalent to
\[
\operatorname{rank}A(x)=m_{\mathfrak l}.
\]
Hence
\[
U_{\mathrm{full}}
=
\{x\in U \, \vert \,\operatorname{rank}A(x)=m_{\mathfrak l}\}.
\]

Since \(\operatorname{rank}A(x_0)=m_{\mathfrak l}\), some \(m_{\mathfrak l}\times m_{\mathfrak l}\) minor \(\Delta(x)\) of \(A(x)\) satisfies
\[
\Delta(x_0)\neq 0.
\]
Because the entries of \(A(x)\) are nontrivial real-analytic and \(\Delta(x)\) is not identically zero, \(\Delta(x)\) too is nontrivial real-analytic. Therefore \(\Delta\) is
not identically zero on \(U\), and its nonvanishing locus
\[
U_\Delta:=\{x\in U \, \vert \,\Delta(x)\neq 0\}
\]
is a nonempty open subset of \(U\). Moreover, \(U_\Delta\subset U_{\mathrm{full}}\), so
\(U_{\mathrm{full}}\) is nonempty. To see that $U_\text{full}$ is open, writing $Z$ as the finite index set of all $m_{\mathfrak l}  \times m_{\mathfrak l} $ minors of $A(x)$ associated with some size-$m_{\mathfrak l} $ subset of columns, and for $z\in Z$ letting $\Delta_{z}$ denote the corresponding minor, we have 
\[
U \setminus U_\text{full} = \bigcap_{z \in Z} \{x \in U \, \vert \, \Delta_z(x) = 0\},
\]
which is closed as a finite intersection of closed sets (each $\Delta_z$ is real-analytic, thus continuous, and in turn the preimage of the closed singleton $\{0\}$ for any such $\Delta_z$ must be closed). Hence $U_\text{full}$ is open as the complement of a closed set in $U$.

On the other hand,
\[
U\setminus U_{\mathrm{full}}
\subset
\{x\in U \, \vert \,\Delta(x)=0\}.
\]
Since the zero set of a nontrivial real-analytic function has Lebesgue measure zero in local
coordinates, \(U\setminus U_{\mathrm{full}}\) is coordinate-Lebesgue null. Because
\(\mu_M=q\,d\mathrm{vol}_M\) is absolutely continuous with respect to the smooth volume measure,
it follows that \(U\setminus U_{\mathrm{full}}\) is also \(\mu_{ M}\)-null. Thus
\[
\operatorname{rank}d{\Phi^{(k)}_{\Psi, x}}=m_{\mathfrak l}
\qquad\text{for }\mu_M\text{-a.e. }x\in U.
\]

If \(m_{\mathfrak l}<m\), surjectivity of \(d{\Phi^{(k)}_{\Psi, x}}\) means exactly that \(\Phi^{(k)}_{\Psi}\) is a submersion at
\(x\) \cite[Ch.~4]{Lee2013SmoothManifolds}. If \(m_{\mathfrak l}=m\), then \(d\Phi^{(k)}_{\Psi,x}:T_x  M\to T_{\Phi^{(k)}_{\Psi}(x)}K^{(k)}_{\Psi}\) is a linear map between
vector spaces of equal dimension, so rank \(m\) means invertibility; by the inverse function
theorem, \(\Phi^{(k)}_{\Psi}\) is then a local diffeomorphism at \(x\). This proves the two cases.

Finally, if \(M\) is covered by a finite atlas (again, \( M\) is compact) of such charts and the union of the
rank-drop loci is \(\mu_{ M}\)-null, then
\[
\operatorname{rank}d{\Phi^{(k)}_{\Psi, x}}=\dim K^{(k)}_{\Psi}
\qquad\text{for }\mu_{ M}\text{-a.e. }x\in  M.
\]
The last claim is then an immediate application of Theorem \ref{thm:DimensionMatchingJacobian}.
\end{proof}

\begin{proposition}[Canonical circuit lift satisfies the analytic Jacobian hypothesis]
\phantomsection
\label{prop:canonical-circuit-lift}
Let \(M\) be the \(m\)-dimensional smooth reachable manifold of a VQA, equipped with
\[
\mu_M = q\,d\mathrm{vol}_M,
\qquad q\in L^1(M), \quad q\ge 0.
\]
Let \(K^{(k)}_{\Psi}\) be a depth-\(k\) represented Krylov-Lie group (\Cref{def:krylov-lie-group}) of dimension
\[
m_{\mathfrak l}:=\dim K^{(k)}_{\Psi}=\dim\mathfrak l^{(k)}_{\Psi}\le m,
\]
associated to the depth-\(k\) Krylov--Lie algebra \(\mathfrak l^{(k)}_{\Psi}\) (\Cref{def:krylov-lie-algebra}), and realized in \(U(d)\) by by projectors onto our subspace (\Cref{def:KrylovAlgebra}):
\[
\iota^{(k)}_{\Psi}(U)=P^{(k)}_{\Psi}UP^{(k)}_{\Psi}+I-P^{(k)}_{\Psi}.
\]
Fix a connected chart \((U,x^1,\dots,x^m)\) on \(M\). Assume that on \(U\) the reachable manifold is
described by a circuit with fixed Hermitian generators \(H_1,\dots,H_n\) and real-analytic
parameter functions \(\theta_1,\dots,\theta_p:U\to\mathbb R\), so that the corresponding circuit
unitary is, for $i_l \in \{1, \dots, n \}$ and $1 \leq l \leq p$,
\[
\mathcal U(x)=e^{-i\theta_p(x)H_{i_p}}\cdots e^{-i\theta_1(x)H_{i_1}}.
\]
Let
\[
\widetilde H_{j}^{(k)}:=P^{(k)}_{\Psi}H_jP^{(k)}_{\Psi}\big|_{\mathcal K^{(k)}_{\Psi}}\in \mathrm{End}(\mathcal K^{(k)}_{\Psi}),
\qquad j=1,\dots,n,
\]
and assume that
\[
\widetilde H_j^{(k)}\in \mathfrak l^{(k)}_{\Psi}
\qquad\text{for all }j=1,\dots,n.
\]
Define the
canonical depth-\(k\) comparison map
\[
\kappa^{(k)}_{\Psi}:U\to K^{(k)}_{\Psi}
\]
by
\[
\kappa^{(k)}_{\Psi}(x)
:=
e^{-i\theta_p(x) \widetilde H_{i_p}^{(k)}}\cdots e^{-i\theta_1(x) \widetilde H_{i_1}^{(k)}}.
\]

Then \(\kappa^{(k)}_{\Psi}\) is well-defined and real-analytic. Moreover, if \(X_1,\dots,X_{m_{\mathfrak l}}\) is any basis of
\(\mathfrak l^{(k)}_{\Psi}\), and if \(A(x)=(A_{\alpha \beta}(x))\in \mathrm{Mat}_{m_{\mathfrak l}\times m}(\mathbb R)\) (\Cref{lem:jacobian-rank-criterion}) is defined by
\[
\left(L_{\kappa^{(k)}_{\Psi}(x)^{-1}}\right)_*
\!\left(\frac{\partial\kappa^{(k)}_{\Psi}}{\partial x^\beta}(x)\right)
=
\sum_{\alpha=1}^{m_{\mathfrak l}} A_{\alpha \beta}(x)\,d\iota^{(k)}_{\Psi}(X_\alpha),
\qquad \beta=1,\dots,m,
\]
then every entry \(A_{\alpha \beta}(x)\) is real-analytic on \(U\).

Consequently, if there exists \(x_0\in U\) such that
\[
\operatorname{rank}d\kappa_{\Psi,x_0}^{(k)}=m_{\mathfrak l},
\]
then the full-rank locus
\[
U_{\mathrm{full}}
:=
\{x\in U \, \vert \,\operatorname{rank}d\kappa_{\Psi,x}^{(k)}=m_{\mathfrak l}\}
\]
is a nonempty open subset of \(U\), and \(U\setminus U_{\mathrm{full}}\) has \(\mu_M\)-measure zero.
Hence:
\begin{enumerate}
    \item if \(m_{\mathfrak l}<m\), then \(\kappa^{(k)}_{\Psi}\) is a submersion on \(U_{\mathrm{full}}\);
    \item if \(m_{\mathfrak l}=m\), then \(\kappa^{(k)}_{\Psi}\) is a local diffeomorphism on \(U_{\mathrm{full}}\);
    \item if \(m_{\mathfrak l}=m\) and \(\kappa^{(k)}_{\Psi}\) is injective on \(U_{\mathrm{full}}\), then
    \[
    \kappa^{(k)}_{\Psi}:U_{\mathrm{full}}\to \kappa^{(k)}_{\Psi}(U_{\mathrm{full}})
    \]
    is a diffeomorphism;
    \item if in addition \(\kappa^{(k)}_{\Psi}(U_{\mathrm{full}})=K^{(k)}_{\Psi}\), then \(\kappa^{(k)}_{\Psi}\) is a diffeomorphism
    from \(U_{\mathrm{full}}\) onto \(K^{(k)}_{\Psi}\).
\end{enumerate}
\end{proposition}

\begin{proof}
Because \(\widetilde H_{i_j}^{(k)}\in\mathfrak l^{(k)}_{\Psi}\) for every \(j\) ($i_j$ represents an index from the generator set), each one-parameter subgroup
\[
t\mapsto e^{-it \widetilde H_{i_j}^{(k)}}
\]
lies in \(K^{(k)}_{\Psi}\). Since \(K^{(k)}_{\Psi}\) is a group, the product
\[
\kappa^{(k)}_{\Psi}(x)
=
e^{-i\theta_p(x) \widetilde H_{i_p}^{(k)}}\cdots e^{-i\theta_1(x) \widetilde H_{i_1}^{(k)}}
\]
also lies in \(K^{(k)}_{\Psi}\) for every \(x\in U\). Thus \(\kappa^{(k)}_{\Psi}\) is well-defined.

Each entry of \(e^{-i\theta_j(x) \widetilde H_{i_j}^{(k)}}\) is real-analytic in \(x\), because the matrix
exponential is real-analytic and \(\theta_j(x)\) is real-analytic. Therefore \(\kappa^{(k)}_{\Psi}\), being a
finite product of such matrix-valued real-analytic maps, is itself real-analytic.

Write
\[
G_j(x):=e^{-i\theta_j(x) \widetilde H_{i_j}^{(k)}},
\qquad
\kappa^{(k)}_{\Psi}(x)=G_p(x)\cdots G_1(x).
\]
Since \( \widetilde H_{i_j}^{(k)}\) commutes with its own exponential,
\[
\frac{\partial G_j}{\partial x^\beta}(x)
=
-\bigl(\partial_\beta\theta_j(x)\bigr)\,i \widetilde H_{i_j}^{(k)}\,G_j(x).
\]
Differentiating the product gives
\[
\frac{\partial\kappa^{(k)}_{\Psi}}{\partial x^\beta}(x)
=
-\sum_{r=1}^p
G_p(x)\cdots G_{r+1}(x)\,
\bigl(\partial_\beta\theta_r(x)\bigr)i \widetilde H_{i_r}^{(k)}G_r(x)\cdots G_1(x).
\]
Left-translating to the identity yields
\[
\left(L_{\kappa^{(k)}_{\Psi}(x)^{-1}}\right)_*
\!\left(\frac{\partial\kappa^{(k)}_{\Psi}}{\partial x^\beta}(x)\right)
=
-\sum_{r=1}^p
\bigl(\partial_\beta\theta_r(x)\bigr)
\! \operatorname{Ad}_{(G_{r-1}(x)\cdots G_1(x))^{-1}}
\!\bigl(i \widetilde H_{i_r}^{(k)}\bigr).
\]
Each term on the right-hand side depends real-analytically on \(x\), since it is obtained from
the real-analytic maps \(G_j(x)\) by matrix multiplication and inversion inside the unitary
group. Hence the left-translated tangent vector is a real-analytic \(\mathfrak l^{(k)}_{\Psi}\)-valued map.
Expanding in the fixed basis \(X_1,\dots,X_{m_{\mathfrak l}}\) (each term of the sum is in the Lie algebra since the adjoint action of Lie group elements preserves the Lie algebra), we conclude that every coefficient
\(A_{\alpha\beta}(x)\) is real-analytic on \(U\).

Thus \(\kappa^{(k)}_{\Psi}\) satisfies the hypotheses of \Cref{lem:jacobian-rank-criterion}.
If there exists \(x_0\in U\) with
\[
\operatorname{rank}d\kappa_{\Psi,x_0}^{(k)}=m_{\mathfrak l},
\]
then the full-rank locus \(U_{\mathrm{full}}\) is nonempty and open, and its complement is
\(\mu_M\)-null. On \(U_{\mathrm{full}}\), the map is therefore a submersion when \(m_{\mathfrak l}<m\), and a
local diffeomorphism when \(m_{\mathfrak l}=m\) \cite[Thms.~4.5 and 4.12]{Lee2013SmoothManifolds}.

If \(m_{\mathfrak l}=m\) and \(\kappa^{(k)}_{\Psi}\) is injective on \(U_{\mathrm{full}}\), then a bijective local
diffeomorphism onto its image is a diffeomorphism onto that image \cite[Ch.~4]{Lee2013SmoothManifolds}. If, moreover,
\(\kappa^{(k)}_{\Psi}(U_{\mathrm{full}})=K^{(k)}_{\Psi}\), then \(\kappa^{(k)}_{\Psi}\) is a diffeomorphism from \(U_{\mathrm{full}}\)
onto \(K^{(k)}_{\Psi}\) \cite[Ch.~4]{Lee2013SmoothManifolds}.
\end{proof}

\section{The Krylov-Lie Approximation Theorem and Related Results}
The canonical comparison map is built from the original compressed generators. Taking the closure supplies a Haar probability measure, but it need not supply tangent directions to this map. Indeed, temporarily distinguishing the original algebra from its completion, the derivative formula in \Cref{prop:canonical-circuit-lift} gives
\[
\operatorname{rank}d\kappa_{\Psi,x}^{(k)}
\leq\dim\mathfrak l_\Psi^{(k)}
\leq\dim\widehat{\mathfrak l}_\Psi^{(k)}
=\dim K_\Psi^{(k)}.
\]
If the second inequality is strict, full rank into the compact KLG is impossible. Geometrically, every finite product still lies in the original immersed subgroup, which has lower dimension than its closure and therefore has zero Haar volume there. The resulting pushforward is singular with respect to Haar, even though that subgroup is dense. This motivates restricting the canonical construction to seeds for which the represented subgroup is already closed. Irreducibility is sufficient by \Cref{cor:irreducible-kla-closed} and \Cref{prop:generic-irreducible-kla-closed}, but closed reducible representations are equally admissible. 

Alternatively, one may instead impose the constraint that our observable has a centralizer containing the new central components of the KLG after compactification (or say that the observable $O$ satisfies a LASA-like condition and $iO$ is in our Lie algebra) and then simply add the new central directions as exponentiated factors to our canonical comparison map as before. For any loss-based observable we would like to compute, the new central directions factor out under this prescription due to the presence of the conjugate transpose term in the loss and do not affect quantitative results.

\begin{theorem}[Krylov--Lie approximation on a certified closed family]
\label{thm:main_klg_approximation}

Let \(M\) be an \(m\)-dimensional smooth reachable manifold, equipped with
\[
\mu_M=q\,d\mathrm{vol}_M,
\qquad
q\in C^1(M),
\qquad
q\geq0.
\]
Fix \(k\geq1\), \(s\geq1\), and a finite Hermitian generator set
\[
S=\{H_1,\dots,H_n\}\subset\mathrm{End}(\mathcal H),
\qquad
\mathcal H\cong\mathbb C^d.
\]
Let \(\Sigma_s\) be the compact unit-norm block-seed space. Fix integers
\(r,\ell\), with \(\ell\leq m\).

Assume that there is a nonempty fixed-basis determinantal chart
\(\mathcal A_{r,\ell}\subset\Sigma_s\), specified by finitely many
vanishing conditions
\[
E_\nu(\Psi)=0,
\qquad
\nu=1,\dots,M,
\]
and nonvanishing conditions
\[
D_j(\Psi)\neq0,
\qquad
j=1,\dots,N,
\]
such that, for every \(\Psi\in\mathcal A_{r,\ell}\),
\[
\dim\mathcal K_\Psi^{(k)}=r,
\qquad
\dim\mathfrak l_\Psi^{(k)}=\ell,
\]
and
\[
\Pi_\Psi^{(k)}\bigl(\widetilde K_\Psi^{(k)}\bigr)
=
K_\Psi^{(k)},
\]
especially one obtained by finding a seed in an irreducible subset of a fixed-rank stratum whose KLA acts on its Krylov-Lie subspace irreducibly (see \Cref{prop:generic-irreducible-kla-closed}).
Equivalently,
\[
\widehat{\mathfrak l}_\Psi^{(k)}
=
\mathfrak l_\Psi^{(k)},
\qquad
\dim K_\Psi^{(k)}=\ell.
\]
Assume that the \(D_j\) include the fixed Krylov-basis minor and all
nonvanishing denominators used in the local projector and compressed-generator
formulas.

Suppose that a common finite atlas of connected circuit charts on \(M\)
uses the same real-analytic angle functions for every seed in
\(\mathcal A_{r,\ell}\), and that the resulting canonical comparison maps
\[
\kappa_\Psi^{(k)}:M\longrightarrow K_\Psi^{(k)}
\]
agree on overlaps.
Then there exists \(\delta>0\) such that the \(\delta\)-nondegenerate
substratum
\[
\mathcal G_{r,\ell}^{\delta}
:=
\bigl\{\Psi\in Y_{r,\ell}\:\big|\:|D_j(\Psi)|\geq\delta
\text{ for all }j=1,\dots,N\bigr\}
\]
is nonempty and compact. Conclusions \emph{(1)--(5)} and \emph{(7)} below
hold for any \(\Psi\in\mathcal G_{r,\ell}^{\delta}\) which passes the
rank test of \Cref{lem:jacobian-rank-criterion} at least once on each
circuit chart. Conclusion \emph{(6)} holds on all of
\(\mathcal G_{r,\ell}^{\delta}\).

Moreover, let \(\Phi\) be a connected real-analytic seed family
contained in this determinantal chart. If one seed
\(\Psi_\circ\in\Phi\) passes the rank test at least once on each
circuit chart, then the seeds in \(\Phi\) which pass that test form a
relatively open, dense set whose complement has smooth-volume measure zero
in \(\Phi\).

\begin{enumerate}
\item The full-rank locus
\[
M_{\mathrm{full}}(\Psi)
:=
\left\{
x\in M
\;\middle|\;
\operatorname{rank}d\kappa_{\Psi,x}^{(k)}=\ell
\right\}
\]
is open and satisfies
\[
\mu_M\bigl(M\setminus M_{\mathrm{full}}(\Psi)\bigr)=0.
\]

\item If \(C\subset M_{\mathrm{full}}(\Psi)\) is compact, then \(C\)
has a finite open cover \(V_1,\dots,V_t\) such that, on each
\(C\cap V_a\), one fixed \(\ell\times\ell\) Jacobian minor is bounded
away from zero.

\item If \(\ell<m\), then \(\kappa_\Psi^{(k)}\) is a submersion at every
point of \(C\), indeed at every point of \(M_{\mathrm{full}}(\Psi)\).

\item If \(\ell=m\), then \(\kappa_\Psi^{(k)}\) is a local
diffeomorphism at every point of \(C\). If
\(\kappa_\Psi^{(k)}|_C\) is injective, then it is a homeomorphism onto
its image. If, in addition, \(C\) is a compact embedded smooth
\(m\)-submanifold of \(M\), possibly with boundary, then
\[
\kappa_\Psi^{(k)}|_C:
C\longrightarrow\kappa_\Psi^{(k)}(C)
\]
is a diffeomorphism onto an embedded smooth submanifold with boundary.

\item Equip \(K_\Psi^{(k)}\) with a bi-invariant metric whose Riemannian
volume measure is its normalized Haar measure
\(\mu_{K_\Psi^{(k)}}\). Then
\[
\nu_{\Psi,C}^{(k)}
:=
(\kappa_\Psi^{(k)})_*(\mu_M|_C)
\]
is absolutely continuous with respect to \(\mu_{K_\Psi^{(k)}}\).
Its Radon--Nikodym derivative satisfies, for
\(\mu_{K_\Psi^{(k)}}\)-almost every \(g\),
\[
\rho_{\Psi,C}^{(k)}(g)
=
\int_{\kappa_\Psi^{(k)\,-1}(g)\cap C}
\frac{q(x)}{J_{\kappa_\Psi^{(k)}}(x)}
\,d\mathcal H^{m-\ell}(x).
\]

If \(C\) is saturated, namely
\[
C=
\kappa_\Psi^{(k)\,-1}\bigl(\kappa_\Psi^{(k)}(C)\bigr),
\]
then there exists \(h\in C^1(K_\Psi^{(k)})\) such that
\[
\rho_{\Psi,C}^{(k)}(g)=h(g)
\quad\text{for Haar-a.e. }g\in\kappa_\Psi^{(k)}(C).
\]
In particular, the restriction of the density to
\(\kappa_\Psi^{(k)}(C)\) admits a bounded Lipschitz extension. No
regularity across the boundary of \(\kappa_\Psi^{(k)}(C)\) is asserted
for the global Radon--Nikodym representative.

\item For every fixed compact \(C_0\subset M\), common to all seeds, the
error
\[
\mathcal E_{C_0}(k;\Psi)
:=
\sup_{x\in C_0}
\left\|
\mathcal U(x)
-
\iota_\Psi^{(k)}\bigl(\kappa_\Psi^{(k)}(x)\bigr)
\right\|_{\mathrm{HS}}
\]
is continuous on \(\mathcal G_{r,\ell}^{\delta}\), and therefore
attains a minimum there.

\item For every compact \(C\subset M_{\mathrm{full}}(\Psi)\) and every
bounded function \(F:U(d)\to\mathbb R\) that is Lipschitz with respect
to the Hilbert--Schmidt metric,
\[
\left|
\int_C F(\mathcal U(x))\,d\mu_M(x)
-
\int_{K_\Psi^{(k)}}
F(\iota_\Psi^{(k)}(g))
\rho_{\Psi,C}^{(k)}(g)\,
d\mu_{K_\Psi^{(k)}}(g)
\right|
\leq
\operatorname{Lip}(F)\,
\mu_M(C)\,
\mathcal E_C(k;\Psi).
\]
\end{enumerate}
\end{theorem}
\begin{proof}
Choose \(\Psi_\circ\in\mathcal A_{r,\ell}\), and set
\[
\delta
:=
\frac12\min_{1\leq j\leq N}|D_j(\Psi_\circ)|
>0.
\]
Then \(\Psi_\circ\in\mathcal G_{r,\ell}^{\delta}\), so this set is
nonempty.

Moreover, \(\mathcal G_{r,\ell}^{\delta}\) is closed in \(\Sigma_s\):
the equations \(E_\nu=0\) are closed conditions, and the inequalities
\(|D_j|\geq\delta\) are closed conditions. Since \(\Sigma_s\) is
compact, \(\mathcal G_{r,\ell}^{\delta}\) is compact. The inequalities
\(|D_j|\geq\delta\) ensure that every seed in this set remains in the
fixed-basis chart \(\mathcal A_{r,\ell}\). Hence, throughout
\(\mathcal G_{r,\ell}^{\delta}\),
\[
\dim K_\Psi^{(k)}
=
\dim\mathfrak l_\Psi^{(k)}
=
\ell.
\]
For the genericity assertion, fix \(\Phi\) and \(\Psi_\circ\) as
above. On the fixed determinantal chart, use the Lie-word basis selected by
the nonvanishing minors to form the matrices in
\Cref{lem:jacobian-rank-criterion}. If \(A_a(\Psi,x)\) denotes the
left-translated Jacobian matrix on the connected circuit chart \(U_a\),
then its entries are jointly real-analytic in \((\Psi,x)\): the projector,
the compressed generators \(\widetilde H_j^{(k)}\), and the selected Lie
words are rational in the real seed coordinates, while the circuit formula
is real analytic; taking coordinates in the selected basis uses an
invertible Gram matrix.

Choose a point \(x_a\in U_a\) and a nonzero \(\ell\times\ell\) minor
\(\Delta_a(\Psi_\circ)\) for every \(a\). Then
\(\Psi\mapsto\Delta_a(\Psi)\), with \(x_a\) fixed, is a nonzero
real-analytic function on \(\Phi\). Its zero set has empty interior
and smooth-volume measure zero. The finite intersection of the complements
of these zero sets is therefore relatively open, dense, and conull in
\(\Phi\), and every seed in that intersection passes the rank test at
the same points \(x_a\).

Fix \(\Psi\in\mathcal G_{r,\ell}^{\delta}\) which passes the rank test in every chart. By
\Cref{prop:canonical-circuit-lift}, the map
\(\kappa_\Psi^{(k)}\) is real analytic in the circuit coordinates.
Since the rank test is passed once on every chart of the finite atlas,
\Cref{lem:jacobian-rank-criterion} implies that
\[
M_{\mathrm{full}}(\Psi)
=
\left\{
x\in M
\;\middle|\;
\operatorname{rank}d\kappa_{\Psi,x}^{(k)}=\ell
\right\}
\]
is open and has \(\mu_M\)-null complement. This proves (1).

Let \(C\subset M_{\mathrm{full}}(\Psi)\) be compact. For every
\(x\in C\), choose a circuit chart and an \(\ell\times\ell\) Jacobian
minor \(\Delta_x\) which is nonzero at \(x\). By continuity, there is
an open neighborhood \(V_x\) of \(x\) such that
\[
|\Delta_x(y)|
\geq
\frac12|\Delta_x(x)|
\qquad
\text{for every }y\in V_x.
\]
A finite subcollection \(V_1,\dots,V_t\) covers \(C\), proving (2).

If \(\ell<m\), full rank means that
\[
d\kappa_{\Psi,x}^{(k)}:
T_xM\longrightarrow
T_{\kappa_\Psi^{(k)}(x)}K_\Psi^{(k)}
\]
is surjective. Thus \(\kappa_\Psi^{(k)}\) is a submersion at every
point of \(M_{\mathrm{full}}(\Psi)\), proving (3).

If \(\ell=m\), the same differential is an isomorphism at every point
of \(M_{\mathrm{full}}(\Psi)\). The inverse function theorem therefore
gives a local diffeomorphism there. If \(\kappa_\Psi^{(k)}|_C\) is
injective, compactness of \(C\) and the Hausdorff property of
\(K_\Psi^{(k)}\) imply that
\[
\kappa_\Psi^{(k)}|_C:
C\longrightarrow\kappa_\Psi^{(k)}(C)
\]
is a homeomorphism. If \(C\) is additionally a compact embedded smooth
\(m\)-submanifold, possibly with boundary, then
\(\kappa_\Psi^{(k)}|_C\) is an injective immersion. Since its domain is
compact, it is proper; hence it is a smooth embedding, and therefore a
diffeomorphism onto its image. This proves (4).

By (1), the hypotheses of \Cref{thm:DimensionMatchingJacobian} hold
with \(U=M_{\mathrm{full}}(\Psi)\). Therefore
\[
(\kappa_\Psi^{(k)})_*\mu_M
\ll
\mu_{K_\Psi^{(k)}}.
\]
Restricting the source measure to \(C\) preserves absolute continuity,
so
\[
\nu_{\Psi,C}^{(k)}
=
(\kappa_\Psi^{(k)})_*(\mu_M|_C)
\ll
\mu_{K_\Psi^{(k)}}.
\]
The coarea formula gives the stated formula for
\(\rho_{\Psi,C}^{(k)}\).

Now suppose that \(C\) is saturated. Since
\(C\subset M_{\mathrm{full}}(\Psi)\) is compact and
\(M_{\mathrm{full}}(\Psi)\) is open, choose
\[
\chi\in C_c^\infty(M_{\mathrm{full}}(\Psi))
\]
such that \(\chi=1\) on a neighborhood of \(C\). The pushforward of the
compactly supported \(C^1\) density
\[
\chi q\,d\mathrm{vol}_M
\]
under the submersion \(\kappa_\Psi^{(k)}\) has a \(C^1\) density
\(h\) with respect to Haar measure. Indeed, this follows locally by
writing the submersion in product coordinates and differentiating the
resulting compactly supported fiber integral.

For every $g\in\kappa_\Psi^{(k)}(C)$,
saturation gives $\kappa_\Psi^{(k)\,-1}(g)\subset C$.
Thus \(\chi=1\) on the entire relevant fiber, and the coarea formula
shows that
\[
h(g)=\rho_{\Psi,C}^{(k)}(g)
\]
for Haar-almost every \(g\in\kappa_\Psi^{(k)}(C)\). Since
\(K_\Psi^{(k)}\) is compact and \(h\in C^1(K_\Psi^{(k)})\), the
restriction \(h|_{\kappa_\Psi^{(k)}(C)}\) is bounded and Lipschitz.
This proves (5).

For (6), on the fixed-basis chart the projectors \(P_\Psi^{(k)}\) and
compressed generators depend continuously on \(\Psi\). In the ambient
unitary representation,
\[
\iota_\Psi^{(k)}\bigl(\kappa_\Psi^{(k)}(x)\bigr)
=
P^{(k)}_{\Psi}e^{-i\theta_p(x) P^{(k)}_{\Psi} H_{i_p}^{(k)} P^{(k)}_{\Psi}}\cdots e^{-i\theta_1(x) P^{(k)}_{\Psi} H_{i_1}^{(k)} P^{(k)}_{\Psi}}P^{(k)}_{\Psi} + (I-P^{(k)}_{\Psi}
)\]
Hence the represented circuit lift is jointly continuous in
\((\Psi,x)\). Uniform continuity on the compact set
\(\mathcal G_{r,\ell}^{\delta}\times C_0\) implies continuity of
\(\mathcal E_{C_0}(k;\Psi)\). The extreme-value theorem then gives a
minimizing seed.

Finally, for (7), the Lipschitz property gives
\[
\left|
F(\mathcal U(x))
-
F\!\left(
\iota_\Psi^{(k)}(\kappa_\Psi^{(k)}(x))
\right)
\right|
\leq
\operatorname{Lip}(F)\,\mathcal E_C(k;\Psi).
\]
Integrating over \(C\) gives
\[
\left|
\int_C F(\mathcal U(x))\,d\mu_M(x)
-
\int_C
F\!\left(
\iota_\Psi^{(k)}(\kappa_\Psi^{(k)}(x))
\right)
\,d\mu_M(x)
\right|
\leq
\operatorname{Lip}(F)\,
\mu_M(C)\,
\mathcal E_C(k;\Psi).
\]
The definition of the pushforward measure and the Radon--Nikodym
identity transform the second integral into
\[
\int_{K_\Psi^{(k)}}
F(\iota_\Psi^{(k)}(g))
\rho_{\Psi,C}^{(k)}(g)\,
d\mu_{K_\Psi^{(k)}}(g).
\]
This proves (7).
\end{proof}

\begin{remark}[Interpretation and uses of the Krylov--Lie approximation theorem]
\label{rem:klg_approx_discussion}
\Cref{thm:main_klg_approximation} is the main structural result linking the
geometric Krylov--Lie construction to the statistical behaviour of finite-depth
variational quantum algorithms. It says, roughly, that after choosing a seed
which passes the rank test, the original circuit ensemble can be
\emph{approximated} on suitable compact regions by a finite-dimensional
represented Krylov--Lie group equipped with Haar measure and a well-behaved
Radon--Nikodym density. The approximation enters only through the replacement
of the true circuit map by its Krylov--Lie surrogate; all measure-theoretic
statements on the Krylov--Lie group are exact for that surrogate.

There are several layers to this statement.

\smallskip

\noindent
\emph{(1) Depth-aware Lie-group surrogate.}
For a fixed commutator depth \(k\) and seed block size \(s\), the construction of
\(\mathfrak l^{(k)}_\Psi\) and \(K^{(k)}_\Psi\) attaches to each seed block
\(\Psi\) a finite-dimensional Lie algebra and a compact Lie group that encode the
order-\(k\) expressive content of the ansatz. On the nondegenerate seed
substratum \(\mathcal G_{r,\ell}^{\delta}\), the dimension data \((r,\ell)\)
are constant and the comparison map
\[
\kappa_\Psi^{(k)}:M\to K_\Psi^{(k)}
\]
has real-analytic Jacobian.

The rank conclusion is not asserted for every seed in
\(\mathcal G_{r,\ell}^{\delta}\). Rather, if one seed in a connected
real-analytic seed family passes the rank test on each circuit chart, then the
seeds in that family which pass the same test form a relatively open, dense,
and conull set. For any such seed, items~(1)--(4) of the theorem give an open
full-measure locus \(M_{\mathrm{full}}(\Psi)\subset M\) on which
\(\kappa_\Psi^{(k)}\) has constant full rank \(\ell\), and on any compact
subset \(C\subset M_{\mathrm{full}}(\Psi)\) the map is a submersion, or a local
diffeomorphism when \(\ell=m\). Thus, locally on
\(M_{\mathrm{full}}(\Psi)\), the reachable manifold admits a depth-\(k\)
Lie-group model of controlled dimension.

\smallskip

\noindent
\emph{(2) Separating geometric approximation from sampling.}
The theorem cleanly separates \emph{geometric} approximation from \emph{sampling}
effects. The geometric error is quantified by
\[
\mathcal E_C(k;\Psi)
=
\sup_{x\in C}
\bigl\|\mathcal U(x)-\iota_\Psi^{(k)}
\bigl(\kappa_\Psi^{(k)}(x)\bigr)\bigr\|_{\mathrm{HS}},
\]
which measures how well the represented Krylov--Lie circuit
\(\iota_\Psi^{(k)}\circ\kappa_\Psi^{(k)}\) reproduces the original circuit map
\(\mathcal U\) on the chosen compact set \(C\). By compactness of the
nondegenerate seed substratum, this error attains a minimum over
\(\mathcal G_{r,\ell}^{\delta}\).

This optimization statement is separate from the rank statement: the theorem
does not assert that an error-minimizing seed in
\(\mathcal G_{r,\ell}^{\delta}\) necessarily passes the rank test. For a seed
which does pass the rank test, sampling effects are encoded in item~(5): the
reachable-manifold measure \(\mu_M|_C\) is pushed forward by the comparison map
\(\kappa_\Psi^{(k)}\) to a measure on \(K_\Psi^{(k)}\) that is absolutely
continuous with respect to Haar measure and admits a Radon--Nikodym density
\(\rho_{\Psi,C}^{(k)}\). Thus, for any integrable \(F\colon U(d)\to\mathbb R\),
\[
\int_{K_\Psi^{(k)}}
F\bigl(\iota_\Psi^{(k)}(g)\bigr)\,
\rho_{\Psi,C}^{(k)}(g)\,d\mu_{K_\Psi^{(k)}}(g)
=
\int_C
F\bigl(\iota_\Psi^{(k)}(\kappa_\Psi^{(k)}(x))\bigr)\,d\mu_M(x).
\]
Hence all moment and variance formulas on \(K_\Psi^{(k)}\) are \emph{exact} for
the surrogate circuit \(\iota_\Psi^{(k)}\circ\kappa_\Psi^{(k)}\). When these
formulas are used to describe the original circuit \(\mathcal U\), one must
additionally account for the geometric approximation error
\(\mathcal E_C(k;\Psi)\), as quantified in item~(7) of the theorem.

\smallskip

\noindent
\emph{(3) Locality on the regular part of the reachable manifold.}
For every seed which passes the rank test, the theorem is \emph{local} on the
regular part of the reachable manifold. One does not require
\(\kappa_\Psi^{(k)}\) to have full rank everywhere on \(M\); it is enough that
the full-rank locus \(M_{\mathrm{full}}(\Psi)\) be open and of full
\(\mu_M\)-measure. The theorem then applies on any compact subset
\(C\subset M_{\mathrm{full}}(\Psi)\), even if the global geometry or
singularity structure of \(M\) is complicated. In practice, this permits the
analysis to focus on well-behaved regions of parameter space without requiring
global control.

\smallskip

\noindent
\emph{(4) Saturated compact sets and regularity of the density.}
For a general compact \(C\subset M_{\mathrm{full}}(\Psi)\), item~(5) gives
absolute continuity of the pushforward measure. If \(C\) is additionally a
compact smooth domain in \(M\), the density
\(\rho_{\Psi,C}^{(k)}\) is obtained locally from the sheet-sum or coarea
formulas and is differentiable away from the image of the boundary
\[
\kappa_\Psi^{(k)}(\partial C).
\]
If such a set \(C\) is saturated with respect to
\(\kappa_\Psi^{(k)}\), namely
\[
C=
\bigl(\kappa_\Psi^{(k)}\bigr)^{-1}
\bigl(\kappa_\Psi^{(k)}(C)\bigr),
\]
then \(C\) is a union of whole fibers of \(\kappa_\Psi^{(k)}\), so no fiber is
truncated by the boundary of \(C\). In this case, the restriction of
\(\rho_{\Psi,C}^{(k)}\) to \(\kappa_\Psi^{(k)}(C)\) admits a \(C^1\) extension
to a neighborhood of that image, and is therefore bounded and Lipschitz there.
This is a statement about the restriction of the density to the image, rather
than about a global Radon--Nikodym representative across the boundary of that
image.

\smallskip

\noindent
\emph{(5) Applications and choice of comparison map.}
Taken together, these properties make
\Cref{thm:main_klg_approximation} the main technical bridge between the
original VQA ensemble and the practical application of the Krylov--Lie
framework developed later in the thesis. It justifies choosing, generically,
seeds for which the comparison map is measure-faithful on its full-rank locus,
optimizing over the compact structural seed core to minimize the geometric
error \(\mathcal E_C(k;\Psi)\), and performing subsequent moment, variance,
and convergence analyses directly on \(K_\Psi^{(k)}\) with respect to Haar
measure reweighted by \(\rho_{\Psi,C}^{(k)}\).

Without an additional closed rank certificate, the seed selected by the
compact error minimization need not coincide with a seed passing the rank
test. Thus, when both properties are required simultaneously, one should
verify the rank test for the chosen seed or impose an additional admissibility
condition. In all cases, the resulting KLG-based formulas are exact for the
chosen surrogate ensemble and are approximations to the original ansatz whose
quality is governed explicitly by \(\mathcal E_C(k;\Psi)\).
\end{remark}

Closedness can also be obtained by changing the surrogate generators slightly. Since the approximating group need not be the compactification of the original compressed circuit, we may approximate its central frequencies by rational ones while retaining the semisimple action. This changes the periods of the central motion so that it closes in the original number of directions. The following result makes this alternative precise for a fixed seed.

\begin{samepage}
\begin{proposition}[Closed surrogates from rational central frequencies]
Fix a seed $\Psi$ and compressed Hermitian generators $\widetilde H_1,\dots,\widetilde H_n$ on $\mathcal K_\Psi$.\phantomsection\label{prop:rational-central-surrogate} Here write
\[
\mathfrak l_\Psi=\operatorname{Lie}\langle \widetilde H_1,\dots,\widetilde H_n\rangle,
\qquad
\widehat{\mathfrak l}_\Psi=\operatorname{Lie}(K_\Psi)
\]
to distinguish the original algebra from its completion. For every $\varepsilon>0$, there are Hermitian generators $\widetilde H'_1,\dots,\widetilde H'_n$ on $\mathcal K_\Psi$ such that
\[
\|\widetilde H'_j-\widetilde H_j\|_{\mathrm{HS}}<\varepsilon,
\qquad j=1,\dots,n,
\]
and their generated connected subgroup $K'_\Psi\subset K_\Psi$ is closed, with
\[
\dim K'_\Psi=\dim\mathfrak l_\Psi.
\]
The assignment $\widetilde H_j\mapsto \widetilde H'_j$ extends to a Lie algebra isomorphism fixing $[\mathfrak l_\Psi,\mathfrak l_\Psi]$ pointwise. Every commutator word containing at least two generators is unchanged as a matrix.
\end{proposition}
\end{samepage}

\begin{proof}
Use the notation of \Cref{prop:kla-central-frequencies}. Since $W^{\mathrm{rat}}$ is a rational subspace, its rational vectors are dense. Approximate the columns of $\Omega$ by rational columns in $W^{\mathrm{rat}}$, keeping them linearly independent, and denote the resulting matrix by $\Omega'$. Its columns define anti-Hermitian matrices
\[
Z'_b=i\bigoplus_{\alpha=1}^{q}\omega'_{\alpha b}I_{V_\alpha}
\in\mathfrak t_\Psi.
\]
Their span has the same dimension as $\mathfrak z_\Psi$ and a full period lattice, so its exponential image is a closed torus. Define $\varphi$ on $\mathfrak l_\Psi=\mathfrak s_\Psi\oplus\mathfrak z_\Psi$ by fixing $\mathfrak s_\Psi$ and sending $Z_b$ to $Z'_b$. By \Cref{cor:kla-central-torus}, $\mathfrak t_\Psi$ commutes with $\mathfrak s_\Psi$ and intersects it trivially. Thus $\varphi$ is an injective Lie homomorphism.

Set $\widetilde H'_j=\varphi(\widetilde H_j)$, using the convention of \Cref{rem:VQAConventions}. These matrices generate $\varphi(\mathfrak l_\Psi)$, whose connected group is the product of $S_\Psi$ and the new central torus. This product is compact and contained in $K_\Psi$, and has the required dimension. Taking $\Omega'$ sufficiently close to $\Omega$ makes every generator difference smaller than $\varepsilon$.

To see why these differences commute with the generators, recall that each $\widetilde H_j$ commutes with the original center and therefore preserves every joint eigenspace $V_\alpha$. Each difference $\widetilde H'_j-\widetilde H_j$ is scalar on each $V_\alpha$. Consequently, for every $j,h$,
\[
[\widetilde H'_j-\widetilde H_j,\widetilde H_h]=0,
\qquad
[\widetilde H'_j-\widetilde H_j,\widetilde H'_h-\widetilde H_h]=0.
\]
It also commutes with $\widetilde H'_h$, by addition. Expanding a commutator of length at least two therefore leaves the original word. A zero center requires no perturbation.
\end{proof}

The effect on the comparison map is particularly simple. On a fixed circuit parameter domain, or a chart with the angle functions of \Cref{prop:canonical-circuit-lift}, define
\[
\begin{aligned}
\kappa_\Psi(x)&=e^{-i\theta_p(x)\widetilde H_{i_p}}\cdots e^{-i\theta_1(x)\widetilde H_{i_1}},\\
\kappa'_\Psi(x)&=e^{-i\theta_p(x)\widetilde H'_{i_p}}\cdots e^{-i\theta_1(x)\widetilde H'_{i_1}}.
\end{aligned}
\]
The commutation identities in the proof imply commutation with every exponential factor, by its power series, and hence with both products. We may therefore collect the perturbations into a single central exponential:
\[
\kappa'_\Psi(x)=
\exp\!\left(-i\sum_{j=1}^{p}\theta_j(x)
(\widetilde H'_{i_j}-\widetilde H_{i_j})\right)\kappa_\Psi(x).
\]
Since $\|e^{-iA}-I\|_{\mathrm{HS}}\leq\|A\|_{\mathrm{HS}}$ for Hermitian $A$, as follows by diagonalizing $A$, any compact region $C$ in that domain satisfies
\[
\sup_{x\in C}\|\kappa'_\Psi(x)-\kappa_\Psi(x)\|_{\mathrm{HS}}
\leq\sup_{x\in C}\sum_{j=1}^{p}|\theta_j(x)|
\|\widetilde H'_{i_j}-\widetilde H_{i_j}\|_{\mathrm{HS}}.
\]
The right-hand side can be made arbitrarily small and, by the triangle inequality, adds directly to any established uniform circuit approximation bound. For a bound on the action on a seed vector, it is multiplied by the norm of that vector. Exact first-order matching with the original generators is no longer asserted.

Moreover, $\varphi$ intertwines the adjoint exponentials and sends each generator to its primed counterpart. Applying it to the left-translated derivative formula in \Cref{prop:canonical-circuit-lift} therefore gives
\[
\operatorname{rank}d\kappa'_{\Psi,x}
=\operatorname{rank}d\kappa_{\Psi,x}.
\]
If the original rank equals $\dim\mathfrak l_\Psi$ almost everywhere, the modified map consequently has full rank into the closed group $K'_\Psi$. An absolutely continuous source law then has an absolutely continuous pushforward with respect to Haar on $K'_\Psi$, by \Cref{thm:DimensionMatchingJacobian}. Thus the perturbation removes the obstruction from compactification whenever the original map reaches the dimension of its generated algebra.

The modified matrices need not be exact compressions of the original ambient generators. They do define a KLG on $\mathcal K_\Psi$, for example by taking a spanning block seed there. This construction is for a fixed seed and a common parameter domain. Changing the central frequencies may change periods, so descent to the same reachable manifold must still be checked if circuit parameters are identified. Rational choices also need not vary continuously with the seed. Incorporating them into the seed-optimization statement of \Cref{thm:main_klg_approximation} therefore requires these additional compatibility conditions. The bound controls circuit values on bounded parameter regions; it does not by itself control Haar averages or long-time behavior. Nevertheless, it supplies closed surrogates of the original Lie algebra dimension without introducing additional independent circuit directions.

\begin{remark}[Multiple Krylov--Lie charts and reconstruction of the ansatz measure]
\label{rem:multiple_KLG_charts}
The comparison map $\kappa_\Psi^{(k)}\colon M\to K_\Psi^{(k)}$ and represented
Krylov--Lie group $K_\Psi^{(k)}$ provide a single ``chart'' on the reachable
manifold, together with a pushforward measure whose Radon--Nikodym density
$\rho_{\Psi,C}^{(k)}$ records the sampling distortion of the surrogate circuit
on $C$. In general, this chart is only an approximation in the sense discussed
in \Cref{rem:klg_approx_discussion}.

In the ideal case where $\dim M = \dim K_\Psi^{(k)}$ and
$\kappa_\Psi^{(k)}|_C\colon C\to \kappa_\Psi^{(k)}(C)$ is a diffeomorphism onto
its image, the situation is substantially better. Treating $d\mu_{K_\Psi^{(k)}}$ as the volume
measure on $K_\Psi^{(k)}$, the Radon--Nikodym derivative simplifies to the
standard change-of-variables form
\[
\rho_{\Psi,C}^{(k)}(g)
=
\frac{d(\kappa_\Psi^{(k)})_*(\mu_M|_C)}{d\mu_{K_\Psi^{(k)}}}(g)
=
q \! \left((\kappa_\Psi^{(k)})^{-1}(g)\right) \!
\bigl|\det d(\kappa_\Psi^{(k)})^{-1}_g\bigr|,
\qquad g\in \kappa_\Psi^{(k)}(C),
\]
and we extend $\rho_{\Psi,C}^{(k)}$ by zero outside $\kappa_\Psi^{(k)}(C)$.
With this convention, for every integrable $F\colon M\to\mathbb R$ we obtain
the exact formula
\[
\int_C F(U)\,d\mu_M(U)
=
\int_{K_\Psi^{(k)}} F\bigl((\kappa_\Psi^{(k)})^{-1}(g)\bigr)\,
\rho_{\Psi,C}^{(k)}(g)\,d\mu_{K_\Psi^{(k)}}(g).
\]
Thus, in this dimension-matched, diffeomorphism-onto-its-image setting, the comparison map genuinely transfers integrals from the reachable manifold to the Krylov--Lie group via a Radon--Nikodym density that is explicitly computable from the Jacobian of $\kappa_\Psi^{(k)}$ and the original sampling density $q$ on $M$.

A natural question is whether one can use \emph{multiple} seeds and depths to
obtain a family of comparison maps
\[
\kappa_{\Psi_j}^{(k_j)} \colon M \to K_{\Psi_j}^{(k_j)}, \qquad j \in J,
\]
where $J$ is an index set, whose associated pushforward measures collectively determine the original reachable-manifold measure $\mu_M$ (or the law of the circuit map $\mathcal U$). Such a family would play the role of an atlas of Krylov--Lie charts:
each map probes the geometry and sampling law along a different finite-depth
projection, and one might hope to reconstruct global information about
$(M,\mu_M)$ from the collection of pushforwards, in analogy with how
tomographic methods reconstruct a function from its integrals along many
directions. In the dimension-matched, diffeomorphic regime above, this
tomographic picture becomes particularly concrete, because each chart admits an exact Radon--Nikodym formula for integrals over its image, so one can in
principle reconstruct the restriction of the ansatz measure to $C$ from an
integral equation on $\kappa_\Psi^{(k)}(C)\subset K_\Psi^{(k)}$.

Making this idea rigorous would require quantitative invertibility of the joint
forward operator
\[
x \longmapsto \bigl(\kappa_{\Psi_j}^{(k_j)}(x)\bigr)_{j\in J},
\]
for some family of comparison maps indexed by $J$, together with control of
the associated Radon--Nikodym densities. In favorable cases, one could imagine
that an appropriate family of Krylov--Lie groups suffices to recover the original ansatz measure (or at least its low-order moments) by solving an inverse problem on the corresponding product space $\prod_{j\in J} K_{\Psi_j}^{(k_j)}$. However, we have not shown that one can construct enough such dimension-matched diffeomorphic charts in typical VQA settings to compute integrals over the entire reachable manifold in this way. For this reason, we do not use these exact formulas in the remainder of the thesis, and we treat the tomography interpretation as a speculative but nonetheless intriguing direction for future work,
analogous in spirit to using multiple coordinate charts or tomographic projections to recover a global geometric or probabilistic object.
\end{remark}

\begin{remark}[Krylov-Lie groups and Lietic spaces]
\phantomsection
\label{rem:lietic-spaces}
The Krylov-Lie charts and atlases discussed in \Cref{rem:multiple_KLG_charts} propound a broader notion of a manifold-like space that is modeled locally on Lie groups rather than Euclidean geometry. Informally, one may say that a second-countable Hausdorff space
\(\mathcal L\) is a \emph{Lietic space} if it is locally diffeomorphic to
finite-dimensional Lie groups: that is, if it admits an open cover
\(\{U_\alpha\}_{\alpha\in A}\) and diffeomorphisms
\[
\phi_\alpha:U_\alpha\to V_\alpha\subset G_\alpha,
\]
where each \(G_\alpha\) is a finite-dimensional Lie group, such that the
transition maps
\[
\phi_\beta\circ\phi_\alpha^{-1}:
\phi_\alpha(U_\alpha\cap U_\beta)\to\phi_\beta(U_\alpha\cap U_\beta)
\]
are smooth diffeomorphisms between open subsets of Lie groups.

Every smooth manifold is automatically Lietic through Euclidean charts regarded as charts into the additive Lie group \((\mathbb R^n,+)\). Additionally, every Lietic space is a smooth manifold, since being locally diffeomorphic to a Lie group implies being locally diffeomorphic to \(\mathbb R^n\) by composing the local Lie-group charts with Euclidean charts on the Lie groups themselves. The point of this terminology is not to enlarge the class of underlying spaces, but to emphasize a preferred local group structure and the role of Haar measures and Lie-algebraic data in local analysis.

From this perspective, the present Krylov--Lie program can respectably be thought of as an attempt to approximate objects (VQA circuits) that are essentially layer-indexed Lie semigroupoids of admissible partial circuits---concretely, subsemigroupoids of the construction $\operatorname{Free}(0 \rightarrow 1 \rightarrow \dots \rightarrow L) \times S$, the product of the free semigroupoid (or free category, if we include units) on the linear quiver of layer boundaries with an ambient Lie semigroup $S$, whose morphisms are the layer-tagged partial circuits and where composition is inherited from the multiplication of $S$, defined whenever the layer data are compatible (the more classical name for such an object, emphasizing the presence of units, would be a Lie category \cite{Grad_2024})---with such local Lie group charts. The comparison maps \(\kappa_\Psi^{(k)}\) provide Lietic charts on the reachable manifold, and the associated pushforward measures describe how the ansatz measure appears in these local coordinates. Under auspicious circumstances (for example, on compact, dimension-matched, full-rank subsets where \(\kappa_\Psi^{(k)}\) is a diffeomorphism onto its image), the corresponding Radon--Nikodym densities admit exact change-of-variables formulas, so integrals over the reachable manifold can be computed directly on Krylov--Lie groups with Haar measures. One can then view this Lietic picture as a means of avoiding many of the technical difficulties of integrating directly over Lie semigroups or semigroupoids, especially in an invariant fashion \cite{Grad_2024, Argabright1966, Granirer1965}: integration and long-time behavior can be analyzed in convenient local Lie group models with their Lie algebras, with the semigroup/semigroupoid dynamics entering through the choice of charts and their overlaps rather than through a single global Haar-like structure. Conversely, one may ask whether this picture bears on Lie semigroup and semigroupoid theory itself: whether linear Lie semigroups and semigroupoids admit relatively canonical atlases of this type, with charts and transition data determined by generator and seed data rather than ad hoc choices, in a way that extends constructs like the global polar-decomposition charts available for Ol'shanskii semigroups to more general wedges and algebraic structures \cite{Lawson1994}. 

Most generally, one need not restrict the local models to Krylov-Lie groups; any suitable family of compatible Lie group charts (e.g. automorphism groups of Jordan algebras) would fit into this picture, and the global measure-theoretic and dynamical questions reduce to gluing local Haar-based descriptions across overlaps. These gluing
and chart-choice issues at face appear manageable in the compact, finite-dimensional regimes considered here, especially when Krylov-Lie groups provide natural charts tied to the generator set and seed, but we do not develop a general theory of Lietic spaces or Lietic integration in this thesis. The term is primarily introduced for posterity and as a concise way to describe the geometric role played by local Lie-group charts in the Krylov--Lie approximation program.
\end{remark}

\Cref{thm:main_klg_approximation} controls the uniform Hilbert–Schmidt error between the original circuit map and its Krylov–Lie surrogate on compact subsets of the full-rank locus, and shows how this error propagates to expectations of bounded Lipschitz observables on $U(d)$. In many variational settings, however, one is interested in observables that depend only on the action of the circuit on a fixed initial state or on a chosen block seed. In that case, it is natural to refine the approximation theory to a seed-based setting, where the error is measured directly on the orbit of the seed and the observable is Lipschitz with respect to the Hilbert-space norm. The following proposition records this refinement in a form aligned with the notation of \Cref{thm:main_klg_approximation}.

\begin{proposition}[Seed-based Krylov--Lie observable approximation]
\phantomsection
\label{prop:seed-KLG-observable}
Let $M$ be the reachable manifold of a fixed variational ansatz architecture,
realized as a smooth subset of $U(d)$, and let $\mu_M = q\,d\mathrm{vol}_M$ be
the parameter sampling measure on $M$, with $q\in C^1(M)$ and $q\ge 0$.
Fix a normalized initial state $\psi\in\mathcal H$, which we assume to be an element of our block seed $\Psi\in\Sigma_s$ (\Cref{rem:SurjectiveParamBlockSeeds}), and let
\[
\mathcal U\colon M\to U(d),
\qquad
\iota_\Psi^{(k)}\circ\kappa_\Psi^{(k)}\colon M\to U(d)
\]
denote the original circuit map and the represented depth-$k$ Krylov--Lie
surrogate associated to the block seed $\Psi$ as in
\Cref{thm:main_klg_approximation}. For a compact subset $C\subset M$ contained
in the full-rank locus $M_{\mathrm{full}}(\Psi)$, define the seed-based uniform
error
\[
\mathcal E_C^{(\psi)}(k;\Psi)
:=
\sup_{x\in C}
\bigl\|\mathcal U(x)\psi
-
\iota_\Psi^{(k)}\bigl(\kappa_\Psi^{(k)}(x)\bigr)\psi\bigr\|.
\]
Now let $\Phi\colon\mathcal H\to\mathbb R$ be a bounded Lipschitz observable with
respect to the Hilbert-space norm on $\mathcal H$, with Lipschitz constant
$\mathrm{Lip}_\psi(\Phi)$, i.e.
\[
\bigl|\Phi(v)-\Phi(w)\bigr|
\le
\mathrm{Lip}_\psi(\Phi)\,\|v-w\|
\qquad
\text{for all }v,w\in\mathcal H.
\]
Define the induced circuit observable
\[
F_{\psi}\colon U(d)\to\mathbb R,
\qquad
F_{\psi}(U):=\Phi\bigl(U\psi\bigr).
\]
Then, for every compact subset $C\subset M_{\mathrm{full}}(\Psi)$, one has
\begin{equation}
\label{eq:seed-KLG-observable-bound}
\left|
\int_C F_{\psi}(\mathcal U(x))\, d\mu_M(x)
-
\int_{K^{(k)}_{\Psi}}
F_{\psi}(\iota_\Psi^{(k)}(g))\, d\nu_{\Psi,C}^{(k)}(g)
\right|
\le
\mathrm{Lip}_\psi(\Phi)\,\mu_M(C)\,\mathcal E_C^{(\psi)}(k;\Psi),
\end{equation}
where $\nu_{\Psi,C}^{(k)}(E)=\int_E\rho_{\Psi,C}^{(k)}\,d\mu_{K_\Psi^{(k)}}$ is the pushforward measure of \Cref{thm:main_klg_approximation}, represented as a Radon-Nikodym derivative integrated against the Haar measure over the KLG.
\end{proposition}

\begin{proof}
For each $x\in C$, by definition of $F_{\psi}$ and $\Phi$,
\[
\bigl|F_{\psi}(\mathcal U(x))-F_{\psi}(\iota_\Psi^{(k)}(\kappa_\Psi^{(k)}(x)))\bigr|
\le
\mathrm{Lip}_\psi(\Phi)\,
\bigl\|\mathcal U(x)\psi
-
\iota_\Psi^{(k)}\bigl(\kappa_\Psi^{(k)}(x)\bigr)\psi\bigr\|
\le
\mathrm{Lip}_\psi(\Phi)\,\mathcal E_C^{(\psi)}(k;\Psi).
\]
Repeating the steps of the proof of (7) from \Cref{thm:main_klg_approximation} with this bound then yields the desired inequality of \cref{eq:seed-KLG-observable-bound}.
\end{proof}

\begin{remark}[Seed-based versus Hilbert--Schmidt error]
\label{rem:seed-vs-HS}
In \Cref{thm:main_klg_approximation} the uniform approximation error appears in
the Hilbert--Schmidt norm on $U(d)$, e.g.
\[
\mathcal E_C(k;\Psi)
=
\sup_{x\in C}
\bigl\|\mathcal U(x)
-
\iota_\Psi^{(k)}\bigl(\kappa_\Psi^{(k)}(x)\bigr)\bigr\|_{\mathrm{HS}}.
\]
By norm compatibility in finite dimensions one has
\[
\bigl\|\mathcal U(x)\psi
-
\iota_\Psi^{(k)}\bigl(\kappa_\Psi^{(k)}(x)\bigr)\psi\bigr\|
\le
\bigl\|\mathcal U(x)
-
\iota_\Psi^{(k)}\bigl(\kappa_\Psi^{(k)}(x)\bigr)\bigr\|_{\mathrm{op}}
\le
\bigl\|\mathcal U(x)
-
\iota_\Psi^{(k)}\bigl(\kappa_\Psi^{(k)}(x)\bigr)\bigr\|_{\mathrm{HS}}
\]
for each $x\in C$ and unit seed $\psi$. Hence
\[
\mathcal E_C^{(\psi)}(k;\Psi)
\le
\mathcal E_C(k;\Psi),
\]
so the seed-based bound in \Cref{prop:seed-KLG-observable} is always controlled
by the Hilbert--Schmidt error appearing in item~(7) of
\Cref{thm:main_klg_approximation}.

Additionally, observables of the form $F_{\psi}(U)=\Phi(U\psi)$ are
particularly natural in variational quantum algorithms, where one typically
fixes an initial state $\psi$ and studies expectation values or cost
functions built from the parameterized circuit $\mathcal U(x)$ acting on that
state. In such settings, the seed-based error $\mathcal E_C^{(\psi)}(k;\Psi)$
directly measures the approximation quality on the physically relevant orbit of
$\psi$, and the corresponding observable bound
\eqref{eq:seed-KLG-observable-bound} exploits the intrinsic Krylov subspace structure of the ansatz.
\end{remark}

\subsection{Krylov--Lie Error Bounds from Commutator Matching}

In this subsection we record seed-based error estimates for the canonical
Krylov--Lie comparison map on a compact subset of the reachable manifold.
For Hermitian involutive generators, commutator matching controls the
part of the generator action that leaves the Krylov subspace. This gives
ordinary word preservation through approximately half the matching grade,
and hence a factorial circuit-error bound. Prefix KLAs preserve ordinary
words through the full grade and admit the corresponding stronger estimate.
Neither argument requires convergence of a BCH logarithm.

Throughout, $\|\cdot\|$ is the operator norm induced by the Hilbert space
norm on $\mathcal H$, or the Hilbert space norm when applied to a vector.
The generators, circuit architecture, and block seed are fixed as the
Krylov grade varies, as are the circuit charts and their angle functions.

Fix a block seed \(\Psi\in\Sigma_s\) (see
\Cref{rem:SurjectiveParamBlockSeeds}), and let \(\psi\in\mathcal H\)
be a normalized initial state which we assume to be one of the
components of \(\Psi\). Let
\[
\kappa_\Psi^{(k)}\colon M\to K^{(k)}_\Psi,
\qquad
\iota_\Psi^{(k)}\colon K^{(k)}_\Psi\hookrightarrow U(d)
\]
be the depth-\(k\) comparison map and represented embedding from
\Cref{thm:main_klg_approximation}.

\begin{definition}[Seed-level BCH matching up to depth \(k\)]
\phantomsection
\label{def:seed-KLG-matching}
We say that the depth-\(k\) Krylov--Lie algebra
\(\mathfrak l^{(k)}_\Psi\) (\Cref{def:krylov-lie-algebra}) is
\emph{BCH-matched} if, for any corresponding commutator words
\[
w(i_1,\ldots,i_r)\in\mathcal W_{\leq k}
\qquad\text{and}\qquad
\widetilde w(i_1,\ldots,i_r)\in\widetilde{\mathcal W}_{\leq k}
\]
of grade \(r\leq k\) (see \Cref{def:CircWords},
\Cref{def:KrylovSubspace}, and \Cref{def:KrylovAlgebra}), one has
\[
w(i_1,\ldots,i_r)\psi
=
\widetilde w(i_1,\ldots,i_r)\psi
\]
for every seed \(\psi\in\Psi\).
\end{definition}

Standard Krylov--Lie algebras need not satisfy
\Cref{def:seed-KLG-matching} in general, although the failure is more
delicate than one might first expect. Already at grade three the property need not hold. In a compressed
word of grade two, every interior projector is adjacent either to the
seed or to a grade-one vector \(H_j\psi\), both of which lie in
\(\mathcal K^{(k)}_\Psi\). At grade three the expansion instead produces
interior configurations of the form \(P_\Psi H_b P_\Psi\,(H_aH_c\psi)\), and the
intermediate vector
\(H_aH_c\psi=\tfrac12[H_a,H_c]\psi+\tfrac12\{H_a,H_c\}\psi\)
carries a symmetric component that generically escapes the
commutator-spanned subspace \(\mathcal{K}^{(k)}_\Psi\). A direct computation,
using \(w\psi\in \mathcal{K}^{(k)}_\Psi\) for all grades \(r\leq k\) as well as the identity $P_\Psi XP_\Psi\psi-XP_\Psi \psi=-X(I-P_\Psi)\psi-(I-P_\Psi)X\psi$, yields for
\(k\geq 3\) the exact seed-level defect
\[
\widetilde w(i_1,i_2,i_3)\psi-w(i_1,i_2,i_3)\psi
=\tfrac12\,P_\Psi H_{i_2}(I-P_\Psi)\{H_{i_1},H_{i_3}\}\psi
-\tfrac12\,P_\Psi H_{i_1}(I-P_\Psi)\{H_{i_2},H_{i_3}\}\psi,
\]
which does not generally vanish. Higher grades
impose analogous conditions on progressively longer intermediate
products, so seed-level matching must be checked separately for standard
Krylov--Lie algebras.

\begin{remark}[When standard Krylov--Lie algebras are BCH-matched]
\label{rem:standard-KLA-matching}
The defect above vanishes whenever every \emph{interior} product
remains in the subspace. Precisely, if
\(\mathcal P^{(k-1)}_\Psi\subset \mathcal K^{(k)}_\Psi\),
where \(\mathcal P^{(k-1)}_\Psi\) is the depth-\((k-1)\) Krylov prefix
subspace of \Cref{def:prefix-subspace-kla}, 
then the induction of \Cref{thm:word-preservation-prefix}
applies verbatim to the standard Krylov--Lie algebra and grants
seed-level matching through grade \(k\): the terminal depth-\(k\)
monomials need not lie in \(\mathcal K^{(k)}_\Psi\) individually, as only their
antisymmetrized combinations do. This absorption condition is strictly
weaker than being a prefix KLA and is realized nontrivially. For
pairwise anticommuting involutive generators (for instance mutually
anticommuting Pauli strings, \(\{\sigma_a,\sigma_b\}=2\delta_{ab}I\)), every
length-two monomial satisfies
\(H_aH_b\psi=\tfrac12[H_a,H_b]\psi+\delta_{ab}\psi\in \mathcal K^{(2)}_\Psi\),
so the grade-three defect vanishes identically and matching extends
through grade three for every seed; it need not extend to grade four,
since the length-three interior products
\((I-P_\Psi)H_aH_bH_c\psi\) generically escape \(\mathcal K^{(k)}_\Psi\).
More generally, on the fixed-rank charts of
\Cref{sec:dim-matching-seeds}, BCH matching is expressed by vanishing of
the graded defect tensors. These give additional real-algebraic
conditions after clearing the nonzero denominators of the projector.
\end{remark}

The prefix Krylov-Lie model with the canonical comparison map satisfies
this matching condition through the full grade.

\begin{corollary}[Prefix KLA is BCH-matched]
\label{cor:prefix-KLA-BCH-matched}
The depth-\(k\) prefix Krylov--Lie algebra associated to \(\Psi\)
satisfies \Cref{def:seed-KLG-matching}.
\end{corollary}

\begin{proof}
This follows immediately from \Cref{thm:word-preservation-prefix}.
\end{proof}

For Hermitian involutive generators, we can use conjugation and the norm
lost under orthogonal projection to pass from commutator matching to
ordinary word preservation. Conjugation expands in the original
commutator words, while the involution relation supplies a positive
quantity that measures whether a generator takes the compressed state
outside the Krylov subspace.

\begin{lemma}[Seed leakage for Hermitian involutions]
\label{lem:involution-seed-leakage}
Let $P_\Psi$ be the orthogonal projector onto $\mathcal K_\Psi^{(k)}$,
and let $\psi\in\Psi$ be normalized. Suppose that the Hermitian generators
satisfy $H_i^2=I$, and write $\widetilde H_i=P_\Psi H_iP_\Psi$ for their
ambient compressions. For any ordered circuit $\mathcal U$ and its
canonical compressed circuit $\widetilde{\mathcal U}$,
\begin{equation}
\label{eq:involution-seed-leakage}
\bigl\|(I-P_\Psi)H_i\widetilde{\mathcal U}\psi\bigr\|^2
\le 2\bigl\|\bigl(\mathcal U^\dagger H_i\mathcal U
-\widetilde{\mathcal U}^{\dagger}\widetilde H_i
 \widetilde{\mathcal U}\bigr)\psi\bigr\|.
\end{equation}
\end{lemma}

\begin{proof}
The compressed state belongs to $\mathcal K_\Psi^{(k)}$, and $H_i$ is
unitary. Orthogonality and unitarity therefore give
\[
\bigl\|(I-P_\Psi)H_i\widetilde{\mathcal U}\psi\bigr\|^2
=1-\bigl\|\widetilde H_i\widetilde{\mathcal U}\psi\bigr\|^2
=1-\bigl\|\widetilde{\mathcal U}^{\dagger}\widetilde H_i
             \widetilde{\mathcal U}\psi\bigr\|^2.
\]
The last norm is at most one, while
$\|\mathcal U^\dagger H_i\mathcal U\psi\|=1$. Factoring the difference
of squares and applying the reverse triangle inequality proves
\eqref{eq:involution-seed-leakage}. In particular, taking square roots
gives an exponent of $1/2$ and a constant of $\sqrt2$, independently of
the grade, circuit parameters, and dimension.
\end{proof}

\begin{proposition}[Ordinary word preservation from commutator matching]
\label{prop:involution-seed-word-matching}
Suppose that $H_i^2=I$ for every generator and that
$\mathfrak l_\Psi^{(k)}$ is BCH-matched in the sense of
\Cref{def:seed-KLG-matching}, with $k\ge1$. Then, for every seed
$\psi\in\Psi$,
\begin{equation}
\label{eq:involution-seed-word-matching}
H_{i_r}\cdots H_{i_1}\psi
=\widetilde H_{i_r}\cdots\widetilde H_{i_1}\psi
\qquad\left(r\le\left\lceil\frac{k}{2}\right\rceil\right).
\end{equation}
\end{proposition}

\begin{proof}
It suffices to consider a normalized seed. Choose any finite sequence
of generator indices and give its gates independent real angles:
\[
\mathcal U(\theta)
=e^{-i\theta_rH_{i_r}}\cdots e^{-i\theta_1H_{i_1}}.
\]
Define
\[
D_i(\theta):=
\bigl(\mathcal U(\theta)^\dagger H_i\mathcal U(\theta)
-\widetilde{\mathcal U}(\theta)^\dagger\widetilde H_i
 \widetilde{\mathcal U}(\theta)\bigr)\psi.
\]
The conjugation series has, in homogeneous degree $d$, linear
combinations of nested commutators of grade $d+1$
\cite{HallLieGroups2015}. The corresponding coefficients for the
compressed circuit are the same expressions in the compressed generators.
Matching through grade $k$ therefore gives
\[
D_i(\theta)=O(\|\theta\|^k).
\]
By \Cref{lem:involution-seed-leakage},
\[
\bigl\|(I-P_\Psi)H_i\widetilde{\mathcal U}(\theta)\psi\bigr\|^2
=O(\|\theta\|^k).
\]
The vector inside this norm is real analytic. If its first nonzero
homogeneous term has degree $d$, the first nonzero term of its squared
norm has degree $2d$. Indeed, the squared norm of a nonzero polynomial
vector cannot vanish identically. Thus all vector coefficients of
degree below $\lceil k/2\rceil$ vanish. Taking mixed derivatives in the
independent angles and canceling the common factor $(-i)^r$ gives
\[
(I-P_\Psi)H_i\widetilde H_{i_r}\cdots\widetilde H_{i_1}\psi=0
\qquad\left(r<\left\lceil\frac{k}{2}\right\rceil\right).
\]
Starting with $P_\Psi\psi=\psi$, induction on ordinary word length now
proves \eqref{eq:involution-seed-word-matching}: whenever the preceding
product agrees with its compression, the next generator leaves that
vector in $\mathcal K_\Psi^{(k)}$, so its projection does not change it.

The auxiliary circuit is arbitrary because
\Cref{def:seed-KLG-matching} concerns every commutator word through the
specified grade. No convergence of a BCH logarithm is needed here.
\end{proof}

\begin{remark}
For involutive generators, applying the proposition to the seeds gives
the necessary condition
\[
\mathcal P_\Psi^{(\lceil k/2\rceil)}
\subset\mathcal K_\Psi^{(k)}
\]
for BCH matching through grade $k$. This complements the sufficient
condition $\mathcal P_\Psi^{(k-1)}\subset\mathcal K_\Psi^{(k)}$ in
\Cref{rem:standard-KLA-matching}. Thus matching imposes ordinary word
preservation on the original commutator-based subspace; it does not
require replacing that subspace by a prefix construction.
\end{remark}

We now record the factorial estimate supplied by ordinary word
preservation. We first state it for prefix KLAs, where preservation
holds through the full grade, and then apply the same argument to
\Cref{prop:involution-seed-word-matching}. This is the polynomial exactness
argument underlying classical Krylov approximations to the matrix
exponential, applied to ordered words in several generators
\cite[Sections~3--4]{Saad1992Exponential}.

\begin{samepage}
\begin{theorem}[Factorial error bound for prefix Krylov--Lie algebras]
\label{thm:prefix-KLG-factorial}
Let $\mathfrak p_\Psi^{(k)}$ be the grade-$k$ prefix KLA, with canonical
comparison map and normalized initial state $\psi\in\Psi$.
Let $C\subset M$ be covered by finitely many compact sets $C_\alpha$
inside canonical circuit charts with continuous angle functions. Set
\[
B_C:=\max_\alpha\sup_{x\in C_\alpha}
\sum_{j=1}^p|\theta_{\alpha,j}(x)|\,\|H_{i_j}\|.
\]
Then, for every $k\ge0$,
\begin{equation}
\label{eq:prefix-KLG-factorial}
\mathcal E_C^{(\psi)}(k;\Psi)
\le 2\frac{B_C^{k+1}}{(k+1)!}.
\end{equation}
This estimate applies to arbitrary Hermitian generators and requires no BCH convergence hypothesis.
\end{theorem}
\end{samepage}

\begin{proof}
On a fixed chart, expand each factor of the ordered circuit exponential.
The homogeneous degree-$r$ coefficient is
\[
\mathcal U_r^{(\alpha)}(x)
=(-i)^r\!\!\sum_{m_1+\cdots+m_p=r}
\frac{\theta_{\alpha,p}(x)^{m_p}\cdots
\theta_{\alpha,1}(x)^{m_1}}{m_p!\cdots m_1!}
H_{i_p}^{m_p}\cdots H_{i_1}^{m_1},
\]
where the sum is over nonnegative integers. Every summand is an ordinary
word of length $r$. Therefore \Cref{thm:word-preservation-prefix} gives
\[
\mathcal U_r^{(\alpha)}(x)\psi
=\widetilde{\mathcal U}_{k,r}^{(\alpha)}(x)\psi
\qquad (r\le k).
\]
To bound the Taylor remainder, let
\[
\mathcal U^{(\alpha)}(t,x)
=e^{-it\theta_{\alpha,p}(x)H_{i_p}}\cdots
 e^{-it\theta_{\alpha,1}(x)H_{i_1}},
\qquad 0\le t\le1,
\]
and use $\widetilde{\mathcal U}_k^{(\alpha)}(t,x)$ for its compressed
counterpart. For real $t$, every exponential factor is unitary.
The product rule, submultiplicativity of the operator norm, and the
multinomial theorem therefore give
\[
\bigl\|\partial_t^r\mathcal U^{(\alpha)}(t,x)\bigr\|
\le\left(\sum_{j=1}^p
|\theta_{\alpha,j}(x)|\,\|H_{i_j}\|\right)^r
\le B_C^r.
\]
The same estimate holds for the compressed circuit, since orthogonal
compression preserves Hermiticity and does not increase the operator norm.
Taylor's integral remainder now gives
\begin{align*}
\left\|\mathcal U(x)-\sum_{r=0}^k\mathcal U_r^{(\alpha)}(x)\right\|
&\le\frac{1}{k!}\int_0^1(1-t)^k
\bigl\|\partial_t^{k+1}\mathcal U^{(\alpha)}(t,x)\bigr\|\,dt\\
&\le\frac{B_C^{k+1}}{(k+1)!},
\end{align*}
and likewise for the compressed circuit. Their Taylor polynomials agree
on $\psi$ through grade $k$, so subtracting and taking the supremum over
$C$ proves \eqref{eq:prefix-KLG-factorial}. Unitarity is what removes the
exponential prefactor that would arise from summing the norms of the
infinite Taylor tail; the analogous improvement for a single matrix
exponential is discussed in \cite[Section~4.2]{Saad1992Exponential}.
\end{proof}

\begin{remark}
The scale $B_C$ includes every gate in the circuit. For normalized Pauli
generators, for example,
\[
B_C\le p\max_{\alpha,j}\sup_{x\in C_\alpha}
|\theta_{\alpha,j}(x)|.
\]
\end{remark}

\begin{samepage}
The same estimate now gives a factorial bound for BCH-matched
Krylov--Lie algebras with involutive generators.

\begin{theorem}[Factorial seed-based error bound from commutator matching]
\label{thm:seed-KLG-factorial}
Fix the Hermitian generators, circuit architecture, and seed $\Psi$,
with normalized initial state $\psi\in\Psi$. Suppose that $H_i^2=I$
for every generator and that the depth-$k$ KLA is BCH-matched, with
$k\ge1$. Let $C\subset M$ be covered by finitely many compact sets
$C_\alpha$ inside canonical circuit charts with continuous angle
functions, and define $B_C$ as in \Cref{thm:prefix-KLG-factorial}.
Then the canonical comparison map satisfies
\begin{equation}
\label{eq:seed-KLG-factorial}
\mathcal E_C^{(\psi)}(k;\Psi)
\le 2\frac{B_C^{\lceil k/2\rceil+1}}
{(\lceil k/2\rceil+1)!}.
\end{equation}
\end{theorem}
\end{samepage}

\begin{proof}
By \Cref{prop:involution-seed-word-matching}, the original and compressed
ordinary words agree on the seed through length $\lceil k/2\rceil$.
The proof of \Cref{thm:prefix-KLG-factorial} uses only this word
preservation and the operator-norm bounds for the original and compressed
generators. Applying its Taylor remainder estimate at order
$\lceil k/2\rceil$ gives \eqref{eq:seed-KLG-factorial}.
\end{proof}

\begin{remark}
The theorem applies in particular to Pauli generators. It also extends
to Hermitian generators with two eigenvalues: subtracting the spectral
midpoint and dividing by half the spectral gap gives an involution.
The scalar phases agree on the compression space, and matching
passes to these normalized generators. The same word-preservation
argument then gives \eqref{eq:seed-KLG-factorial}, with $B_C$ computed
from the original generators.

For fixed architecture and compact region, $B_C$ is independent of $k$.
The factorial estimate therefore has no unspecified transfer constant;
for $B_C>0$, the logarithm of its right-hand side is
$-\tfrac{k}{2}\log k+O(k)$. Its conclusion applies to the grades at
which matching holds. Increasing the grade of a standard KL subspace
does not establish this hypothesis by itself.
\end{remark}

\subsection{Seed-Based Approximation of Circuit Logarithms}

The factorial estimates control the output of the circuit on the seed.
For questions involving an effective generator, it is also useful to
compare the action of the circuit logarithms. This requires a choice of
branch. We use the principal logarithm and assume that the spectra of
both circuits remain uniformly separated from its branch cut
\cite[Ch.~11]{Higham2008Functions}. The resulting estimate holds
throughout the compact region, without requiring convergence of the BCH
series in the original gate angles.

\begin{samepage}
\begin{lemma}[Polynomial approximation of the logarithm on a spectral arc]
\label{lem:log-polynomial-approximation}
Fix $0<\delta<\pi$ and set
\[
E_\delta:=\{e^{it}:|t|\le\pi-\delta\}.
\]
There exist constants $M_\delta,A_\delta\ge1$ and $0<q_\delta<1$,
and polynomials $p_n(z)=\sum_{j=0}^n a_{n,j}z^j$, such that, for every
$n\ge1$,
\begin{equation}
\label{eq:log-polynomial-approximation}
\sup_{z\in E_\delta}|\log z-p_n(z)|\le M_\delta q_\delta^n,
\qquad
\sum_{j=0}^n|a_{n,j}|\le M_\delta A_\delta^n.
\end{equation}
The constants depend only on $\delta$.
\end{lemma}
\end{samepage}

\begin{proof}
The arc has connected, regular complement, and the principal logarithm
is analytic on a neighborhood of it. Geometric polynomial approximation
therefore follows from Walsh's theorem
\cite[Theorem~4.1]{Saff2010PotentialTheory}. The approximating polynomials
are uniformly bounded on the arc. The Bernstein--Walsh inequality
\cite[Lemma~3.7]{Saff2010PotentialTheory} then gives a constant
$A_\delta\ge1$ for which
\[
\max_{|z|=2}|p_n(z)|
\le A_\delta^n\max_{z\in E_\delta}|p_n(z)|.
\]
Cauchy's coefficient estimate on $|z|=2$ gives
$|a_{n,j}|\le2^{-j}\max_{|z|=2}|p_n(z)|$.
Summing over $j$ and enlarging $M_\delta$ proves the second bound.
\end{proof}

\begin{theorem}[Uniform seed-based logarithm approximation]
\label{thm:seed-KLG-logarithm}
In the setting of \Cref{thm:seed-KLG-factorial}, write
\[
\widetilde{\mathcal U}_k(x)
:=\iota_\Psi^{(k)}\bigl(\kappa_\Psi^{(k)}(x)\bigr).
\]
Suppose that, for some fixed $0<\delta<\pi$,
\[
\operatorname{spec}\mathcal U(x)\cup
\operatorname{spec}\widetilde{\mathcal U}_k(x)
\subset E_\delta
\qquad (x\in C)
\]
for every grade under consideration. Then
\begin{equation}
\label{eq:seed-KLG-logarithm}
\begin{split}
&\sup_{x\in C}
\bigl\|(\log\mathcal U(x)-\log\widetilde{\mathcal U}_k(x))\psi\bigr\|\\
&\qquad\le 2M_\delta\inf_{n\ge1}
\left\{q_\delta^n+
A_\delta^n\frac{(nB_C)^{\lceil k/2\rceil+1}}
{(\lceil k/2\rceil+1)!}\right\},
\end{split}
\end{equation}
where the infimum is over positive integers. For prefix KLAs in the
setting of \Cref{thm:prefix-KLG-factorial}, the same estimate holds
with $\lceil k/2\rceil$ replaced by $k$.

In particular, for fixed $\delta$ and $B_C$, the right-hand side is
bounded by $D\sigma^{\lceil k/2\rceil+1}$, or $D\sigma^{k+1}$ in the
prefix case, for constants $D>0$ and $0<\sigma<1$ independent of $k$.
\end{theorem}

\begin{proof}
Put $r=\lceil k/2\rceil+1$, or $r=k+1$ for a prefix KLA. In either
case, ordinary words agree on $\psi$ through length $r-1$.
For each positive integer $j$, repeat the original circuit $j$ times
and repeat the same compressed circuit. Scaling every gate angle by
$t$, the Taylor coefficients of degree $m$ of these repeated circuits
are linear combinations of ordinary words of length $m$, with the
same scalar coefficients. They therefore agree on the seed for $m<r$.
The total weighted gate-angle sum is at most $jB_C$, so the Taylor
remainder proof of \Cref{thm:prefix-KLG-factorial} gives
\begin{equation}
\label{eq:seed-KLG-repeated-circuit}
\sup_{x\in C}
\bigl\|(\mathcal U(x)^j-\widetilde{\mathcal U}_k(x)^j)\psi\bigr\|
\le 2\frac{(jB_C)^r}{r!}.
\end{equation}
The same compression is used for all repetitions, since word preservation
holds independently of the number and ordering of the gates.

Apply \Cref{lem:log-polynomial-approximation} to the common spectral
arc. Since both circuits are unitary, the spectral theorem bounds each
operator-norm error in replacing its logarithm by $p_n$ by
$M_\delta q_\delta^n$ \cite[Ch.~1]{Higham2008Functions}. Thus
\begin{align*}
\bigl\|(\log\mathcal U(x)-\log\widetilde{\mathcal U}_k(x))\psi\bigr\|
&\le 2M_\delta q_\delta^n+
\sum_{j=1}^n|a_{n,j}|
\bigl\|(\mathcal U(x)^j-\widetilde{\mathcal U}_k(x)^j)\psi\bigr\|\\
&\le 2M_\delta q_\delta^n+
\frac{2B_C^r}{r!}\sum_{j=1}^n|a_{n,j}|j^r\\
&\le 2M_\delta\left(q_\delta^n+
A_\delta^n\frac{(nB_C)^r}{r!}\right).
\end{align*}
Taking the supremum over $C$ and then the infimum over $n$ proves
\eqref{eq:seed-KLG-logarithm}.

To obtain geometric decay, choose $\eta>0$ sufficiently small that
$A_\delta^\eta e\eta B_C<1$, and take $n=\lfloor\eta r\rfloor$
for sufficiently large $r$. Since $r!\ge(r/e)^r$,
\[
q_\delta^n\le q_\delta^{-1}(q_\delta^\eta)^r,
\qquad
A_\delta^n\frac{(nB_C)^r}{r!}
\le\bigl(A_\delta^\eta e\eta B_C\bigr)^r.
\]
Both bases are strictly below one. Taking their maximum and increasing
the prefactor to cover the finitely many smaller values of $r$ gives
the asserted constants $D$ and $\sigma$.
\end{proof}

\begin{remark}
The infimum in \eqref{eq:seed-KLG-logarithm} retains the factorial
remainder for quantitative estimates. Geometric decay requires increasing
the polynomial degree with the matching grade; fixing it leaves a nonzero
polynomial-approximation bound. 
\end{remark}

\subsection{Commutator Budget and Krylov--Lie Approximation Quality}

We now formalize the intuitive picture that higher-order BCH commutator
directions live in a quotient of the depth-$k$ Krylov--Lie algebra by the
span of the compressed Hermitian generators, with a dimension budget that
constrains how much genuinely new geometry remains beyond first order. This
quotient describes the available higher-order directions. Quantitative
control of the approximation is supplied by
\Cref{thm:seed-KLG-factorial} under its generator and matching
hypotheses, or by \Cref{thm:prefix-KLG-factorial} for prefix KLAs.

For this dimension count, only the individual homogeneous BCH
coefficients are needed. They are finite Lie polynomials, defined by
scaling the gate angles by $t$ and expanding the logarithm near $t=0$
\cite{BCHDynkin}\cite[Ch.~5]{HallLieGroups2015}. Their sum need not converge
at $t=1$. Thus the budget applies throughout the parameter region; where
the series does converge at $t=1$, it also describes the corresponding
BCH logarithm. The principal-logarithm estimate of
\Cref{thm:seed-KLG-logarithm} is a separate statement and does not
establish this convergence \cite{BiagiBonfiglioliMatone2018}.

\begin{lemma}[Higher-commutator BCH sector modulo the generator span]
\label{lem:BCH-commutator-sector-quotient}
Let $M\subset U(d)$ be a reachable manifold of dimension $m$, and let
$\mathcal S=\{H_1,\dots,H_n\}$ be a fixed Hermitian generator set for a variational ansatz in the sense of \Cref{rem:VQAConventions}. Fix a normalized block
seed $\Psi\in\Sigma_s$ (\Cref{rem:SurjectiveParamBlockSeeds}), and let
\[
\mathcal K_k:=\mathcal K_{\Psi}^{(k)}(\mathcal S)\subset\mathcal H
\]
denote the depth-$k$ Krylov--Lie subspace generated by $\mathcal S$ and $\Psi$, as in
\Cref{def:krylov-lie-subspace}. Let
\[
\mathfrak l_k:=\mathfrak l_{\Psi}^{(k)}(\mathcal S)\subset\End(\mathcal K_k)
\]
be the corresponding depth-$k$ Krylov--Lie algebra, as in
Definition~3.2.8. Assume we are in the dimension-matched regime
\[
\dim\mathfrak l_k = m.
\]

Let $C\subset M$ be compact and fix canonical circuit charts $U_\alpha$
around $C$, with angle functions $\theta_{\alpha,j}$. Write
\[
\widetilde H_j^{(k)}:=P^{(k)}_\Psi H_jP^{(k)}_\Psi\big|_{\mathcal K_k},
\qquad j\in\{1,\dots,n\},
\]
where $P^{(k)}_\Psi$ is the orthogonal projector onto $\mathcal K_k$.
Let $\Omega_{k,r}^{(\alpha)}(x)$ denote the homogeneous coefficients
obtained by scaling every gate angle by $t$. Using the Hermitian
identification of \Cref{rem:VQAConventions}, they are specified near
$t=0$ by
\[
\log\!\left(
 e^{-it\theta_{\alpha,p}(x)\widetilde H_{i_p}^{(k)}}\cdots
 e^{-it\theta_{\alpha,1}(x)\widetilde H_{i_1}^{(k)}}
\right)
=i\sum_{r\ge1}t^r\Omega_{k,r}^{(\alpha)}(x).
\]
No convergence at $t=1$ is assumed.

Define the real generator subspace
\[
G_k
:=
\mathrm{span}_{\mathbb R}\{\widetilde H_1^{(k)},\dots,\widetilde H_n^{(k)}\}
\subset\mathfrak l_k,
\qquad
n_k:=\dim G_k,
\]
and let
\[
\pi_k:\mathfrak l_k\to \mathfrak l_k/G_k
\]
be the quotient map. Define the higher-commutator BCH sector over $C$ by
\[
\mathcal C_k(C)
:=
\mathrm{span}_{\mathbb R}
\Bigl\{
\pi_k\bigl(\Omega_{k,r}^{(\alpha)}(x)\bigr)
\, \Big \vert \,\;
x\in C,\ r\ge 2,\ U_\alpha\ni x
\Bigr\}
\subset \mathfrak l_k/G_k.
\]
Then:
\begin{enumerate}
\item The degree-one BCH term lies in the generator subspace:
\[
\Omega_{k,1}^{(\alpha)}(x)\in G_k
\qquad\text{for all }x\in C\text{ and all charts }U_\alpha\ni x,
\]
and therefore $\pi_k\bigl(\Omega_{k,1}^{(\alpha)}(x)\bigr)=0$.

\item All genuinely new BCH directions beyond the generators are encoded by
$\mathcal C_k(C)$, in the sense that any class
$\pi_k\bigl(\Omega_{k,r}^{(\alpha)}(x)\bigr)\neq 0$ with $r\ge 2$ represents
a genuinely new commutator direction modulo $G_k$.

\item One has the dimension bound
\begin{equation}
\label{eq:BCH-commutator-budget-quotient}
\dim \mathcal C_k(C)
\le
\dim\bigl(\mathfrak l_k/G_k\bigr) =
m-n_k.
\end{equation}
In particular, the total number of independent directions that higher-order
BCH commutators of the canonical comparison map can contribute modulo the
generator span is at most $m-n_k$.
\end{enumerate}
\end{lemma}

\begin{proof}
For each fixed $x$, the scaled circuit equals the identity at $t=0$.
The matrix logarithm series therefore converges near $t=0$. More
explicitly, the exponential and logarithm series admit an absolute
majorant whenever
\[
|t|\sum_{j=1}^p|\theta_{\alpha,j}(x)|\,\|H_{i_j}\|<\log 2,
\]
since orthogonal compression does not increase the operator norm
\cite[Ch.~11]{Higham2008Functions}. Such a neighborhood always exists,
regardless of the unscaled gate angles. The BCH formula identifies each
coefficient as a finite Lie polynomial of grade $r$ in the compressed
generators, so $\Omega_{k,r}^{(\alpha)}(x)\in\mathfrak l_k$
\cite{BCHDynkin}\cite[Ch.~5]{HallLieGroups2015}. These coefficients are
well-defined throughout the chart even if the series diverges at $t=1$.
The degree-one coefficient, in the Hermitian convention above, is
\[
\Omega_{k,1}^{(\alpha)}(x)
=-\sum_{j=1}^p\theta_{\alpha,j}(x)\widetilde H_{i_j}^{(k)}\in G_k.
\]
Hence $\pi_k\bigl(\Omega_{k,1}^{(\alpha)}(x)\bigr)=0$ for all
$x\in C$ and all charts $U_\alpha\ni x$, proving (1).

Every homogeneous coefficient $\Omega_{k,r}^{(\alpha)}(x)$ lies in $\mathfrak l_k$, so
passing to the quotient $\mathfrak l_k/G_k$ removes precisely the directions
already visible at first order, namely the generator span
\[
G_k=\mathrm{span}_{\mathbb R}\bigl\{\widetilde H_j^{(k)} \, \big \vert \, 1 \leq j \leq n\bigr\}.
\]
Thus any nonzero class $\pi_k\bigl(\Omega_{k,r}^{(\alpha)}(x)\bigr)$ with
$r\ge 2$ represents a genuinely new commutator direction modulo $G_k$, and
the span of all such classes over $x\in C$, $r\ge 2$, and charts $U_\alpha$
meeting $C$ is exactly $\mathcal C_k(C)$. This proves (2).

Finally, $\mathcal C_k(C)$ is a real subspace of $\mathfrak l_k/G_k$, so
\[
\dim \mathcal C_k(C)
\le
\dim\bigl(\mathfrak l_k/G_k\bigr)
=
\dim\mathfrak l_k-\dim G_k
=
m-n_k,
\]
which is \eqref{eq:BCH-commutator-budget-quotient}.
\end{proof}
\begin{proposition}[Layered generator growth consumes commutator budget]
\label{prop:generator-growth-budget}
In the setting of Lemma~\ref{lem:BCH-commutator-sector-quotient}, consider a
family of ansatz architectures indexed by an external layer parameter $L$,
each with generator set $\mathcal S_L=\{H_{L,1},\dots,H_{L,n_L}\}$, reachable
manifold $M_L\subset U(d)$ of dimension $m_L$, depth-$k$ Krylov--Lie
subspace
\[
\mathcal K_k^{(L)}:=\mathcal K_\Psi^{(k)}(\mathcal S_L)\subset\mathcal H,
\]
and depth-$k$ Krylov--Lie algebra
\[
\mathfrak l_k^{(L)}:=\mathfrak l_\Psi^{(k)}(\mathcal S_L)
\subset\End\bigl(\mathcal K_k^{(L)}\bigr)
\]
in the dimension-matched regime
\[
\dim\mathfrak l_k^{(L)} = m_L.
\]
Let
\[
G_k^{(L)}
:=
\mathrm{span}_{\mathbb R}\{\widetilde H_{L,1}^{(k)},\dots,
\widetilde H_{L,n_L}^{(k)}\}
\subset\mathfrak l_k^{(L)},
\qquad
n_k^{(L)}:=\dim G_k^{(L)},
\]
denote the generator subspace, where
$\widetilde H_{L,j}^{(k)}:=P_k^{(L)} H_{L,j} P_k^{(L)}\big|_{\mathcal K_k^{(L)}}$
are the compressed Hermitian generators on $\mathcal K_k^{(L)}$ and
$P_k^{(L)}$ is the corresponding orthogonal projector. For any compact subset
$C\subset M_L$ one has
\[
\dim \mathcal C_k^{(L)}(C)\le m_L - n_k^{(L)},
\]
where $\mathcal C_k^{(L)}(C)$ is defined as in
Lemma~\ref{lem:BCH-commutator-sector-quotient} for the architecture
$\mathcal S_L$.

If passing from layer $L$ to layer $L+1$ introduces $\delta_L\ge 1$ new
linearly independent generator directions at depth $k$, so that
\[
n_k^{(L+1)}\ge n_k^{(L)}+\delta_L,
\]
while the reachable manifold dimension $m_L$ remains fixed along this
family, then the available BCH commutator budget is reduced by at least
$\delta_L$:
\[
m_L-n_k^{(L+1)}
\le
m_L-n_k^{(L)}-\delta_L.
\]
In particular, if $n_k^{(L)}$ becomes close to $m_L$ for some architecture
$L$, then
\[
m_L-n_k^{(L)}\ll m_L,
\]
and the higher-commutator sector $\mathcal C_k^{(L)}(C)$ has very limited
dimension.
\end{proposition}

\begin{proof}
Applying Lemma~\ref{lem:BCH-commutator-sector-quotient} to the homogeneous
coefficients of each architecture $L$ gives
\[
\dim \mathcal C_k^{(L)}(C)\le m_L-n_k^{(L)}.
\]
If $n_k^{(L+1)}\ge n_k^{(L)}+\delta_L$ and $m_{L+1}=m_L$, then
\[
m_L-n_k^{(L+1)}
\le
m_L-n_k^{(L)}-\delta_L.
\]
The final claim is just a restatement of the fact that
$\dim \mathcal C_k^{(L)}(C)\le m_L-n_k^{(L)}$.
\end{proof}

\begin{remark}[Quotient budget and approximation quality]
\label{rem:quotient-budget-approximation}
In the regimes where the dimension-matched depth-$k$ Krylov--Lie algebra
$\mathfrak l_k^{(L)}$ has generator span $G_k^{(L)}$ well below full
dimension, the quotient budget $m_L-n_k^{(L)}$ is large and there is room for
higher-order BCH commutators of the canonical comparison map
$\kappa_\Psi^{(k)}$ to contribute genuinely new directions beyond the
compressed generators. The budget counts independent directions modulo the generators, rather
than their coefficients or their action on the seed.

When $n_k^{(L)}$ approaches $m_L$, the quotient dimension
$m_L-n_k^{(L)}$ shrinks, leaving fewer independent higher-commutator
directions beyond the generator span. When the higher-order structure to be captured requires comparably many
independent directions outside the generator span, a smaller budget
therefore places greater pressure on the approximation. This is a
heuristic interpretation: higher-order terms may also lie in directions
already present, and the dimension bound alone does not control their
size or the resulting error. Moreover,
changing $k$ changes the compressed algebra, so dimension matching must
still be checked at each grade to which the budget is applied.

Matching supplies a separate constraint on the compressed subspace.
For involutive generators, \Cref{prop:involution-seed-word-matching}
gives the necessary condition
\[
\mathcal P_\Psi^{(\lceil k/2\rceil)}\subset\mathcal K_\Psi^{(k)}
\]
for commutator matching through grade $k$. This constrains the dimension
of a subspace capable of matching to that grade, in addition to the
quotient budget on independent commutator directions.

Under the hypotheses of \Cref{thm:seed-KLG-factorial}, the
factorial seed-based Krylov--Lie error bound
\[
\mathcal E_C^{(\psi)}(k;\Psi)
\le
\frac{2B_C^{\lceil k/2\rceil+1}}{(\lceil k/2\rceil+1)!}
\]
controls the error of the canonical comparison map at finite $k$.
For prefix KLAs, \Cref{thm:prefix-KLG-factorial} gives the estimate at
the full grade. Under its additional spectral-gap hypothesis,
\Cref{thm:seed-KLG-logarithm} also controls the principal logarithms on
the seed. The quotient budget describes the dimension available to
higher commutators; it does not by itself establish matching or any
of these error estimates.
\end{remark}

\section{Concentration of Measure for Krylov--Lie Approximators}

We close this chapter by recording two concentration of measure statements:
a general geometric concentration result on compact represented Krylov--Lie
groups, and a transport principle showing that, for seed-based
observables, concentration properties of the Krylov--Lie proxy transfer to
the original circuit up to the factorial approximation error from
\Cref{thm:seed-KLG-factorial}, or the full-grade prefix error from
\Cref{thm:prefix-KLG-factorial}, under their respective hypotheses.

\begin{theorem}[Geometric concentration on compact Krylov--Lie groups]
\label{thm:KLG-concentration}
Let $K$ be a compact connected represented Krylov--Lie group equipped with a
bi-invariant Riemannian metric and normalized Haar measure $\mu_K$. Then
$\mu_K$ satisfies a logarithmic Sobolev inequality with some constant
$\alpha>0$:
\[
\Ent_{\mu_K}(f^2)
\le
2\alpha
\int_K |\nabla f|^2\,d\mu_K
\qquad\text{for all Lipschitz }f:K\to\mathbb R,
\]
where
\[
\Ent_{\mu_K}(f^2)
:=
\int_K f^2\log(f^2)\,d\mu_K
-
\left(\int_K f^2\,d\mu_K\right)
\log\left(\int_K f^2\,d\mu_K\right)
\]
denotes the entropy of $f^2$ with respect to $\mu_K$. In particular, if
$F:K\to\mathbb R$ is $\lambda$-Lipschitz with respect to the Riemannian
distance on $K$, then for every $\varepsilon>0$,
\begin{equation}
\label{eq:KLG-subgaussian}
\mu_K\left(
\bigl|F-\int_K F\,d\mu_K\bigr|\ge \varepsilon
\right)
\le
2\exp\left(
-\frac{\varepsilon^2}{2\alpha\lambda^2}
\right).
\end{equation}
\end{theorem}

\begin{proof}
Since $K$ is a compact connected Lie group, its bi-invariant Riemannian metric
has nonnegative Ricci curvature bounded below, and the associated heat kernel
satisfies a logarithmic Sobolev inequality with some finite constant
$\alpha>0$; see, for example, Naor for the general
Riemannian theory \cite{NaorConcentration}. The Herbst argument applied to the logarithmic Sobolev
inequality yields the sub-Gaussian concentration estimate
\eqref{eq:KLG-subgaussian} for every Lipschitz $F$, as claimed.
\end{proof}

\begin{remark}
The theorem above should be viewed as a geometric baseline: once the reachable
manifold is faithfully approximated by a compact dimension-matched or
dimension-reduced Krylov--Lie group $K$ on a compact seed orbit, Lipschitz
losses on $K$ already concentrate under the Haar measure at a scale governed by
the intrinsic geometry of $K$. The genuinely finite-depth, non-Haar structure
of a variational ansatz then shows up as a deviation from this geometric
baseline, rather than as a complete absence of concentration phenomena, as we shall now show.
\end{remark}

\begin{proposition}[Concentration for the pushforward measure on $K$]
\phantomsection
\label{prop:pushforward-KLG-concentration}
Let $K$ be a compact represented Krylov--Lie group equipped with a
bi-invariant Riemannian metric and normalized Haar measure $\mu_K$. Let
$\nu$ be a probability measure on $K$ absolutely continuous with respect
to $\mu_K$, with density $q=d\nu/d\mu_K$. Assume $\|q\|_\infty := \operatorname{ess\,sup}_{g\in K} q(g) < \infty$.

Let $F:K\to\mathbb R$ be $\lambda$-Lipschitz with respect to the Riemannian
distance on $K$, and let
\[
m_K:=\int_K F\,d\mu_K,
\qquad
m_\nu:=\int_K F\,d\nu.
\]
Set
\[
\Delta := |m_\nu-m_K|.
\]
Then, for every $\varepsilon>0$,
\begin{equation}
\label{eq:pushforward-concentration-shifted}
\nu\bigl(\{|F-m_\nu|\ge \varepsilon\}\bigr)
\le
\min\!\left\{
1,\,
2\|q\|_\infty
\exp\left(
-\frac{(\varepsilon-\Delta)_+^2}{2\alpha\lambda^2}
\right)
\right\},
\end{equation}
where $\alpha>0$ is the logarithmic Sobolev constant from
\Cref{thm:KLG-concentration} and $(x)_+ := \max\{x,0\}$.

In particular, if $\Delta<\varepsilon$, then
\[
\nu\bigl(\{|F-m_\nu|\ge \varepsilon\}\bigr)
\le
2\|q\|_\infty
\exp\left(
-\frac{(\varepsilon-\Delta)^2}{2\alpha\lambda^2}
\right).
\]
\end{proposition}

\begin{proof}
Fix $\varepsilon>0$ and write $\Delta:=|m_\nu-m_K|$. By the triangle inequality,
\[
|F-m_\nu|
\le
|F-m_K|+\Delta.
\]
Hence, if $\varepsilon>\Delta$ and $|F-m_\nu|\ge\varepsilon$, then necessarily
\[
|F-m_K|\ge \varepsilon-\Delta.
\]
Therefore, because the above implies that
\[
\{|F-m_\nu|\ge\varepsilon\}
\subset
\{|F-m_K|\ge \varepsilon-\Delta\}
\]
we have the inequality
\[
\nu\bigl(\{|F-m_\nu|\ge\varepsilon\}\bigr)
\le
\nu\bigl(\{|F-m_K|\ge \varepsilon-\Delta\}\bigr).
\]
Using absolute continuity and the density bound,
\[
\nu\bigl(\{|F-m_K|\ge \varepsilon-\Delta\}\bigr)
\le
\|q\|_\infty\,
\mu_K\bigl(\{|F-m_K|\ge \varepsilon-\Delta\}\bigr).
\]
By \Cref{thm:KLG-concentration},
\[
\mu_K\bigl(\{|F-m_K|\ge t\}\bigr)
\le
2\exp\left(-\frac{t^2}{2\alpha\lambda^2}\right)
\qquad\text{for all }t>0.
\]
Applying this with $t=\varepsilon-\Delta$ yields
\[
\nu\bigl(\{|F-m_\nu|\ge\varepsilon\}\bigr)
\le
2\|q\|_\infty
\exp\left(
-\frac{(\varepsilon-\Delta)^2}{2\alpha\lambda^2}
\right)
\qquad (\varepsilon>\Delta).
\]
If $\varepsilon\le\Delta$, we use the trivial bound
\[
\nu\bigl(\{|F-m_\nu|\ge\varepsilon\}\bigr)\le 1.
\]
Combining the two cases gives \eqref{eq:pushforward-concentration-shifted}.
\end{proof}

\begin{remark}
The trivial bound for the case $\Delta \geq \varepsilon$ is not a technical failure of the argument. It marks the regime in which the pushforward mean has drifted too far from the Haar mean for Haar concentration on $K$ to determine concentration of $\nu$ about its own mean without added control on the density.
\end{remark}

We now show that, for observables depending on the seed orbit, concentration properties of the Krylov--Lie surrogate transfer directly to the original circuit up to its seed-based approximation error. Substituting \Cref{thm:seed-KLG-factorial,thm:prefix-KLG-factorial} gives factorial control at the respective orders under their hypotheses. This is the concentration-of-measure
analogue of the approximation theorem.
\begin{theorem}[Transport of concentration from Krylov--Lie proxy]
\label{prop:KLG-concentration-transport}
Let $C\subset M$ be a compact subset of the reachable manifold, and fix a
normalized block seed $\Psi\in\Sigma_s\subset\mathcal H$ (see \Cref{rem:SurjectiveParamBlockSeeds}) together with a
choice of initial state $\psi\in\Psi$. Let
\[
\mathcal U: M \to U(d),
\qquad
\iota_{\Psi}^{(k)} \circ \kappa_{\Psi}^{(k)}: M \to U(d)
\]
denote the original circuit and the ambient canonical depth-$k$
Krylov--Lie comparison unitary associated to $\Psi$, as in
\Cref{prop:seed-KLG-observable}. Define the seed-based uniform
approximation error
\[
\mathcal E_C^{(\psi)}(k; \, \Psi)
:=
\sup_{x\in C}
\bigl\|\mathcal U(x)\psi-\iota_{\Psi}^{(k)}\bigl( \kappa_{\Psi}^{(k)}(x)\bigr)\psi\bigr\|.
\]
Let $\mu_C:=\mu_M|_C$ be a probability measure on $C$ obtained
by restricting a probability measure $\mu_M$ on $M$, and let
$\nu$ denote the pushforward of $\mu_C$ under $\kappa_\Psi^{(k)}$, transported
to $K_\Psi^{(k)}$ via the represented unitary action (so that observables
composed with $\iota_{\Psi}^{(k)} \, \circ \, \kappa_{\Psi}^{(k)}$ are naturally evaluated
against $\nu$). 

Let $\Phi:\mathcal H\to\mathbb R$ be bounded and Lipschitz with constant
$\mathrm{Lip}_\psi(\Phi)$ with respect to the Hilbert space norm on
$\mathcal H$, and define the seed-based circuit observables
\[
F_\psi(U):=\Phi\bigl(U\psi\bigr).
\]
Then for every $x\in C$ one has
\begin{equation}
\label{eq:seed-observable-pointwise}
\bigl|F_\psi \bigl(\mathcal U(x) \bigr)-F_\psi \bigl(\iota_{\Psi}^{(k)}\bigl( \kappa_{\Psi}^{(k)}(x)\bigr)\bigr)\bigr|
\le
\mathrm{Lip}_\psi(\Phi)\,\mathcal E_C^{(\psi)}(k; \, \Psi).
\end{equation}
Consequently,
\begin{align}
\bigl|\mathbb E_{\mu_C} \bigl [ F_\psi \bigl(\mathcal U(x) \bigr) \bigr] -\mathbb E_{\nu} \bigl [ F_\psi \bigl (\iota_{\Psi}^{(k)}(g)\bigr) \bigr ]\bigr|
&\le
\mathrm{Lip}_\psi(\Phi)\,\mathcal E_C^{(\psi)}(k; \, \Psi),
\\
\bigl|\mathrm{Var}_{\mu_C}\bigl[F_\psi \bigl(\mathcal U(x) \bigr)\bigr]-\mathrm{Var}_{\nu}\bigl[F_\psi \bigl (\iota_{\Psi}^{(k)}(g)\bigr) \bigr]\bigr|
&\le
4\|\Phi\|_\infty\,\mathrm{Lip}_\psi(\Phi)\,\mathcal E_C^{(\psi)}(k; \, \Psi),
\end{align}
where
\[
\|\Phi\|_\infty:=\operatorname{}\sup_{v\in\mathcal H}|\Phi(v)|.
\]

Moreover, if $F_\psi \bigl (\iota_{\Psi}^{(k)}\bigl( \kappa_{\Psi}^{(k)}(x)\bigr)\bigr)$ satisfies a concentration inequality
with respect to $\nu$,
\[
\nu\bigl(\bigl\{\bigl|F_\psi \bigl (\iota_{\Psi}^{(k)}\bigl( \kappa_{\Psi}^{(k)}(x)\bigr)\bigr)-\mathbb E_{\nu} \bigl[F_\psi \bigl (\iota_{\Psi}^{(k)}(g)\bigr) \bigr]\bigr|\ge \varepsilon \bigr\}\bigr)
\le
2\exp\left(-\frac{\varepsilon^2}{2\sigma_k^2}\right)
\quad\text{for all }\varepsilon>0
\]
for some variance proxy $\sigma_k^2>0$, then $F_\psi \bigl ({\mathcal U}(x)\bigr)$ satisfies
\begin{equation}
\label{eq:transport-concentration}
\mu_C\bigl( \bigl \{|F_\psi \bigl ({\mathcal U}(x)\bigr)-\mathbb E_{\mu_C} \bigl[ F_\psi \bigl ({\mathcal U}(x)\bigr)\bigr]|
\ge \varepsilon+2\mathrm{Lip}_\psi(\Phi)\,\mathcal E_C^{(\psi)}(k; \, \Psi) \bigr \}\bigr)
\le
2\exp\left(-\frac{\varepsilon^2}{2\sigma_k^2}\right)
\quad\forall \,\varepsilon>0.
\end{equation}
In particular, under the hypotheses of \Cref{thm:seed-KLG-factorial}, use the factorial error bound
\[
\mathcal E_C^{(\psi)}(k; \, \Psi)
\le
\frac{2B_C^{\lceil k/2\rceil+1}}{(\lceil k/2\rceil+1)!}
\]
to obtain
\begin{align}
\bigl|\mathbb E_{\mu_C} \bigl [ F_\psi \bigl(\mathcal U(x) \bigr) \bigr] -\mathbb E_{\nu} \bigl [ F_\psi \bigl (\iota_{\Psi}^{(k)}(g)\bigr) \bigr ]\bigr|
&\le
\mathrm{Lip}_\psi(\Phi)\,\frac{2B_C^{\lceil k/2\rceil+1}}{(\lceil k/2\rceil+1)!},
\\
\bigl|\mathrm{Var}_{\mu_C}\bigl[F_\psi \bigl(\mathcal U(x) \bigr)\bigr]-\mathrm{Var}_{\nu}\bigl[F_\psi \bigl (\iota_{\Psi}^{(k)}(g)\bigr) \bigr]\bigr|
&\le
4\|\Phi\|_\infty\,\mathrm{Lip}_\psi(\Phi)\,
\frac{2B_C^{\lceil k/2\rceil+1}}{(\lceil k/2\rceil+1)!},
\end{align}
and the threshold shift in \eqref{eq:transport-concentration} is of order
\[
2\mathrm{Lip}_\psi(\Phi)\,\frac{2B_C^{\lceil k/2\rceil+1}}{(\lceil k/2\rceil+1)!}.
\]
For a prefix KLA, \Cref{thm:prefix-KLG-factorial} gives the same
expectation, variance, and threshold-shift estimates with
$\lceil k/2\rceil$ replaced by $k$.
\end{theorem}
\begin{proof}
The first inequality in the theorem is given directly by \Cref{prop:seed-KLG-observable}. For the variance comparison, write
\[
\mathrm{Var}_{\mu_C}\bigl[F_\psi\bigl(\mathcal U(x)\bigr)\bigr]
=
\int_C
\Bigl(
F_\psi\bigl(\mathcal U(x)\bigr)
-
\mathbb E_{\mu_C}F_\psi\bigl(\mathcal U(x)\bigr)
\Bigr)^2
\,d\mu_C(x),
\]
and
\[
\mathrm{Var}_{\nu}\bigl[F_\psi\bigl(\iota_{\Psi}^{(k)}(g)\bigr)\bigr]
=
\int_{K_\Psi^{(k)}}
\Bigl(
F_\psi\bigl(\iota_{\Psi}^{(k)}(g)\bigr)
-
\mathbb E_{\nu}F_\psi\bigl(\iota_{\Psi}^{(k)}(g)\bigr)
\Bigr)^2
\,d\nu(g).
\]
Using the pushforward relation, the second variance can be written as
\[
\mathrm{Var}_{\nu}\bigl[F_\psi\bigl(\iota_{\Psi}^{(k)}(g)\bigr)\bigr]
=
\int_C
\Bigl(
F_\psi\bigl(\iota_{\Psi}^{(k)}\bigl(\kappa_{\Psi}^{(k)}(x)\bigr)\bigr)
-
\mathbb E_{\nu}F_\psi\bigl(\iota_{\Psi}^{(k)}(g)\bigr)
\Bigr)^2
\,d\mu_C(x).
\]

Using the identity $a^2-b^2=(a-b)(a+b)$ and the pushforward
relation, we obtain after integration
\begin{align*}
&\mathrm{Var}_{\mu_C}\bigl[F_\psi(\mathcal U(x))\bigr]
-\mathrm{Var}_{\nu}\bigl[F_\psi(\iota_\Psi^{(k)}(g))\bigr]\\
&\quad=\int_C
\Bigl(F_\psi(\mathcal U(x))
-F_\psi\bigl(\iota_\Psi^{(k)}(\kappa_\Psi^{(k)}(x))\bigr)\Bigr)\\
&\qquad\qquad\cdot
\Bigl(F_\psi(\mathcal U(x))
+F_\psi\bigl(\iota_\Psi^{(k)}(\kappa_\Psi^{(k)}(x))\bigr)\\
&\qquad\qquad\qquad
-\mathbb E_{\mu_C}F_\psi(\mathcal U(x))
-\mathbb E_\nu F_\psi(\iota_\Psi^{(k)}(g))\Bigr)\,d\mu_C(x).
\end{align*}
Indeed, the additional term containing the difference of the two means
integrates to zero, since the second factor has mean zero.
Since \(|\Phi|\le\|\Phi\|_\infty\) pointwise and both expectations are bounded
in absolute value by \(\|\Phi\|_\infty\), we have
\[
\Bigl|
F_\psi\bigl(\mathcal U(x)\bigr)
+
F_\psi\bigl(\iota_{\Psi}^{(k)}\bigl(\kappa_{\Psi}^{(k)}(x)\bigr)\bigr)
-
\mathbb E_{\mu_C}F_\psi\bigl(\mathcal U(x)\bigr)
-
\mathbb E_{\nu}F_\psi\bigl(\iota_{\Psi}^{(k)}(g)\bigr)
\Bigr|
\le
4\|\Phi\|_\infty.
\]
Combining this with \eqref{eq:seed-observable-pointwise} and integrating
over $C$ gives
\begin{align*}
&\bigl|\mathrm{Var}_{\mu_C}\bigl[F_\psi\bigl(\mathcal U(x)\bigr)\bigr]
-
\mathrm{Var}_{\nu}\bigl[F_\psi\bigl(\iota_{\Psi}^{(k)}(g)\bigr)\bigr]\bigr|\\
&\qquad\le
4\|\Phi\|_\infty
\int_C
\Bigl|
F_\psi\bigl(\mathcal U(x)\bigr)
-
F_\psi\bigl(\iota_{\Psi}^{(k)}\bigl(\kappa_{\Psi}^{(k)}(x)\bigr)\bigr)
\Bigr|
\,d\mu_C(x)\\
&\qquad\le
4\|\Phi\|_\infty\,\mathrm{Lip}_\psi(\Phi)\,\mathcal E_C^{(\psi)}(k;\Psi),
\end{align*}
which is the variance bound.

For the concentration transport, note that for all $x\in C$,
\begin{align*}
&\Bigl|
F_\psi\bigl(\mathcal U(x)\bigr)
-
\mathbb E_{\mu_C}F_\psi\bigl(\mathcal U(x)\bigr)
\Bigr|\\
&\qquad\le
\Bigl|
F_\psi\bigl(\iota_{\Psi}^{(k)}\bigl(\kappa_{\Psi}^{(k)}(x)\bigr)\bigr)
-
\mathbb E_{\nu}F_\psi\bigl(\iota_{\Psi}^{(k)}(g)\bigr)
\Bigr|
+
\Bigl|
F_\psi\bigl(\mathcal U(x)\bigr)
-
F_\psi\bigl(\iota_{\Psi}^{(k)}\bigl(\kappa_{\Psi}^{(k)}(x)\bigr)\bigr)
\Bigr|\\
&\qquad\quad \, +
\Bigl|
\mathbb E_{\mu_C}F_\psi\bigl(\mathcal U(x)\bigr)
\: \, - \, \:
\mathbb E_{\nu}F_\psi\bigl(\iota_{\Psi}^{(k)}(g)\bigr)
\Bigr|\\
&\qquad\le
\Bigl|
F_\psi\bigl(\iota_{\Psi}^{(k)}\bigl(\kappa_{\Psi}^{(k)}(x)\bigr)\bigr)
-
\mathbb E_{\nu}F_\psi\bigl(\iota_{\Psi}^{(k)}(g)\bigr)
\Bigr|
+
2\mathrm{Lip}_\psi(\Phi)\,\mathcal E_C^{(\psi)}(k;\Psi),
\end{align*}
using the pointwise and expectation bounds. Hence, if
\[
\Bigl|
F_\psi\bigl(\mathcal U(x)\bigr)
-
\mathbb E_{\mu_C}F_\psi\bigl(\mathcal U(x)\bigr)
\Bigr|
\ge
\varepsilon
+
2\mathrm{Lip}_\psi(\Phi)\,\mathcal E_C^{(\psi)}(k;\Psi),
\]
then necessarily
\[
\Bigl|
F_\psi\bigl(\iota_{\Psi}^{(k)}\bigl(\kappa_{\Psi}^{(k)}(x)\bigr)\bigr)
-
\mathbb E_{\nu}F_\psi\bigl(\iota_{\Psi}^{(k)}(g)\bigr)
\Bigr|
\ge
\varepsilon.
\]
Therefore
\begin{align*}
&\mu_C\Bigl(
\Bigl\{
\bigl|F_\psi\bigl(\mathcal U(x)\bigr)
-\mathbb E_{\mu_C}F_\psi\bigl(\mathcal U(x)\bigr)\bigr|
\ge
\varepsilon
+
2\mathrm{Lip}_\psi(\Phi)\,\mathcal E_C^{(\psi)}(k;\Psi)
\Bigr\}
\Bigr)\\
&\qquad\le
\nu\Bigl(
\Bigl\{
\bigl|F_\psi\bigl(\iota_{\Psi}^{(k)}\bigl(\kappa_{\Psi}^{(k)}(x)\bigr)\bigr)
-\mathbb E_{\nu}F_\psi\bigl(\iota_{\Psi}^{(k)}(g)\bigr)\bigr|
\ge
\varepsilon
\Bigr\}
\Bigr).
\end{align*}
Under the hypothesis that the right-hand side is bounded by
\(2\exp\bigl(-\varepsilon^2/(2\sigma_k^2)\bigr)\) for all \(\varepsilon>0\),
we obtain \eqref{eq:transport-concentration}.

Finally, substituting the factorial bounds of
\Cref{thm:seed-KLG-factorial,thm:prefix-KLG-factorial} into the expectation, variance, and
threshold-shift estimates gives the stated conclusions under their
respective hypotheses.
\end{proof}
\begin{remark}
\Cref{prop:KLG-concentration-transport} shows that, for seed-dependent
observables, the concentration behavior of the original circuit on a
compact subset $C\subset M$ is controlled by that of the Krylov--Lie
proxy, with a threshold shift proportional to
$\mathcal E_C^{(\psi)}(k;\Psi)$. The variance comparison also transfers
lower bounds on the proxy variance whenever the approximation error is
smaller than the variance scale in question. This provides a
Lie-theoretic bridge between the geometric concentration properties of
Krylov--Lie groups and the observable concentration properties of
variational quantum circuits on their reachable manifolds.

For fixed architecture and compact parameter region, the factorial bound
uses order $\lceil k/2\rceil$ under \Cref{thm:seed-KLG-factorial}, and
the full order $k$ for prefix KLAs under \Cref{thm:prefix-KLG-factorial}. These estimates
concern the Krylov grade. Comparison with convergence to unitary designs
in circuit depth requires separate control of the architecture-dependent
constants as the number of layers or system size changes.
\end{remark}

\chapter{Applications to VQA Landscape Theory}  
\label{chap:applications}

In this chapter we suppress the seed and commutator-depth labels, since the formulas are
structural and hold for any fixed represented Krylov-Lie group. In applications, the group,
representation, induced density, and other pertinent quantities will of course depend on those choices.

\section{Moment Formulas for Krylov-Lie Approximators}\label{sec:moment-formulas}

\subsection{Expectation Value of the Loss under a Weighted Sampling Measure}\label{subsec:weighted-expectation-value-reduced}
We now derive the expectation value of the reduced loss function—that is, the component of the loss function directly seen by the Krylov-Lie algebra. The reduced loss is the function for which the Krylov-Lie theory expresses most of its interpretative power. In practice, it is natural (though not strictly necessary) to choose a prefix KLA and block seed that includes the initial state on which the circuit acts, since in that case, by \Cref{prop:KLG-concentration-transport}, the restriction to the associated Krylov-Lie subspace yields approximation errors that decay factorially in the Krylov grade, while any remaining components of the ambient loss live in directions orthogonal to this subspace. For Hermitian involution generators, the same conclusion holds with effective order $\lceil k/2\rceil$ whenever seed-level commutator matching holds through grade $k$ (\Cref{thm:seed-KLG-factorial}). Separately, we note that while we shall perform our derivation under the assumption of a single KLG chart onto the compact subset $C$ of the reachable manifold $M$, most generally one may have to use an atlas of such charts to perform the integration due to the presence of critical points in the canonical comparison map(s). However, as one can simply interpret such an integral as a weighted sum of Haar integrals over the entire relevant KLGs against their appropriate Radon-Nikodym densities, our formulas will retain the same schematic form. One may also choose larger and larger compact subsets away from the measure zero set of critical points to approximate the integral to arbitrary precision without changing the KLG chart.

In addition, by analogy with the relationship between upgrading the DLA theory of Ragone et al. to Anschuetz’ JAWS framework \cite{ragone_lie_2024, anschuetz_2025}, it is reasonable to assume that a fully fledged Krylov–Jordan theory would produce algebraic moment formulas with significant qualitative similarities to those obtained here for the reduced loss function of the Krylov-Lie theory. For this reason, the reduced loss function provides the most informative first approximation to the VQA loss landscape, even when the ultimate object of interest is the ambient loss. One can certainly perform these calculations in an ambient setting, but they become much more unwieldy and difficult to interpret in a clean, algebraic fashion, which reinforces the need for a Jordan-centric theory \cite{anschuetz_2025}. This is amplified by the fact that the compression mechanism of Krylov algebras forces them to exist in a subspace of the ambient space in which our observables and density matrices tend to live. 

Let \(K\) be the compact Krylov-Lie group (\Cref{def:krylov-lie-group}), let \(\mu_K\) denote its normalized Haar probability measure, and let
\[
\pi : K \longrightarrow U(\mathcal K)
\]
be the finite-dimensional unitary representation on the Krylov-Lie subspace \(\mathcal K\) (\Cref{def:krylov-lie-subspace}). Define
\begin{equation} \label{eq:reduced-loss-and-adjoint}
\Ad_g(X) := \pi(g) X \pi(g)^\dagger,
\qquad
\ell(g) := \Tr\bigl[O\Ad_g(\rho)\bigr],
\end{equation}
where $\ell(g)$ denotes the reduced loss function that implicitly works only with the reduced observable $P_{\mathcal K} O P_{\mathcal{K}}$ and the reduced density matrix $P_{\mathcal K} \rho P_{\mathcal{K}}$. Here $P_\mathcal{K}$ is the projector onto $\mathcal K$. By \Cref{ReductiveKLA}, the represented Krylov-Lie algebra is reductive, with orthogonal decomposition
\[
\mathfrak l
=
\mathfrak z
\oplus
\bigoplus_{j=1}^m \mathfrak l_j,
\]
where \(\mathfrak z\) is the center and each \(\mathfrak l_j\) is a simple ideal. Let
\[
\Pi_{\mathfrak z} : \End(\mathcal K) \to \mathfrak z,
\qquad
\Pi_j : \End(\mathcal K) \to \mathfrak l_j
\]
denote the Hilbert-Schmidt orthogonal projections. Write
\begin{equation} \label{eq:KLA-projections}
\rho_{\mathfrak z} := \Pi_{\mathfrak z}(\rho),
\qquad
O_{\mathfrak z} := \Pi_{\mathfrak z}(O),
\qquad
\rho_j := \Pi_j(\rho),
\qquad
O_j := \Pi_j(O).
\end{equation}
Without loss of generality, let \(\nu\) be a probability measure on \(K\) absolutely continuous with respect to Haar measure, and write
\begin{equation} \label{eq:measures-for-calculations}
d\nu(g) = q(g)\,d\mu_K(g),
\qquad
q \in L^1(K,\mu_K),
\qquad
q \ge 0,
\qquad
\int_K q(g)\,d\mu_K(g)=1.
\end{equation}
Next, define the centered density
\begin{equation} \label{eq:centered-density}
f := q-1,
\qquad
\int_K f(g)\,d\mu_K(g)=0.
\end{equation}
For each simple component \(\mathfrak l_j\), define the visible first-moment correction operator
\begin{equation}  \label{eq:visible-first-moment-correction}
\mathcal{V}_j^{(1)}(f)
:=
\int_K d\mu_K(g) f(g)\Ad_g^{(j)}
\in \End(\mathfrak l_j),
\end{equation}
where \(\Ad_g^{(j)}\) denotes the restriction of \(\Ad_g\) to \(\mathfrak l_j\). On the center \(\mathfrak z\), define similarly
\begin{equation}  \label{eq:visible-first-moment-correction-center}
\mathcal V_{\mathfrak z}^{(1)}(f)
:=
\int_K \,d\mu_K(g) f(g)\Ad_g^{(\mathfrak z)}
\in \End(\mathfrak z).
\end{equation}
Since the center is pointwise fixed by the adjoint action, we have
\[
\Ad_g^{(\mathfrak z)} = I_{\mathfrak z}
\qquad
\text{for all } g \in K,
\]
and therefore
\[
\mathcal{V}^{(1)}_{\mathfrak z}(f)
=
\left(\int_K f(g)\,d\mu_K(g)\right) I_{\mathfrak z}
=
0.
\]

\begin{theorem}[Weighted expectation decomposition]
\label{thm:weighted_expectation_decomposition}
Assume that either \(\rho \in \mathfrak l\) or \(O \in \mathfrak l\). Then the expectation value of the reduced loss function (Equation \ref{eq:reduced-loss-and-adjoint}) with respect to \(\nu=q\,\mu_K\) (Equation \ref{eq:measures-for-calculations}) is given by
\begin{equation}
\label{eq:weighted-expectation-decomposition}
\mathbb E_{\nu}[\ell]
=
\Tr \bigl[\rho_{\mathfrak z} O_{\mathfrak z}\bigr]
+
\sum_{j=1}^m
\bigl\langle O_j, \mathcal{V}^{(1)}_j(f)\rho_j \bigr\rangle_{\mathrm{HS}},
\end{equation}
where the subscripts denote projection onto the simple Krylov-Lie algebra components (\ref{eq:KLA-projections}), and the $\mathcal{V}_j^{(1)}(f)$ are the visible first-moment correction operators (\ref{eq:visible-first-moment-correction}). In particular, if \(q \equiv 1\), equivalently \(f=0\) (\ref{eq:centered-density}), then
\begin{equation}
\label{eq:haar-expectation-recovered}
\mathbb E_{\mu_K}[\ell]
=
\Tr \bigl[\rho_{\mathfrak z} O_{\mathfrak z}\bigr].
\end{equation}
If, moreover, the Lie algebra is centerless, then
\[
\mathbb E_{\mu_K}[\ell] = 0.
\]
\end{theorem}

\begin{proof}
We give the proof in the case \(O \in \mathfrak l\); the case \(\rho \in \mathfrak l\) is analogous by symmetry.

Decompose
\[
O = O_{\mathfrak z} + \sum_{j=1}^m O_j,
\qquad
\rho = \rho_{\mathfrak z} + \sum_{j=1}^m \rho_j + \rho_\perp,
\]
where \(\rho_\perp\) is orthogonal to \(\mathfrak l\). Since \(O \in \mathfrak l\), the orthogonal component \(\rho_\perp\) does not contribute to the Hilbert--Schmidt pairing with \(O\). Thus
\[
\ell(g)
=
\Tr \bigl[O_{\mathfrak z}\,\Ad_g(\rho_{\mathfrak z})\bigr]
+
\sum_{j=1}^m
\bigl\langle O_j,\Ad_g^{(j)}(\rho_j)\bigr\rangle_{\mathrm{HS}}.
\]
Because \(\mathfrak z\) is central, the adjoint action fixes it pointwise, so
\[
\Ad_g(\rho_{\mathfrak z}) = \rho_{\mathfrak z}
\qquad
\text{for all } g \in K.
\]
Hence
\[
\ell(g)
=
\Tr\bigl[\rho_{\mathfrak z} O_{\mathfrak z}\bigr]
+
\sum_{j=1}^m
\bigl\langle O_j,\Ad_g^{(j)}(\rho_j)\bigr\rangle_{\mathrm{HS}}.
\]
Now integrate against \(d\nu = q\,d\mu_K = (1+f)\,d\mu_K\). We obtain
\[
\mathbb E_{\nu}[\ell]
=
\Tr \bigl[\rho_{\mathfrak z} O_{\mathfrak z}\bigr]
+
\sum_{j=1}^m
\int_K
\bigl\langle O_j,\Ad_g^{(j)}(\rho_j)\bigr\rangle_{\mathrm{HS}}
\,d\mu_K(g)
+
\sum_{j=1}^m
\int_K
\bigl\langle O_j,\Ad_g^{(j)}(\rho_j)\bigr\rangle_{\mathrm{HS}}
\,f(g)\,d\mu_K(g).
\]
For each simple ideal \(\mathfrak l_j\), the Haar first moment vanishes:
\[
\int_K \Ad_g^{(j)}\,d\mu_K(g)=0
\qquad
\text{on } \mathfrak l_j.
\]
Indeed, the Haar first moment is the orthogonal projection onto the \(K\)-invariant vectors in the adjoint representation, and these invariants are precisely the center, which is trivial on a simple Lie algebra \cite{ragone_lie_2024}. Therefore the middle term vanishes, and the final term becomes
\[
\sum_{j=1}^m
\bigl\langle O_j,\; \mathcal{V}^{(1)}_j(f)\rho_j \bigr\rangle_{\mathrm{HS}}.
\]
This proves \eqref{eq:weighted-expectation-decomposition}. The Haar formula \eqref{eq:haar-expectation-recovered} follows by setting \(f=0\). If the Lie algebra is centerless, then \(\mathfrak z = \{0\}\), and so \(\mathbb E_{\mu_K}[\ell]=0\).
\end{proof}

\begin{corollary}[Expectation gap relative to Haar]
\label{cor:expectation-gap-relative-to-haar}
Under the hypotheses of Theorem~\ref{thm:weighted_expectation_decomposition},
\begin{equation}
\label{eq:expectation-gap-relative-to-haar}
\mathbb E_{\nu}[\ell] - \mathbb E_{\mu_K}[\ell]
=
\sum_{j=1}^m
\bigl\langle O_j,\; \mathcal{V}^{(1)}_j(f)\rho_j \bigr\rangle_{\mathrm{HS}}.
\end{equation}
Thus every deviation of the expectation from its Haar value is carried by the centered density \(f=q-1\) through the visible first-moment correction operators \(\mathcal{V}^{(1)}_j(f)\).
\end{corollary}

\begin{proof}
Subtract \eqref{eq:haar-expectation-recovered} from \eqref{eq:weighted-expectation-decomposition}.
\end{proof}

\begin{remark}[Compatibility with the Lie-algebraic Haar formula]
Theorem~\ref{thm:weighted_expectation_decomposition} is the Krylov-Lie weighted first-moment analogue of the usual dynamical Lie group computation for the expectation values of our loss when our distribution is Haar \cite{ragone_lie_2024}. It is written so that compatibility with the Lie-algebraic Haar formula is immediate: the Haar expectation is exactly the unperturbed contribution, and every non-Haar effect enters only through the centered first-moment correction operators \(\mathcal{V}^{(1)}_j(f)\).
\end{remark}

\subsection{A Weighted Variance Formula Compatible with Haar Distributions}\label{subsec:weighted-variance-reduced}

We now derive a variance formula for sampling measures that are absolutely continuous with respect to the Haar measure on the Krylov--Lie group over the reduced loss (\ref{eq:reduced-loss-and-adjoint}). The resulting expression will again be written so as to make the Haar case manifest: when the density is constant, the correction terms vanish identically and one recovers the Lie-algebraic variance formula of Ragone et al. as a special case. As in \Cref{thm:weighted_expectation_decomposition}, comparison with the original moments over the reachable manifold is controlled by the factorial seed-based error under \Cref{thm:seed-KLG-factorial}, or by the stronger prefix estimate under \Cref{thm:prefix-KLG-factorial} (cf. \Cref{rem:klg_approx_discussion,prop:KLG-concentration-transport}). 

Throughout this subsection, we shall employ the same conventions and variable names as in \Cref{subsec:weighted-expectation-value-reduced}. 
We also write
\begin{equation} \label{eq:purities}
\mathcal P_j(A) := \|\Pi_j(A)\|_{\mathrm{HS}}^2
\end{equation}
for the Hilbert--Schmidt purity of \(A\) in the \(j\)-th simple component. Moreover, we define the centered visible second-moment correction operator as
\begin{equation} \label{eq:visible-second-moment-correction}
\mathcal{V}^{(2)}_{jk}(f)
:=
\int_K d\mu_K(g) f(g)\,
\bigl(\Ad_g^{(j)} \otimes \Ad_g^{(k)}\bigr)
\in \End(\mathfrak l_j \otimes \mathfrak l_k).
\end{equation}

\begin{theorem}[Weighted variance decomposition over simple ideals]
\label{thm:weighted_ragone_compatible_variance}
Assume that either \(\rho \in \mathfrak l\) or \(O \in \mathfrak l\). Then the variance of the reduced loss function \(\ell(g)=\Tr[O\Ad_g(\rho)]\) (see \ref{eq:reduced-loss-and-adjoint}) with respect to the measure \(\nu=q\,\mu_K\) (see \ref{eq:measures-for-calculations}) is given exactly by
\begin{equation}
\label{eq:weighted-ragone-compatible-variance}
\Var_{\nu}(\ell)
=
\sum_{j=1}^m
\frac{\mathcal P_j(\rho)\,\mathcal P_j(O)}{\dim(\mathfrak l_j)}
+
\sum_{j,k=1}^m
\bigl\langle
O_j \otimes O_k,\,
\mathcal{V}^{(2)}_{jk}(f)(\rho_j \otimes \rho_k)
\bigr\rangle_{\mathrm{HS}}
-
\left(
\sum_{j=1}^m
\bigl\langle O_j,\; \mathcal{V}^{(1)}_j(f)\rho_j \bigr\rangle_{\mathrm{HS}}
\right)^2,
\end{equation}
where the $\mathcal{P}_j$ are the $\mathfrak l_j$-purities (\ref{eq:purities}), the subscripts denote projections onto the simple components of our Krylov-Lie algebra (\ref{eq:KLA-projections}),  $\mathcal{V}_j^{(1)}(f)$ is the visible first-moment correction operator (\ref{eq:visible-first-moment-correction}), and $\mathcal{V}_{jk}^{(2)}(f)$ is the visible second-moment correction operator (\ref{eq:visible-second-moment-correction}).

In particular, if \(q \equiv 1\), equivalently \(f=0\), then
\begin{equation}
\label{eq:haar-ragone-recovered}
\Var_{\mu_K}(\ell)
=
\sum_{j=1}^m
\frac{\mathcal P_j(\rho)\,\mathcal P_j(O)}{\dim(\mathfrak l_j)}.
\end{equation}
\end{theorem}

\begin{proof}
We give the proof in the case \(O \in \mathfrak l\); the case \(\rho \in \mathfrak l\) follows by the same argument after exchanging the roles of \(\rho\) and \(O\).

Since \(\mathfrak l = \mathfrak z \, \oplus \, \bigoplus_{j=1}^m \mathfrak l_j\) is a reductive decomposition into pairwise orthogonal \(\Ad\)-invariant ideals, we may write
\[
O = \sum_{j=1}^m O_j,
\qquad
\rho = \rho_{\mathfrak z} + \sum_{j=1}^m \rho_j + \rho_\perp,
\]
where \(\rho_{\mathfrak z}\in \mathfrak z\), each \(\rho_j \in \mathfrak l_j\), and \(\rho_\perp\) is orthogonal to \(\mathfrak l\). Since \(O \in \mathfrak l\), the component \(\rho_\perp\) does not contribute to the Hilbert--Schmidt pairing with \(O\), and therefore
\[
\ell(g)
=
\sum_{j=1}^m \langle O_j,\Ad_g^{(j)}(\rho_j)\rangle_{\mathrm{HS}}
+
\langle O,\rho_{\mathfrak z}\rangle_{\mathrm{HS}}.
\]
The central term is independent of \(g\), hence contributes to the mean but not to the variance. It therefore suffices to compute the fluctuating part
\[
\ell_{\mathrm{fluc}}(g)
:=
\sum_{j=1}^m \langle O_j,\Ad_g^{(j)}(\rho_j)\rangle_{\mathrm{HS}}.
\]

We first compute the mean. Since \(q=1+f\), we have
\[
\mathbb E_{\nu}[\ell_{\mathrm{fluc}}]
=
\int_K \ell_{\mathrm{fluc}}(g)\,d\mu_K(g)
+
\int_K \ell_{\mathrm{fluc}}(g)\,f(g)\,d\mu_K(g).
\]
For each simple ideal \(\mathfrak l_j\), the Haar average of the adjoint representation vanishes, because the only invariant vectors in the adjoint representation would lie in the center, whereas \(\mathfrak l_j\) is simple and centerless. Hence
\[
\int_K \Ad_g^{(j)}\,d\mu_K(g)=0
\qquad
\text{on } \mathfrak l_j.
\]
Therefore
\[
\mathbb E_{\nu}[\ell_{\mathrm{fluc}}]
=
\sum_{j=1}^m
\bigl\langle O_j,\; \mathcal{V}^{(1)}_j(f)\rho_j \bigr\rangle_{\mathrm{HS}}.
\]

We next compute the second moment. Expanding the square gives
\[
\ell_{\mathrm{fluc}}(g)^2
=
\sum_{j,k=1}^m
\bigl\langle O_j,\Ad_g^{(j)}(\rho_j)\bigr\rangle_{\mathrm{HS}}
\bigl\langle O_k,\Ad_g^{(k)}(\rho_k)\bigr\rangle_{\mathrm{HS}}.
\]
Equivalently,
\[
\ell_{\mathrm{fluc}}(g)^2
=
\sum_{j,k=1}^m
\bigl\langle
O_j \otimes O_k,\,
(\Ad_g^{(j)} \otimes \Ad_g^{(k)})(\rho_j \otimes \rho_k)
\bigr\rangle_{\mathrm{HS}}.
\]
Integrating against \(q\,d\mu_K = (1+f)\,d\mu_K\), we obtain
\[
\mathbb E_{\nu}[\ell_{\mathrm{fluc}}^2]
=
\mathbb E_{\mu_K}[\ell_{\mathrm{fluc}}^2]
+
\sum_{j,k=1}^m
\bigl\langle
O_j \otimes O_k,\,
\mathcal{V}^{(2)}_{jk}(f)(\rho_j \otimes \rho_k)
\bigr\rangle_{\mathrm{HS}}.
\]
The Haar contribution will be given by the familiar Ragone expression \cite{ragone_lie_2024}:
\[
\mathbb E_{\mu_K}[\ell_{\mathrm{fluc}}^2]
=
\sum_{j=1}^m
\frac{\mathcal P_j(\rho)\,\mathcal P_j(O)}{\dim(\mathfrak l_j)}.
\]
Combining the preceding identities with the usual formula $\Var_{\nu}(\ell)
=
\mathbb E_{\nu}[\ell_{\mathrm{fluc}}^2]
-
\mathbb E_{\nu}[\ell_{\mathrm{fluc}}]^2$, we obtain
\[
\Var_{\nu}(\ell)
=
\sum_{j=1}^m
\frac{\mathcal P_j(\rho)\,\mathcal P_j(O)}{\dim(\mathfrak l_j)}
+
\sum_{j,k=1}^m
\bigl\langle
O_j \otimes O_k,\,
\mathcal{V}^{(2)}_{jk}(f)(\rho_j \otimes \rho_k)
\bigr\rangle_{\mathrm{HS}}
-
\left(
\sum_{j=1}^m
\bigl\langle O_j,\; \mathcal{V}^{(1)}_j(f)\rho_j \bigr\rangle_{\mathrm{HS}}
\right)^2,
\]
which is \eqref{eq:weighted-ragone-compatible-variance}. If \(q\equiv 1\), then \(f=0\), so \(\mathcal{V}^{(1)}_j(f)=0\) and \(\mathcal{V}^{(2)}_{jk}(f)=0\) for all \(j,k\), and \eqref{eq:haar-ragone-recovered} follows immediately.
\end{proof}

\begin{remark}[Relation to the Lie-algebraic Haar formula]
Equation \eqref{eq:weighted-ragone-compatible-variance} is written so that the Haar contribution appears as the leading term and every deviation from Haar sampling is isolated into the centered visible correction operators \(\mathcal{V}^{(1)}_j(f)\) and \(\mathcal{V}^{(2)}_{jk}(f)\). The operator \(\mathcal{V}^{(1)}_j(f)\) records the first visible non-Haar bias of the sampling measure on the \(j\)-th simple ideal, while \(\mathcal{V}^{(2)}_{jk}(f)\) records the corresponding second-moment bias coupling the \(j\)-th and \(k\)-th simple sectors. In particular, the formula is manifestly compatible with the Lie-algebraic variance formula of Ragone et al.: the latter is recovered by the single substitution \(q\equiv 1\).
\end{remark}

\begin{corollary}[Variance gap relative to Haar]
\label{cor:variance-gap-relative-to-haar}
Under the hypotheses of Theorem~\ref{thm:weighted_ragone_compatible_variance},
\begin{equation}
\label{eq:variance-gap-relative-to-haar}
\Var_{\nu}(\ell)-\Var_{\mu_K}(\ell)
=
\sum_{j,k=1}^m
\bigl\langle
O_j \otimes O_k,\,
\mathcal{V}^{(2)}_{jk}(f)(\rho_j \otimes \rho_k)
\bigr\rangle_{\mathrm{HS}}
-
\left(
\sum_{j=1}^m
\bigl\langle O_j,\; \mathcal{V}^{(1)}_j(f)\rho_j \bigr\rangle_{\mathrm{HS}}
\right)^2.
\end{equation}
In particular, all deviation from the Haar variance is carried by the centered density \(f=q-1\).
\end{corollary}

\begin{proof}
Subtract \eqref{eq:haar-ragone-recovered} from \eqref{eq:weighted-ragone-compatible-variance}.
\end{proof}


\begin{remark}[Main implications]
\label{rem:main-implications-nonhaar}
The preceding expectation and variance formulas have several noteworthy conceptual ramifications.

First, they provide a canonical framework for analyzing non-Haar loss landscapes on compact Krylov-Lie groups. Since the Krylov-Lie group is compact by \Cref{thm:DimensionMatchingJacobian}, it carries a unique normalized Haar probability measure \(\mu_K\). Hence any probability measure \(\nu\) on the group that is absolutely continuous with respect to \(\mu_K\) admits a Radon--Nikodym density
\[
q = \frac{d\nu}{d\mu_K},
\]
which is uniquely determined up to \(\mu_K\)-null sets. Consequently, the centered density
\[
f := q-1
\]
is likewise canonical up to \(\mu_K\)-null sets, and the expectation and variance decompositions depend only on the induced group-level sampling law \(\nu\), not on any auxiliary choice of coordinates or parameterization.

Second, these formulas extend the Haar-analysis of Ragone et al. to a genuinely non-Haar setting. The Haar terms are recovered by the single substitution \(q \equiv 1\), while every non-Haar effect is isolated into explicit first- and second-moment correction operators built from the centered density \(f\). In this sense, the non-Haar theory is not an ad-hoc perturbation of the Haar case, but its natural measure-theoretic extension.

Third, the formulas show that loss-landscape statistics can be studied intrinsically on the compact Krylov--Lie group itself. Once a physical or algorithmic sampling procedure induces a pushforward probability measure \(\nu\) on the group, the analytically relevant object is precisely this induced measure. Different sampling protocols that yield the same pushforward law on the group therefore lead to the same expectation values, variances, and higher moment formulas. This makes the framework coordinate-free at the level of the observable statistics.

Fourth, as the Krylov-Lie algebra is intentionally defined so that its dimension will correspond to the manifold dimension of the VQA (or potentially be even smaller), often compressing to work over a vector space with significantly smaller dimension than the ambient Hilbert space in the shallow-circuit regime, the dimension of the simple sectors of the KLA are far more constrained than those of the DLA. This means that the physically relevant Haar sector variance of any given variational quantum circuit is inversely proportional to effective dimensions dependent on the number of unique parameters in the circuit and the composition of the corresponding KLA. As a consequence, informally speaking, for noiseless shallow circuits whose effective dimension does not grow exponentially with \(N\), Haar-induced barren plateaus are not expected to arise from expressivity alone; rather, under the variance formula, exponential concentration can only occur through sufficiently small generalized purity of the initial state or observable relative to the relevant simple sectors. This is compatible with the polynomially increasing variance behavior (cragged terrains) observed in the separate analysis of shallow QAOA on Max Independent Set found in joint work \cite{EHMQAOAPaper}. More interestingly, inclusion of the centered density terms predicts that certain choices of initial parameter distributions---those that yield a variance gap (\Cref{cor:variance-gap-relative-to-haar}) asymptotically larger than the Haar part as a function of $N$---could at least in principle enhance the variance actually encountered during training by reweighting the visible sectors of the landscape in a manner loosely analogous to importance sampling. We stress that this observation does not contradict the standard claim that no optimizer can efficiently escape a barren plateau once optimization is taking place inside a region of exponential concentration \cite{larocca_barren_2025}. Indeed, in such a region, both gradient information and finite cost differences are exponentially suppressed, so feasible shot budgets do not suffice to distinguish neighboring points in the landscape with the fidelity required for reliable local updates \cite{larocca_barren_2025}. However, this local obstruction does not preclude a global change in the distribution of parameters visited during training. In the present framework, a nonuniform optimizer-induced sampling density may strongly downweight barren-plateau sectors and upweight sectors with larger visible variance, so that the optimization process need not exhibit the same concentration properties as the fixed Haar reference ensemble. We leave the question of proving rigorously whether this strategy can be performed without a priori knowledge of the landscape to future work.

Finally, for fixed finite moments, only finitely many representation-theoretic sectors of the density are visible to the corresponding moment operators. Thus the exact non-Haar corrections are controlled by finite-dimensional data extracted from \(q\). This makes the problem well-posed in principle and structurally amenable to explicit harmonic and Lie-theoretic computation.
\end{remark}
\singlespacing
\section{Deviation from Haar and Variance Bounds on the Krylov-Lie Group}
\phantomsection
\label{sec:haar-deviation-variance}
\onehalfspacing

We shall now use the reduced loss to analyze the effects of the Haar deviation terms on our moments. The analogous ambient formulas too follow from straightforward calculations, but they again lack the insight and algebraic organization of the reduced formulas and so we omit them from our discussion. 

Let \(K\) be a compact represented Krylov--Lie group (\Cref{def:krylov-lie-group}) acting unitarily on a finite-dimensional
Krylov subspace \(\mathcal K\) (\Cref{def:krylov-lie-subspace}) through a unitary representation
\[
\kappa:K\to U(\mathcal K),
\]
and let \(\mu_K\) denote Haar probability measure on \(K\). Without loss of generality, say that
\[
\nu\ll \mu_K,
\qquad
q:=\frac{d\nu}{d\mu_K},
\qquad
f:=q-1.
\]
Then
\[
\int_K f(g)\,d\mu_K(g)=0.
\]
For \(t\in\mathbb N\) and $\eta$ some distribution of unitaries over $K$, define the \(t\)-th moment superoperator
\[
\mathcal M^{(t)}_\eta(X)
:=
\int_K \kappa(g)^{\otimes t}\,X\,(\kappa(g)^{\dagger})^{\otimes t}\,d\eta(g),
\qquad
X\in \End(\mathcal K^{\otimes t}),
\]
and the corresponding moment-error operator
\[
\Delta^{(t)}
:=
\mathcal M^{(t)}_{\nu}-\mathcal M^{(t)}_{\mu_K}.
\]
Equivalently,
\[
\Delta^{(t)}(X)
=
\int_K f(g)\,\kappa(g)^{\otimes t}X(\kappa(g)^{\dagger})^{\otimes t}\,d\mu_K(g).
\]
\begin{lemma}[Vectorized form and Peter-Weyl block decomposition]
\label{lem:vectorized-moment-error}
Let
\[
\Delta^{(t)}
=
\mathcal M^{(t)}_{\nu}-\mathcal M^{(t)}_{\mu_K}
\]
be the \(t\)-th moment-error superoperator, so that
\[
\Delta^{(t)}(X)
=
\int_K f(g)\,\kappa(g)^{\otimes t}X(\kappa(g)^{\dagger})^{\otimes t}\,d\mu_K(g),
\qquad
f:=q-1.
\]
Let
\[
V^{(t)}:=\mathcal K^{\otimes t}\otimes \overline{\mathcal K}^{\otimes t},
\]
and let
\[
\mathrm{vec}:\End(\mathcal K^{\otimes t})\to V^{(t)}
\]
denote the standard vectorization isomorphism. Define the \(t\)-moment representation (where the overline denotes the entrywise complex conjugate)
\[
R^{(t)}(g):=\kappa(g)^{\otimes t}\otimes \overline{\kappa(g)}^{\otimes t}
\in \End(V^{(t)}).
\]
Then the vectorized error operator
\[
\widetilde\Delta^{(t)}
:=
\mathrm{vec}\circ \Delta^{(t)}\circ \mathrm{vec}^{-1}
\]
is given by
\[
\widetilde\Delta^{(t)}
=
\int_K d\mu_K(g)f(g)R^{(t)}(g).
\]
Moreover, if
\[
V^{(t)}
\cong
\bigoplus_{\lambda\in \mathrm{Vis}_t(K)\cup\{\mathbf 1\}}
\mathbb C^{m_\lambda}\otimes V_\lambda
\]
is the decomposition of \(R^{(t)}\) into irreducibles \cite[Thm.~5.2]{Folland1995Harmonic}, then
\[
\widetilde\Delta^{(t)}
\cong
\bigoplus_{\lambda\in \mathrm{Vis}_t(K)\cup\{\mathbf 1\}}
I_{m_\lambda}\otimes \widehat f(\lambda),
\qquad
\widehat f(\lambda):=\int_K f(g)\,\pi_\lambda(g)\,d\mu_K(g),
\]
and the trivial block vanishes:
\[
\widehat f(\mathbf 1)=0.
\]
Consequently,
\[
\|\widetilde\Delta^{(t)}\|_{\op}= \alpha_t
=
\max_{\lambda\in \mathrm{Vis}_t(K)}
\sigma_{\max}\bigl(\widehat f(\lambda)\bigr).
\]
\end{lemma}
\begin{proof}
Fix \(X\in \End(\mathcal K^{\otimes t})\). We use the standard identity
\[
\mathrm{vec}(AXB^\dagger)
=
(A\otimes \overline B)\,\mathrm{vec}(X),
\]
valid for all operators \(A,B\) on \(\mathcal K^{\otimes t}\), where the overline denotes the entrywise complex conjugate. Applying this with
\[
A=B=\kappa(g)^{\otimes t}
\]
gives
\[
\mathrm{vec}\bigl(\kappa(g)^{\otimes t}X(\kappa(g)^{\dagger})^{\otimes t}\bigr)
=
\bigl(\kappa(g)^{\otimes t}\otimes \overline{\kappa(g)}^{\otimes t}\bigr)\mathrm{vec}(X)
=
R^{(t)}(g)\,\mathrm{vec}(X).
\]
Therefore,
\[
\mathrm{vec}\bigl(\Delta^{(t)}(X)\bigr)
=
\int_K d\mu_K(g)
f(g)\,
\mathrm{vec}\bigl(\kappa(g)^{\otimes t}X(\kappa(g)^{\dagger})^{\otimes t}\bigr)
\]
\[
=
\int_K d\mu_K(g)
f(g)
R^{(t)}(g)\,\mathrm{vec}(X).
\]
Since \(V^{(t)}\) is finite-dimensional, the vector \(\mathrm{vec}(X)\) may be pulled outside the integral, yielding
\[
\mathrm{vec}\bigl(\Delta^{(t)}(X)\bigr)
=
\left(
\int_K d\mu_K(g) f(g) R^{(t)}(g)
\right)\mathrm{vec}(X).
\]
As this holds for every \(X\), we conclude that
\[
\widetilde\Delta^{(t)}
=
\int_K d\mu_K(g)f(g)R^{(t)}(g).
\]
Now decompose the unitary representation \(R^{(t)}\) into irreducibles \cite[Thm.~5.2]{Folland1995Harmonic}:
\[
V^{(t)}
\cong
\bigoplus_{\lambda}
\mathbb C^{m_\lambda}\otimes V_\lambda,
\qquad
R^{(t)}(g)
\cong
\bigoplus_{\lambda}
I_{m_\lambda}\otimes \pi_\lambda(g).
\]
Substituting this into the preceding integral and integrating blockwise gives
\[
\widetilde\Delta^{(t)}
\cong
\bigoplus_{\lambda}
I_{m_\lambda}\otimes
\left(
\int_K f(g)\,\pi_\lambda(g)\,d\mu_K(g)
\right)
=
\bigoplus_{\lambda}
I_{m_\lambda}\otimes \widehat f(\lambda).
\]
For the trivial representation \(\mathbf 1\), one has
\[
\widehat f(\mathbf 1)
=
\int_K f(g)\,d\mu_K(g)
=
\int_K (q(g)-1)\,d\mu_K(g)
=
1-1=0.
\]
Hence only the nontrivial visible sectors contribute.

Finally, the operator norm of a block diagonal operator is the maximum of the operator norms of its blocks, and tensoring with an identity does not change singular values. Therefore
\[
\|\widetilde\Delta^{(t)}\|_{\op}
=
\max_{\lambda\in \mathrm{Vis}_t(K)}
\|I_{m_\lambda}\otimes \widehat f(\lambda)\|_{\op}
=
\max_{\lambda\in \mathrm{Vis}_t(K)}
\|\widehat f(\lambda)\|_{\op}
=
\max_{\lambda\in \mathrm{Vis}_t(K)}
\sigma_{\max}\bigl(\widehat f(\lambda)\bigr).
\]
\end{proof}

Thus the moment-error superoperator is nothing but the noncommutative Fourier transform of the
centered density \(f=q-1\), evaluated on the \(t\)-moment representation. The quantity
\(\|\widetilde\Delta^{(t)}\|_{\op}\) therefore measures the largest non-Haar Fourier block of
\(f\) that is visible to \(t\)-copy observables.

For a density matrix \(\rho\) and Hermitian observable \(O\), define
\[
\ell_g(\rho,O):=\Tr\bigl[\kappa(g)\rho\kappa(g)^\dagger O\bigr],
\qquad
m^{(t)}_\eta(\rho,O):=\int_K \ell_g(\rho,O)^t\,d\eta(g).
\]
Then because the moment superoperator is Hermitian \cite{ragone_lie_2024},
\[
m^{(t)}_\eta(\rho,O)
=
\Tr\bigl[\rho^{\otimes t}\,\mathcal M^{(t)}_{\eta}(O^{\otimes t})\bigr],
\]
so the scalar \(t\)-moment non-Haar for the distribution $\eta$ over $K$ is
\[
\varepsilon^{(t)}_{\eta}(\rho,O)
:=
m^{(t)}_\eta(\rho,O)-m^{(t)}_{\mu_K}(\rho,O)
=
\Tr \bigl[\rho^{\otimes t}\,\Delta^{(t)}(O^{\otimes t})\bigr].
\]

\begin{proposition}
For every \(t\ge 1\),
\[
|\varepsilon^{(t)}_{\eta}(\rho,O)|
\le
\alpha_t\,\|O\|_\infty^t.
\]
Moreover,
\[
|\varepsilon^{(t)}_{\eta}(\rho,O)|
\le
\|q-1\|_{L^1(\mu_K)}\,\|O\|_\infty^t.
\]
\end{proposition}

\begin{proof}
Since \(\rho\) is a density matrix, \(\|\rho^{\otimes t}\|_1=1\). Hence
\[
|\varepsilon^{(t)}_{\eta}(\rho,O)|
\le
\|\Delta^{(t)}(O^{\otimes t})\|_\infty
\le
\|\widetilde{\Delta}^{(t)}\|_{\op}\,\|O^{\otimes t}\|_\infty
=
\alpha_t\,\|O\|_\infty^t.
\]
The \(L^1\)-bound follows from
\[
\|\Delta^{(t)}(X)\|_\infty
\le
\int_K |f(g)|\,\|\kappa(g)^{\otimes t}X\kappa(g)^{\dagger\otimes t}\|_\infty\,d\mu_K(g)
=
\|q-1\|_{L^1(\mu_K)}\,\|X\|_\infty.
\]
\end{proof}

The constants \(\alpha_t\) are monotone in \(t\), which yields a useful uniform estimate for all
moments up to a prescribed order.

\begin{proposition}[Monotonicity of the moment-error constants]
For every \(t\ge 1\),
\[
\alpha_t\le \alpha_{t+1}.
\]
Consequently, for every \(1\le t\le T\),
\[
|\varepsilon^{(t)}_{\eta}(\rho,O)|
\le
\alpha_T\,\|O\|_\infty^t.
\]
\end{proposition}

\begin{proof}
Define
\[
J_t:\End(\mathcal K^{\otimes t})\to \End(\mathcal K^{\otimes (t+1)}),
\qquad
J_t(X):=X\otimes I.
\]
Then \(J_t\) intertwines the \(t\)- and \((t+1)\)-moment actions:
\[
\mathcal M^{(t+1)}_\eta\circ J_t
=
J_t\circ \mathcal M^{(t)}_{\eta}
\]
for every probability measure \(\eta\) on \(K\), because by the mixed-product property of the tensor product
\[
(\kappa(g)^{\otimes t} \otimes \kappa(g))(X\otimes I)({(\kappa(g)^{\dagger})}^{\otimes t} \otimes\kappa(g)^{\dagger} )
=
\bigl(\kappa(g)^{\otimes t}X{(\kappa(g)^{\dagger})}^{\otimes t}\bigr)\otimes I.
\]
Subtracting the Haar and ansatz moment operators gives
\[
\Delta^{(t+1)}\circ J_t
=
J_t\circ \Delta^{(t)}.
\]

Now \(J_t\) is an isometric embedding for the operator norm, since
\[
\|X\otimes I\|_\infty=\|X\|_\infty\|I\|_\infty=\|X\|_\infty.
\]
Therefore, for every \(X\neq 0\),
\[
\frac{\|\Delta^{(t)}(X)\|_\infty}{\|X\|_\infty}
=
\frac{\|J_t(\Delta^{(t)}(X))\|_\infty}{\|J_t(X)\|_\infty}
=
\frac{\|\Delta^{(t+1)}(J_t(X))\|_\infty}{\|J_t(X)\|_\infty}
\le
\|\Delta^{(t+1)}\|_{\infty\to\infty}.
\]
Taking the supremum over \(X\neq 0\) yields
\[
\|\Delta^{(t)}\|_{\infty\to\infty}\le \|\Delta^{(t+1)}\|_{\infty\to\infty}.
\]
Since \(\alpha_t=\|\widetilde{\Delta}^{(t)}\|_{\op}=\|\Delta^{(t)}\|_{\infty\to\infty}\), we obtain
\[
\alpha_t\le \alpha_{t+1}.
\]
The final estimate follows immediately.
\end{proof}

Since
\[
\Var_\eta(\ell_{\rho,O})
=
m^{(2)}_\eta(\rho,O)-\bigl(m^{(1)}_\eta(\rho,O)\bigr)^2,
\]
we obtain the following comparison with the Haar variance.

\begin{corollary} \label{cor:VarianceIneq}
One has
\[
\bigl|\Var_{\nu}(\ell_g(\rho,O))-\Var_{\mu_K}(\ell_g(\rho,O))\bigr|
\le
\bigl(\alpha_2+2\alpha_1\bigr)\|O\|_\infty^2
\le
3\alpha_2\,\|O\|_\infty^2.
\]
In addition,
\[
\bigl|\Var_{\nu}(\ell_g(\rho,O))-\Var_{\mu_K}(\ell_g(\rho,O))\bigr|
\le
3\,\|q-1\|_{L^1(\mu_K)}\,\|O\|_\infty^2.
\]
\end{corollary}

\begin{proof}
We have
\[
\Var_{\nu}(\ell_{\rho,O})-\Var_{\mu_K}(\ell_{\rho,O})
=
\varepsilon^{(2)}_{\nu}(\rho,O)
-
\bigl(m^{(1)}_\nu(\rho,O)-m^{(1)}_{\mu_K}(\rho,O)\bigr)
\bigl(m^{(1)}_\nu(\rho,O)+m^{(1)}_{\mu_K}(\rho,O)\bigr).
\]
Also, because $\rho$ is a density matrix, by the von Neumann-H\"older inequalities for the trace,
\[
|\ell_g(\rho,O)|\le \|O\|_\infty,
\qquad
|m^{(1)}_\nu(\rho,O)|\le \|O\|_\infty,
\qquad
|m^{(1)}_{\mu_K}(\rho,O)|\le \|O\|_\infty.
\]
Hence
\[
\bigl|\Var_{\nu}(\ell_g(\rho,O))-\Var_{\mu_K}(\ell_g(\rho,O))\bigr|
\le
|\varepsilon^{(2)}_{\nu}(\rho,O)|+2\|O\|_\infty |\varepsilon^{(1)}_{\nu}(\rho,O)|.
\]
Applying the previous propositions gives
\[
\bigl|\Var_{\nu}(\ell_g(\rho,O))-\Var_{\mu_K}(\ell_g(\rho,O))\bigr|
\le
(\alpha_2+2\alpha_1)\|O\|_\infty^2
\le
3\alpha_2\|O\|_\infty^2,
\]
and likewise for the \(L^1\)-bound.
\end{proof}

\subsection{Comparison with Ragone et al.}

This inequality plays the same role as the approximate-design correction in Ragone et al.: the Haar
variance is computed exactly first, and one then bounds the difference between the actual variance
and the Haar variance by controlling lower-order moment errors. In the present setting, the
correction is governed exactly by the visible Peter-Weyl blocks of the centered density \(q-1\) of the ansatz measure expressed relative to Haar on the represented
Krylov-Lie group,
through the constants
\[
\alpha_t
=
\max_{\lambda\in \mathrm{Vis}_t(K)}
\sigma_{\max}\bigl(\widehat f(\lambda)\bigr),
\]
or more coarsely by \(\|q-1\|_{L^1(\mu_K)}\). To reemphasize, however, comparison with moments over the original circuits incurs the factorial approximation error supplied by \Cref{thm:seed-KLG-factorial,thm:prefix-KLG-factorial} under their respective hypotheses (see \Cref{prop:KLG-concentration-transport}); by ``exact", we mean that non-Haar terms are directly computed by our theory and the omission of their data is not a part of our approximation scheme. 

\singlespacing
\section{Variance Inequalities, Non-Haar Asymptotics, and the Failure of Lemma~6}
\phantomsection
\label{sec:variance-inequalities-non-haar}
\onehalfspacing

At first sight, the error inequalities derived in \Cref{cor:VarianceIneq} suggest a possibility that may seem to run against the standard intuition in the barren-plateau literature. Indeed, our bounds are naturally expressed in terms of the visible Fourier content of the sampling error, and therefore they leave open the possibility that certain choices of parameter distribution and dynamical Lie algebra may systematically retain nontrivial visible components, even in the large-depth limit, rather than washing out to the Haar moment operator. Put differently, the structure of our estimates already suggests that there may exist ans\"atze whose asymptotic sampling measure is not driven to Haar on the represented group, because some visible sector is not mixed at all.

This appears, at least superficially, to conflict with a common heuristic assumption, made particularly explicit in Ragone et al.'s Theorem~2, that sufficiently increasing the depth should force the circuit ensemble toward an approximate $G$ \(t\)-design. In their notation, if \(\mathcal E_L\) is the ensemble generated by an \(L\)-layer circuit and
\[
\mathcal A_{\mathcal E_L}^{(t)}
=
\mathcal M_{\mathcal E_L}^{(t)}-\mathcal M_G^{(t)},
\]
then Theorem~2 asserts that \(\mathcal E_L\) becomes an \(\varepsilon\)-approximate $G$ \(t\)-design once
\[
L
\ge
\frac{\log(1/\varepsilon)}
{\log\!\bigl(1/\|\mathcal A_{\mathcal E_1}^{(t)}\|_\infty\bigr)}.
\]
Thus the philosophy is that depth alone should eventually force convergence to the Haar \(t\)-moment operator.

However, this conclusion is not valid in full generality. The issue is with the strict-contraction input used to make useful the formal algebraic identity
\[
\mathcal A_{\mathcal E_L}^{(t)}
=
\bigl(\mathcal A_{\mathcal E_1}^{(t)}\bigr)^L.
\]
More precisely, the proof of Theorem~2 relies on Lemma~6 of Ragone et al., which claims that every eigenvalue of the single-layer moment operator on the orthogonal complement \(\mathfrak p\) of the \(G\)-invariant sector has modulus strictly less than \(1\). Once that statement fails, one no longer obtains
\(\|\mathcal A_{\mathcal E_1}^{(t)}\|_\infty<1\), and the claimed universal exponential convergence to Haar breaks down.

The cycle graph QAOA on Max Cut example provides a pedagogical counterexample to precisely this step. In that setting, one finds a nontrivial visible sector outside the \(G\)-invariants on which the one-layer averaging operator may have eigenvalue \(1\) under certain choices of measure on the ensemble. Consequently, the associated component of the \(t\)-moment error is not damped by increasing depth, so the ensemble does not converge to the Haar \(t\)-moment operator on the full represented group. This shows that Lemma~6 of Ragone et al.---and therefore their Theorem~2---cannot hold without additional hypotheses excluding such unmixed visible sectors. The rigorous details of this counterexample are available in \Cref{appendix:cycle-graph-qaoa}.

Conceptually, the underlying obstruction is very simple: the averaging measure may annihilate many visible modes while leaving others untouched. Even if the sampling measure has positive density almost everywhere with respect to Haar, convergence to Haar still fails whenever there exists a nontrivial visible representation sector on which the moment operator is not strictly contractive. Since Haar convergence requires decay of every nontrivial Fourier mode, the survival of a single such mode is sufficient to obstruct convergence. In the especially transparent abelian situation, this happens through characters. If \(\chi\) is a one-dimensional unitary character of a visible block and \(\nu\) is the sampling measure, then the corresponding averaged eigenvalue is
\[
\lambda_\nu=\int \chi(g)\,d\nu(g),
\qquad
|\lambda_\nu|\le 1.
\]
The equality case is rigid: \(\lambda_\nu=1\) holds if and only if \(\chi(g)=1\) for \(\nu\)-almost every \(g\). Hence there is no requirement that \(\nu\) be supported on a discrete set. It is enough that the support of \(\nu\) lie inside the kernel of the visible character. 

This is important for interpreting the counterexample correctly. The phenomenon is not an artifact of finitely supported parameter distributions (as is seen). One may equally well choose measures with continuous support on a nontrivial interval and still obtain eigenvalue \(1\), provided that the relevant visible character restricts to \(1\) throughout that interval. More generally, if \(\pi\) is a unitary representation and
\[
\left(\int \pi(g)\,d\nu(g)\right)v=v,
\]
then equality in the norm-convexity argument forces \(\pi(g)v=v\) for \(\nu\)-almost every \(g\). Thus the real obstruction is not discreteness but support inside a stabilizer, or equivalently inside the kernel of a visible representation sector. This of particular relevance to semisimple sectors, which could still in principle fail to contract when the sampling measure is supported in the stabilizer of a nontrivial visible subspace, despite the fact that such sectors only admit trivial characters. To be clear, this would not mean that the vector $v$ is a $G$-invariant, but that it has a nontrivial stabilizer in $G$ from which $\nu$ samples almost all of its elements. 

From the perspective of the present theory, this possibility is not surprising at all. Our operator-level error formulas isolate exactly the visible Fourier blocks present in
\[
\Delta^{(t)}
=
\mathcal M_{\nu}^{(t)}-\mathcal M_{\mu_K}^{(t)},
\]
so they immediately suggest that asymptotic non-Haar behavior can occur whenever the sampling rule fails to suppress some visible block. In other words, the error inequalities do not merely permit such counterexamples; they almost advertise them. A visible character, or more generally a visible irreducible block, that remains fixed under the sampling measure is precisely the sort of mechanism our framework is designed to detect. More generally, the same representation-theoretic mechanism suggests analogous obstructions for singular sampling laws, although that regime lies outside the Radon–Nikodym formalism developed in \Cref{sec:moment-formulas}.

This also clarifies the sense in which the present framework is sharper than design-only criteria. A design-centric inequality for a particular loss might reveal that some scalar quantity fails to concentrate as predicted, but it does not directly identify the representation-theoretic mechanism behind this phenomenon and may in fact neglect convergence behavior of higher-order moments. This is because a design-centric inequality necessarily forgets the representation-theoretic structure inherent to $\nu$ expressed as a measure relative to $\mu_K$. By contrast, the moment-error superoperator with the appropriate Krylov-Lie reductions exhibits the offending sector explicitly. In the cycle-graph QAOA counterexample, the failure of approximate-design convergence is therefore seen not as an accidental pathology, but as the natural consequence of an undamped visible block, made transparent by the Peter-Weyl form of our error bounds.

The significance of this obstruction  cannot be understated. Without depth-dependent Haar convergence, much of contemporary VQA theory cannot be unconditionally applied to the overwhelming majority of realistic circuits (nor can said theory be properly understood as working within universal asymptotic regimes), especially Ragone et al. and Anschuetz \cite{anschuetz_2025, ragone_lie_2024}, who require $\varepsilon$-approximate $t$-design convergence to an asymptotic structure for their formulas to be applicable. In contrast, the Krylov-Lie construction explicitly allows one to find groups that inherently approximate the initial manifold and its parameter distribution in a canonical fashion, and then relate the difference via the Radon-Nikodym derivative, thereby overcoming these limitations.

\chapter{An Observable Kawada--It\^o Theorem}
\label{subsec:observable-kawada-ito}

As it turns out, there is a classical result from the ergodic theory of random
walks on compact groups that allows us to recover the claimed convergence of
Ragone et al.\ under the right conditions: the Kawada--It\^o theorem
\cite{KawadaIto1940}, in the sharpened form due to Stromberg
\cite{Stromberg1960} (see also \cite{Applebaum2014} for a modern textbook
treatment, and \cite{NoteKawadaIto2022} for the Ces\`aro-averaged variant).  Some care
is required, however, in identifying the group on which the ergodic hypotheses must be imposed.  The $t$-th moment operators of an ensemble only
probe the matrix coefficients of the \emph{balanced} tensor-power
representations $g \mapsto g^{\otimes t}\otimes \overline{g}^{\otimes t}$;
irreducible representations of the dynamical Lie group that are nontrivial on
its subgroup of global phases occur in no balanced power and are therefore
invisible to moments of every order.  Consequently, no support condition on
the dynamical Lie group itself can be \emph{necessary} for moment
convergence, and the Peter--Weyl algebra of \emph{all} matrix coefficients is
strictly larger than what $t$-copy experiments can see.  The correct
formulation lives on a canonical quotient, the \emph{observable group}, on
which the visible coefficients are uniformly dense and the Kawada--It\^o
conditions become necessary \emph{and} sufficient.  We develop this
formulation in detail in the present chapter.

Throughout this chapter, $G \subset U(d)$ denotes a compact subgroup
(automatically a closed Lie subgroup, by Cartan's closed-subgroup theorem)
and
\[
\rho : G \longrightarrow U(\mathcal H), \qquad \dim_{\mathbb C}\mathcal H =
d_\rho < \infty,
\]
a continuous unitary representation. In our applications, $G$ is a compact
dynamical Lie group generated by the DLA $\mathfrak g \subset
\mathfrak u(d)$, $\mathcal H$ is the $N$-qubit
Hilbert space, and $\rho$ is the defining representation.  We write $\overline{\mathcal H}$ for the
conjugate Hilbert space and $\overline{A}$ for the entrywise complex
conjugate of an operator or vector with respect to this basis, consistent
with the conventions of \Cref{sec:haar-deviation-variance}.

For each integer $t \ge 0$ define the \emph{balanced $t$-th tensor-power
representation}
\begin{equation}
\label{eq:balanced-rep}
\rho^{(t)}(g) \;:=\; \rho(g)^{\otimes t} \otimes
\overline{\rho(g)}^{\otimes t}
\;\in\; U\!\big(\mathcal H_t\big),
\qquad
\mathcal H_t := \mathcal H^{\otimes t} \otimes
\overline{\mathcal H}^{\otimes t},
\end{equation}
with the convention that $\rho^{(0)}$ is the trivial representation on
$\mathcal H_0 = \mathbb C$.  Under the vectorization isomorphism
$\operatorname{vec} : \operatorname{End}(\mathcal H^{\otimes t}) \to
\mathcal H_t$ of \Cref{sec:haar-deviation-variance}, the representation $\rho^{(t)}$ is precisely
the matrix of the $t$-fold conjugation superoperator
$X \mapsto \rho(g)^{\otimes t} X (\rho(g)^{\dagger})^{\otimes t}$, i.e.\ the
analogue for $G$ of the representation $R^{(t)}$ used there for the
Krylov--Lie group.

All measures below are Borel probability measures; since $G$ is a compact
metrizable group, every Borel probability measure on $G$ is automatically
Radon \cite[Ch.~7]{Folland1999RealAnalysis}, and its support $\operatorname{supp}\mu$
(the smallest closed set of full measure) is well defined and nonempty.  For
probability measures $\mu,\nu$ on $G$, the convolution $\mu * \nu$ is the
pushforward of $\mu \otimes \nu$ under the multiplication map, and we write
$\mu^{*L}$ for the $L$-fold convolution power.  Given a closed subset
$S \subset G$, we write $[S]$ for the smallest closed subgroup of $G$
containing $S$.  Finally, $\mu_G$ denotes the normalized Haar probability
measure of $G$.

Since $\dim \mathcal H_t = d_\rho^{2t} < \infty$, all norms on
our space are equivalent, with constants depending only on
$t$ and $d_\rho$; in particular, for fixed $t$, convergence of moment
operators in the operator norm $\|\cdot\|_\infty$, in any Schatten norm, or
in the induced Schatten $p\to p$ norms of the associated superoperators are
all equivalent notions (for $p = 2$ the vectorization
$\operatorname{vec}$ is a unitary between Hilbert--Schmidt space and
$\mathcal H^{\otimes t}\otimes\overline{\mathcal H}^{\otimes t}$, so the
identification is even isometric).  We therefore state convergence in
$\|\cdot\|_\infty$ and use the other norms freely.

\begin{definition}[Observable moment algebra]
\label{def:observable-algebra}
For $t \ge 0$ let
\[
\mathcal F_t \;:=\;
\Big\{\, g \mapsto \big\langle V, \rho^{(t)}(g)\, W \big\rangle
\; \, \Big \vert \,\; V, W \in \mathcal H_t \,\Big\}
\;\subset\; C(G)
\]
be the space of matrix coefficients of $\rho^{(t)}$, and set
\[
\mathcal F := \bigcup_{t = 0}^\infty \mathcal F_t,
\qquad
\mathcal A := \operatorname{span}_{\mathbb C} \mathcal F .
\]
We call $\mathcal A$ the \emph{observable moment algebra} of the pair
$(G,\rho)$. This title is motivated in \Cref{appendix:observable-quotient-proofs}, specifically \Cref{lem:algebra}.
\end{definition}

Two comments on this definition are in order.  First, since every
subrepresentation $\pi \subset \rho^{(t)}$ embeds isometrically into
$\mathcal H_t$, the matrix coefficients of subrepresentations of
$\rho^{(t)}$ already belong to $\mathcal F_t$ (cf. \Cref{appendix:observable-quotient-proofs} for greater detail); conversely, decomposing
$\rho^{(t)}$ into irreducibles exhibits every element of $\mathcal F_t$ as a
finite sum of coefficients of irreducible subrepresentations.  Hence
defining $\mathcal F_t$ via subrepresentations, as one might prefer, yields
the same span.  Second, the name is justified by the observation that
$\mathcal A$ contains exactly the $t$-copy loss statistics of the ansatz:
for any density matrix $\varrho$, observable $O$, and $t \ge 1$, the
function
\[
g \;\longmapsto\; \ell_g(\varrho, O)^t
= \operatorname{Tr}\!\big[\rho(g)\varrho\rho(g)^\dagger O\big]^t
= \Big\langle \operatorname{vec}\!\big(O^{\otimes t}\big),\;
\rho^{(t)}(g)\, \operatorname{vec}\!\big(\varrho^{\otimes t}\big)
\Big\rangle
\]
lies in $\mathcal F_t$, and conversely $\mathcal F_t$ is spanned by such
functions with $\varrho^{\otimes t}, O^{\otimes t}$ replaced by arbitrary
operators on $\mathcal H^{\otimes t}$.  Thus $\mathcal A$ is precisely the
function system on $G$ accessible to $t$-th order moment experiments,
uniformly in $t$.

We now identify the canonical compact group on which $\mathcal A$ becomes a full function algebra.

\begin{definition}[Observable group]
\label{def:observable-group}
Let
\[
\Phi := \rho^{(1)} :\; G \longrightarrow
U\big(\mathcal H \otimes \overline{\mathcal H}\big),
\qquad
\Phi(g) = \rho(g) \otimes \overline{\rho(g)},
\]
and define the \emph{observable group} of the pair $(G,\rho)$ as
\[
G_{\mathrm{obs}} := \Phi(G)
\;\subset\; U\big(\mathcal H \otimes \overline{\mathcal H}\big).
\]
Given a probability measure $\mu$ on $G$, its \emph{observable pushforward}
is the probability measure
$\mu_{\mathrm{obs}} := \Phi_*\mu$ on $G_{\mathrm{obs}}$, defined by
$\Phi_*\mu(B) = \mu(\Phi^{-1}(B))$ for Borel $B \subset G_{\mathrm{obs}}$.
\end{definition}

Since $\Phi$ is a continuous homomorphism, $\Phi(G)$ is a compact, hence closed, subgroup of the ambient unitary
group, and it is again a compact Lie group.  Under vectorization, $\Phi(g)$
is the matrix of the unitary channel $X \mapsto \rho(g) X \rho(g)^\dagger$,
so $G_{\mathrm{obs}}$ is canonically the image of $G$ in the projective
unitary group $PU(\mathcal H)$, i.e.\ the group of \emph{channels}
implemented by the ansatz.  The next lemma collects its basic properties.
\begin{lemma}[Structure of the observable group]
\phantomsection
\label{lem:structure}
Let
\[
Z := \rho^{-1}\big(\mathbb T \cdot I\big)
= \{ g \in G \, \vert \, \rho(g) \in \mathbb T\cdot I\},
\]
where \(\mathbb T = \bigl\{\lambda \in \mathbb C \, \big \vert \, |\lambda| = 1\bigr\}\).
Then:
\begin{enumerate}
\item
\(\ker \Phi = Z = \bigcap_{t = 1}^\infty \ker \rho^{(t)}\), and \(Z\) is a closed
normal subgroup of \(G\) containing the global phases
\(G \cap \mathbb T\cdot I\) of the defining representation (with equality
\(Z = G \cap \mathbb T \cdot  I\) when \(\rho\)
is the defining inclusion \(G \subset U(\mathcal H)\)).

\item
\(\Phi\) induces an isomorphism of compact groups
\[
G/Z \;\cong\; G_{\mathrm{obs}},
\]
and for every \(t \ge 0\) there is a unique continuous unitary representation
\[
\widetilde\rho^{(t)} : G_{\mathrm{obs}} \to U(\mathcal H_t)
\]
with
\[
\rho^{(t)} = \widetilde\rho^{(t)} \circ \Phi .
\]

\item
For probability measures \(\mu, \nu\) on \(G\):
\[
\Phi_*(\mu * \nu) = \Phi_*\mu * \Phi_*\nu,
\]
\[
\operatorname{supp}(\Phi_*\mu) = \Phi(\operatorname{supp}\mu),
\]
and
\[
\Phi([S]) = [\Phi(S)]
\qquad\text{for every closed }S \subset G.
\]

\item
\(\Phi_*\mu_G = \mu_{G_{\mathrm{obs}}}\), the normalized Haar measure of
\(G_{\mathrm{obs}}\).

\item
For every \(t \ge 0\) and every probability measure \(\mu\) on \(G\),
\[
\mathcal M^{(t)}_\mu
= \int_{G_{\mathrm{obs}}} \widetilde\rho^{(t)}(x)\, d(\Phi_*\mu)(x)
= \widehat{\Phi_*\mu}\big(\widetilde\rho^{(t)}\big),
\]
the noncommutative Fourier coefficient of the pushforward at
\(\widetilde\rho^{(t)}\).
\end{enumerate}
\end{lemma}

\begin{proof}
(1). If \(\Phi(g) = I\) then, by vectorization,
\[
\rho(g) X \rho(g)^\dagger = X
\qquad\text{for all }X \in \operatorname{End}(\mathcal H),
\]
so \(\rho(g)\) lies in the commutant of the full matrix algebra and is
therefore a scalar unitary,
\[
\rho(g) = \lambda I
\qquad\text{with }|\lambda| = 1.
\]
Conversely, if \(\rho(g) = \lambda I\) then
\[
\rho^{(t)}(g)
=
\lambda^t \bar\lambda^t\, I
=
I
\qquad\text{for every }t \ge 1.
\]
This proves both equalities in (1). Normality and closedness are clear,
since \(Z\) is the kernel of a continuous homomorphism.

(2). The induced map
\[
G/Z \longrightarrow G_{\mathrm{obs}}
\]
is a continuous bijective homomorphism between compact Hausdorff groups,
hence a homeomorphic isomorphism.

Since \(\ker\Phi = Z \subset \ker\rho^{(t)}\) by part 1, the representation
\(\rho^{(t)}\) factors uniquely through \(G/Z \cong G_{\mathrm{obs}}\),
defining \(\widetilde\rho^{(t)}\). Continuity of \(\widetilde\rho^{(t)}\) follows because
\(G \to G/Z\) is an open quotient map (equivalently, because
\(\widetilde\rho^{(t)}\) is obtained by composing with the inverse of the
homeomorphism \(G/Z \cong G_{\mathrm{obs}}\)).

(3). The first identity holds because \(\Phi\) is a homomorphism: pushforward
along \(\Phi\) of the multiplication map of \(G\) is the multiplication map
of \(G_{\mathrm{obs}}\) applied to the pushforwards. Explicitly,
\[
\Phi_*(\mu * \nu)
=
\Phi_*\mu * \Phi_*\nu.
\]
For the support identity, \(\Phi(\operatorname{supp}\mu)\) is compact (hence
closed) and carries full \(\Phi_*\mu\)-measure. Conversely, if
\(U\subset G_{\mathrm{obs}}\) is open with
\(U \cap \Phi(\operatorname{supp}\mu) \neq \emptyset\), then
\(\Phi^{-1}(U)\) is open and meets \(\operatorname{supp}\mu\). Hence
\[
\Phi_*\mu(U)
\ge
\mu(\Phi^{-1}(U))
> 0.
\]
This shows
\[
\operatorname{supp}(\Phi_*\mu)
=
\Phi(\operatorname{supp}\mu).
\]

For the third claim, \(\Phi([S])\) is a compact (closed) subgroup containing
\(\Phi(S)\), so
\[
[\Phi(S)] \subset \Phi([S]).
\]
Conversely, \(\Phi^{-1}([\Phi(S)])\) is a closed subgroup containing \(S\),
hence containing \([S]\). Applying \(\Phi\) gives
\[
\Phi([S]) \subset [\Phi(S)].
\]
Thus \(\Phi([S]) = [\Phi(S)]\).

(4). The pushforward \(\Phi_*\mu_G\) is a probability measure on
\(G_{\mathrm{obs}}\) invariant under left translation by every element of
\(\Phi(G) = G_{\mathrm{obs}}\). Indeed, for any \(h\in G\), the translate
of \(\Phi_*\mu_G\) by \(\Phi(h)\) equals \(\Phi_*\) applied to the translate
of \(\mu_G\) by \(h\), which leaves \(\mu_G\) invariant. By uniqueness of normalized Haar measure on a compact group
\cite[Thm.~2.10, 2.20]{Folland1995Harmonic}, we must have $\Phi_*\mu_G = \mu_{G_{\mathrm{obs}}}$.

(5). Let \(t\ge0\) and let \(\mu\) be a probability measure on \(G\).

By part (2),
\[
\rho^{(t)} = \widetilde\rho^{(t)} \circ \Phi.
\]
Therefore
\[
\mathcal M^{(t)}_\mu
=
\int_G \rho^{(t)}(g)\, d\mu(g)
=
\int_G \widetilde\rho^{(t)}(\Phi(g))\, d\mu(g)
=
\int_{G_{\mathrm{obs}}} \widetilde\rho^{(t)}(x)\, d(\Phi_*\mu)(x),
\]
which is the noncommutative Fourier coefficient
\(\widehat{\Phi_*\mu}(\widetilde\rho^{(t)})\).

This proves (5).
\end{proof}

We next tackle ergodic hypotheses on the observable group.

\begin{definition}[Adaptedness and strict aperiodicity]
\label{def:adapted-aperiodic}
Let $K$ be a compact group and $\nu$ a probability measure on $K$.  We say
that $\nu$ is \emph{adapted} if $[\operatorname{supp}\nu] = K$;
$\nu$ is \emph{strictly aperiodic} if there exists no proper closed
normal subgroup $N \triangleleft K$, $N \neq K$, and element $x \in K$
such that $\operatorname{supp}\nu \subset xN$. A probability measure $\mu$ on $G$ is called \emph{observably adapted}
(resp.\ \emph{observably strictly aperiodic}) if its observable pushforward
$\Phi_*\mu$ is adapted (resp.\ strictly aperiodic) on $G_{\mathrm{obs}}$.
\end{definition}

We emphasize that for strict aperiodicity the subgroup $N$ is required to be \emph{proper},
but may perfectly well be trivial: the case $N = \{e\}$ excludes point
masses $\delta_x$ (whose convolution powers $\delta_{x^L}$ never converge to
Haar measure on a nontrivial group), while allowing $N = K$ would make the
condition vacuously false for every measure.  Strict aperiodicity neither
implies nor is implied by adaptedness; the classical Kawada--It\^o theorem
requires both \cite{KawadaIto1940, Stromberg1960, Applebaum2014}.

The hypotheses of adaptedness and strict aperiodicity at the level of $G$ are
strictly stronger than the observable ones:

\begin{lemma}[Descent of the ergodic hypotheses]
\label{lem:descent-ergodic}
If $\mu$ is adapted on $G$, then $\mu$ is observably adapted; if $\mu$ is
strictly aperiodic on $G$, then $\mu$ is observably strictly aperiodic.
Neither converse holds in general.
\end{lemma}

\begin{proof}
Write $\nu = \Phi_*\mu$.  If $[\operatorname{supp}\mu] = G$, then by
Lemma~\ref{lem:structure}(c),
\[
[\operatorname{supp}\nu]
= [\Phi(\operatorname{supp}\mu)]
= \Phi([\operatorname{supp}\mu])
= \Phi(G) = G_{\mathrm{obs}} .
\]
For the second claim we argue by contraposition.  Suppose
$\operatorname{supp}\nu \subset \bar x \bar N$ for some proper closed
normal $\bar N \triangleleft G_{\mathrm{obs}}$ and
$\bar x \in G_{\mathrm{obs}}$.  Put $N := \Phi^{-1}(\bar N)$, a closed
normal subgroup of $G$, and pick $x \in \Phi^{-1}(\{\bar x\})$.  Then $N$ is
proper (else $\bar N = \Phi(N) = G_{\mathrm{obs}}$ by surjectivity of
$\Phi$), and
\[
\operatorname{supp}\mu
\subset \Phi^{-1}\big(\operatorname{supp}\nu\big)
\subset \Phi^{-1}(\bar x \bar N) = xN,
\]
where the first inclusion holds because
$\mu\big(\Phi^{-1}(\operatorname{supp}\nu)\big) =
\nu(\operatorname{supp}\nu) = 1$ and
$\Phi^{-1}(\operatorname{supp}\nu)$ is closed.  Hence $\mu$ is not strictly
aperiodic on $G$.
\end{proof}

We first record the equivalence between moment convergence and weak-$*$
convergence relative to the observable algebra.  This part requires no
convolution structure whatsoever and holds for arbitrary sequences of
measures; it is the rigorous form of the
``moment-operator/observable correspondence.''

\begin{proposition}[Moment--observable correspondence]
\label{prop:moment-observable}
Let $(\mu_L)_{L \in \mathbb N}$ be any sequence of probability measures on
$G$ and let $\mu_\infty$ be a probability measure on $G$.  The following
are equivalent:
\begin{enumerate}
\item
For every $t \ge 1$, as $L \to \infty$,
$\big\| \mathcal M^{(t)}_{\mu_L} - \mathcal M^{(t)}_{\mu_\infty} \big\|_\infty
\longrightarrow 0$.
\item
for every $f \in \mathcal A$, as $L \to \infty$,
$\displaystyle\int_G f\, d\mu_L \longrightarrow \int_G f\, d\mu_\infty$.
\item
As $L \to \infty$, $\Phi_*\mu_L \longrightarrow \Phi_*\mu_\infty$ in the weak-$*$ topology of
$C(G_{\mathrm{obs}})^*$. That is to say, \
$\int_{G_{\mathrm{obs}}} h\, d(\Phi_*\mu_L) \to
\int_{G_{\mathrm{obs}}} h\, d(\Phi_*\mu_\infty)$ for every
$h \in C(G_{\mathrm{obs}})$.
\end{enumerate}
\end{proposition}

\begin{proof}
(1) $\Rightarrow$ (2).  Let $f \in \mathcal F_t$, say
$f(g) = \langle V, \rho^{(t)}(g)W\rangle$.  Then for every probability
measure $\mu$ on $G$, interchanging the (entrywise, finite-dimensional)
integral with the inner product,
\begin{equation}
\label{eq:coefficient-pairing}
\int_G f\, d\mu = \big\langle V,\,\mathcal M^{(t)}_{\mu}\, W \big\rangle,
\end{equation}
whence
\[
\Big| \int_G f\, d\mu_L - \int_G f\, d\mu_\infty \Big|
\le \|V\|\,\|W\|\,
\big\|\mathcal M^{(t)}_{\mu_L} - \mathcal M^{(t)}_{\mu_\infty}\big\|_\infty
\longrightarrow 0 .
\]
The case of general $f \in \mathcal A$ follows by linearity, since
$\mathcal A = \operatorname{span}\bigcup_t \mathcal F_t$ (no products need
be considered separately, by Lemma~\ref{lem:algebra}).

(2) $\Rightarrow$ (1).  Fix $t$ and apply part (2) to the
$d_\rho^{2t} \times d_\rho^{2t}$ matrix coefficients
$g \mapsto \langle e_i, \rho^{(t)}(g) e_j\rangle \in \mathcal F_t$, where
$(e_i)$ is an orthonormal basis of $\mathcal H_t$.  By
\eqref{eq:coefficient-pairing} this gives entrywise convergence
$\mathcal M^{(t)}_{\mu_L} \to \mathcal M^{(t)}_{\mu_\infty}$, which on this finite-dimensional
space implies convergence in norm.

(2) $\Leftrightarrow$ (3).  By Lemma~\ref{lem:density} and the
change-of-variables formula, for $f = \widetilde f\circ\Phi \in \mathcal A$,
\[
\int_G f\, d\mu_L = \int_{G_{\mathrm{obs}}} \widetilde f\, d(\Phi_*\mu_L),
\qquad
\int_G f\, d\mu_\infty
= \int_{G_{\mathrm{obs}}} \widetilde f\, d(\Phi_*\mu_\infty),
\]
so (3) $\Rightarrow$ (2) is immediate since
$\widetilde f \in C(G_{\mathrm{obs}})$.  Conversely, assume (2), letting
$h \in C(G_{\mathrm{obs}})$ and $\varepsilon > 0$.  By
Lemma~\ref{lem:density}(a) choose $\widetilde f \in \widetilde{\mathcal A}$
with $\|h - \widetilde f\|_\infty \le \varepsilon$.  Since
$\Phi_*\mu_L$ and $\Phi_*\mu_\infty$ are probability measures,
\[
\Big| \int h\, d(\Phi_*\mu_L) - \int h\, d(\Phi_*\mu_\infty) \Big|
\le 2\varepsilon +
\Big| \int \widetilde f\, d(\Phi_*\mu_L)
- \int \widetilde f\, d(\Phi_*\mu_\infty) \Big|,
\]
and the last term tends to $0$ by (2) applied to
$f = \widetilde f \circ \Phi \in \mathcal A$.  As $\varepsilon$ was
arbitrary, (3) follows.
\end{proof}

We can now state and prove the observable Kawada--It\^o criterion for convergence to the
DLG $t$-design hierarchy.

\begin{theorem}[Observable Kawada--It\^o criterion for convergence to the
DLG \(t\)-design hierarchy]
\label{thm:observable-kawada-ito}
Let \(G \subset U(d)\) be the compact Lie group generated by the DLA
\(\mathfrak g\), let \(\rho\) be as above, and let \(\mathcal E_1\) be the
single-layer ensemble, viewed as a probability measure on \(G\). For \(L \ge 1\)
define
\[
\mathcal E_L := \mathcal E_1^{*L},
\qquad
\nu := \Phi_* \mathcal E_1,
\qquad
P_{t} := \mathcal M^{(t)}_{\mu_G}.
\]
The following are equivalent:
\begin{enumerate}
\item
\(\mathcal E_1\) is observably adapted and observably strictly aperiodic,
i.e.\ \(\nu\) is adapted and strictly aperiodic on \(G_{\mathrm{obs}}\);
\item
for every \(t \ge 1\), as $L\to\infty$
\[
\big\|\mathcal A^{(t)}_{\mathcal E_L}\big\|_\infty
=
\big\|\mathcal M^{(t)}_{\mathcal E_L}-P^{(t)}\big\|_\infty
\longrightarrow 0.
\]
\item
for every \(f\in\mathcal A\), as $L \to \infty$,
\[
\int_G f\,d\mathcal E_L
\longrightarrow
\int_G f\,d\mu_G.
\]
\item
\[
\Phi_*\mathcal E_L
=
\nu^{*L}
\longrightarrow
\mu_{G_{\mathrm{obs}}}
\]
in the weak-\(*\) topology of \(C(G_{\mathrm{obs}})^*\).
\end{enumerate}
Moreover, if these conditions hold, then for each \(t\ge1\),
\begin{equation}
\label{eq:spectral-radius}
r_t
:=
\max\Big\{
|\lambda|:\lambda\in
\operatorname{spec}\!\big(\mathcal M^{(t)}_{\mathcal E_1}\big)\setminus\{1\}
\Big\}
<1,
\end{equation}
with the convention \(r_t:=0\) if the set is empty, and
\[
\ker\!\big(\mathcal M^{(t)}_{\mathcal E_1}-I\big)
=
\operatorname{im}P_{t}.
\]
Hence, for every \(r\in(r_t,1)\), there exists \(C_{t,r}>0\) such that
\begin{equation}
\label{eq:gelfand-rate}
\big\|\mathcal A^{(t)}_{\mathcal E_L}\big\|_\infty
\le
C_{t,r}\,r^L
\qquad\text{for all }L\ge1,
\end{equation}
and therefore
\begin{equation}
\label{eq:gelfand-depth}
L
\;\ge\;
\frac{\log C_{t,r}+\log(1/\varepsilon)}{\log(1/r)}
\quad\Longrightarrow\quad
\big\|\mathcal A^{(t)}_{\mathcal E_L}\big\|_\infty
\le
\varepsilon.
\end{equation}
If, in addition, \(\mathcal E_1\) is symmetric in the sense that
\[
\mathcal E_1(B)=\mathcal E_1(B^{-1})
\qquad\text{for every Borel }B\subset G,
\]
then \(\mathcal M^{(t)}_{\mathcal E_1}\) is self-adjoint for the
Hilbert--Schmidt inner product, and hence
\begin{equation}
\label{eq:ragone-rate}
\big\|\mathcal A^{(t)}_{\mathcal E_L}\big\|_\infty
=
\big\|\mathcal A^{(t)}_{\mathcal E_1}\big\|_\infty^L
=
r_t^{\,L}
\qquad\text{for all }L\ge1.
\end{equation}
Equivalently,
\begin{equation}
\label{eq:ragone-depth}
L
\;\ge\;
\frac{\log(1/\varepsilon)}
{\log\!\big(1/\|\mathcal A^{(t)}_{\mathcal E_1}\|_\infty\big)}
\quad\Longrightarrow\quad
\big\|\mathcal A^{(t)}_{\mathcal E_L}\big\|_\infty
\le
\varepsilon.
\end{equation}
\end{theorem}

\begin{proof}
Write
\[
M_t := \mathcal M^{(t)}_{\mathcal E_1},
\qquad
P_t := P^{(t)}=\mathcal M^{(t)}_{\mu_G}.
\]
By Lemma~\ref{lem:structure}(1)--(3), for every probability measure \(\mu\) on \(G\),
\begin{equation}
\label{eq:obs-identities}
\Phi_*(\mu*\nu')=\Phi_*\mu*\Phi_*\nu',
\qquad
\Phi_*\mu_G=\mu_K,
\qquad
\mathcal M^{(t)}_\mu
=
\widehat{\Phi_*\mu}\big(\widetilde\rho^{(t)}\big).
\end{equation}
In particular,
\begin{equation}
\label{eq:obs-powers-proof}
\Phi_*\mathcal E_L
=
(\Phi_*\mathcal E_1)^{*L}
=
\nu^{*L}.
\end{equation}
Next, by Ragone's convolution identity for moment operators
\cite{ragone_lie_2024},
\begin{equation}
\label{eq:moment-power-proof}
\mathcal M^{(t)}_{\mathcal E_L}
=
(\mathcal M^{(t)}_{\mathcal E_1})^L.
\end{equation}
Also, by Haar averaging
\cite{ragone_lie_2024},
\begin{equation}
\label{eq:haar-proj-proof}
P_t
=
\int_G d\mu_G(g) \, \rho^{(t)}(g)
\end{equation}
is the orthogonal projection onto
\[
(\mathcal H_t)^G
:=
\{v\in\mathcal H_t \, \vert \,\rho^{(t)}(g)v=v \; \; \forall g\in G\}.
\]
Hence
\begin{equation}
\label{eq:absorption-proof}
\mathcal M^{(t)}_{\mathcal E_1}P_t=P_t\mathcal M^{(t)}_{\mathcal E_1}=P_t.
\end{equation}
Combining \eqref{eq:moment-power-proof} and \eqref{eq:absorption-proof},
\begin{equation}
\label{eq:error-factor-proof}
\mathcal A^{(t)}_{\mathcal E_L}
=
\mathcal M^{(t)}_{\mathcal E_L}-P_t
=
(\mathcal M^{(t)}_{\mathcal E_1})^L-P_t
=
(\mathcal M^{(t)}_{\mathcal E_1}-P_t)^L
=
\big(\mathcal A^{(t)}_{\mathcal E_1}\big)^L.
\end{equation}

\noindent\emph{(1) \(\Leftrightarrow\) (4).}
By \eqref{eq:obs-powers-proof}, condition (4) is exactly
\[
\nu^{*L}\longrightarrow\mu_K
\qquad\text{weak-*}\text{ on } G_{\operatorname{obs}}.
\]
This is precisely the classical Kawada--It\^o theorem on the compact group \(G_{\operatorname{obs}}\)
(\cite{KawadaIto1940}; more generally \cite{Stromberg1960}), so (4) holds if
and only if \(\nu\) is adapted and strictly aperiodic, i.e. if and only if (1)
holds.

\noindent\emph{(2) \(\Leftrightarrow\) (3) \(\Leftrightarrow\) (4).}
This is Proposition~\ref{prop:moment-observable}, applied with
\[
\mu_L=\mathcal E_L,
\qquad
\mu_\infty=\mu_G,
\]
using \eqref{eq:obs-powers-proof} and \(\Phi_*\mu_G=\mu_{G_{\operatorname{obs}}}\) from
\eqref{eq:obs-identities}. This proves the equivalence of (2), (3), and (4),
hence of all four conditions.

\noindent\emph{General quantitative estimate.}
Assume now that the equivalent conditions hold, and fix \(t\ge1\).
By part (5) of Lemma~\ref{lem:structure},
\begin{equation}
\label{eq:Mt-Fourier-proof}
\mathcal M^{(t)}_{\mathcal E_1}
=
\widehat{\nu}\big(\widetilde\rho^{(t)}\big).
\end{equation}
Since \(\nu\) is adapted and strictly aperiodic on \(G_{\operatorname{obs}}\), Lemma~\ref{lem:spectral}
applied to \(\widetilde\rho^{(t)}\) implies that every eigenvalue of \(\mathcal M^{(t)}_{\mathcal E_1}\) on
the unit circle is equal to \(1\), and that the \(1\)-eigenspace is exactly the
subspace of \(G_{\operatorname{obs}}\)-invariant vectors. Because
\(\widetilde\rho^{(t)}\circ\Phi=\rho^{(t)}\) and \(\Phi\) is surjective, this is
precisely \((\mathcal H_t)^G\), hence
\begin{equation}
\label{eq:1eigenspace-proof}
\ker(\mathcal M^{(t)}_{\mathcal E_1}-I)
=
(\mathcal H_t)^G
=
\operatorname{im}P_t.
\end{equation}
Therefore \(\mathcal A^{(t)}_{\mathcal E_1}=\mathcal M^{(t)}_{\mathcal E_1}-P_t\) vanishes on \(\operatorname{im}P_t\),
while on \(\operatorname{im}(I-P_t)\) it agrees with the restriction
\[
\mathcal M^{(t)}_{\mathcal E_1}\big|_{\operatorname{im}(I-P_t)}.
\]
By \eqref{eq:1eigenspace-proof}, this restriction has no eigenvalue of modulus
\(1\). Since \(\mathcal H_t\) is finite-dimensional, it follows that
\[
r\!\left(\mathcal A^{(t)}_{\mathcal E_1}\right)
=
r_t
<
1,
\]
where \(r_t\) is the quantity in \eqref{eq:spectral-radius}. Gelfand's formula and enlargement of constants over finitely many initial powers therefore yields that for every \(r\in(r_t,1)\) there is a \(C_{t,r}>0\) such that
\begin{equation}
\label{eq:gelfand-proof}
\big\|\big(\mathcal A^{(t)}_{\mathcal E_1}\big)^L\big\|_\infty
\le
C_{t,r}\,r^L
\qquad\text{for all }L\ge1.
\end{equation}
Combining \eqref{eq:gelfand-proof} with \eqref{eq:error-factor-proof} gives
\eqref{eq:gelfand-rate}.
Thus, if
\[
L \ge \frac{\log C_{t,r}+\log(1/\varepsilon)}{\log(1/r)},
\]
then \(C_{t,r}r^L\le \varepsilon\), proving \eqref{eq:gelfand-depth}.

Now assume in addition that
\[
\mathcal E_1(B)=\mathcal E_1(B^{-1})
\qquad\text{for every Borel }B\subset G.
\]
Under what amounts to this condition, Ragone et al. proves that \(\mathcal M^{(t)}_{\mathcal E_1}=\mathcal M^{(t)}_{\mathcal E_1}\) is self-adjoint
with respect to the Hilbert--Schmidt inner product
\cite{ragone_lie_2024}. Since \(P_t\) is an orthogonal projection and
\(\mathcal M^{(t)}_{\mathcal E_1}P_t=P_t\mathcal M^{(t)}_{\mathcal E_1}=P_t\), the operator
\[
\mathcal A^{(t)}_{\mathcal E_1}=\mathcal M^{(t)}_{\mathcal E_1}-P_t
\]
is also self-adjoint. Hence its operator norm equals its spectral radius, so
\[
\big\|\mathcal A^{(t)}_{\mathcal E_1}\big\|_\infty
=
r\!\left(\mathcal A^{(t)}_{\mathcal E_1}\right)
=
r_t
<
1.
\]
Using \eqref{eq:error-factor-proof},
\[
\mathcal A^{(t)}_{\mathcal E_L}
=
\big(\mathcal A^{(t)}_{\mathcal E_1}\big)^L,
\]
and self-adjointness gives
\[
\big\|\mathcal A^{(t)}_{\mathcal E_L}\big\|_\infty
=
\big\|\mathcal A^{(t)}_{\mathcal E_1}\big\|_\infty^L
=
r_t^L,
\]
which is \eqref{eq:ragone-rate}. The depth bound \eqref{eq:ragone-depth} is then
immediate.
\end{proof}

\begin{remark}\label{rem:qaoaHaar}
These hypotheses are not abstract: for a randomized QAOA model with
full-support angle sampling, the one-layer measure is adapted and strictly aperiodic on the
compact group generated by the DLA even before considering the observable group, so the convolution powers converge to Haar in the $L \rightarrow \infty$ limit by \Cref{lem:descent-ergodic} (see \Cref{appendix:qaoaHaar}).  In fact, by essentially the same argument as \Cref{appendix:qaoaHaar}, it is
not hard to see that any fully-supported angle sampling law on the corresponding compact parameter torus for a one-layer VQA will converge as a channel ensemble
to Haar over the dynamical Lie group as $L \rightarrow \infty$. That said, since in algorithmic design we are not constrained to keep our measure the same layer-by-layer nor possess observable adaptedness or observable strict aperiodicity (for instance,
one could selectively decide which subset of layers to optimize with a Bayesian algorithm),
for if this is what instills a deadly Haar convergence, then we should simply avoid strategies
that implicitly permit such choices of measure.
\end{remark}

Relatedly, when Fontana et al. prove approximate 2-design mixing for a specific randomized LASA construction obtained by sampling from an orthogonal basis of DLA generators with uniformly distributed angles over full periods, they also explicitly note, immediately after their Theorem 2.6, that this assumption of approximate 2-design mixing may fail for other initialization schemes or correlated parameters \cite{fontana_adjoint_2024}. This result, \Cref{thm:observable-kawada-ito}, shows precisely when this is the case.

The following practical criterion reduces the verification of (i) to a
character computation, in the spirit of the abelian character discussion of \Cref{sec:variance-inequalities-non-haar}.
\begin{corollary}[Character test for the observable hypotheses]
\label{cor:character-test}
Let \(\nu = \Phi_*\mathcal E_1\) be the observable pushforward of the
one-layer law, viewed as a probability measure on \(G_{\mathrm{obs}}\), and
suppose \(\nu\) is adapted. Then \(\nu\) is strictly aperiodic if and only if
\[
\Big| \int_{G_{\mathrm{obs}}} \chi(x)\, d\nu(x) \Big| < 1
\qquad\text{for every nontrivial continuous character }
\chi : G_{\mathrm{obs}} \to \mathbb T .
\]
Equivalently, \(\mathcal E_1\) satisfies condition~\textnormal{(i)} of
\Cref{thm:observable-kawada-ito} if and only if
\[
[\operatorname{supp}\Phi_*\mathcal E_1] = G_{\mathrm{obs}}
\]
and no nontrivial continuous character of \(G_{\mathrm{obs}}\) has
unit-modulus mean under \(\Phi_* \mathcal E_1\).
\end{corollary}

\begin{proof}
\emph{(\(\Rightarrow\)\,)} 
Suppose there exists a nontrivial character
\(\chi : G_{\mathrm{obs}}\to\mathbb T\) with
\(|\widehat\nu(\chi)| = 1\). Then \(\widehat\nu(\chi) = \lambda\) for some
\(\lambda\in\mathbb C\) with \(|\lambda|=1\); this is exactly the
one-dimensional case of a unit-modulus eigenvalue. By Lemma~\ref{lem:spectral}(b) applied to the representation \(\chi\),
we must have
\[
\chi(x) \equiv \lambda
\qquad\text{for all }x\in \operatorname{supp}\nu.
\]
Fix \(x_0 \in \operatorname{supp}\nu\). For any
\(x \in \operatorname{supp}\nu\),
\[
\chi(x_0^{-1}x)
=
\overline{\chi(x_0)}\,\chi(x)
=
\overline{\lambda}\,\lambda
=
1.
\]
Hence \(x_0^{-1}x \in \ker\chi\), so
\[
\operatorname{supp}\nu \subset x_0 \ker\chi.
\]
Since \(\chi\) is nontrivial, \(\ker\chi\) is a proper closed normal
subgroup of \(G_{\mathrm{obs}}\). Thus \(\nu\) is supported in a coset of a
proper closed normal subgroup, so \(\nu\) is not strictly aperiodic.

\emph{(\(\Leftarrow\)\,)}
Conversely, suppose \(\nu\) is not strictly aperiodic. Then there exist a
proper closed normal subgroup
\(N \triangleleft G_{\mathrm{obs}}\) and some \(x_0\in G_{\mathrm{obs}}\)
such that
\[
\operatorname{supp}\nu \subset x_0 N.
\]
Let
\[
q : G_{\mathrm{obs}} \to Q := G_{\mathrm{obs}}/N
\]
be the quotient map. Then
\[
q\big(\operatorname{supp}\nu\big) = \{\bar x_0\},
\]
where \(\bar x_0 = q(x_0)\). Since \(\nu\) is adapted,
\[
[\operatorname{supp}\nu] = G_{\mathrm{obs}},
\]
and therefore
\[
Q
=
q\big([\operatorname{supp}\nu]\big)
=
\big[q(\operatorname{supp}\nu)\big]
=
[\{\bar x_0\}].
\]
In other words, \(Q\) is the closure of the cyclic subgroup generated by
\(\bar x_0\). Hence \(Q\) is a nontrivial compact abelian group. Every nontrivial compact abelian group admits a nontrivial continuous
character \(\eta : Q \to \mathbb T\)
(see, e.g., \cite[Cor.~4.41 or Thm.~5.11]{Folland1995Harmonic}). Define
\[
\chi := \eta \circ q : G_{\mathrm{obs}} \to \mathbb T.
\]
Then \(\chi\) is a nontrivial continuous character of \(G_{\mathrm{obs}}\). Moreover, for any \(x \in x_0 N\),
\[
\chi(x)
=
\eta(q(x))
=
\eta(\bar x_0),
\]
so \(\chi\) is constant on \(x_0 N\), and in particular constant on
\(\operatorname{supp}\nu\subset x_0 N\). Therefore
\[
\left|\int_{G_{\mathrm{obs}}} \chi(x)\, d\nu(x)\right|
=
|\eta(\bar x_0)|
=
1.
\]
Thus failure of strict aperiodicity produces a nontrivial character with
unit-modulus mean. This proves the contrapositive of the desired
implication.
\end{proof}

\begin{remark}[Consistency with the counterexample of Appendix~A]
\label{rem:appendix-A}
Corollary~\ref{cor:character-test} places the cycle-graph QAOA
counterexample of Appendix~A in its natural context.  There, the central
torus of the DLA gives rise to genuinely \emph{visible} one-dimensional
character sectors of the represented group---characters of
$G_{\mathrm{obs}}$ in the present language, since they occur inside the
balanced representations $\rho^{(t)}$---and the offending parameter
distributions are exactly those whose observable pushforward has
unit-modulus mean on such a character.  By
\Cref{thm:observable-kawada-ito}, this is not merely a counterexample to
Lemma~6 of \cite{ragone_lie_2024}: it identifies the \emph{only} mechanism by
which an observably adapted ensemble can fail to converge to the design
hierarchy.  Invisible sectors (characters of $G$ nontrivial on the phase
kernel $Z$) are, by contrast, irrelevant to design convergence---they
cannot obstruct it, but their mixing can also never be certified by
moments.
\end{remark}

\section{Implications and Interpretation of Observable Kawada-It\^o}

The preceding discussion and proof of \Cref{thm:observable-kawada-ito} shows that, under the observable adaptedness and observable strict aperiodicity hypotheses of observable Kawada--It\^o (or more strictly, by \Cref{lem:descent-ergodic}, adaptedness and strict aperiodicity), repeated convolution of the one-layer sampling measure converges to the Haar measure on the compact Lie group generated by the underlying dynamical Lie algebra. In particular, if $\mu_{1,N}$ denotes the one-layer measure for an $N$-qubit variational architecture and $G_N$ denotes the corresponding compact dynamical Lie group, then for each fixed system size $N$ one has
\[
\mu_{1,N}^{*L}\xrightarrow[L\to\infty]{}\mu_{G_N}
\]
in the weak sense with respect to the observable moment algebra, and hence the associated moment operators converge to their Haar values. Consequently, the Lie-algebraic variance formula of Ragone et al. over the dynamical Lie group is asymptotically recovered whenever the hypotheses of observable Kawada-It\^o are satisfied.

It is important, however, not to overinterpret this asymptotic statement. The convergence just described is a limit in the circuit depth $L$ for each fixed $N$; it is not a statement that a single depth bound, or even a single approximation error $\varepsilon$, may be chosen uniformly in the system size. In particular, the depth required for $\mu_{1,N}^{*L}$ to become $\varepsilon$-close to Haar on $G_N$ may itself depend strongly on $N$. This dependence is essential for reconciling the DLA-based asymptotic picture with the depth-aware Krylov--Lie framework developed in the present work.

We now make precise the dimension-theoretic observation underlying this reconciliation. Let $\Theta_{N,L}\subset \mathbb{R}^{p(N,L)}$ denote the parameter space of a depth-$L$ variational ansatz on $N$ qubits, where $p(N,L)$ is the number of independent real parameters. Let
\[
F_{N,L}:\Theta_{N,L}\longrightarrow M_{N,L}
\]
denote the smooth map sending parameters to the corresponding point in the reachable manifold $M_{N,L}$ of states or represented unitaries. Since $M_{N,L}$ is the image of a smooth map from a $p(N,L)$-dimensional manifold, one always has
\[
\dim M_{N,L}\leq p(N,L).
\]
Hence, for any architecture whose parameter count satisfies $p(N,L)=\mathrm{poly}(N,L)$, the reachable manifold has at most polynomial dimension whenever $L=\mathrm{poly}(N)$.

This elementary geometric fact has a substantial consequence for the interpretation of barren plateau results. In either of the dimension-matched or dimension-reduced Krylov-Lie regimes of \Cref{chap:krylov-lie-approximations}, one chooses a compact Krylov-Lie group $K_{N,L}$ together with a comparison map
\[
\Phi_{N,L}:M_{N,L}\longrightarrow K_{N,L}
\]
such that $\dim K_{N,L} \leq \dim M_{N,L}$ and $\Phi_{N,L}$ has full rank almost everywhere. Under these hypotheses, the pushforward of the sampling measure on $M_{N,L}$ is absolutely continuous with respect to Haar measure on $K_{N,L}$, so that the Radon-Nikodym density required in \Cref{thm:weighted_ragone_compatible_variance} exists and the weighted variance formula applies for the reduced loss. Comparison with the original loss has factorial error for Hermitian involution generators satisfying seed-level commutator matching under \Cref{thm:seed-KLG-factorial}, or for a prefix KLA under \Cref{thm:prefix-KLG-factorial}, through \Cref{prop:KLG-concentration-transport}. But in this same regime one necessarily has
\[
\dim K_{N,L} \leq \dim M_{N,L}\leq p(N,L).
\]
Therefore, if $p(N,L)$ is polynomial in $N$, then the depth-aware reference group on which our variance theory is built is likewise at most polynomial-dimensional. Thus, by \Cref{thm:weighted_ragone_compatible_variance}, the fastest decay one can obtain in a circuit whose parameter count is polynomial in $N$ solely from Lie-algebraic, Haar-type sources is in fact polynomial decay, not exponential. Of course, these sectors may very well be regions where training is slower and which are undesirable for practical implementation, but they are \textit{not} barren plateaus as traditionally defined (\Cref{def:barren-plateau}). Transferring a polynomial variance scale from the reduced loss to the original circuit requires the approximation error in \Cref{prop:KLG-concentration-transport} to be smaller than that variance scale. This comparison must include the gate-angle scale $B_C$, the attainable matching grade, and the observable-dependent constants as system size and circuit depth vary; the fixed-architecture grade estimates alone do not supply it.
 
This observation shows that one must distinguish two genuinely different regimes.

First, there is the \emph{asymptotic DLA-Haar regime}. Here, one fixes $N$, lets $L\to\infty$, and uses Kawada-It\^o to conclude convergence of the channel ensemble to the Haar measure on the full compact dynamical Lie group $G_N$. When the Lie algebra has large simple components, the corresponding Haar variance may be exponentially small in $N$, as in the Lie-algebraic barren plateau theory.

Second, there is the \emph{finite-depth, dimension-aware regime}. Here one studies a fixed depth scaling $L=L(N)$, typically polynomial in $N$, and models the actual reachable ensemble through a dimension-matched Krylov-Lie group $K_{N,L}$. In this setting, the group dimension is controlled by the parameter count, and hence cannot itself account for an exponential suppression arising solely from the size of the full DLA. Put differently, for polynomial-depth, polynomial-parameter ans\"atze, one cannot faithfully regard the circuit as already exploring an exponential-dimensional manifold simply because the DLA generated by the architecture is exponentially large.

This does \emph{not} imply that barren plateaus cannot occur at shallow depth. Rather, it clarifies that the mechanism responsible for such concentration need not be, and in general should not be interpreted as, the dimension of the full DLA alone. Indeed, shallow barren plateaus may still arise from the structure of the observable, the initial state, or the blockwise design properties of the ansatz. This is entirely consistent with the cost-function-dependent results of Cerezo et al. \cite{cerezo_cost_2021}, where global observables can exhibit exponentially small gradient variance even at shallow depth, while local costs remain trainable up to significantly greater depths. The point is that such shallow-depth concentration is driven by the observable and state sectors being probed and by local design behavior on subsystems, not by a faithful exploration of the full DLA group by a polynomial-depth ansatz.

The distinction may be summarized as follows. Convergence to Haar on the full DLA group is an $L\to\infty$ statement for each fixed $N$. By contrast, the present Krylov--Lie theory controls the geometry and measure of the ensemble actually explored at finite depth. If $L=L(N)$ grows only polynomially, then the reachable manifold remains at most polynomial-dimensional whenever the parameterization does, and any dimension-matched Krylov-Lie model must reflect this fact. Consequently, the exponentially small Haar variance associated with a large DLA should be understood as an asymptotic limiting value, not as a uniform finite-depth prediction valid throughout the polynomial-depth regime.

One should therefore resist the heuristic inference that the mere existence of an exponentially large dynamical Lie algebra forces a barren plateau for every polynomial depth. Such a conclusion would require a depth bound ensuring sufficiently rapid convergence to Haar on $G_N$ uniformly in $N$, together with compatibility between the observable-state pair and the high-dimensional sectors of the adjoint representation that actually contribute to the variance. Absent such additional hypotheses, the finite-depth landscape is more faithfully described by the non-Haar weighted variance formula of \Cref{thm:weighted_ragone_compatible_variance}, whose correction terms encode precisely which visible sectors remain unsuppressed.

This perspective also elucidates that the failure of Lemma 6 in the argument of Ragone et al. is not a minor technicality but a structural issue. Without a valid mechanism forcing uniform contraction of the nontrivial Fourier sectors, one cannot conclude that polynomial-depth circuit ensembles have already entered the full DLA-Haar regime. The Krylov-Lie framework isolates the finite-depth geometry and shows that deviations from Haar may persist in precisely those sectors that control trainability. Thus, while Haar-induced concentration on the full DLA may indeed emerge asymptotically under suitable ergodic hypotheses, the practically relevant shallow-depth regime must be analyzed through the depth-aware geometry of the reachable manifold rather than through the dimension of the ambient DLA alone.

\begin{remark}
The preceding discussion should be interpreted as a statement about the measure-faithful regime of the present theory. If one instead chooses to embed the reachable manifold into a larger canonical Lie group whose dimension exceeds that of the manifold, or even more simply when we do not have a pushforward who rank is full almost everywhere (frequently the case for pathological sampling measures), then the pushforward measure generally has singular components with respect to Haar and the Radon-Nikodym formalism of \Cref{sec:moment-formulas} no longer applies directly. Such enlarged ambient models, as well as more general ones with singular components, may still be heuristically useful, but they describe a different regime from the one treated rigorously in this thesis.
\end{remark}

\chapter{Discussion of Limitations and Possible Extensions}\label{chap:discussion}

The central contribution of this thesis has been the introduction of Krylov--Lie
structures as depth-aware approximators for the reachable manifolds of variational quantum
algorithms. By replacing the usual appeal to asymptotic Haar-randomness with a family of
finite-depth, seed-dependent Lie-theoretic models, the preceding chapters developed a
framework that is both more geometric and more sensitive to the actual circuit structure
than existing approaches based purely on dynamical Lie algebras or approximate
\(t\)-design heuristics. In particular, the thesis established that represented
Krylov--Lie groups provide canonical compact reference spaces carrying Haar measures,
derived seed-stratification and genericity results for the associated Krylov--Lie
subspaces and algebras, proved approximation results for the reachable manifold on compact
subsets (especially of the full-rank locus), obtained factorial seed-based error bounds under
commutator matching for involutive generators and under word preservation for prefix KLAs,
extended these estimates to geometric approximation of principal logarithms on seeds
under a uniform spectral-gap hypothesis, and then used these structures to obtain weighted moment
and variance formulas for non-Haar loss landscapes on the resulting compact groups. These
results suggest that finite-depth geometric structure, rather than asymptotic randomness
alone, should play a central role in any realistic theory of trainability for shallow or
intermediate-depth variational circuits.

At the same time, the present theory is still incomplete. Several of the most powerful results
obtained here are existence, regularity, or approximation statements, and they therefore
raise natural questions about constructive algorithms, sharper structural classifications,
and extensions to broader physical settings. In particular, while the Lie-theoretic
framework developed in this thesis is well adapted to noiseless unitary dynamics, it is
not yet sufficient for a full treatment of noise, dissipative evolution, or optimization
procedures that interact nontrivially with the geometry of the ansatz. The purpose of this
chapter is therefore to summarize the principal limitations of the current framework and to
highlight the most important directions in which it should be extended.

\section{The Need for Jordan Structure}

One of the clearest limitations of the present theory is that it is fundamentally
Lie-theoretic. This is natural for unitary variational circuits, since their infinitesimal
structure is governed by commutators and the resulting reachable symmetries are encoded by
Lie algebras and compact Lie groups. However, once one attempts to study open-system
effects, noise channels, or more general dissipative quantum dynamics, the Lie bracket
alone is no longer sufficient to capture the relevant operator structure.

The basic reason is that the Lindbladian generator of a quantum Markov semigroup naturally
contains both commutator and anticommutator contributions. The Hamiltonian part of the
evolution is Lie-theoretic, but the dissipative part is built from products of the form
\[
L_\alpha \rho L_\alpha^\dagger
\qquad \text{and} \qquad
\{L_\alpha^\dagger L_\alpha,\rho\},
\]
so any attempt to construct a genuinely structure-preserving finite-depth approximation to
noisy variational dynamics must accommodate not only Lie brackets but also Jordan-type
products. Put differently, while the present Krylov--Lie framework is tailored to the
antisymmetric algebraic content of unitary evolution, a theory of noisy and nonunitary ansatzes appears
to require a simultaneous treatment of the symmetric algebraic content as well.

For this reason, one of the most important next steps is the development of a
Krylov--Jordan or, more ambitiously, a unified Jordan--Lie framework. Such a theory would
ideally retain the depth-awareness and seed-flexibility of the present construction while
also encoding the symmetric products needed to analyze Lindbladian dynamics, effective
noise models, and open-system trainability.  If one hopes to understand whether the geometric compression ideas developed
here remain useful in realistic noisy devices, then one must identify the correct
finite-dimensional algebraic object that simultaneously sees coherent control and
dissipative drift. Additionally, noisy and nonunitary dynamics may ultimately prove vital to demonstrating quantum advantage by granting faster algorithms in a manner akin to randomized numerical linear algebra, or to how stochastic gradient descent often improves convergence times by effectively incorporating Hessian information without direct computations \cite{Xie2020DiffusionArxiv}.

A related issue is that many observables relevant to trainability are not controlled solely
by the Lie structure of the ansatz, even in the noiseless setting. The present moment and
variance formulas show that compact represented Krylov--Lie groups carry the right Haar
background for analyzing non-Haar sampling laws, but they do not yet furnish an exact
finite-depth analogue of the more fully state-space-sensitive Jordan-algebraic machinery
that appears in other parts of the landscape literature. In this sense, the present
framework should be viewed as a robust geometric and measure-theoretic foundation, but not
yet as the final algebraic language for all aspects of trainability.

\section{Open Problems}

A first major open problem concerns the dimension theory of Krylov--Lie subspaces and
Krylov--Lie algebras. \Cref{chap:krylov-lie-algebras} showed that the relevant seed loci admit determinantal
stratifications and constructible decompositions, and that maximal dimension is generic on
appropriate supports. However, these results do not yet answer the most practically
important dimension questions. How sparse are the attainable dimensions in a fixed ansatz?
When dimension matching is possible, is it rare or abundant? Do the attainable dimensions
typically fill long integer intervals, or do they occur in thin arithmetic or
representation-theoretic patterns? More generally, to what extent can one predict in
advance whether a given family of generators and a chosen commutator depth can realize a
dimension-matched Krylov-Lie structure at all?

These questions are not immaterial for VQA landscape theory. The feasibility of the entire approximation
program depends greatly on whether dimension-matched Krylov-Lie groups are abundant
enough to be found in practice. The present results show that, once a desired dimension is
realized, one often obtains useful genericity statements on the corresponding seed support.
What remains unclear is how difficult it is to reach that support in the first place. A
full theory of attainable dimensions, including lower bounds, sparsity phenomena, and
possible interval results, would therefore substantially strengthen the practical content
of the approximation theorem.

A second open problem is whether the canonical comparison map used in this thesis is in any
meaningful sense optimal. The canonical map is attractive because it is natural, explicit, compatible with the compressed circuit data, and it satisfies the analytic Jacobian
hypothesis needed for the approximation theorems proved earlier. Nevertheless, there is no
obvious inherent reason that it must necessarily be the best possible map from the reachable manifold into a
represented Krylov-Lie group. It may be that other choices of comparison map preserve
more of the geometry, yield larger full-rank loci, improve injectivity properties, or
reduce the ambient approximation error. Accordingly, one of the central conceptual
questions for future work is whether one can find a better representation of the reachable
manifold than the canonical map, or instead prove a genuine optimality theorem for the
canonical construction under natural geometric or operator-theoretic criteria.

A third direction, also of a more geometric flavor, concerns the possibility of reconstructing the original reachable manifold and its statistics from Krylov–Lie structures themselves. In \Cref{chap:krylov-lie-approximations} we observed that, on suitable compact subsets of the full-rank locus, comparison maps into represented Krylov–Lie groups behave like submersions or local diffeomorphisms, and that the pushforward measure admits a well-controlled Radon–Nikodym density. This suggests that one might be able to construct atlases of Krylov–Lie charts whose transition data encode the geometry and measure of the reachable manifold in a purely Krylov-theoretic language, or more broadly to view reachable manifolds as admitting a Lietic structure in the informal sense of \Cref{rem:lietic-spaces}. Determining when such atlases exist, and to what extent they allow faithful reconstruction of manifold-level structure and statistics from Krylov models alone, is a potentially far-reaching open problem.

A fourth direction concerns the role of nonuniform sampling as well as the possibility of
reweighting the loss landscape. The weighted expectation and variance formulas in
\Cref{chap:applications} show that deviations from Haar are encoded by the centered Radon-Nikodym
density and its visible representation-theoretic sectors. This strongly suggests that the
sampling distribution itself can have a substantial effect on the observed trainability of
a variational algorithm. In particular, it raises the possibility that one could
deliberately reweight the effective landscape so as to downweight regions exhibiting barren
plateau behavior and upweight regions with more favorable variance structure. At present,
however, it is not clear whether such a strategy can be made algorithmically meaningful
without already knowing detailed information about the landscape over which one is trying to optimize.
Determining whether there exist principled, a priori reweighting schemes capable of
mitigating barren plateaus is therefore an important open problem. The theory of importance sampling will likely prove relevant here.

A fifth issue is interpretive rather than purely technical: the geometric meaning of
submersions in the Krylov-Lie approximation theory remains only partially understood from
the perspective of the original algorithm. In the present framework, a lower-dimensional
measure-faithful Krylov-Lie approximation naturally leads to a comparison map that is
generically a submersion rather than a local diffeomorphism. This suggests that there are
fiber directions in the reachable manifold that are invisible to the effective compact Lie
geometry. One possible interpretation is that these directions correspond to geometric or
algorithmic redundancy, in the sense that distinct parameter-space motions may induce
equivalent or nearly equivalent behavior at the level captured by the approximating group.
If that interpretation is correct, then submersions may encode a useful notion of
compressibility or overparameterization intrinsic to the ansatz. Further work is
needed to determine whether these hidden fiber directions correlate with trainability,
symmetry, gauge freedom, or optimization degeneracy in the original circuit family.

A sixth open problem concerns the computational meaning of ``good" Krylov-Lie approximations. Since the present framework seeks to approximate reachable manifolds by
lower-dimensional or otherwise structured compact Lie models, it is natural to ask whether
the quality of such approximations correlates with classical simulability. Roughly
speaking, if a variational circuit admits a highly accurate approximation by a small or
highly structured represented Krylov-Lie group, then one might expect the effective
degrees of freedom of the circuit to be limited in a way that makes certain observables,
moments, or optimization trajectories easier to simulate classically. Conversely, the
failure of low-dimensional Krylov--Lie approximations might signal genuinely hard
many-body behavior. At present this relationship is speculative, but it is sufficiently
natural that it deserves systematic investigation, especially in families of ans\"atze where
both classical simulation heuristics and Krylov--Lie constructions can be analyzed
explicitly.

Finally, a more practical but equally important direction is the development of algorithms
for finding good seeds and certifying good approximators. The theoretical results of this
thesis show that genericity statements are available once one works on the correct support,
and that optimal seeds (in the sense of the supremum of the deviation of representatives of a unitary in the reachable manifold in the Hilbert-Schmidt norm) exist after restricting to compact unit-norm seed families with certified closedness and comparison-map rank (\Cref{cor:KLA-generic-rank-drop}, \Cref{thm:main_klg_approximation}). However,
an existence theorem is not yet a practical search procedure. Future work should therefore
address the algorithmic side of the theory: how to efficiently construct seeds that pass
finite full-rank tests, how to certify commutator matching and determine its attainable grade,
how to estimate uniform approximation error without exhaustive computation,
how to detect promising dimensions and depths, and how to couple these tasks to ansatz
design (in particular, assessing how optimal the initial state is as a seed).
For involutive generators, the necessary prefix inclusion of
\Cref{prop:involution-seed-word-matching} supplies one obstruction to
matching. Its relation to the sufficient condition in
\Cref{rem:standard-KLA-matching} may help make these tests more effective.

The present theory provides a strong first step toward a depth-aware geometric
theory of variational quantum algorithms, but it also leaves a large and mathematically
interesting research program behind. Extending the framework from Lie structure to
Jordan-Lie structure, clarifying the attainable dimensions of Krylov-Lie objects,
characterizing the optimality of comparison maps, understanding reweighting strategies for
trainability, interpreting the geometric meaning of submersions, and relating approximation
quality to classical simulability are all natural next questions. If these directions can
be developed successfully, then Krylov-based approaches may eventually furnish a broader
replacement for purely Haar-centric theories of variational quantum dynamics.

\begingroup
\setstretch{1}
\bibliography{references}

@article{Grad_2024,
   title={Fundamentals of {Lie} categories},
   volume={19},
   ISSN={1661-6960},
   url={http://dx.doi.org/10.4171/JNCG/563},
   DOI={10.4171/jncg/563},
   number={1},
   journal={Journal of Noncommutative Geometry},
   publisher={European Mathematical Society - EMS - Publishing House GmbH},
   author={Grad, {\v{Z}}an},
   year={2024},
   month=Feb, pages={211–248} }

@article{Granirer1965,
  author  = {Granirer, Edmond E.},
  title   = {On the invariant mean on topological semigroups and on topological groups},
  journal = {Pacific Journal of Mathematics},
  volume  = {15},
  number  = {1},
  year    = {1965},
  pages   = {107--140}
}

@article{Argabright1966,
  author  = {Argabright, L. N.},
  title   = {A note on invariant integrals on locally compact semigroups},
  journal = {Proceedings of the American Mathematical Society},
  volume  = {17},
  number  = {2},
  year    = {1966},
  pages   = {377--382}
}

@article{Lawson1994,
author = {Lawson, Jimmie D.},
journal = {Journal für die reine und angewandte Mathematik},
pages = {191-219},
title = {Polar and {Ol'shanskii} decompositions.},
url = {http://eudml.org/doc/153602},
volume = {448},
year = {1994},
}

@article{BCHDynkin,
  author  = {Dynkin, E. B.},
  title   = {Calculation of the coefficients in the {Campbell--Hausdorff} formula},
  journal = {Doklady Akademii Nauk SSSR (N.S.)},
  volume  = {57},
  pages   = {323--326},
  year    = {1947}
}

@article{KawadaIto1940,
  author  = {Kawada, Yukiyosi and It{\^o}, Kiyosi},
  title   = {On the probability distribution on a compact group. {I}},
  journal = {Proceedings of the Physico-Mathematical Society of Japan.
             3rd Series},
  volume  = {22},
  year    = {1940},
  pages   = {977--998}
}

@article{Stromberg1960,
  author  = {Stromberg, Karl},
  title   = {Probabilities on a compact group},
  journal = {Transactions of the American Mathematical Society},
  volume  = {94},
  year    = {1960},
  pages   = {295--309}
}

@book{Applebaum2014,
  author    = {Applebaum, David},
  title     = {{Probability on Compact Lie Groups}},
  series    = {Probability Theory and Stochastic Modelling},
  volume    = {70},
  publisher = {Springer International Publishing},
  address   = {Cham},
  year      = {2014}
}

@book{BtD1985,
  author    = {Br{\"o}cker, Theodor and Tom Dieck, Tammo},
  title     = {{Representations of Compact Lie Groups}},
  series    = {Graduate Texts in Mathematics},
  volume    = {98},
  publisher = {Springer},
  address   = {New York},
  year      = {1985}
}

@misc{EHMQAOAPaper,
      title={Dynamical {Lie} Algebras Cannot Describe Shallow {QAOA}: Cragged Terrains, Barren Plateaus, and Empirical Hardness Models}, 
      author={Harrison Copp and Charlton Li and Anžej Margeta-Cacace and Amy Qiao},
      year={2026},
      eprint={2608.04252},
      archivePrefix={arXiv},
      primaryClass={quant-ph},
      url={https://arxiv.org/abs/2608.04252}, 
}

@book{DummitFoote2004,
  author    = {David S. Dummit and Richard M. Foote},
  title     = {{Abstract Algebra}},
  edition   = {3rd},
  publisher = {John Wiley \& Sons},
  address   = {Hoboken, NJ},
  year      = {2004},
  isbn      = {978-0-471-43334-7}
}

@misc{allcock2026dynamicalliealgebrasquantum,
      title={On the dynamical {Lie} algebras of quantum approximate optimization algorithms}, 
      author={Jonathan Allcock and Miklos Santha and Pei Yuan and Shengyu Zhang},
      year={2026},
      eprint={2407.12587},
      archivePrefix={arXiv},
      primaryClass={quant-ph},
      doi={https://doi.org/10.22331/q-2026-05-29-2119},
      url={https://arxiv.org/abs/2407.12587}, 
}

@misc{Xie2020DiffusionArxiv,
      title={A Diffusion Theory For Deep Learning Dynamics: Stochastic Gradient Descent Exponentially Favors Flat Minima}, 
      author={Zeke Xie and Issei Sato and Masashi Sugiyama},
      year={2021},
      eprint={2002.03495},
      archivePrefix={arXiv},
      primaryClass={cs.LG},
      url={https://arxiv.org/abs/2002.03495}, 
}

@article{NoteKawadaIto2022,
title = {A note on the {K}awada–{I}tô theorem},
journal = {Statistics \& Probability Letters},
volume = {181},
pages = {109261},
year = {2022},
issn = {0167-7152},
doi = {https://doi.org/10.1016/j.spl.2021.109261},
url = {https://www.sciencedirect.com/science/article/pii/S0167715221002236},
author = {Heybetkulu Mustafayev}
}

@misc{NaorConcentration,
  author       = {Assaf Naor},
  title        = {{Concentration of Measure}},
  note         = {Lecture notes, Courant Institute of Mathematical Sciences, Fall 2008. Scribed by Lingjiong Zhu},
  year         = {2008},
  url = {https://web.math.princeton.edu/~naor/homepage%20files/Concentration%20of%20Measure.pdf} 
}

@incollection{Petrogradsky2000,
  author    = {V. M. Petrogradsky},
  title     = {On {Witt's} Formula and Invariants for Free {Lie} Superalgebras},
  booktitle = {Formal Power Series and Algebraic Combinatorics},
  publisher = {Springer},
  year      = {2000},
  pages     = {543--551},
  doi       = {10.1007/978-3-662-04166-6_52}
}

@book{Harris1992,
  author    = {Joe Harris},
  title     = {{Algebraic Geometry: A First Course}},
  series    = {Graduate Texts in Mathematics},
  volume    = {133},
  publisher = {Springer},
  address   = {New York},
  year      = {1992},
  isbn      = {978-0-387-97716-4}
}

@book{BeltramettiAlgebraicGeometry2009,
  author    = {Mauro C. Beltrametti and Ettore Carletti and Dionisio Gallarati and Giacomo Monti Bragadin},
  title     = {Lectures on Curves, Surfaces and Projective Varieties},
  subtitle  = {A Classical View of Algebraic Geometry},
  series    = {EMS Textbooks in Mathematics},
  publisher = {European Mathematical Society},
  address   = {Z{\"u}rich},
  year      = {2009},
  isbn      = {978-3-03719-064-7}
}

@misc{sottile2016realalgebraicgeometry,
      title={Real Algebraic Geometry for Geometric Constraints}, 
      author={Frank Sottile},
      year={2016},
      eprint={1606.03127},
      archivePrefix={arXiv},
      primaryClass={math.AG},
      url={https://arxiv.org/abs/1606.03127}, 
}

@book{bochnak1998realalgebraicgeometry,
  author    = {Jacek Bochnak and Michel Coste and Marie-Fran{\c{c}}oise Roy},
  title     = {Real Algebraic Geometry},
  publisher = {Springer},
  address   = {Berlin, Heidelberg}, 
  year      = {1998}
}

@book{HallLieGroups2015,
  author    = {Brian C. Hall},
  title     = {Lie Groups, Lie Algebras, and Representations: An Elementary Introduction},
  edition   = {2nd},
  series    = {Graduate Texts in Mathematics},
  publisher = {Springer},
  address   = {Cham},
  year      = {2015},
  doi       = {10.1007/978-3-319-13467-3},
  isbn      = {978-3-319-13466-6}
}

@book{antoulas2005approximation,
  author    = {Athanasios C. Antoulas},
  title     = {Approximation of Large-Scale Dynamical Systems},
  publisher = {SIAM},
  address   = {Philadelphia},
  year      = {2005}
}

@book{saad2003iterative,
  author    = {Yousef Saad},
  title     = {Iterative Methods for Sparse Linear Systems},
  edition   = {2nd},
  publisher = {SIAM},
  address   = {Philadelphia},
  year      = {2003}
}

@article{gutknecht2009blockgrade,
  author  = {Martin H. Gutknecht and Thomas Schmelzer},
  title   = {The Block Grade of a Block {Krylov} Space},
  journal = {Linear Algebra and its Applications},
  volume  = {430},
  number  = {1},
  pages   = {174--185},
  year    = {2009}
}

@book{chavel2006riemannian,
  author    = {Isaac Chavel},
  title     = {Riemannian Geometry: A Modern Introduction},
  edition   = {2},
  publisher = {Cambridge University Press},
  year      = {2006},
  isbn      = {9780521853682}
}

@book{evans_gariepy_2015,
  author    = {Lawrence C. Evans and Ronald F. Gariepy},
  title     = {Measure Theory and Fine Properties of Functions},
  edition   = {Revised},
  publisher = {Chapman and Hall/CRC},
  year      = {2015},
  doi       = {10.1201/b18333}
}

@book{Lee2013SmoothManifolds,
  author    = {Lee, John M.},
  title     = {Introduction to Smooth Manifolds},
  edition   = {2nd},
  series    = {Graduate Texts in Mathematics},
  volume    = {218},
  publisher = {Springer},
  address   = {New York},
  year      = {2013},
  isbn      = {9781441999818}
}

@book{Folland1995Harmonic,
  author    = {Folland, Gerald B.},
  title     = {A Course in Abstract Harmonic Analysis},
  publisher = {CRC Press},
  address   = {Boca Raton, FL},
  year      = {1995},
  series    = {Studies in Advanced Mathematics},
  isbn      = {9780849384905}
}

@book{Folland1999RealAnalysis,
  author    = {Folland, Gerald B.},
  title     = {Real Analysis: Modern Techniques and Their Applications},
  edition   = {2nd},
  series    = {Pure and Applied Mathematics},
  publisher = {John Wiley \& Sons},
  address   = {New York},
  year      = {1999},
  isbn      = {9780471317166}
}

@phdthesis{OMeara2014LieTheoretic,
  author       = {O'Meara, Corey Patrick},
  title        = {A Lie Theoretic Framework for Controlling Open Quantum Systems},
  school       = {Technische Universit{\"a}t M{\"u}nchen},
  year         = {2014},
  archivePrefix= {arXiv},
  eprint       = {2510.04719},
  primaryClass = {quant-ph},
  doi          = {10.48550/arXiv.2510.04719},
}

@article{Park2024hamiltonian,
  title     = {Hamiltonian variational ansatz without barren plateaus},
  author    = {Park, Chae-Yeun and Killoran, Nathan},
  journal   = {Quantum},
  volume    = {8},
  pages     = {1239},
  year      = {2024},
  month     = feb,
  doi       = {10.22331/q-2024-02-01-1239},
  url       = {https://doi.org/10.22331/q-2024-02-01-1239},
  issn      = {2521-327X},
  publisher = {Verein zur F{\"o}rderung des Open Access Publizierens in den Quantenwissenschaften}
}

@article{Zhang2024absence,
   title={Absence of Barren Plateaus in Finite Local-Depth Circuits with Long-Range Entanglement},
   volume={132},
   ISSN={1079-7114},
   url={http://dx.doi.org/10.1103/PhysRevLett.132.150603},
   DOI={10.1103/physrevlett.132.150603},
   number={15},
   journal={Physical Review Letters},
   publisher={American Physical Society (APS)},
   author={Zhang, Hao-Kai and Liu, Shuo and Zhang, Shi-Xin},
   year={2024},
   month=Apr 
   }

@article{Holmes2022,
   title={Connecting Ansatz Expressibility to Gradient Magnitudes and Barren Plateaus},
   volume={3},
   ISSN={2691-3399},
   url={http://dx.doi.org/10.1103/PRXQuantum.3.010313},
   DOI={10.1103/prxquantum.3.010313},
   number={1},
   journal={PRX Quantum},
   publisher={American Physical Society (APS)},
   author={Holmes, Zoë and Sharma, Kunal and Cerezo, M. and Coles, Patrick J.},
   year={2022},
   month=Jan }

@article{HafterkampRandomCircuits,
  title     = {Random quantum circuits are approximate unitary
               $t$-designs in depth $O\!\left(nt^{5+o(1)}\right)$},
  author    = {Haferkamp, Jonas},
  journal   = {Quantum},
  volume    = {6},
  pages     = {795},
  year      = {2022},
  month     = {September},
  publisher = {Verein zur F\"{o}rderung des Open Access Publizierens
               in den Quantenwissenschaften},
  doi       = {10.22331/q-2022-09-08-795}
}

@misc{Mityagin2015,
  author        = {Mityagin, Boris},
  title         = {The Zero Set of a Real Analytic Function},
  year          = {2015},
  eprint        = {1512.07276},
  archivePrefix = {arXiv},
  primaryClass  = {math.CA},
  url           = {https://arxiv.org/abs/1512.07276}
}

@book{kirillov_lie_groups,
  author    = {Kirillov, Jr., Alexander},
  title     = {Introduction to Lie Groups and Lie Algebras},
  publisher = {Cambridge University Press},
  year      = {2008}
}

@article{cerezo_variational_2021,
   title={Variational quantum algorithms},
   volume={3},
   ISSN={2522-5820},
   url={http://dx.doi.org/10.1038/s42254-021-00348-9},
   DOI={10.1038/s42254-021-00348-9},
   number={9},
   journal={Nature Reviews Physics},
   publisher={Springer Science and Business Media LLC},
   author={Cerezo, M. and Arrasmith, Andrew and Babbush, Ryan and Benjamin, Simon C. and Endo, Suguru and Fujii, Keisuke and McClean, Jarrod R. and Mitarai, Kosuke and Yuan, Xiao and Cincio, Lukasz and Coles, Patrick J.},
   year={2021},
   month=Aug, pages={625–644} }

@article{fontana_adjoint_2024,
   title={Characterizing barren plateaus in quantum ansätze with the adjoint representation},
   volume={15},
   ISSN={2041-1723},
   url={http://dx.doi.org/10.1038/s41467-024-49910-w},
   DOI={10.1038/s41467-024-49910-w},
   number={1},
   journal={Nature Communications},
   publisher={Springer Science and Business Media LLC},
   author={Fontana, Enrico and Herman, Dylan and Chakrabarti, Shouvanik and Kumar, Niraj and Yalovetzky, Romina and Heredge, Jamie and Sureshbabu, Shree Hari and Pistoia, Marco},
   year={2024},
   month=Aug }

@misc{brady_iterative_2023,
      title={Iterative Quantum Algorithms for Maximum Independent Set: A Tale of Low-Depth Quantum Algorithms}, 
      author={Lucas T. Brady and Stuart Hadfield},
      year={2023},
      eprint={2309.13110},
      archivePrefix={arXiv},
      primaryClass={quant-ph},
      url={https://arxiv.org/abs/2309.13110}, 
}

@article{ragone_lie_2024,
  author  = {Ragone, Michael and Bakalov, Bojko N. and Sauvage, Fr\'{e}d\'{e}ric
             and Kemper, Alexander F. and Ortiz Marrero, Carlos
             and Larocca, Mart\'{\i}n and Cerezo, M.},
  title   = {A {Lie} algebraic theory of barren plateaus for deep parameterized quantum circuits},
  journal = {Nature Communications},
  volume  = {15},
  pages   = {7172},
  year    = {2024},
  doi     = {10.1038/s41467-024-49909-3},
}

@misc{anschuetz_2025,
      title={A Unified Theory of Quantum Neural Network Loss Landscapes}, 
      author={Eric R. Anschuetz},
      year={2025},
      eprint={2408.11901},
      archivePrefix={arXiv},
      primaryClass={quant-ph},
      url={https://arxiv.org/abs/2408.11901}, 
}

@article{larocca_barren_2025,
   title={Barren plateaus in variational quantum computing},
   volume={7},
   ISSN={2522-5820},
   url={http://dx.doi.org/10.1038/s42254-025-00813-9},
   DOI={10.1038/s42254-025-00813-9},
   number={4},
   journal={Nature Reviews Physics},
   publisher={Springer Science and Business Media LLC},
   author={Larocca, Martín and Thanasilp, Supanut and Wang, Samson and Sharma, Kunal and Biamonte, Jacob and Coles, Patrick J. and Cincio, Lukasz and McClean, Jarrod R. and Holmes, Zoë and Cerezo, M.},
   year={2025},
   month=Mar, pages={174–189} }

@article{cerezo_cost_2021,
	title = {Cost function dependent barren plateaus in shallow parametrized quantum circuits},
	volume = {12},
	copyright = {2021 The Author(s)},
	issn = {2041-1723},
	url = {https://www.nature.com/articles/s41467-021-21728-w},
	doi = {10.1038/s41467-021-21728-w},
	language = {en},
	number = {1},
	urldate = {2025-08-16},
	journal = {Nature Communications},
	author = {Cerezo, M. and Sone, Akira and Volkoff, Tyler and Cincio, Lukasz and Coles, Patrick J.},
	month = mar,
	year = {2021},
	pages = {1791},
}

@book{Higham2008Functions,
  author    = {Higham, Nicholas J.},
  title     = {Functions of Matrices: Theory and Computation},
  publisher = {Society for Industrial and Applied Mathematics},
  address   = {Philadelphia},
  year      = {2008},
  doi       = {10.1137/1.9780898717778}
}

@article{Saad1992Exponential,
  author  = {Saad, Yousef},
  title   = {Analysis of Some {Krylov} Subspace Approximations to the Matrix Exponential Operator},
  journal = {SIAM Journal on Numerical Analysis},
  volume  = {29},
  number  = {1},
  pages   = {209--228},
  year    = {1992},
  doi     = {10.1137/0729014}
}

@misc{BiagiBonfiglioliMatone2018,
  author        = {Biagi, Stefano and Bonfiglioli, Andrea and Matone, Marco},
  title         = {On the {Baker--Campbell--Hausdorff} Theorem: non-convergence and prolongation issues},
  year          = {2018},
  eprint        = {1805.10089},
  note          = {arXiv:1805.10089},
  archivePrefix = {arXiv},
  primaryClass  = {math-ph},
  url           = {https://arxiv.org/abs/1805.10089}
}

@article{Saff2010PotentialTheory,
  author  = {Saff, Edward B.},
  title   = {Logarithmic Potential Theory with Applications to Approximation Theory},
  journal = {Surveys in Approximation Theory},
  volume  = {5},
  pages   = {165--200},
  year    = {2010},
  url     = {https://surveys-in-approximation-theory.com/papers/14/index.html},
  eprint  = {1010.3760},
  archivePrefix = {arXiv}
}
\endgroup

\appendix
\crefalias{chapter}{appendix}

\chapter{Central-Character Obstructions to Unit Eigenvalues in Cycle-Graph QAOA}\label{appendix:cycle-graph-qaoa}

In this section we show that cycle-graph QAOA on MaxCut presents a natural obstruction to the
usual strict-overlap argument used to deduce a bound of the form $|\lambda|<1$ from mere
non-invariance. The key point is that the dynamical Lie algebra for $C_n$ has a nontrivial
center, and the induced conjugation action of the corresponding central torus on operator
space contains genuine one-dimensional character sectors.

Throughout, let
\[
\mathfrak g:=\mathfrak g_{C_n}\subset \mathfrak{su}(2^n),
\qquad
X:=i\sum_{j=0}^{n-1}X_j,
\qquad
ZZ:=i\sum_{j=0}^{n-1}Z_jZ_{j+1},
\]
where indices are understood modulo $n$. By the cycle-graph QAOA Lie-algebra
classification, one has the reductive decomposition
\[
\mathfrak g
=
\mathfrak c
\oplus
\mathfrak g_1
\oplus\cdots\oplus
\mathfrak g_{n-1},
\qquad
\mathfrak g_k\cong \mathfrak{su}(2),
\qquad
\dim \mathfrak c =2,
\]
and the center is explicitly
\[
\mathfrak c=\operatorname{span}_{\mathbb R}\{c_1,c_2\}.
\]
For odd $n$,
\[
c_1=-X+\sum_{t=1}^{(n-1)/2}\bigl(ZX^{2t-1}Z+YX^{2t-1}Y\bigr),
\qquad
c_2=X^{n-1}+\sum_{t=0}^{(n-3)/2}\bigl(ZX^{2t}Z+YX^{2t}Y\bigr),
\]
whereas for even $n$,
\[
c_1=X^{n-1}-X+\sum_{t=1}^{(n-2)/2}\bigl(ZX^{2t-1}Z+YX^{2t-1}Y\bigr),
\qquad
c_2=\sum_{t=0}^{(n-2)/2}\bigl(ZX^{2t}Z+YX^{2t}Y\bigr).
\]
See \cite{allcock2026dynamicalliealgebrasquantum}.

\begin{lemma}[Joint weight decomposition of the central action]
\label{lem:joint-weight-decomposition}
There exists an orthogonal decomposition
\[
\mathcal H
=
\bigoplus_{\lambda\in\Lambda}\mathcal H_\lambda
\]
into simultaneous eigenspaces of $c_1$ and $c_2$ such that, for every
$\lambda=(\lambda_1,\lambda_2)\in\Lambda$,
\[
c_r|_{\mathcal H_\lambda}=i\lambda_r\,I_{\mathcal H_\lambda},
\qquad r=1,2.
\]
\end{lemma}

\begin{proof}
Since $c_1,c_2\in\mathfrak{su}(2^n)$ are skew-Hermitian and commute, the Hermitian
operators $-ic_1$ and $-ic_2$ commute and are therefore simultaneously unitarily
diagonalizable. The joint eigenspace decomposition of $-ic_1$ and $-ic_2$ yields the stated
decomposition for $c_1,c_2$.
\end{proof}

\begin{remark}
Since $\mathfrak c\neq 0$, the set of joint weights $\Lambda$ is nonempty. Moreover, because
$c_1$ and $c_2$ are nonzero traceless operators, at least one of them is non-scalar, so
there exist distinct joint weights whenever the center acts nontrivially on $\mathcal H$.
\end{remark}

Let
\[
T:=\exp(\mathfrak c)
=
\{\exp(\theta_1 c_1+\theta_2 c_2) \, \vert \,(\theta_1,\theta_2)\in\mathbb R^2\}.
\]
This is a compact Abelian subgroup of the represented dynamical Lie group generated by the
center.

\begin{lemma}[Characters from weight differences]
\label{lem:characters-from-weight-differences}
Fix joint weights $\lambda,\mu\in\Lambda$, and let
\[
A_{\lambda\mu}\in \Hom(\mathcal H_\mu,\mathcal H_\lambda)
\subset \End(\mathcal H).
\]
Then for every $(\theta_1,\theta_2)\in\mathbb R^2$,
\[
\exp(\theta_1 c_1+\theta_2 c_2)\,
A_{\lambda\mu}\,
\exp(-\theta_1 c_1-\theta_2 c_2)
=
e^{\,i[(\lambda_1-\mu_1)\theta_1+(\lambda_2-\mu_2)\theta_2]}\,
A_{\lambda\mu}.
\]
Equivalently, the conjugation representation of $T$ on
$\Hom(\mathcal H_\mu,\mathcal H_\lambda)$ is the one-dimensional character
\[
\chi_{\lambda-\mu}(\theta_1,\theta_2)
:=
\exp\!\bigl(i[(\lambda_1-\mu_1)\theta_1+(\lambda_2-\mu_2)\theta_2]\bigr).
\]
\end{lemma}

\begin{proof}
Let $v\in\mathcal H_\mu$. Since $A_{\lambda\mu}v\in\mathcal H_\lambda$, Lemma
\ref{lem:joint-weight-decomposition} gives
\[
(\theta_1 c_1+\theta_2 c_2)A_{\lambda\mu}v
=
i(\lambda_1\theta_1+\lambda_2\theta_2)A_{\lambda\mu}v,
\]
whereas
\[
A_{\lambda\mu}(\theta_1 c_1+\theta_2 c_2)v
=
i(\mu_1\theta_1+\mu_2\theta_2)A_{\lambda\mu}v.
\]
Hence
\[
[(\theta_1 c_1+\theta_2 c_2),A_{\lambda\mu}]
=
i[(\lambda_1-\mu_1)\theta_1+(\lambda_2-\mu_2)\theta_2]A_{\lambda\mu}.
\]
Since the commutator acts by a scalar on $A_{\lambda\mu}$, exponentiating yields the claim.
\end{proof}

\begin{corollary}[Failure of strict overlap]
\label{cor:strict-overlap-fails}
Assume $\lambda\neq\mu$ and $A_{\lambda\mu}\neq 0$. Then $A_{\lambda\mu}$ is not
$T$-invariant, but for every $(\theta_1,\theta_2)\in\mathbb R^2$ one has
\[
\left|
\left\langle
A_{\lambda\mu},
\exp(\theta_1 c_1+\theta_2 c_2)\,
A_{\lambda\mu}\,
\exp(-\theta_1 c_1-\theta_2 c_2)
\right\rangle_{HS}
\right|
=
\|A_{\lambda\mu}\|_{HS}^2.
\]
In particular, non-invariance alone does not imply the existence of a group element whose
Hilbert--Schmidt overlap with $A_{\lambda\mu}$ is strictly smaller in modulus than
$\|A_{\lambda\mu}\|_{HS}^2$.
\end{corollary}

\begin{proof}
By Lemma \ref{lem:characters-from-weight-differences},
\[
\exp(\theta_1 c_1+\theta_2 c_2)\,
A_{\lambda\mu}\,
\exp(-\theta_1 c_1-\theta_2 c_2)
=
\chi_{\lambda-\mu}(\theta_1,\theta_2)A_{\lambda\mu},
\]
with $|\chi_{\lambda-\mu}(\theta_1,\theta_2)|=1$. Therefore
\[
\left|
\langle A_{\lambda\mu},\chi_{\lambda-\mu}(\theta_1,\theta_2)A_{\lambda\mu}\rangle_{HS}
\right|
=
|\chi_{\lambda-\mu}(\theta_1,\theta_2)|\,\|A_{\lambda\mu}\|_{HS}^2
=
\|A_{\lambda\mu}\|_{HS}^2.
\]
Since $\lambda\neq\mu$, the character is nontrivial, so $A_{\lambda\mu}$ is not
$T$-invariant.
\end{proof}

The preceding corollary shows that the usual strict overlap argument cannot be justified
solely from the statement that an operator lies outside the invariant subspace. In the
presence of a nontrivial central torus, one must additionally control the Fourier
coefficients of the sampling measure against the characters arising from the weight
differences.

We now identify the central components of the two standard cycle-QAOA generators.

\begin{lemma}[Central projections of the QAOA generators]
\label{lem:central-projections}
With respect to the orthogonal decomposition
\[
\mathfrak g
=
\mathfrak c\oplus[\mathfrak g,\mathfrak g],
\]
the central components of $X$ and $ZZ$ are
\[
X_{\mathfrak c}=-\frac{1}{n}c_1,
\qquad
(ZZ)_{\mathfrak c}=\frac{1}{n}c_2.
\]
\end{lemma}

\begin{proof}
The cycle-graph orbit-sums form an orthogonal basis of $\mathfrak g$ \cite{allcock2026dynamicalliealgebrasquantum}, and the central basis
vectors $c_1,c_2$ are sums of exactly $n$ such orbit-sums with coefficients $\pm 1$;
hence
\[
\|c_1\|_{HS}^2=\|c_2\|_{HS}^2=n^2 2^n.
\]
Moreover, by inspection of the explicit formulas for $c_1,c_2$: the coefficient of $X$ in $c_1$ is $-1$, and $X$ does not appear in $c_2$, while the coefficient of $ZZ$ in $c_2$ is $+1$, and $ZZ$ does not appear in $c_1$. Therefore
\[
\langle X,c_1\rangle_{HS}=-n2^n,
\qquad
\langle X,c_2\rangle_{HS}=0,
\]
and
\[
\langle ZZ,c_1\rangle_{HS}=0,
\qquad
\langle ZZ,c_2\rangle_{HS}=n2^n.
\]
Since $c_1,c_2$ are orthogonal, the orthogonal projection formulas give
\[
X_{\mathfrak c}
=
\frac{\langle X,c_1\rangle_{HS}}{\|c_1\|_{HS}^2}c_1
=
-\frac{1}{n}c_1,
\qquad
(ZZ)_{\mathfrak c}
=
\frac{\langle ZZ,c_2\rangle_{HS}}{\|c_2\|_{HS}^2}c_2
=
\frac{1}{n}c_2.
\]
\end{proof}

\begin{proposition}[Central character of a single QAOA layer]
\label{prop:central-character-single-layer}
Let
\[
U(\beta,\gamma):=e^{\beta X}e^{\gamma ZZ}.
\]
Then its central factor is
\[
U_{\mathfrak c}(\beta,\gamma)
=
\exp\!\left(-\frac{\beta}{n}c_1+\frac{\gamma}{n}c_2\right).
\]
Consequently, for every $A_{\lambda\mu}\in\Hom(\mathcal H_\mu,\mathcal H_\lambda)$,
\[
U_{\mathfrak c}(\beta,\gamma)\,
A_{\lambda\mu}\,
U_{\mathfrak c}(\beta,\gamma)^\dagger
=
e^{\,i\phi_{\lambda,\mu}(\beta,\gamma)}A_{\lambda\mu},
\]
where
\[
\phi_{\lambda,\mu}(\beta,\gamma)
=
-\frac{\lambda_1-\mu_1}{n}\beta
+
\frac{\lambda_2-\mu_2}{n}\gamma.
\]
Thus the induced character on this sector is
\[
\chi_{\lambda-\mu}(\beta,\gamma)
=
\exp\!\left(
i\left[
-\frac{\lambda_1-\mu_1}{n}\beta
+
\frac{\lambda_2-\mu_2}{n}\gamma
\right]
\right).
\]
\end{proposition}

\begin{proof}
Write
\[
X=X_{\mathfrak c}+X_{\mathrm{ss}},
\qquad
ZZ=(ZZ)_{\mathfrak c}+(ZZ)_{\mathrm{ss}},
\]
with $X_{\mathfrak c},(ZZ)_{\mathfrak c}\in\mathfrak c$ and
$X_{\mathrm{ss}},(ZZ)_{\mathrm{ss}}\in[\mathfrak g,\mathfrak g]$.
Because $\mathfrak g=\mathfrak c\oplus[\mathfrak g,\mathfrak g]$ is a direct sum of ideals,
all central elements commute with all semisimple elements. Hence
\[
e^{\beta X}=e^{\beta X_{\mathfrak c}}e^{\beta X_{\mathrm{ss}}},
\qquad
e^{\gamma ZZ}=e^{\gamma (ZZ)_{\mathfrak c}}e^{\gamma (ZZ)_{\mathrm{ss}}},
\]
and the central factor of $U(\beta,\gamma)$ is
\[
e^{\beta X_{\mathfrak c}}e^{\gamma (ZZ)_{\mathfrak c}}
=
\exp\!\left(\beta X_{\mathfrak c}+\gamma (ZZ)_{\mathfrak c}\right).
\]
Applying Lemma \ref{lem:central-projections} yields
\[
U_{\mathfrak c}(\beta,\gamma)
=
\exp\!\left(-\frac{\beta}{n}c_1+\frac{\gamma}{n}c_2\right).
\]
The formula for the induced action on $A_{\lambda\mu}$ now follows from Lemma
\ref{lem:characters-from-weight-differences}.
\end{proof}

\begin{corollary}[Sampling-measure Fourier coefficient]
\label{cor:sampling-fourier-coefficient}
Let $\nu$ be any probability measure on the parameter space of one-layer cycle-QAOA.
Then on the weight-transfer sector $\Hom(\mathcal H_\mu,\mathcal H_\lambda)$, the averaged
central action is multiplication by the Fourier coefficient
\[
\widehat{\nu}(\lambda-\mu)
:=
\int
\chi_{\lambda-\mu}(\beta,\gamma)\,d\nu(\beta,\gamma).
\]
Hence a strict contraction on this sector requires a uniform bound
\[
\bigl|\widehat{\nu}(\lambda-\mu)\bigr|<1
\]
for every nonzero weight difference $\lambda-\mu$ that occurs.
\end{corollary}

\begin{proof}
This is immediate from Proposition \ref{prop:central-character-single-layer} by averaging the
scalar character action over $\nu$.
\end{proof}

\begin{remark}
Corollary \ref{cor:strict-overlap-fails} should be interpreted as an obstruction to a proof
strategy, not as a standalone claim that the full moment operator of cycle-QAOA must have
eigenvalue $1$. The semisimple factor
$\bigoplus_{k=1}^{n-1}\mathfrak g_k\cong\bigoplus_{k=1}^{n-1}\mathfrak{su}(2)$ may still
produce contraction on the same operator sector. What the argument proves is that no strict
spectral inequality statement can, in general, be deduced from non-invariance alone once a nontrivial central
torus is present as in this case certain choices of sampling measure can allow for an eigenvalue of 1. We shall next demonstrate a concrete instance of this being the case.
\end{remark}

\begin{proposition}[Kernel-curve counterexample for all moments]
\phantomsection
\label{prop:kernel-curve-counterexample}
Let $\mathcal E_L$ denote the depth-$L$ cycle-graph QAOA ensemble with
parameters $(\boldsymbol{\beta},\boldsymbol{\gamma})$ and let
$\mathcal{M}^{(t)}_{\mathcal E_L}$ be the corresponding $t$-th moment operator
acting on $\mathrm{End}(\mathcal H^{\otimes t})$ via
\[
\mathcal{M}^{(t)}_{\mathcal E_L}(X)
:=
\int_{\mathcal E_L}
U(\boldsymbol{\beta},\boldsymbol{\gamma})^{\otimes t}
\,X\,
(U(\boldsymbol{\beta},\boldsymbol{\gamma})^{\otimes t})^\dagger
\,d\nu_L(\boldsymbol{\beta},\boldsymbol{\gamma}),
\qquad
X\in\mathrm{End}(\mathcal H^{\otimes t}),
\]
where $U(\boldsymbol{\beta},\boldsymbol{\gamma})$ is the DLG representation of
the circuit unitary.

Assume there exists a nonzero DLA-visible weight difference
$\lambda-\mu$ and a corresponding visible subspace
$\mathcal H_{\lambda\mu}\subset \Hom(\mathcal H_\mu,\mathcal H_\lambda)$ such that the restriction of the DLG
representation to $\mathcal H_{\lambda\mu}$ is given by a central character
\[
U(\boldsymbol{\beta},\boldsymbol{\gamma})\big|_{\mathcal H_{\lambda\mu}}
=
\chi_{\lambda-\mu}(\boldsymbol{\beta},\boldsymbol{\gamma})\,
I_{\mathcal H_{\lambda\mu}},
\]
with $\chi_{\lambda-\mu}$ the nontrivial central character on the
$(\lambda,\mu)$-weight transfer sector (as defined in \Cref{appendix:cycle-graph-qaoa}). Assume further that there exists a nonconstant smooth curve on the space of all parameters $\Theta_L$
\[
c : \mathcal I \to \Theta_L,\qquad
t\mapsto (\boldsymbol{\beta}(t),\boldsymbol{\gamma}(t)),
\]
defined on a nontrivial open interval $\mathcal I\subset\mathbb{R}$, such that
\[
\chi_{\lambda-\mu} \big(c(t)\big)=1
\quad\text{for all }t\in \mathcal I.
\]

Let $q\in L^1(\mathcal I)$ with $q\ge 0$ and $\int_\mathcal{I} q(t)\,dt=1$, and define a probability
measure $\nu_L$ on $\Theta_L$, and therefore on $\mathcal E_L$, by the pushforward:
\[
\nu_L(E):=
\int_{c^{-1}(E)} q(t)\,dt,
\qquad
E\subset\Theta_L\text{ Borel}.
\]

Then:

\begin{enumerate}
\item $\nu_L$ is absolutely continuous with respect to one-dimensional Lebesgue 
measure on the parameter curve $c(\mathcal I)$.

\item For every integer $t\ge 1$, the $t$-th moment operator
$\mathcal{M}^{(t)}_{\mathcal E_L}$ acts as the identity on
$\mathrm{End}(\mathcal H_{\lambda\mu}^{\otimes t})$:
\[
\mathcal{M}^{(t)}_{\mathcal E_L}(A_{\lambda \mu})
=
A_{\lambda \mu}
\quad\text{for all }A_{\lambda \mu}\in\mathrm{End}(\mathcal H_{\lambda\mu}^{\otimes t}).
\]
Equivalently,
\[
\mathcal{M}^{(t)}_{\mathcal E_L}\big|_{\mathrm{End}(\mathcal H_{\lambda\mu}^{\otimes t})}
=
\mathrm{id}_{\mathrm{End}(\mathcal H_{\lambda\mu}^{\otimes t})}.
\]
In particular, for every $t\ge 1$ there is an eigenvalue $1$ in this
DLA-visible tensor sector, so there is no contraction there despite the
parameter distribution having continuous density along $\mathcal I$.
\end{enumerate}
\end{proposition}

\begin{proof}
(1). By construction, the parameter law on $\mathcal I$ is $q(t)\,dt$, which is absolutely 
continuous with respect to the Lebesgue measure $dt$ on $\mathcal I$. The pushforward 
$\nu_L=c_*(q(t)\,dt)$ is therefore absolutely continuous with respect to the 
one-dimensional Hausdorff (Lebesgue) measure on the image curve 
$c(\mathcal I)\subset \mathcal E_L$.

(2). Fix an integer $t\ge 1$ and an operator
$A_{\lambda \mu}\in\Hom(\mathcal H_\mu,\mathcal H_\lambda)$ supported on $\mathcal H_{\lambda\mu}^{\otimes t}$,
i.e.\ $X$ acts nontrivially only on that tensor factor. By the central-character
assumption, for each parameter pair $(\boldsymbol{\beta},\boldsymbol{\gamma})$
we have
\[
U(\boldsymbol{\beta},\boldsymbol{\gamma})\big|_{\mathcal H_{\lambda\mu}}
=
\chi_{\lambda-\mu}(\boldsymbol{\beta},\boldsymbol{\gamma})\,
I_{\mathcal H_{\lambda\mu}},
\]
hence on the $t$-fold tensor power
\[
U(\boldsymbol{\beta},\boldsymbol{\gamma})^{\otimes t}
\big|_{\mathcal H_{\lambda\mu}^{\otimes t}}
=
\chi_{\lambda-\mu}(\boldsymbol{\beta},\boldsymbol{\gamma})^t\,
I_{\mathcal H_{\lambda\mu}^{\otimes t}}.
\]

Along the curve $c(t)$, the kernel condition
$\chi_{\lambda-\mu}(c(s))\equiv 1$ for all $s\in \mathcal I$ implies
\[
U\big(c(s)\big)^{\otimes t}
\big|_{\mathcal H_{\lambda\mu}^{\otimes t}}
=
\chi_{\lambda-\mu}\big(c(s)\big)^t\,
I_{\mathcal H_{\lambda\mu}^{\otimes t}}
=
1^t\,I_{\mathcal H_{\lambda\mu}^{\otimes t}}
=
I_{\mathcal H_{\lambda\mu}^{\otimes t}}
\quad\text{for all }s\in \mathcal I.
\]
Thus for any such $A_{\lambda \mu}$ and any $s\in \mathcal I$,
\[
U\big(c(s)\big)^{\otimes t}\,A_{\lambda \mu}\,
(U\big(c(s)\big)^{\otimes t})^ \dagger
=
A_{\lambda \mu}.
\]

By definition of $\mathcal{M}^{(t)}_{\mathcal E_L}$ and of the pushforward
measure $\nu_L$, we have
\begin{equation*}
\begin{split}
\mathcal{M}^{(t)}_{\mathcal E_L}(A_{\lambda \mu})
&= \int_{\Theta_L}
U(\boldsymbol{\beta},\boldsymbol{\gamma})^{\otimes t}
\,A_{\lambda \mu}\,
\bigl(U(\boldsymbol{\beta},\boldsymbol{\gamma})^{\otimes t}\bigr)^\dagger
\,d\nu_L(\boldsymbol{\beta},\boldsymbol{\gamma}) \\
&= \int_\mathcal{I}
U\bigl(c(s)\bigr)^{\otimes t}
\,A_{\lambda \mu}\,
\bigl(U(c(s))^{\otimes t}\bigr)^\dagger
\,q(s)\,ds.
\end{split}
\end{equation*}
Since the integrand is identically $A_{\lambda \mu}$ for all $s\in \mathcal I$, this reduces to
\[
\mathcal{M}^{(t)}_{\mathcal E_L}(A_{\lambda \mu})
=
\int_\mathcal{I} A_{\lambda \mu}\,q(s)\,ds
=
\left(\int_\mathcal{I} q(s)\,ds\right) A_{\lambda \mu}
=
A_{\lambda \mu}.
\]
Therefore
\(
\mathcal{M}^{(t)}_{\mathcal E_L}\big|_{\mathrm{End}(\mathcal H_{\lambda\mu}^{\otimes t})}
=
\mathrm{id}
\),
proving (2).
\end{proof}

\begin{example}[Explicit cosine--sine kernel curve for all moments]
\label{ex:explicit-cos-sine-kernel}
In the setting of the preceding proposition, starting with the single-layer cycle graph QAOA on Max Cut, take any
$(\lambda,\mu)$-weight transfer sector where the central character has the
schematic form
\[
\chi_{\lambda-\mu}(\beta,\gamma)
=
\exp\!\Big(i\big(a\,\beta + b\,\gamma\big)\Big),
\]
for some nonzero integers $a,b$ determined by the DLA-visible weights
(cf.\ Appendix~A). Consider
the special case $a=b\neq 0$.

Set
\[
\alpha
:= \frac{2\pi k}{a}, \ \text{fix } k\in\mathbb{Z},
\]
and define a smooth curve $c:\mathcal I\to\Theta_L$ by
\[
c(t)
=
\big(\beta(t),\gamma(t)\big)
:=
\big(\alpha\cos^2 t,\;\alpha\sin^2 t\big),
\qquad t\in \mathcal I,
\]
where $\mathcal I\subset\mathbb{R}$ is any nontrivial open interval on which
$\cos t$ and $\sin t$ are not both constant.

For every $t\in \mathcal I$ we have
\begin{align*}
a\,\beta(t) + a\,\gamma(t)
&=
a\,\alpha\big(\cos^2 t + \sin^2 t\big) \\
&=
a\,\alpha \\
&=
a\,\frac{2\pi k}{a} \\
&=
2\pi k.
\end{align*}
Hence the phase in the central character is constant:
\[
\chi_{\lambda-\mu}\big(c(t)\big)
=
\exp\!\big(i(a\,\beta(t)+a\,\gamma(t))\big)
=
\exp(2\pi i k)
=
1
\quad\text{for all }t\in \mathcal I.
\]

Now fix any $q\in L^1(\mathcal I)$ with $q\ge 0$ and $\int_\mathcal{I} q(t)\,dt=1$ (e.g. $q(t)=1/\text{len}(\mathcal I)$ if $\mathcal I$ is compact), and define
$\nu_L$ on $\mathcal E_L$ by the pushforward
\[
\nu_L(E)
:=
\int_{c^{-1}(E)} q(t)\,dt,
\qquad
E\subset \Theta_L\ \text{Borel}.
\]
By the preceding proposition, $\nu_L$ is absolutely continuous with
respect to one-dimensional Lebesgue measure on $c(\mathcal I)$, and for every
integer $t\ge 1$ and every
$A_{\lambda\mu}\in\mathrm{End}(\mathcal H_{\lambda\mu}^{\otimes t})$ one has
\[
\mathcal{M}^{(t)}_{\mathcal E_L}(A_{\lambda\mu})
=
A_{\lambda\mu}.
\]
Thus this explicit cosine--sine curve gives a parameter law with
continuous density on a nontrivial interval for which all $t$-th moment
operators have eigenvalue $1$ on the DLA-visible tensor sector
$\mathcal H_{\lambda\mu}^{\otimes t}$. This construction, at some notational cost, can be straightforwardly extended to arbitrary layers.
\end{example}

\chapter{The Observable Quotient Formalism and Auxiliary Proofs}\label{appendix:observable-quotient-proofs}

We collect here some structural facts about the observable algebra
\(\mathcal A\) and the observable group \(G_{\mathrm{obs}}\) used in the main
text. All notation is as in the body of \Cref{subsec:observable-kawada-ito}.

\subsection{\(\mathcal A\) is already an algebra}

The first lemma justifies calling \(\mathcal A\) an “algebra”: it shows that
we do not need to close under products or conjugation by hand; those
operations are already present at the level of matrix coefficients.

\begin{lemma}[$\mathcal A$ is a unital $*$-algebra]
\phantomsection
\label{lem:algebra}
For all integers \(t,s \ge 0\),
\[
\mathcal F_t \cdot \mathcal F_s \subset \mathcal F_{t+s}
\quad\text{and}\quad
\overline{\mathcal F_t} = \mathcal F_t .
\]
Consequently, \(\mathcal A
= \operatorname{span}_{\mathbb C} \bigcup_{t=0}^\infty \mathcal F_t\)
is a unital $*$-subalgebra of \(C(G)\): it contains the constants and is
closed under products and complex conjugation.
\end{lemma}

\begin{proof}
Let \(f_1 \in \mathcal F_{t}\) and \(f_2 \in \mathcal F_{s}\), say
\[
f_1(g)
=
\big\langle V_1, \rho^{(t)}(g) W_1 \big\rangle,
\qquad
f_2(g)
=
\big\langle V_2, \rho^{(s)}(g) W_2 \big\rangle
\]
for some \(V_1,W_1 \in \mathcal H_t\) and \(V_2,W_2 \in \mathcal H_s\).

Then, for every \(g\in G\),
\begin{align*}
f_1(g)\, f_2(g)
&=
\big\langle V_1, \rho^{(t)}(g) W_1 \big\rangle\,
\big\langle V_2, \rho^{(s)}(g) W_2 \big\rangle
\\
&=
\big\langle V_1 \otimes V_2,\;
\big(\rho^{(t)} \otimes \rho^{(s)}\big)(g)\,
(W_1 \otimes W_2)
\big\rangle .
\end{align*}

Let
\[
\Sigma_{t,s}
:\;
\mathcal H_t \otimes \mathcal H_s
\longrightarrow
\mathcal H_{t+s}
\]
be the unitary that permutes tensor factors by first collecting the
\(t+s\) unbarred factors and then the \(t+s\) barred factors. By construction
of the balanced tensor powers, \(\Sigma_{t,s}\) intertwines
\(\rho^{(t)} \otimes \rho^{(s)}\) with \(\rho^{(t+s)}\): both act by
\(\rho(g)\) on each unbarred tensor factor and by \(\overline{\rho(g)}\) on
each barred factor. Therefore
\begin{align*}
f_1(g) f_2(g)
&=
\big\langle
V_1 \otimes V_2,\,
\big(\rho^{(t)} \otimes \rho^{(s)}\big)(g)\,
(W_1 \otimes W_2)
\big\rangle
\\
&=
\big\langle
\Sigma_{t,s}(V_1 \otimes V_2),\,
\rho^{(t+s)}(g)\,
\Sigma_{t,s}(W_1 \otimes W_2)
\big\rangle
\\
&\in
\mathcal F_{t+s}.
\end{align*}
This proves \(\mathcal F_t \cdot \mathcal F_s \subset \mathcal F_{t+s}\).

For complex conjugation, write
\[
f(g) = \langle V, \rho^{(t)}(g) W\rangle
\qquad (V,W\in\mathcal H_t).
\]
In the fixed product basis we have
\[
\overline{\rho^{(t)}}
=
\overline{\rho}^{\,\otimes t} \otimes \rho^{\otimes t}.
\]
Let
\[
\tau :
\overline{\mathcal H}^{\otimes t} \otimes \mathcal H^{\otimes t}
\longrightarrow
\mathcal H^{\otimes t} \otimes \overline{\mathcal H}^{\otimes t}
\]
be the tensor-swap unitary that exchanges the two blocks. By construction,
\(\tau\) intertwines \(\overline{\rho^{(t)}}\) with \(\rho^{(t)}\).

Thus, for all \(g\in G\),
\begin{align*}
\overline{f(g)}
&=
\big\langle \overline V,\, \overline{\rho^{(t)}(g)}\, \overline W \big\rangle
\\
&=
\big\langle \tau \overline V,\,
\rho^{(t)}(g)\, \tau \overline W
\big\rangle
\\
&\in \mathcal F_t .
\end{align*}
Hence \(\overline{\mathcal F_t} = \mathcal F_t\).

Finally, by definition \(\mathcal F_0\) consists of the constant functions
(matrix coefficients of the trivial representation), so \(1 \in \mathcal A\).

The span of \(\bigcup_{t=0}^\infty \mathcal F_t\) is closed under products by the
first part and bilinearity, and closed under conjugation by the second part,
so \(\mathcal A\) is a unital $*$-subalgebra of \(C(G)\).
\end{proof}

\subsection{Observable Stone--Weierstrass and Peter--Weyl}

The next lemma formalizes the idea that the observable algebra \(\mathcal A\)
is, after pushing forward to the observable quotient \(G_{\mathrm{obs}}\), as
large as one could reasonably hope: it becomes a dense $*$-subalgebra of
\(C(G_{\mathrm{obs}})\), and its balanced tensor powers see every irreducible
representation of \(G_{\mathrm{obs}}\).

\begin{lemma}[Observable Stone--Weierstrass and observable Peter--Weyl]
\label{lem:density}
For every \(f \in \mathcal A\) there is a unique
\(\widetilde f \in C(G_{\mathrm{obs}})\) with
\[
f = \widetilde f \circ \Phi,
\]
and the resulting map \(f \mapsto \widetilde f\) is a unital $*$-algebra
isomorphism of \(\mathcal A\) onto
\[
\widetilde{\mathcal A}
:=
\operatorname{span}_{\mathbb C}
\Big\{ x \mapsto \big\langle V, \widetilde\rho^{(t)}(x) W\big\rangle \, \Big \vert \,
t \ge 0,\; V,W \in \mathcal H_t \Big\}
\;\subset\; C(G_{\mathrm{obs}}) .
\]
Moreover:
\begin{enumerate}
\item
\(\widetilde{\mathcal A}\) is uniformly dense in \(C(G_{\mathrm{obs}})\).

\item
Every irreducible representation of \(G_{\mathrm{obs}}\) is unitarily
equivalent to a subrepresentation of \(\widetilde\rho^{(t)}\) for some
\(t \ge 1\). Equivalently, an irreducible representation \(\pi\) of \(G\)
occurs in some \(\rho^{(t)}\) if and only if \(\pi\) is trivial on
\(Z = \ker\Phi\).

\item
Every \(f \in \mathcal A\) is constant on cosets of \(Z\). In particular, if
\(Z \neq \{e\}\), then \(\mathcal A\) is not uniformly dense in \(C(G)\),
and its uniform closure is the proper subalgebra
\[
\{ f \in C(G) \, \vert \, f(gz) = f(g) \ \forall g \in G,\, z \in Z \}
\;\cong\;
C(G_{\mathrm{obs}}).
\]
\end{enumerate}
\end{lemma}

\begin{proof}
Every generator of \(\mathcal A\) has the form
\[
f(g)
=
\big\langle V, \rho^{(t)}(g) W \big\rangle
\qquad (t\ge0,\, V,W\in\mathcal H_t).
\]
By Lemma~\ref{lem:structure}(2) the balanced representation \(\rho^{(t)}\)
factors through \(\Phi\), that is,
\[
\rho^{(t)}
=
\widetilde\rho^{(t)} \circ \Phi
\]
for a unique \(\widetilde\rho^{(t)} : G_{\mathrm{obs}} \to U(\mathcal H_t)\).
Thus
\[
f(g)
=
\big\langle V,\,
\widetilde\rho^{(t)}(\Phi(g)) W
\big\rangle
=
\widetilde f(\Phi(g)),
\]
where we set
\[
\widetilde f(x)
:=
\big\langle V, \widetilde\rho^{(t)}(x) W\big\rangle,
\qquad x\in G_{\mathrm{obs}}.
\]

This constructs, for each generator \(f\in\mathcal A\), a corresponding
\(\widetilde f \in \widetilde{\mathcal A}\) with \(f=\widetilde f\circ\Phi\).
The map extends by linearity to all of \(\mathcal A\). Surjectivity of
\(\Phi : G \to G_{\mathrm{obs}}\) implies that \(\widetilde f\) is uniquely
determined by \(f\), and that the assignment
\[
\mathcal A \longrightarrow \widetilde{\mathcal A} \quad\text{by}
\quad f \longmapsto \widetilde f
\]
is an injective unital $*$-algebra homomorphism whose image is exactly
\(\widetilde{\mathcal A}\). It is also an isometry for the uniform norms,
since
\[
\|f\|_\infty
=
\sup_{g\in G} |f(g)|
=
\sup_{x\in G_{\mathrm{obs}}} |\widetilde f(x)|
=
\|\widetilde f\|_\infty .
\]

(1). By Lemma~\ref{lem:algebra}, \(\mathcal A\) is a unital $*$-algebra, so its
image \(\widetilde{\mathcal A}\) is a unital $*$-subalgebra of
\(C(G_{\mathrm{obs}})\). $\widetilde {\mathcal A}$ separates points: if \(x_1,x_2\in G_{\mathrm{obs}}\) with
\(x_1\neq x_2\), then some matrix entry of the unitary
\[
x \in G_{\mathrm{obs}} \subset
U(\mathcal H \otimes \overline{\mathcal H})
\]
takes different values on \(x_1\) and \(x_2\). But those matrix entries are
precisely the coefficients of \(\widetilde\rho^{(1)}\), which lie in
\(\widetilde{\mathcal A}\). Since \(G_{\mathrm{obs}}\) is a compact Hausdorff space, the
Stone–Weierstrass theorem \cite[Thm.~4.45]{Folland1995Harmonic} implies that
\(\widetilde{\mathcal A}\) is uniformly dense in \(C(G_{\mathrm{obs}})\).

(2). Assume towards contradiction that there exists an irreducible representation
\(\sigma : G_{\mathrm{obs}}\to U(W)\) that is not equivalent to a
subrepresentation of any \(\widetilde\rho^{(t)}\). Decompose each \(\widetilde\rho^{(t)}\) into irreducibles, and use the Schur
orthogonality relations \cite[Thm.~5.8]{Folland1995Harmonic} on
\(G_{\mathrm{obs}}\). Every matrix coefficient of \(\sigma\) is orthogonal
in \(L^2(G_{\mathrm{obs}}, \mu_{G_{\mathrm{obs}}})\) to every matrix
coefficient of every subrepresentation of \(\widetilde\rho^{(t)}\), hence to
every element of \(\widetilde{\mathcal A}\).

By part (1), \(\widetilde{\mathcal A}\) is uniformly dense in
\(C(G_{\mathrm{obs}})\), and thus dense in
\(L^2(G_{\mathrm{obs}},\mu_{G_{\mathrm{obs}}})\). Therefore every matrix
coefficient of \(\sigma\) must vanish identically, which contradicts
\(\langle v,\sigma(e)v\rangle = \|v\|^2 \ne 0\) for any nonzero \(v\in W\).
Thus every irreducible representation of \(G_{\mathrm{obs}}\) appears in
some \(\widetilde\rho^{(t)}\).

For the equivalent restatement on \(G\), note first that every \(\rho^{(t)}\)
is trivial on \(Z=\ker\Phi\) by Lemma~\ref{lem:structure}(1), so any
irreducible \(\pi\) occurring in some \(\rho^{(t)}\) must be \(Z\)-trivial.

Conversely, let \(\pi\) be an irreducible representation of \(G\) that is
trivial on \(Z\). Then \(\pi\) descends to an irreducible representation of
\(G/Z \cong G_{\mathrm{obs}}\). By the first part, this descendant occurs in
some \(\widetilde\rho^{(t)}\), and hence \(\pi\) occurs in the pullback
\[
\rho^{(t)}
=
\widetilde\rho^{(t)} \circ \Phi.
\]
This is the balanced analogue of the classical fact that tensor powers
\(\rho^{\otimes m} \otimes \overline\rho^{\otimes n}\) of a faithful
representation \(\rho\) contain every irreducible representation of the
underlying compact group; see, for instance, \cite[Ch.~3]{BtD1985}.

(3). For any \(f\in\mathcal A\), we have \(f = \widetilde f \circ \Phi\). If
\(z\in Z\), then \(\Phi(gz)=\Phi(g)\) for all \(g\in G\), so
\[
f(gz)
=
\widetilde f(\Phi(gz))
=
\widetilde f(\Phi(g))
=
f(g).
\]
Thus \(f\) is constant on every coset \(gZ\).

If \(Z\neq\{e\}\), choose \(z\in Z\setminus\{e\}\). By Urysohn’s lemma there
exists \(h\in C(G)\) with
\[
h(e)=0,
\qquad
h(z)\neq 0.
\]
Then for every \(f\in\mathcal A\),
\[
\|f-h\|_\infty
\;\ge\;
\frac{1}{2}\,|h(z)-h(e)|
\;>\;
0,
\]
since \(f(e)=f(z)\) by coset-constancy and \(h(e)\ne h(z)\). Thus
\(\mathcal A\) is not dense in \(C(G)\) when \(Z\neq\{e\}\).

Finally, the uniform closure of \(\mathcal A\) consists exactly of those
continuous functions on \(G\) that are constant on cosets of \(Z\):
\[
\overline{\mathcal A}
=
\{f\in C(G) \, \vert \,f(gz)=f(g)\ \forall g\in G, z\in Z\}.
\]
The quotient map \(G\to G/Z\cong G_{\mathrm{obs}}\) identifies this algebra
with \(C(G_{\mathrm{obs}})\), in agreement with part (1).
\end{proof}

\subsection{Unit-modulus eigenvalues of Fourier coefficients}

The next lemma is the key spectral mechanism used in the observable
Kawada–Itô theorem. It gives an “equality case” characterization: an
eigenvalue of modulus one of a Fourier coefficient \(\widehat\nu(\sigma)\)
forces the measure to be supported inside a stabilizer coset, and under
adaptedness and strict aperiodicity that can only occur in a
one-dimensional character sector, and only with eigenvalue \(1\).

\begin{lemma}[Unit-modulus eigenvalues of Fourier coefficients]
\phantomsection
\label{lem:spectral}
Let \(K\) be a compact group, \(\nu\) a probability measure on \(K\), and
\(\sigma : K \to U(V)\) a continuous unitary representation on a
finite-dimensional Hilbert space \(V\). Write
\[
\widehat\nu(\sigma)
:=
\int_K \sigma(x)\, d\nu(x).
\]
Then:
\begin{enumerate}
\item
\(\|\widehat\nu(\sigma)\|_\infty \le 1\).

\item
If \(\widehat\nu(\sigma)v = \lambda v\) for some \(0 \neq v \in V\) and
\(\lambda \in \mathbb C\) with \(|\lambda| = 1\), then
\[
\sigma(x)v
=
\lambda v
\qquad
\text{for every }x \in \operatorname{supp}\nu.
\]

\item
If, in addition, \(\nu\) is adapted (i.e. \([\operatorname{supp}\nu]=K\)),
then \(\mathbb C v\) is \(\sigma(K)\)-invariant, and the function
\(\chi : K \to \mathbb T\) defined by
\[
\sigma(x)v
=
\chi(x)v
\]
is a continuous character of \(K\) satisfying \(\chi \equiv \lambda\) on
\(\operatorname{supp}\nu\).

\item
If, in addition, \(\nu\) is strictly aperiodic, then \(\chi \equiv 1\); in
particular \(\lambda = 1\) and \(v\) is \(\sigma(K)\)-invariant.
\end{enumerate}
\end{lemma}

\begin{proof}
(1). For any unit vector \(u\in V\),
\begin{align*}
\|\widehat\nu(\sigma)u\|
&=
\Big\| \int_K \sigma(x)u\, d\nu(x) \Big\|
\\
&\le
\int_K \|\sigma(x)u\|\, d\nu(x)
\\
&=
\int_K \|u\|\, d\nu(x)
=
1,
\end{align*}
using unitarity of \(\sigma(x)\) and Jensen’s inequality. Taking the
supremum over all unit vectors \(u\) gives
\(\|\widehat\nu(\sigma)\|_\infty \le 1\).

(2). Assume \(\widehat\nu(\sigma)v = \lambda v\) with \(|\lambda|=1\) and \(v\neq 0\). Define the continuous real-valued function
\[
\varphi(x)
:=
\operatorname{Re}\big\langle v,\,
\bar\lambda\,\sigma(x) v\big\rangle,
\qquad
x\in K.
\]
By Cauchy–Schwarz and unitarity,
\[
\varphi(x)
\le
\big| \langle v, \bar\lambda \sigma(x)v\rangle \big|
\le
\|v\|^2
\qquad
\text{for all }x\in K.
\]
On the other hand,
\begin{align*}
\int_K \varphi(x)\, d\nu(x)
&=
\operatorname{Re}
\Big\langle v,\,
\bar\lambda \int_K \sigma(x) v\, d\nu(x)
\Big\rangle
\\
&=
\operatorname{Re}
\big\langle v,\, \bar\lambda\, \widehat\nu(\sigma) v\big\rangle
\\
&=
\operatorname{Re}
\big\langle v,\, \bar\lambda \lambda v\big\rangle
\\
&=
\|v\|^2.
\end{align*}
Thus
\[
\int_K (\|v\|^2 - \varphi(x))\, d\nu(x) = 0
\]
with a continuous nonnegative integrand. Therefore
\(\varphi(x) = \|v\|^2\) for all \(x\in \operatorname{supp}\nu\): if
\(\varphi(x_0)<\|v\|^2\) at some \(x_0\in\operatorname{supp}\nu\), then
\(\varphi\le\|v\|^2-\delta\) on an open neighborhood of \(x_0\) of positive
\(\nu\)-measure, contradicting the integral identity.

Now fix \(x\in\operatorname{supp}\nu\), and set
\[
u
:=
\bar\lambda\,\sigma(x)v.
\]
Then \(\|u\|=\|v\|\) by unitarity, and
\[
\operatorname{Re}\langle v,u\rangle
=
\varphi(x)
=
\|v\|^2
=
\|v\|\,\|u\|.
\]
By Cauchy–Schwarz,
\[
|\langle v,u\rangle|
\le
\|v\|\,\|u\|,
\]
and equality together with
\(\operatorname{Re}\langle v,u\rangle = \|v\|\,\|u\|\) forces
\(\langle v,u\rangle = \|v\|\,\|u\|\), so \(u\) is a positive scalar
multiple of \(v\). Since \(\|u\|=\|v\|\), we have \(u=v\).

Thus $\bar\lambda\,\sigma(x)v = v$, i.e. 
\[
\sigma(x)v = \lambda v
\]
for every \(x\in\operatorname{supp}\nu\).

(3). Define the subset
\[
K_v
:=
\{x\in K \, \vert \, \sigma(x)v \in \mathbb C v\}.
\]
This is a subgroup: if \(\sigma(x)v = \alpha v\) and
\(\sigma(y)v=\beta v\) with \(\alpha,\beta\in\mathbb T\), then
\[
\sigma(xy)v
=
\sigma(x)(\sigma(y)v)
=
\alpha\beta v,
\qquad
\sigma(x^{-1})v
=
\bar\alpha v.
\]
It is also closed, since
\[
x\longmapsto \|(I - P_{\mathbb C v})\sigma(x)v\|
\]
is continuous and \(K_v\) is its zero-set.

By part (2),
\(\sigma(x)v=\lambda v\) for all \(x\in\operatorname{supp}\nu\), so
\(\operatorname{supp}\nu\subset K_v\). Adaptedness gives
\[
K
=
[\operatorname{supp}\nu]
\subset
K_v,
\]
and hence \(K_v=K\). Thus every \(\sigma(x)\) preserves the line
\(\mathbb C v\), i.e. \(\mathbb C v\) is \(\sigma(K)\)-invariant.

Define \(\chi : K \to \mathbb T\) by
\[
\chi(x)
:=
\frac{\langle v, \sigma(x)v\rangle}{\|v\|^2}.
\]
Then \(\sigma(x)v=\chi(x)v\), \(|\chi(x)|=1\) by unitarity,
\(\chi(xy)=\chi(x)\chi(y)\), and \(\chi\) is continuous. By part (2),
\(\chi(x)=\lambda\) for all \(x\in\operatorname{supp}\nu\), so
\(\chi\equiv\lambda\) on \(\operatorname{supp}\nu\).

(4). Assume \(\nu\) is strictly aperiodic. Suppose \(\chi\) from part (3) were
nontrivial. Then
\[
N := \ker\chi
\]
is a proper closed normal subgroup of \(K\): normality follows because for
any \(k\in K\), the conjugate \(kNk^{-1}\) also lies in \(\ker\chi\).

The support of a probability measure on a compact space is nonempty, so
choose \(x_0\in\operatorname{supp}\nu\). For any \(x\in\operatorname{supp}\nu\),
\[
\chi(x_0^{-1}x)
=
\overline{\chi(x_0)}\,\chi(x)
=
\bar\lambda \lambda
=
1,
\]
so \(x_0^{-1} x \in N\), and hence \(x\in x_0N\). Thus
\[
\operatorname{supp}\nu \subset x_0N,
\]
contradicting strict aperiodicity.

Therefore \(\chi\) must be trivial, i.e. \(\chi(x)=1\) for all \(x\in K\).
In particular \(\lambda = \chi(x_0)=1\), and
\[
\sigma(x)v = v
\quad
\text{for all }x\in K.
\]
That is, \(v\) is a \(\sigma(K)\)-invariant vector.
\end{proof}

\subsection{Proof of i.i.d. QAOA convergence to Haar}\label{appendix:qaoaHaar}

Finally, we shall prove the claim of \Cref{rem:qaoaHaar}.
\begin{proof}
Let \(\rho\) be the probability law of \((\beta,\gamma)\) on the periodic
QAOA parameter torus \(T\), and let
\[
\mathcal U:T\to G,\qquad \mathcal U(\beta,\gamma)=e^{i\beta H_M}e^{i\gamma H_C},
\]
so that \(\mu_1 := \mathcal U_*\rho\) is the induced one-layer measure on \(G\).
Since \(\rho\) has full support on \(T\), we have \(\operatorname{supp}\rho=T\).
Using the definition of the pushforward measure,
\[
\mu_1(B)=\rho\big(\mathcal U^{-1}(B)\big)
\quad\text{for Borel }B\subset G,
\]
it follows that \(\mathcal U(T)\subset \operatorname{supp}\mu_1\): indeed, if
\(x\in T\) and \(O\subset G\) is open with \(\mathcal U(x)\in O\), then
\(\mathcal U^{-1}(O)\) is an open neighborhood of \(x\), hence
\(\rho(\mathcal U^{-1}(O))>0\), so \(\mu_1(O)>0\).

In particular, for every allowed \(\beta,\gamma\),
\[
e^{i\beta H_M}=\mathcal U(\beta,0)\in \operatorname{supp}\mu_1,\qquad
e^{i\gamma H_C}=\mathcal U(0,\gamma)\in \operatorname{supp}\mu_1,
\]
and also
\[
e=\mathcal U(0,0)\in \operatorname{supp}\mu_1.
\]
Therefore the closed subgroup generated by \(\operatorname{supp}\mu_1\)
contains both one-parameter subgroups
\[
\{e^{i\beta H_M}\}\quad\text{and}\quad \{e^{i\gamma H_C}\},
\]
and hence contains the closed subgroup generated by them. By definition,
that closed subgroup is exactly \(G\). Thus
\[
[\operatorname{supp}\mu_1]=G,
\]
so \(\mu_1\) is adapted.

Now suppose toward a contradiction, that \(\mu_1\) is not strictly
aperiodic. Then there exists a proper closed normal subgroup
\(N\triangleleft G\) and an element \(g\in G\) such that
\[
\operatorname{supp}\mu_1\subset gN.
\]
Since \(e\in\operatorname{supp}\mu_1\), we have \(e\in gN\), so
\(g^{-1}\in N\). Hence \(gN=N\) \cite[§3.1]{DummitFoote2004}, and therefore
\[
\operatorname{supp}\mu_1\subset N.
\]
But then
\[
[\operatorname{supp}\mu_1]\subset N\neq G,
\]
contradicting adaptedness. Therefore no such \(g\) and \(N\) exist, and
\(\mu_1\) is strictly aperiodic.
\end{proof}
\end{document}